\documentclass[10pt]{article}
\usepackage[margin=1.2in]{geometry}

\usepackage{xr,xr-hyper}

\usepackage{amsmath, amsthm, amssymb, dsfont, mathtools, mathrsfs, nccmath}
\usepackage{bbm}
\usepackage{graphicx}
\usepackage{verbatim}
\usepackage[round]{natbib}
\usepackage{caption}
\usepackage{makecell}
\usepackage{algorithm}        
\usepackage{algpseudocode}

\usepackage{subcaption}
\usepackage{fancyvrb}
\usepackage{enumerate}
\usepackage[inline]{enumitem}
\usepackage{relsize}
\usepackage{hyperref}

\usepackage{xr}

\makeatletter
\newcommand{\purelink}[1]{{\let\Hy@colorlink\@gobble\let\Hy@endcolorlink\@empty\ref{#1}}}
\makeatother

\makeatletter
\newcommand{\bluelink}[1]{%
    {\color{blue}\ref{#1}}%
}
\makeatother

\hypersetup{colorlinks,citecolor=blue,urlcolor=blue,linkcolor=blue}
\usepackage{diagbox}
\usepackage{nikos_tex}
\usepackage{float}
\usepackage{centernot}
\usepackage{array}  
\usepackage[skip=0.333\baselineskip]{caption} 
\usepackage{booktabs, tabu}
\usepackage{blindtext}

\allowdisplaybreaks

\DeclareMathOperator*{\argmax}{arg\,max}

\usepackage[table]{xcolor}
\usepackage{siunitx}

\definecolor{ggplotRed}{HTML}{F08080}
\definecolor{ggplotBlue2}{HTML}{00688B}

\definecolor{ggplotGreen}{HTML}{458B74} 
\definecolor{ggplotBlue}{HTML}{A4D3EE} 

\newcommand{\dshead}[2]{%
  \rowcolor{black!7}%
  \multicolumn{7}{@{}l@{}}{#1}\\[-0.3em]
  \rowcolor{black!7}%
  \multicolumn{7}{@{}l@{}}{{\footnotesize #2}}\\[0.2em]
}
\newcommand{\simhead}[1]{%
  \rowcolor{black!7}%
  \multicolumn{10}{@{}l@{}}{\textbf{#1}}\\[0.2em]
}
\newcommand{\redS}[1]{%
  \multicolumn{1}{S[table-format=3.1, table-space-text-post=\,]<{\,\%}}{\color{red}#1}%
}

\newcommand{\ct}[1]{\hspace{1em}#1}
\newcommand{\limma}{\texttt{limma}}
\newcommand{\limmatrd}{\texttt{limma-trend}}
\newcommand{\ttest}{\texttt{t}}
\newcommand{\At}{\mathtt A}
\newcommand{\Bt}{\mathtt B}
\newcommand{\rmD}{\mathrm{D}}
\newcommand{\rmG}{G'}
\newcommand{\rmH}{H'}
\newcommand{\calD}{\mathcal D}
\newcommand{\mis}{\mathrm{mis}}

\newcommand{\misV}[1]{V_{#1,\mis}^2}
\newcommand{\misT}[1]{\tau_{#1,\mis}^2}
\newcommand{\misG}{G_\mis}

\newcommand{\misP}[1]{P^\trd_{#1,\mis}}

\usepackage{mathtools}
\mathtoolsset{showonlyrefs,showmanualtags}

\usepackage{multirow}

\usepackage[utf8]{inputenc}
\usepackage[english]{babel}

\graphicspath{{./figures/}}

\theoremstyle{definition}
\newtheorem{proposition}{Proposition}
\newtheorem{assumption}[proposition]{Assumption}
\newtheorem{example}[proposition]{Example}
\newtheorem{definition}[proposition]{Definition}

\newtheorem{lemm}[proposition]{Lemma}
\newtheorem{theorem}[proposition]{Theorem}
\newtheorem{rem}[proposition]{Remark}

\newtheoremstyle{remark}
  {10pt}
  {10pt}
  {\color{gray!150}\normalfont}  
  {}
  {\color{gray!150}\bfseries}    
  {.}
  {.5em}
  {}
  
\theoremstyle{remark}

\theoremstyle{remark}

\date{Draft Manuscript: May 2026}

\title{How does \texttt{limma-trend} work?\\An empirical partially Bayes perspective}

\author{
  Sagnik Nandy\thanks{These authors contributed equally to this work.} \\
  \small Ohio State University \\
  \small \texttt{nandy.15@osu.edu}
  \and
  Wanyi Ling\footnotemark[1] \\
  \small University of Chicago \\
  \small \texttt{wanyiling@uchicago.edu}
  \and
  Nikolaos Ignatiadis \\
  \small University of Chicago \\
  \small \texttt{ignat@uchicago.edu}
}

\begin{document}

\maketitle

\begin{abstract}
In high-throughput biology, it is common to fit thousands of linear regressions---one per gene, protein, or other unit---with very few samples per unit. \texttt{Limma-trend}, one of the most widely used methods in this setting, improves power by shrinking variance estimates parametrically toward a fitted curve (the trend) relating variance to a unit-level summary (e.g., average intensity, peptide count), before computing p-values and applying the Benjamini-Hochberg procedure to control the false discovery rate (FDR). We study \texttt{limma-trend} through the lens of empirical partially Bayes inference, a paradigm in which a prior is posited and estimated for the nuisance parameters while parameters of interest remain fixed. From this perspective, \texttt{limma-trend} computes approximate partially Bayes p-values that condition on the residual sample variance and the unit-level summary. The same framework explains why MAnorm2, a popular variant for ChIP-seq, can sometimes fail to control FDR.
We then derive a nonparametric generalization of \texttt{limma-trend} that estimates the residual variance prior using nonparametric maximum likelihood. Under dense signals, this procedure asymptotically controls the FDR---even when the trend is misspecified or inconsistently estimated. To allow the full shape of the conditional variance distribution to depend on the unit-level summary, we develop a second procedure that learns it directly.
\end{abstract}

\section{Introduction}

To this day, one of the most common data analysis tasks in high-throughput biology consists of fitting separate linear models for each of $n$ units (e.g., genes, proteins) and performing statistical analysis using the outputs of these fits. For the $i$-th unit, the adopted statistical framework is standard: $K$ samples, $p$ covariates with $p<K$, homoscedastic Gaussian errors with the responses $Y_{ij}$ modeled as
\begin{equation}
\label{eq:limma_regression}
Y_{ij} = x_j^\top  \beta_i + \varepsilon_{ij},\;\; \varepsilon_{ij} \simindep \mathrm{N}(0,\, \sigma_i^2),\;\;j=1,\dotsc,K,\;\;i=1,\dotsc,n.
\end{equation}
The coefficients $\beta_i \in \RR^{p}$ and noise variances $\sigma_i^2$ are unknown, while the covariates $x_j \in \RR^p$ include, e.g., the intercept, treatment status, as well as confounding (batch effect) variables to adjust for. The interest of the practitioners typically centers on a single contrast $\primary_i = \beta_i^\top  \contrastprimary$ for pre-specified $\contrastprimary \in \RR^p$ (e.g., the treatment effect). The statistical task is to test the hypotheses $H_i: \primary_i=0$ for $i=1,\dots,n$. The number of units $n$ is in the thousands, and the degrees of freedom $K-p$ is often minuscule, sometimes in the single digits. 

The textbook solution to this problem is immediate: for each unit $i$, compute the ordinary least squares estimate $\wh{\primary}_i$ and apply a $t$-test, followed by multiple testing correction. The relevant question is: \emph{Does this allow us to make a significant amount of discoveries while controlling the false discoveries?} In genomics, small sample sizes (small $K$) severely limit the power of standard approaches; after multiple testing corrections, naïve $t$-tests often yield few or no discoveries. The \limma{} framework \citep{smyth2004linear} addresses this by shrinking variance estimates toward a common value before conducting inference, effectively borrowing strength across units. The \limmatrd{} variant~\citep{law2014voom}, building on ideas of \citet{sartor2006intensitybased}, allows the shrinkage target to depend on unit-level summaries such as average intensity (defined below in \eqref{eq:avg_intensity}). This refinement has proven important: \limmatrd{} underlies methods for RNA-seq ~\citep{ritchie2015limma}, proteomics~\citep{zhu2020deqms, messner2023proteomic}, ChIP-seq~\citep{tu2021manorm2}, methylation variability~\citep{phipson2014diffvar}, cytometry~\citep{weber2019diffcyt}, lipidomics~\citep{townsend2025establishing}, and other modalities. 

Despite its widespread adoption, \limmatrd{} has received relatively little formal statistical attention. \citet{ignatiadis2025empirical} recently provided theoretical foundations for the original (untrended) \limma{} procedure. In the trend setting, the unit-level summary toward which variances are shrunk is typically the average intensity, and the average intensity is itself computed from the same outcomes used for inference. This opens up three issues that have no analog in the untrended case: (i)~the validity of the partially Bayes interpretation requires an orthogonality condition on the design that is not automatic, and is violated in the popular ChIP-seq package MAnorm2, (we show in simulations that this leads to an inflation in the FDR); (ii)~the prior on variances is no longer a single shared distribution across units, but must vary with the summary; and (iii)~the analysis must accommodate the plug-in of an estimated trend into the marginal likelihood used to learn that prior. In this paper, we develop a partially Bayes \citep{cox1975note} account of \limmatrd{} that handles all three issues, and we propose two nonparametric procedures motivated by this account.

\subsection{\texttt{Limma-trend} and our contributions}
\label{sec:basic_limma}

\begin{figure}
\centering
\begin{tabular}{@{}lll@{}}
(a) RNA-Seq & (b) ChIP-Seq & (c) Proteomics \\ 
\includegraphics[width=0.32\linewidth]{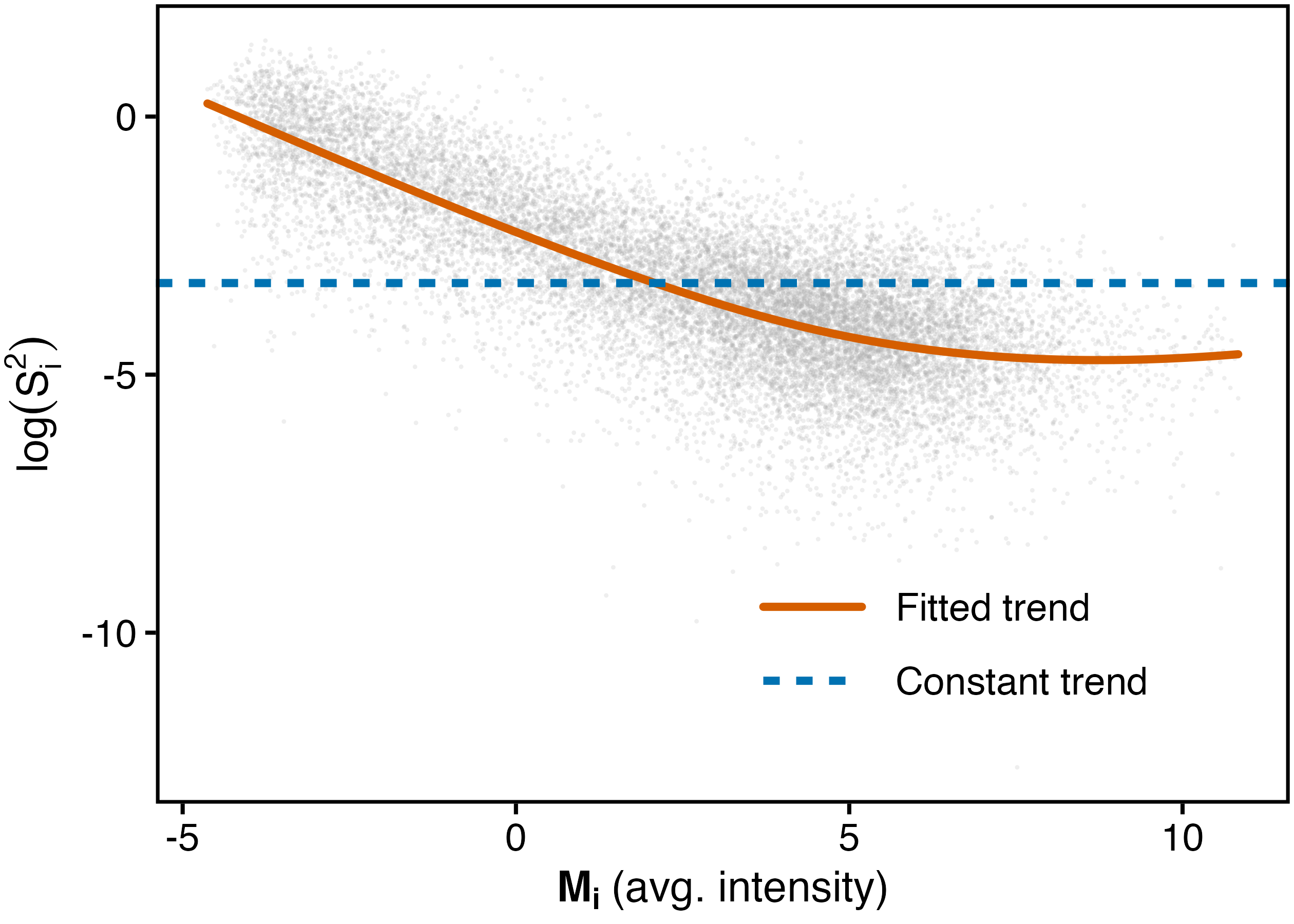} & \includegraphics[width=0.304\linewidth]{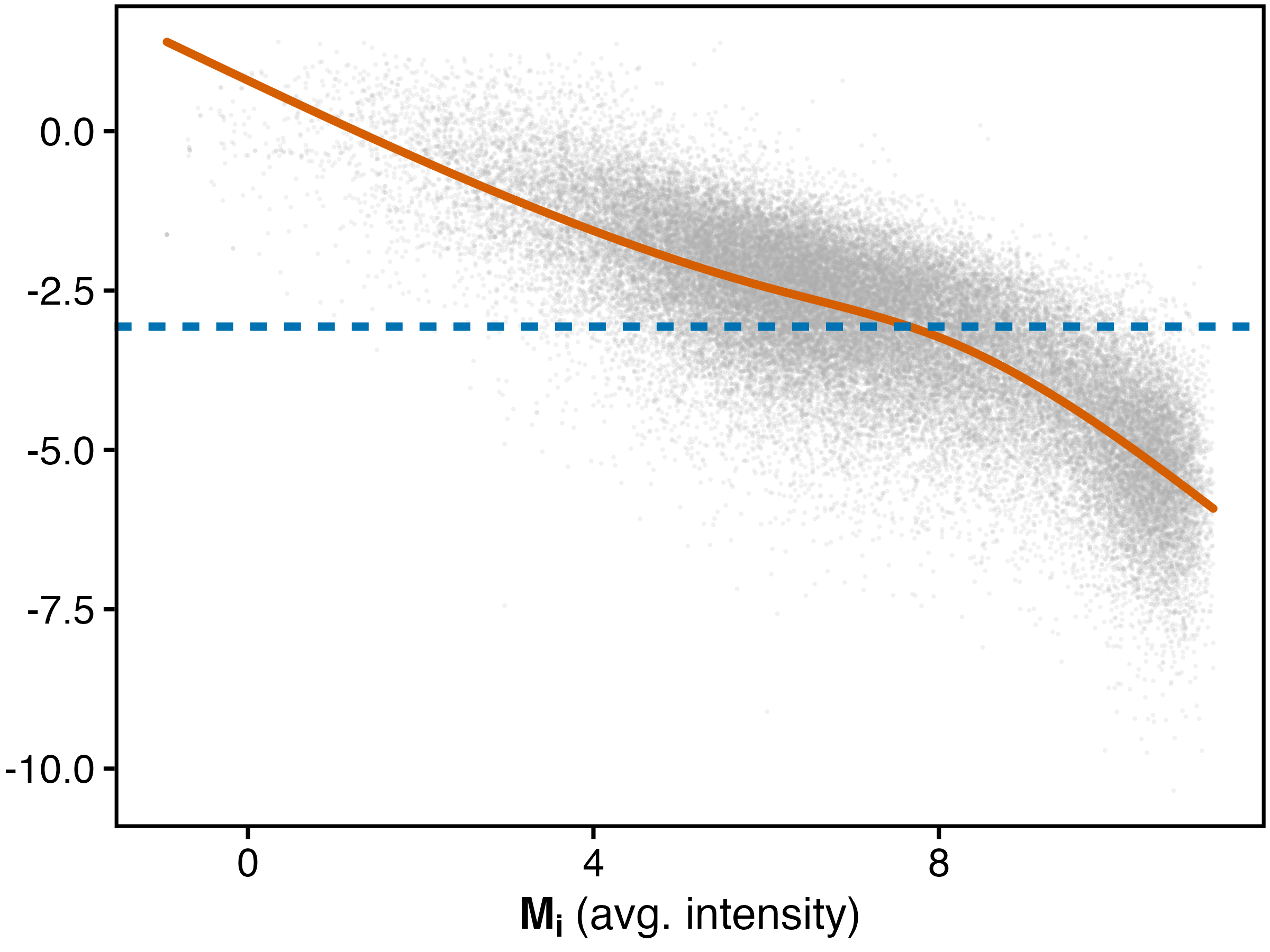} & \includegraphics[width=0.304\linewidth]{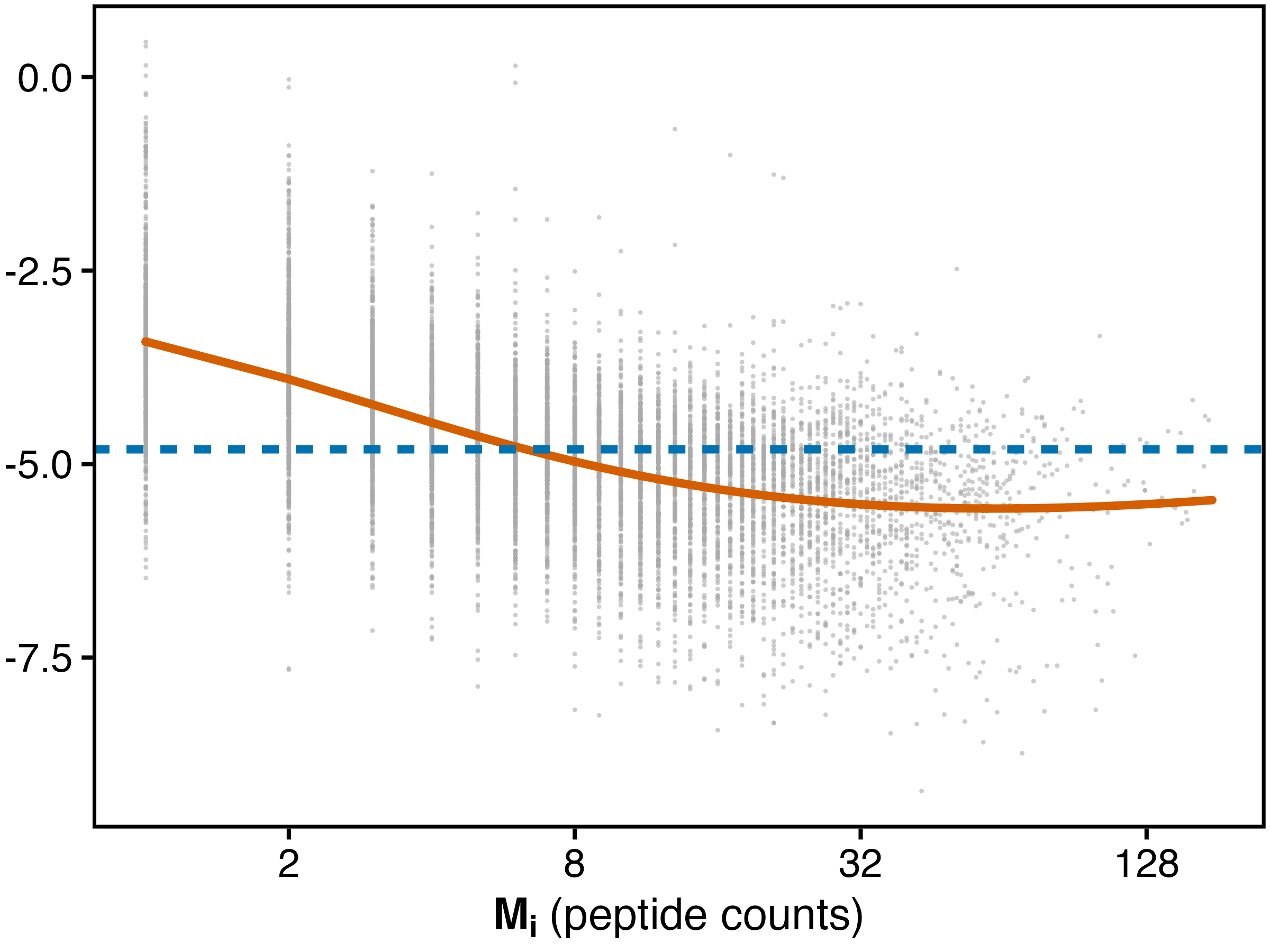} 
\end{tabular}
\caption{Mean-variance trend observed in datasets from three different biological modalities. Grey points represent the pair, unit-wise side-information (average intensity for (a) and (b), $\log_2$ peptide counts for (c)), and log sample variance. Their relation is represented by two lines: a fitted natural cubic spline curve, and a constant trend.  \textbf{(a)} RNA-Seq: each unit $i$ represents a gene in
the CD4$^{+}$ T cells melanoma dataset from \cite{goswami2018modulation}. \textbf{(b)} ChIP-Seq: each unit $i$ represents a genomic interval in the H3K4me3 lymphoblastoid cell dataset from \cite{tu2021manorm2}. \textbf{(c)} Proteomics:  each unit $i$ represents a protein in the miRNA-mimic treatment dataset on A431 cells from \cite{zhu2020deqms}.
}
\label{fig:intro_trend}
\end{figure}

In \limmatrd{}, the practitioner begins by computing 
the ordinary least squares estimates of $\beta_i$ in \eqref{eq:limma_regression},  the primary contrast of interest $\theta_i$ and the residual variance $\sigma_i^2$ (for all units),
\begin{equation}
\wh{\beta}_i := (X^\top X)^{-1}X^\top Y_i,\;\;\;\;  Z_i := \wh{\beta}_i^\top  \contrastprimary,\;\;\;\; S^2_i := \frac{1}{K-p}\sum_{j=1}^K(Y_{ij}- x_j^\top \wh{\beta}_i)^2, \qquad \mbox{for $i=1,\ldots,n$,}
\label{eq:limma_summary_stats}
\end{equation}
where $X \in \RR^{K \times p}$ denotes the design matrix (assumed to be full rank) with $j$-th row equal to \smash{$x_j^\top$} and \smash{$Y_i:=(Y_{i1},\ldots, Y_{iK})^\top \in \RR^K$}.
The textbook approach to compute a p-value for testing $\primary_i=0$ uses the fact that $T_i := Z_i/(\nu S_i) \sim t_{K-p}$ (the t-distribution with $K-p$ degrees of freedom) under the null, where  \smash{$\nu := \{\contrastprimary^\top (X^\top X)^{-1}\contrastprimary\}^{1/2}$}, giving $P_i^{\ttest} := 2\,\wb{F}_{t_{K-p}}(|T_i|)$ (the survival function of $t_{K-p}$).

These $t$-test p-values often lack power because $S_i^2$ is a noisy estimate of $\sigma_i^2$, as reflected by the heavy tails of the $t$-distribution at small $K-p$. The key idea underlying \limmatrd{} is that, in common high-throughput biology applications, there is contextual information $M_i$, distinct from $S_i^2$, that is highly predictive of $\sigma_i^2$. Figure~\ref{fig:intro_trend} plots estimates of the conditional expectation $\EE{\log(S_i^2) \mid M_i}$ across three datasets from different modalities (with different choices of $M_i$); we further revisit these datasets in Section~\ref{sec:case_studies}. We refer to the function $\EE{\log(S_i^2) \mid M_i = \cdot}$ and estimates thereof as \emph{trends}. The trends in Figure~\ref{fig:intro_trend} are far from constant.

How should one use the trend? Suppose for a moment that the \smash{$\sigma_i^2$} are a deterministic function of $M_i$, $\sigma_i^2 = h^2(M_i)$, and that we have a consistent estimate \smash{$\widehat{h}^2(M_i)$} thereof. Then we could report the p-values \smash{$2\Phi(-|Z_i|/\{\nu \widehat{h}(M_i)\})$}, where $\Phi$ is the standard Gaussian CDF, sidestepping the heavy tails of the $t$-distribution. This idea is not new, e.g., it appears in the MAP (model-based analysis of proteomic data) procedure of \citet{li2019map}. In practice, however, $\sigma_i^2$ is \emph{not} determined by $M_i$: there is residual heterogeneity that no trend can absorb, and methods that ignore this heterogeneity can inflate false discoveries, as we confirm for MAP in simulations in Section~\ref{sec:numerical_exp}.

This brings us to \limmatrd{}, which is summarized in Algorithm~\ref{algo:limma_trend} below. Step 1 fits the trend \smash{$\widehat{\xi}^2(M_i)$}. Step 2 estimates the distribution (prior) of the residual heterogeneity \smash{$\sigma_i^2/\widehat{\xi}^2(M_i)$} by empirical Bayes. Step 3 then computes p-values that combine $Z_i$, $S_i^2$, the trend evaluated at $M_i$ (\smash{$\widehat{\xi}^2(M_i)$}), and the prior estimated in Step 2; in particular $S_i^2$ enters the test statistic rather than being discarded. Steps 2 and 3 embed a specific parametric model of residual heterogeneity.

\begin{algorithm}[H]
\caption{\limmatrd{} (\reginvc)}
\begin{algorithmic}[1]
\State Fit the trend of $\log(S_i^2)$ vs. $M_i$ (e.g., as described in Supplement~\bluelink{subsec:reg_trend}), call this trend $\widehat{m}(\cdot)$ and its exponential $\widehat{\xi}^2(\cdot) = \exp(\widehat{m}(\cdot))$. 
\State Estimate $\wh{\kappa}_0$ and $\wh{s}_0^2$ in the model $\widehat{\xi}^2(M_i)/\sigma_i^2 \sim \chi^2_{\kappa_0}/(\kappa_0 s_0^2)$, where $\chi^2_{\kappa_0}$ is the chi-squared distribution with $\kappa_0$ degrees of freedom.
\State Compute $\widetilde{S}_i^2 := \{(K-p)S_i^2 + \wh{\kappa}_0 \wh{s}_0^2\widehat{\xi}^2(M_i)\}/\{(K-p) + \wh{\kappa}_0\}$, $\widetilde{T}_i := Z_i / \{\nu \widetilde{S}_i\}$ and estimated p-values $\wh P^\ltrdp_i =  2 \wb{F}_{t_{(K-p) + \wh\kappa_0}}(|\widetilde{T}_i|)$.
\end{algorithmic}
\label{algo:limma_trend}
\end{algorithm}
Our interpretation of \limmatrd{} is that it computes p-values conditional on $S_i^2$ and $M_i$. This conditioning is not possible under a purely frequentist analysis, but becomes possible when the nuisance parameters carry a prior distribution, with the primary parameters $\theta_i$ kept fixed. Such an analysis is called \emph{partially Bayes} \citep{cox1975note}, and we develop it for \limmatrd{} in Section~\ref{sec:stat_setting}.\footnote{Although our focus is on the partially Bayes interpretation of \limmatrd{}, we also provide purely frequentist results, with all parameters fixed, following the framework of compound decision theory. These results are briefly reviewed in Section~\ref{sec:compound_decision_theory} and presented in detail in Supplement~\purelink{sec:cmp_bayes_ltrd}.}

An important wrinkle in applications is the choice of $M_i$. The default choice in \limmatrd{} is,
\begin{equation} 
A_i := \frac{1}{K}\sum_{j=1}^K Y_{ij},
\label{eq:avg_intensity}
\end{equation}
which is called the \emph{average intensity}. 
But $A_i$ is computed from the same outcomes $Y_i$ as $Z_i$ and $S_i^2$, and the partially Bayes setup quietly assumes their conditional independence given $(\beta_i, \sigma_i^2)$. Section~\ref{subsec:independence} shows that a simple orthogonality condition on the design and contrast (Assumption~\ref{assu:design}, Proposition~\ref{prop:orthog}) restores this independence, and verifies that it holds for the standard two-sample and treatment-effect designs. The condition has not, to our knowledge, been spelled out before, and its violation is not harmless: the MAnorm2~\citep{tu2021manorm2} package for ChIP-seq replaces $A_i$ with a quantity that fails it, and our simulations in Section~\ref{sec:numerical_exp} show that this violation can inflate the false discovery rate (FDR).

From the partially Bayes perspective, the $\mathrm{Inv}\chi^2$ specification (in Step 2 of Algorithm~\ref{algo:limma_trend}) is one choice rather than a necessity. Replacing it with a prior estimated by the nonparametric maximum likelihood estimator (NPMLE)~\citep{robbins1950generalization, keifer_npmle} yields \reglitrd{}  (Algorithm~\ref{algo:limma_trend_reg}, Section~\ref{sec:reg_limma_trend}). The resulting empirical Bayes p-values approximate the partially Bayes p-values one would form if the prior were known, at a rate dominated by the trend-estimation error (Theorem~\ref{thm:conv_p_trnd_orc}), and Benjamini-Hochberg~\citeyearpar{benjamini1995controlling} applied to them controls the FDR asymptotically (Theorem~\ref{thm:final_rate_cs}). More notably, even if the fitted trend \smash{$\widehat{\xi}^2$} is inconsistent for the true trend, the nonparametric prior on residual heterogeneity absorbs the error and FDR control is preserved (Theorem~\ref{thm:final_rate_mis}). The prior on residual heterogeneity captures variation that no deterministic trend can absorb; the same flexibility also buffers errors in the fitted trend. Procedures like MAP, which treat $\sigma_i^2$ as a deterministic function of $M_i$, do neither.

A more ambitious procedure, \jtlitrd{} (Algorithm~\ref{algo:limma_trend_jt}, Section~\ref{sec:joint_limma_trend}), foregoes direct trend-fitting. Rather than modeling residual heterogeneity around a fitted trend, it estimates the joint distribution of $\sigma_i^2$ and the population mean intensity $\mu_i := \EE{A_i \mid \beta_i}$ across units, using a bivariate NPMLE applied to $(S_i^2, A_i)$. An Eddington-Tweedie representation of the partially Bayes p-values (Theorem~\ref{thm:tweedie_2_d}) drives the analysis and yields a near-parametric convergence rate of the p-values (Theorem~\ref{prop:asymp_p_val}), despite the non-parametric estimation of the bivariate prior; asymptotic FDR control again follows (Theorem~\ref{thm:final_rate}).

Section~\ref{sec:numerical_exp} confirms the theory in simulation. \reglitrd{} and \jtlitrd{} retain FDR control across the regimes we consider. The trended methods also recover more discoveries than their untrended counterparts. Section~\ref{sec:case_studies} applies the methods to bulk RNA-seq, ChIP-seq, and proteomics datasets. The proteomics analysis motivates a discrete variant of \jtlitrd{} that fits a separate NPMLE within each peptide-count stratum, which we view as a methodological successor to DEqMS~\citep{zhu2020deqms}, a \limmatrd{} variant popular for proteomics analyses.

\subsection{Related work: mean-variance modeling in genomics}
\label{sec:related_work}

Mean-variance modeling underlies several important methods in genomics. We provide an incomplete and simplified taxonomy here, to emphasize the main aspects of our contribution.

Several authors have focused on consistent estimation of mean-variance trends, e.g.,~\citet{carroll2008nonparametric, wang2009variance, fan2010nonparametric, mandel2013variance, li2024robust}. These works vary in their exact objective, e.g., estimating trends as functions of latent parameter (e.g., as a function of $\mu_i$) or observables, using nonparametric or parametric methods etc. A key message of these papers is that in general a naïve nonparametric regression of $\log(S_i^2)$ (or $S_i^2$) on $A_i$ will not yield a consistent estimate of the true  trend function and propose more elaborate methods.

Although most of the above papers are motivated by the testing problem we consider, 
only~\citet{mandel2013variance} formally develop p-values that use the learned trend.
Specifically,~\citet{mandel2013variance} assume that (under the null),  \smash{$\sigma_i^2 = \zeta^2(\mu_i)$} for a function \smash{$\zeta^2$} that can be estimated consistently as \smash{$\widehat{\zeta}^2$}.  They then propose to compute a p-value for $H_i: \theta_i=0$ as \smash{$\beta+\sup_{\mu_i \in \mathcal{I}_i} 2\Phi( -|Z_i| / \{\nu \widehat{\zeta}(\mu_i)\})$}, where $\mathcal{I}_i$ is a preliminary $(1-\beta)$-confidence interval for $\mu_i$ at some (small) level $\beta \in (0,1)$. A notable difference from our approach is that the trend ($\zeta$) is a function of an unobserved quantity $\mu_i$ (rather than the observed $A_i$), which makes it harder to estimate. The approach is also more conservative (because of the supremum in the p-value definition).

Earlier, we mentioned the MAP method~\citep{li2019map} as an example of an approach that uses the estimated trend to compute p-values, but that in general does not control FDR.
We note that the authors of MAP were aware of the lack of FDR control of their method and proposed a heuristic to further inflate variances. However, it is unclear whether this heuristic provides FDR control. This highlights the importance of a rigorous statistical analysis such as the one we provide for \limmatrd{}, which avoids the need for ad hoc corrections.

Finally, we note that several other key statistical methods in genomics
rely on learning mean-variance (or dispersion) relationships, e.g.,
edgeR~\citep{mccarthy2012differential,chen2025edger}, voom~\citep{law2014voom}, and DESeq2~\citep{love2014moderated}. We hope that our investigation of \limmatrd{} will motivate analogous statistical inquiries into these and related procedures.

\section{Statistical modeling and partially Bayes testing}
\label{sec:stat_setting}
\subsection{Background on untrended \limma{}}
\label{subsec:background}

Before turning to \limmatrd{}, we review the partially Bayes framework for untrended \limma{}~\citep{smyth2004linear, ignatiadis2025empirical}; the trended generalization in Section~\ref{subsec:limma_trd_interpretation} builds on this setup. Under model~\eqref{eq:limma_regression}, the summary statistics in~\eqref{eq:limma_summary_stats} satisfy
\begin{equation}
\label{eq:def_z_i_dis}
(Z_i,S^2_i) \mid \beta_i, \sigma_i^2 \, \sim \dnorm(\theta_i,\, \nu^2 \sigma_i^2) \otimes \frac{\sigma_i^2}{K-p}\chi^2_{K-p}.
\end{equation}
\texttt{Limma} only operates on the summary statistics $(Z_i, S_i^2)$, and so effectively for each unit, there are only two unknown parameters, $\theta_i$ (the parameter of interest) and \smash{$\sigma_i^2$} (a nuisance parameter). The distinctive feature of \limma{} is that it imposes a prior on the
nuisance parameter \smash{$\sigma_i^2$},
\begin{equation}
\label{eq:exchangeable_sigma}
\frac{1}{\sigma_i^2} \sim \frac{\chi^2_{\kappa_0}}{\kappa_0 s_0^2}, \quad \text{for some $\kappa_0, s_0^2 >0$,}
\end{equation}
but not on the primary parameter $\theta_i$, which is instead treated as fixed (as in a frequentist analysis). \citet{cox1975note} and~\citet{mccullagh1990note} use the terminology ``partially Bayes'' for such an analysis that imposes a prior on only the nuisance parameters. 
The upshot is that integrating over~\eqref{eq:exchangeable_sigma}, we get:
\begin{equation}
P_i^{\limma} := 2 \wb{F}_{t_{(K-p) + \kappa_0}}(|\widetilde{T}_i|) \mid (S_i^2, \theta_i=0) \sim \mathrm{Unif}[0,1],\text{ for }\, \widetilde{T}_i := \frac{ Z_i}{ \nu \widetilde{S}_i},\, \widetilde{S}_i^2:= \frac{(K-p)S_i^2 + \kappa_0 s_0^2}{(K-p) + \kappa_0 }.
\label{eq:limma_pvalues}
\end{equation}
In words, the limma p-value is computed by
replacing $S_i^2$ by the ``moderated'' version $\widetilde{S}_i^2$ that shrinks $S_i^2$ toward $s_0^2$ before studentizing and increasing the degrees of freedom from $(K-p)$ to $(K-p) + \kappa_0$. The p-values remain uniformly distributed under the null (marginalizing over~\eqref{eq:exchangeable_sigma}), and improves power.
In practice, one estimates \smash{$\wh{\kappa}_0$} and \smash{$\wh{s}_0^2$} using empirical Bayes~\citep{robbins1956empirical, efron2010largescale} from $S_1^2,\ldots,S_n^2$, yielding estimated p-values \smash{$\widehat{P}_i^{\limma}$}. We interpret \smash{$\wh{\kappa}_0$} as the degrees of freedom gained by sharing information across units.

The p-values in~\eqref{eq:limma_pvalues} admit an equivalent representation
that will be central to our extension to \limmatrd{}. Integrating $\sigma_i^2$
against its posterior $\Pi(\,\cdot\mid S_i^2)$ under the prior
in~\eqref{eq:exchangeable_sigma} yields
\begin{equation}
  P_i^{\limma{}} = 2\int \Phi(-|Z_i|/\{\nu\sigma_i\})\,\Pi(\dd\sigma_i^2\mid S_i^2);
  \label{eq:oracle_limma}
\end{equation}
see, e.g.,~\citet[Proposition 3]{ignatiadis2025empirical}. The advantage
of~\eqref{eq:oracle_limma} over~\eqref{eq:limma_pvalues} is that it is
meaningful for any prior $G$ on \smash{$\sigma_i^2$}, including nonparametric
priors. \citet{ignatiadis2025empirical} pursue this by taking
\smash{$\sigma_i^2 \simiid G$} for an unknown $G$ estimated by the
NPMLE~\citep{robbins1950generalization, keifer_npmle}; also
see~\citet{lu2016variance}. They show that the resulting estimated
p-values \smash{$\widehat{P}_i^{\limma{}}$} converge to the oracle p-values
\smash{$P_i^{\limma{}}$} at a nearly parametric rate (up to log factors), and
that applying the Benjamini-Hochberg procedure to \smash{$\widehat{P}_i^{\limma{}}$}
asymptotically controls the FDR.

An important property of the oracle p-values in~\eqref{eq:oracle_limma}
is that they are conditionally uniform given~$S_i^2$ under the null:
$P_i^{\limma{}} \mid (\theta_i=0, S_i^2) \sim \mathrm{Unif}[0,1]$.
The ordinary $t$-test p-value $P_i^{\ttest}$ does not satisfy this
conditional property: for null units whose $S_i^2$
underestimates~$\sigma_i^2$, the $t$-statistic is inflated 
\citep[Proposition 5]{ignatiadis2025empirical}. In a multiple testing
context, rejections from the $t$-test are thus enriched for units that by chance have unusually small~$S_i^2$. This is an advantage of the partially Bayes specification over a purely frequentist analysis: it is the prior~$G$ on~$\sigma_i^2$ and the integration over its posterior in~\eqref{eq:oracle_limma} that calibrates the p-values at every realized value of~$S_i^2$. (Such conditional calibration is even more important in the \limmatrd{} setting, see below.)

\subsection{Probabilistic interpretation of \limmatrd{}}
\label{subsec:limma_trd_interpretation}

Extending Section~\ref{subsec:background} to \limmatrd{} requires
specifying the conditional prior of $\sigma_i^2$ given $M_i$. We treat
two cases: $M_i$ as external side-information to model~\eqref{eq:limma_regression}, where
the extension is direct, and $M_i = A_i$ the average intensity in~\eqref{eq:avg_intensity}, where
two further issues arise.

\paragraph{\Litrd{} with external side-information.}
Suppose
that $M_i$ is external unit-specific side-information, e.g., $M_i$ could be the peptide count covariate in Fig~\ref{fig:intro_trend}(c), in that the distributional result in~\eqref{eq:def_z_i_dis} continues to hold conditional on $M_i$: 
\begin{equation}
\label{eq:bayesian_cs_model_1}
(Z_i,S^2_i) \mid \beta_i, \sigma_i^2, M_i \, \sim \dnorm(\theta_i,\, \nu^2 \sigma_i^2) \otimes \frac{\sigma_i^2}{K-p}\chi^2_{K-p}.
\end{equation}
As before, for each unit there are two unknown parameters ($\theta_i$ and $\sigma_i^2$) and following the partially Bayes principle, the nuisance parameter $\sigma_i^2$ is assumed to follow a prior $G$ that is now a function of $M_i$, i.e., $\sigma_i^2 \mid M_i \sim G(\cdot \mid M_i)$. This conditional prior is specified as follows in \limmatrd{}; see
also~\citet[Section 5]{phipson2016robust}:
\begin{align}
\label{eq:reg-model-trnd}
\frac{1}{\sigma_i^2} \mid M_i \; \sim \; \frac{\chi^2_{\kappa_0}}{\kappa_0 s_0^2 \xi^2_0(M_i)},
\end{align}
where $\xi_0(\cdot)$ is a smooth trend and $\kappa_0, s_0^2 >0$ are analogous to~\eqref{eq:exchangeable_sigma}. In Section~\ref{sec:reg_ltrd_ext} we will work primarily with the nonparametric generalization:
\begin{equation}
\label{eq:np-model-trnd}
\sigma_i^2 \mid M_i \; \sim \;  \xi_0^2(M_i) G, \quad \mbox{where $G$ is specified nonparametrically.}
\end{equation}
Taking logarithms and writing $m(\cdot)=\log(\xi^2_0(\cdot))$, this model can be rewritten as $\log(\sigma_i^2) = m(M_i) + \eta_i$, where $\eta_i \indep M_i$ and $\exp(\eta_i) \sim G$. In particular, we additively decompose $\log(\sigma_i^2)$ into a part that is fully explained by $M_i$ and
unexplained remaining independent heterogeneity. In this sense, this model is similar to related empirical Bayes models in the Gaussian sequence model that incorporate external side-information~\citep{fayiii1979estimates, ignatiadis2019covariatepowered}. 

Given~\eqref{eq:bayesian_cs_model_1} and~\eqref{eq:reg-model-trnd} (or its nonparametric extension in~\eqref{eq:np-model-trnd}), we interpret \limmatrd{} as computing partially Bayes p-values in analogy to~\eqref{eq:oracle_limma}, now conditioning on both $S_i^2$ and $M_i$,
\begin{equation}
P_i^{\ltrdp} :=  2 \int \Phi(-|Z_i|/\{\nu \sigma_i\}) \Pi(
\dd \sigma_i^2 \mid S_i^2, M_i).
\label{eq:oracle_limma_trend}
\end{equation}
In practice, we estimate $\xi_0(\cdot)$ by fitting a trend function \smash{$\wh \xi(\cdot)$} using $\log(S_i^2)$ and $M_i$ (as in Step 1 of Algorithm~\ref{algo:limma_trend}) and  plug in the estimated trend to estimate \smash{$\wh {G}$} in~\eqref{eq:np-model-trnd} (as in Step 2 of Algorithm~\ref{algo:limma_trend}). Then, \smash{$\wh \xi $} and \smash{$\wh G$} yield an estimated posterior \smash{$\widehat{\Pi}(\,\cdot\mid S_i^2, M_i)$} and estimated p-values $\widehat{P}_i^{\ltrdp}$, for $i=1,\ldots,n$ (Step 3 of Algorithm~\ref{algo:limma_trend}). In Supplement~\bluelink{sec:param_prior} we show that this general formulation encompasses parametric \limmatrd{} as a special case (see Algorithm~\ref{algo:limma_trend_reg} for our proposed nonparametric version of Algorithm~\ref{algo:limma_trend}).

As in the untrended case, the oracle p-values
in~\eqref{eq:oracle_limma_trend} are conditionally uniform under the
null, now given both~$S_i^2$ and~$M_i$:
$P_i^{\ltrdp} \mid (\theta_i=0, S_i^2, M_i) \sim \mathrm{Unif}[0,1]$; see Lemma~\ref{lem:bayes_trd_p_val_cs} below.
The additional conditioning on~$M_i$ ensures calibration within each stratum of units sharing $M_i$, a quantity that the analyst observes directly. For instance, in proteomics, \citet{zhu2020deqms}
observe that without adjusting for a mean-variance trend, analyses
produce excess false positives among proteins quantified with few
peptides (small $M_i$) and excess false negatives among those with many (large $M_i$).
Conditioning on~$M_i$ prevents this; we revisit this issue in our
proteomics application in Section~\ref{sec:seq_peptide}.

\paragraph{\reglitrd{} with average intensity.}
As we mentioned earlier, the most common choice of $M_i$ for \limmatrd{} is the average intensity $A_i$ defined in~\eqref{eq:avg_intensity} whose distribution under~\eqref{eq:limma_regression} satisfies
\begin{equation}
\label{eq:limma_trend_summary}
A_i \mid \beta_i, \sigma_i^2\,\sim \, \mathrm{N}\left(\mu_i,\, \frac{1}{K}\sigma_i^2\right),\;\text{ where }\; \mu_i:= \contrastintensity^\top \beta_i,\;\;  \contrastintensity := \frac{1}{K}\sum_{j=1}^K x_j.
\end{equation}
When using $A_i$, \limmatrd{} proceeds exactly as described above for the case of external side-information. This poses a conceptual challenge for two reasons:
\begin{enumerate}[label=(\Roman*),leftmargin=*,wide]
\item $A_i$ is computed on the same outcomes $Y_{ij}$ in~\eqref{eq:limma_regression} as $Z_i$ and $S_i^2$, so~\eqref{eq:bayesian_cs_model_1} need not hold for $M_i=A_i$.
\item The distributional specifications are circular: the prior of $\sigma_i^2$ is specified as a function of $M_i=A_i$ in~\eqref{eq:reg-model-trnd} (and~\eqref{eq:np-model-trnd}), while the likelihood of $A_i$ in~\eqref{eq:limma_trend_summary} depends on $\sigma_i^2$. 
\end{enumerate}
Section~\ref{subsec:independence} resolves (I) via an orthogonality
condition on the design (Assumption~\ref{assu:design}) that restores independence of $A_i$ from $(Z_i, S_i^2)$;
Section~\ref{subsec:bivariate_nuisance} resolves (II) by treating
$(\mu_i, \sigma_i^2)$ jointly as nuisance parameters with a bivariate
prior, which also motivates the \jtlitrd{} procedure of
Section~\ref{sec:joint_limma_trend}.

\subsection{Independence, orthogonal contrasts, and average intensity}
\label{subsec:independence}
We first explain that (I) does not pose a concern as long
as the analyst verifies a simple condition on the design and contrast in~\eqref{eq:limma_regression}. It will be convenient to state our results under slightly more generality. Fix $c_{\wt{A}} \in \RR^p$ and let 
$\wt{A}_i  := c_{\wt{A}}^\top \wh{\beta}_i$.
Note that with the choice $c_{\wt{A}} = \contrastintensity$ defined in~\eqref{eq:limma_trend_summary}, we recover the average intensity defined in~\eqref{eq:avg_intensity}, i.e., $\wt{A}_i= A_i$.

\begin{assumption}[Design]
The ones vector lies in the column space of $X$, i.e., $\mathbf{1} \in \mathcal{C}(X)$, and the primary contrast $\contrastprimary$ satisfies
$\contrastprimary^\top  (X^\top  X)^{-1} c_{\wt{A}} =0$.
\label{assu:design}
\end{assumption}
It is straightforward to empirically check this assumption before running a \limmatrd{} analysis.
We have the following result:
\begin{proposition}
\label{prop:orthog}
Suppose that Assumption~\ref{assu:design} holds. Then, $(Z_i, S_i^2, \wt{A}_i)$ are mutually independent conditional on $\beta_i, \sigma_i^2$, and so in particular,~\eqref{eq:bayesian_cs_model_1} holds with $M_i=\wt{A}_i$.
\end{proposition}

The simplest, yet prevalent, situation that satisfies this assumption is the two-sample comparison.

\begin{example}[Two-sample comparison]
\label{ex:two_sample}
We seek to test equality of means for each unit between $K_1$ treatment and $K_2$ control subjects, with $K=K_1+K_2$. We encode \smash{$x_j = (1, 0)^\top$} for $j=1,\ldots,K_1$ and \smash{$x_j = (0, 1)^\top$} for $j=K_1+1,\ldots,K$. Our contrast of interest is \smash{$c_{\theta}=(1,-1)^\top$} so that 
\smash{$\wh \theta_i=\wb Y_{i,\mathrm{treated}}-\wb Y_{i,\mathrm{control}}$}, where \smash{$\wb Y_{i,\mathrm{treated}}$} and \smash{$\wb Y_{i,\mathrm{control}}$} denote the treated/control group means for the $i$-th unit. The average intensity in~\eqref{eq:avg_intensity} can be written as $A_i = (K_1  \wb Y_{i,\mathrm{treated}} + K_2 \wb Y_{i,\mathrm{control}})/K$.
In this case, $c_{\theta}$ and $\contrastintensity$ satisfy Assumption~\ref{assu:design}.
\end{example}

The two-sample comparison generalizes to the following commonly used linear model in which the analyst also adjusts for other variables, e.g., confounders.

\begin{example}[Treatment effect]
\label{ex:treatment_effect}
Consider the linear model $Y_{ij} = \alpha_i + \tau_i w_j + \wt{x}_j^\top\gamma_i + \varepsilon_{ij},$ where $w_j \in \{0,1\}$ is a treatment indicator, $\wt{x}_j \in \mathbb R^q$ is a vector of control variables, and an intercept is included. Let $X$ denote the corresponding design matrix and let $c$ denote the contrast vector selecting the coefficient $\tau_i$.
Then $c^\top (X^\top X)^{-1} \contrastintensity = 0$. 
Consequently, our analysis of \limmatrd{} extends to testing for the significance of treatment effects through the hypotheses $H_{i}:\tau_i=0$.
\end{example}

Most applications of~\limmatrd{} we have encountered indeed satisfy Assumption~\ref{assu:design}. Nevertheless, as far as we are aware, Assumption~\ref{assu:design} has not been spelled out explicitly before. 
This is of import. 
As one example, in MAnorm2~\citep{tu2021manorm2}, a popular method for two-sample ChIP-Seq comparisons, the average intensity $A_i$ is replaced by $\wt{A}_i^{\mathrm{MA2}} := (\wb Y_{i,\mathrm{treated}}+\wb Y_{i,\mathrm{control}})/2$  with the motivation of ``alleviating the influence of unbalanced group sizes.'' (Note that \smash{$\wt{A}_i^{\mathrm{MA2}} = A_i$} when $K_A=K_B$ but they are unequal otherwise.) This choice of \smash{$\wt{A}_i$} does not satisfy Assumption~\ref{assu:design}, and we show in simulations (Section~\ref{sec:numerical_exp}) that MAnorm2 can lead to substantial inflation of the FDR.

\subsection{Partially Bayes with average intensity: bivariate nuisances}
\label{subsec:bivariate_nuisance}
We now clarify how to reconcile~\eqref{eq:reg-model-trnd} (or its nonparametric generalization in~\eqref{eq:np-model-trnd}) with $M_i=A_i$ and~\eqref{eq:limma_trend_summary}. We posit Assumption~\ref{assu:design} holds for $\contrastintensity := \frac{1}{K}\sum_{j=1}^K x_j$. Then by Proposition~\ref{prop:orthog}, $Z_i, S_i^2, A_i$, are mutually independent and their distribution only depends on $\theta_i, \mu_i, \sigma_i^2$:
\begin{align}
\label{eq:def_z_i_dis_2d}
(Z_i, S_i^2, A_i)     \mid (\beta_i, \sigma_i^2) \, \sim \mathrm{N}(\theta_i,\, \nu^2 \sigma_i^2)\otimes \frac{\sigma_i^2}{K-p}\chi^2_{K-p} \otimes \mathrm{N}\left(\mu_i,\, \frac{1}{K}\sigma_i^2\right).
\end{align}
We will base inference for the primary parameter $\theta_i$ only on $Z_i, S_i^2,A_i$. Thus, we will treat $\mu_i$ and $\sigma_i^2$ as the nuisance parameters. Following the partially Bayes principle, we posit that $(\mu_i, \sigma_i^2)$ are drawn from an unknown bivariate prior, while $\theta_i$ is treated as fixed,
\begin{equation}
\label{eq:limma_trend_2d_prior}
(\mu_i,\sigma_i^2) \sim H, \text{ and } \theta_i \,\text{ is deterministic,}\; \qquad i=1,\ldots,n.
\end{equation}
Now, under~\eqref{eq:def_z_i_dis_2d} and~\eqref{eq:limma_trend_2d_prior}, the posterior distribution of $\sigma_i^2$ given $A_i$ is well-defined and (implicitly) marginalizes over $\mu_i$.  We say that $H$ is \emph{compatible} with \limmatrd{} if the aforementioned posterior takes the form in~\eqref{eq:np-model-trnd} with $M_i=A_i$. As an example, in Proposition~\bluelink{prop:joint_parametric} in the supplement, we show that the conjugate Gaussian/Inv$\chi^2$ prior $H$ is compatible. 

The above interpretation has two important consequences in this paper. First, in Section~\ref{sec:reg_ltrd_int} 
we show that our proposed \reglitrd{} procedure asymptotically controls the FDR \emph{even} if $H$ in~\eqref{eq:limma_trend_2d_prior} is not compatible; this is a strong form of robustness. However, without compatibility, the procedure loses its full partially Bayes interpretation in that case. To address this, in Section~\ref{sec:joint_limma_trend} we develop an alternative procedure that directly targets the prior in~\eqref{eq:limma_trend_2d_prior}, thereby retaining a fully partially Bayes interpretation even when $H$ is not compatible.

\section{\reglitrd{}: \limmatrd{} based on the NPMLE}
\label{sec:reg_limma_trend}
Throughout this section, we make the following assumption:
\begin{align}
        \label{eq:bayesian_cs_model_2}
        &(Z_i,S^2_i) \mid \theta_i,\sigma_i^2,M_i \sim \dnorm\left(\theta_i,\nu^2\sigma_i^2\right) \otimes \frac{\sigma_i^2}{K-p}\chi^2_{K-p}.
\end{align}
Our formulation of~\eqref{eq:bayesian_cs_model_2} captures two models. First, $M_i$ can be external side-information (cf. Section~\ref{subsec:background}), so that the above model corresponds to~\eqref{eq:bayesian_cs_model_1}, noting that the conditional distributions depend on $\beta_i$ only through $\theta_i$. Second, we allow for $M_i=A_i$, the average intensity, positing that Assumption~\ref{assu:design} holds, and moreover, that $(\mu_i, \sigma_i^2) \sim H$ as in~\eqref{eq:limma_trend_2d_prior}. In the latter case,~\eqref{eq:bayesian_cs_model_2} must be interpreted as marginalizing over $\mu_i$ (as described in Section~\ref{subsec:bivariate_nuisance}).

\subsection{Well-specified \reglitrd{}}
\label{sec:reg_ltrd_ext}

We now consider our nonparametric version of \limmatrd{}. In addition to~\eqref{eq:bayesian_cs_model_2},
 we posit~\eqref{eq:np-model-trnd} which we restate here: we posit that 
there exists a function $\xi_0(\cdot)$ such that:
\begin{equation}
\tau_i^2 \mid  M_i \, \sim\,  G,\; \text{ where }\; \tau_i^2 := \sigma_i^2 / \xi_0^2(M_i),\; \text{ and } \theta_i \text{ are deterministic}.
\label{eq:well-specified-limma-reg}
\end{equation}
 Both $\xi_0^2$ and $G$ are unknown. This model encompasses external side information, and the average intensity, in which case the bivariate prior $H$ in~\eqref{eq:limma_trend_2d_prior} must be compatible (as defined in Section~\ref{subsec:bivariate_nuisance}).

Using the notation of~\eqref{eq:well-specified-limma-reg}, we can equivalently express~\eqref{eq:reg-model-trnd} as $(\tau_i^{2})^{-1} \mid M_i \sim \chi^2_{\kappa_0}/(\kappa_0\, s_0^2)$. Our goal is to replace this parametric assumption by the assumption that $G \in \mathcal G_{\trd}$, where $\mathcal G_{\trd}$ is a nonparametric class of distributions. To this end, we start by rewriting the oracle partially Bayes p-values $P_i^{\ltrdp}$ as an explicit function of $G$ and $\xi_0$ in~\eqref{eq:bayesian_cs_model_2}, namely,\footnote{We specialize $P_i^{\ltrdp}$ in three different settings, each time introducing new notation (e.g., $P^\trd_i$ below) to indicate that the p-values pertain to the setting at hand.} 
\begin{align}
\label{eq:popn_emp_bayes_p_val}
P^\trd_i:=\frac{\int_{0}^{\infty}2\,\Phi\!\left(-\frac{|Z_i|}{\nu \,\xi_0(M_i)\tau_i}\right)
\,p_{\chi^2}(S^2_i\mid K-p,\xi^2_0(M_i)\tau^2_i)\,G(\dd\tau^2_i)}{\int_0^\infty p_{\chi^2}(S^2_i\mid K-p,\xi^2_0(M_i)\tau^2_i)\,G(\dd\tau^2_i)}, \qquad \mbox{for $i \in [n]$,}
\end{align}
where $p_{\chi^2}(\,\cdot\,\mid K-p,\tau^2)$ is the density of $\tau^2\chi^2_{K-p}/(K-p)$ random variable.

The following lemma shows that $P^\trd_i$ are uniform conditional on $S_i^2,M_i$, a desirable property (see Section~\ref{subsec:background}) implied by the partially Bayes specification.
\begin{lemm}
    \label{lem:bayes_trd_p_val_cs}
    Under the data generating model \eqref{eq:bayesian_cs_model_2} and \eqref{eq:well-specified-limma-reg}, 
    for all $i \in [n]$ such that $\theta_i=0$, we have 
    $
    \mathbb P\left[P^\trd_i \le t \mid S^2_i,M_i\right]=t, \; \mbox{for all $t \in [0,1]$, almost surely.}
    $
\end{lemm}

Given the form of the oracle partially Bayes p-values in~\eqref{eq:popn_emp_bayes_p_val}, it remains to explain how to estimate $\xi_0$ and $G$. 
We will state our results below in terms of a general estimator $\widehat{\xi}(\cdot)$ that is consistent for $\xi_0(\cdot)$, that is, we do not prescribe a specific estimation strategy for this step. However, in practice, we use the approach implemented in the \limma{} package~\citep{ritchie2015limma} using natural splines (see,  Step 1 of Algorithm~\ref{algo:limma_trend} and Supplement~\bluelink{subsec:reg_trend}).
Next, we focus on estimating $G$ from $\{(S^2_i,M_i):i \in [n]\}$. The marginal density of $V^2_i :=S^2_i/\xi^2_0(M_i)$ conditioned on $M_i$ is given by
\begin{align}
\label{eq:lik_1d_limma_trend}
f_{G,K-p}(v^2)
    = \int_0^\infty 
        p_{\chi^2}(v^2\mid K-p,\tau^2)\, G(\dd\tau^2).
\end{align}
Thus, if $\xi_0$ were known, we could compute the NPMLE that optimize the joint marginal log-likelihood of $V^2_1,\ldots,V^2_n$ over all $G \in \mathcal{G}_{\tr}$. Instead, we compute the same NPMLE with $V_i^2$ replaced by plug-in estimates $\wh V_i^2$ that use $\wh \xi$ in lieu of the true $\xi_0$:
\begin{align}
\label{eq:def_small_ell}
\widehat{G}_{\trd}
   \in \argmax_{G  \in \mathcal G_{\trd}}
    \frac{1}{n}\sum_{i=1}^n
        \log f_{G,K-p}\!\left(\wh V^2_i\right),\;\;\text{ where }\; \wh V^2_i:=S^2_i/\wh \xi^2(M_i).
\end{align}
We emphasize that the estimator $\wh G_\trd$ is not an NPMLE per se, but rather an approximation to an oracle NPMLE that knows $V_i^2$ (and our theory will account for this discrepancy). We observe~\eqref{eq:def_small_ell} is a convex optimization that can be efficiently solved using the proposal of~\citet{koenker2014convex}. More specifically, we solve~\eqref{eq:def_small_ell} through approximating $\rmG$ using a discrete prior and reducing \eqref{eq:def_small_ell} to a conic programming problem, which 
is solved using MOSEK~\citep{aps2020mosek}.
Further details are provided in Supplement~\bluelink{sec:modek_reg_npmle}. 
We note that our theory does not account for discretization and for the data-driven choice of the support of the prior class.
Previous works, \citet{dicker2016highdimensional} and \citet{soloff2024multivariate}, have conducted analyses of the NPMLE while considering the effects of discretization.

Using the estimated trends $\wh \xi$ and prior $\wh G_\trd$, one can estimate $P^\trd_i$ using $\wh P^\trd_i$, where for all $i \in [n]$, the estimators are computed using \eqref{eq:popn_emp_bayes_p_val} by replacing $\xi_0$ and $G$ with $\wh \xi$ and $\wh G_\trd$, respectively. Our full proposal is summarized in Algorithm~\ref{algo:limma_trend_reg}:

\begin{algorithm}[H]
\caption{\limmatrd{} with nonparametric prior (\reglitrd{})}
\begin{algorithmic}[1]
\State Fit the trend of $\log(S_i^2)$ versus $M_i$, call this trend $\widehat{m}(\cdot)$ and its exponential $\widehat{\xi}^2(\cdot) = \exp(\widehat{m}(\cdot))$. 
\State Estimate $G$ as $\wh G_\trd$ using the approximate NPMLE with plugged-in $\wh{\xi}$ as in~\eqref{eq:def_small_ell}.
\State Compute the p-values $\wh P^\trd_i$ using \eqref{eq:popn_emp_bayes_p_val} by replacing $\xi_0$ with $\wh \xi$ and $G$ with $\wh G_\trd$.
\end{algorithmic}
\label{algo:limma_trend_reg}
\end{algorithm}
We now develop theory for \reglitrd{}. We assume the following on the data-generating process.
\begin{assumption}
\label{assu:limma_trnd_var}
The tuples $(\tau_i^2, M_i, Z_i, S_i^2)$ are generated according to~\eqref{eq:bayesian_cs_model_2} and~\eqref{eq:well-specified-limma-reg} for $K -p\ge 2$ and are jointly independendent across $i \in [n]$.
We have that
$\mathbb P\!\left(\max_{i \in [n]}|M_i| > \mathrm W\sqrt{\log n}\right)\le n^{-3}$ for some $W>0$. The trend-adjusted variance prior satisfies $G \in \mathcal G_{\trd}$, where $\mathcal G_{\trd}:=\left\{\rmG: \rmG\left([\underline L_\trd,\wb U_\trd]\right)=1\right\}$ for absolute constants $\underline L_\trd,\wb U_\trd>0$ and the true trend function $\xi_0$ satisfies $\xi_0(x)\in[\underline M_\trd,\overline M_\trd]$ for all $x\in\mathbb R$, where $\underline M_\trd,\overline M_\trd>0$ are absolute constants.
\end{assumption}

The choice of nonparametric class $\mathcal G_{\trd}$ makes no smoothness assumptions whatsoever, but it does impose lower and upper bounds (see the discussion following Assumption 8 in~\citet{ignatiadis2025empirical}).
We next make a high-level assumption on the estimation of the trend.
\begin{assumption}[Trend estimation]
\label{assum:trend_estimation}
Let $\mathrm W_n:=\mathrm W \sqrt{\log n}$ for $\mathrm W$ defined in Assumption~\ref{assu:limma_trnd_var}. For any two functions $f,g$, define the semi-norm $\|f-g\|_{\mathrm W_n}:=\sup_{x \in [-\mathrm W_n,\mathrm W_n]}|f(x)-g(x)|$. We assume that the estimated trend $\widehat\xi(\cdot)$ satisfies 
$
\|\widehat\xi-\xi_0\|_{\mathrm W_n}\le \Delta_n
$,
where $\Delta_n=o(1)$ is a deterministic sequence.  Moreover, we assume that $\widehat{\xi} \in \mathcal{X}$ almost surely, where $\mathcal{X}$ is a class of functions that contains $\xi_0$ and is separable under the supremum norm on $\RR$. Furthermore, all $\xi \in \mathcal X$ satisfy $\xi(\cdot)\in[\underline M_\trd,\overline M_\trd]$ and $\log N(\varepsilon,\mathcal X,\|\cdot\|_{\mathrm W_n})\lesssim (1/\varepsilon)^{h_1} \cdot |\log\varepsilon|^{h_2}$ for some constants $h_1,h_2 \in [0,1]$, 
where $N(\varepsilon,\mathcal X,\|\cdot\|_{\mathrm W_n})$ is the $\varepsilon$-covering number of the class $\mathcal X$ in $\|\cdot\|_{\mathrm W_n}$. 
\end{assumption}
Our main result quantifies the rate at which we approximate the oracle partially Bayes p-values and accounts for the combined impact of prior approximation and trend estimation error.
\begin{theorem}
\label{thm:conv_p_trnd_orc}
Suppose Assumptions~\ref{assu:limma_trnd_var} and~\ref{assum:trend_estimation} hold. Then, for any $\zeta \in \left(\frac12,1\right)$, there exists a constant $\mathrm{C}_1>0$, depending on $\underline M_\trd,\overline M_\trd,\underline L_\trd,\overline U_\trd,K,p,\nu$, and $\zeta$, such that,
\begin{align}
&\frac{1}{n}\sum_{i=1}^n 
\mathbb E\!\left[\left|P^\trd_i \wedge \zeta - \wh P_i^{\trd} \wedge \zeta\right|\right]\;\le\; \mathrm{C}_1\cdot \mathfrak L_n(\Delta_n) \,\text{ for all }\, n \ge 1,
\end{align} 
where
$
\mathfrak L_n(\Delta_n):=\max\left\{\Delta_n \log^{5/2} n,\,\frac{\log^{5/2} n}{n^{1/4}}\Delta^{1/2}_n,\,\frac{\log^2 n}{n^{1/4}}|\log \Delta_n|^{h_2/4}\Delta^{\frac{1}{2}(1-h_1/2)}_n,\,\frac{\log^{5/2} n}{n^{1/2}}\right\}.
$
\end{theorem}
If $\mathcal X$ contains parametric functions, then $h_1=0$ and $h_2=1$. In that case, the optimal function estimation rate is $\Delta_n=n^{-1/2}$ and $\mathfrak L_n(\Delta_n)$ is also equal to $n^{-1/2}$ up to log factors. Similarly, if $\mathcal X$ is $\varrho$-H\"older smooth on $[-\mathrm W_n,\mathrm W_n]$ for $\varrho \geq 1$, then we can take $h_1=1/\varrho$ and $h_2 \geq 0$ and our result yields $\mathfrak L_n(\Delta_n)$ of order $n^{-\varrho/(2\varrho+1)}$ (up to log factors), which is also the minimax rate for estimating $\xi_0$.
Thus, the statistical cost of p-value estimation matches (and is dominated by) that of the trend function estimation up to logarithmic factors for parametric and H\"older-smooth trends.

\begin{rem}[Proof strategy: marginal density convergence]
\label{rem:marginal}
The proof has two primary ingredients. First, we control the squared Hellinger distance between the true and fitted marginal densities (Lemma~\ref{lem:hell_limm_trnd_cs}), $$\mathcal H^2\left(f_{G,K-p},f_{\wh G_\trd,K-p}\right):=\int_{0}^\infty\left(\sqrt{f_{G,K-p}(x)}-\sqrt{f_{\wh G_\trd,K-p}(x)}\right)^2\dd x,$$ accounting for the fact that $\wh G_\trd$ is an approximate NPMLE under a misspecified likelihood in which the true trend $\xi_0$ is replaced by an estimate $\wh \xi$; we handle this by extending techniques from \cite{chen2024empiricalbayesestimationprecision}. Second, an Eddington--Tweedie-type identity (see~(\bluelink{eq:tweedie_trnd}) of Supplement~\bluelink{sec:proof_trend_pval}, analogous to Proposition 11 in~\citet{ignatiadis2025empirical}) expresses the p-values $P^\trd_i$ as functionals of the marginal density, so Hellinger consistency of the marginal transfers directly into rates for the plug-in p-values. 

A practical consequence is that our guarantees go through the marginal density, not the latent prior $G$. We therefore recommend visual diagnostics that compare the fitted marginal to the observed data, rather than diagnostics that target $G$ itself; see Section~\ref{sec:case_studies}.
\end{rem}

The following results are established under the same assumptions as Theorem \ref{thm:conv_p_trnd_orc}. First, we show that the conditional uniformity of Lemma~\ref{lem:bayes_trd_p_val_cs} holds approximately for the estimated p-values. 

\begin{proposition}
    \label{prop:avg_sign_limma_trnd}
    Suppose Assumptions~\ref{assu:limma_trnd_var} and~\ref{assum:trend_estimation} hold.
    Then, there exists a constant $C_\trd>0$ depending on $\nu,K,p,\underline M_\trd,\overline M_\trd,\underline L_\trd,\overline U_\trd, \zeta$, such that
    \begin{align}
        \max_{i \in [n]\,:\, \theta_i=0}\left\{\mathbb E_G\left[\sup_{t \in [0,\zeta]}\left|\mathbb P_G\!\left(\widehat P_i^{\trd} \le t \mid \{(S_j^2,M_j)\}_{j=1}^n\right)-t\right|\right]\right\} \le C_\trd\,\mathfrak L_n (\Delta_n). 
    \end{align}
\end{proposition}
In practice, \limma{}-type p-values are used alongside the BH procedure~\citep{benjamini1995controlling} to control the false discovery rate (FDR). For this reason, below we show asymptotic FDR control using BH on $\wh P_1^{\trd},\ldots, \wh P_n^{\trd}$.

Let us briefly review the BH procedure applied at level $\alpha \in (0,1)$ to a generic collection of
p-values $P^{\mathrm{gen}}_1, \ldots, P^{\mathrm{gen}}_n$. One first sorts the p-values as $P^{\mathrm{gen}}_{(1)} \le P^{\mathrm{gen}}_{(2)} \le \cdots \le P^{\mathrm{gen}}_{(n)}$ and determines the largest $j$ (called, $\wh j$) such that $P^{\mathrm{gen}}_{(j)} \le (j/n) \alpha$. If no such $j$ exists, then $\wh{j}=0$. Finally, the hypotheses corresponding to $P^{\mathrm{gen}}_{(1)},\cdots,P^{\mathrm{gen}}_{\left(\wh j\right)}$ are rejected. If the p-values are super uniform and independent, then this procedure controls the FDR at $\alpha$~\citep{benjamini1995controlling}. 

To provide asymptotic results on the BH procedure, we introduce the following definition.
\begin{definition}
\label{asm:1d_limma_trend_bh}
A sequence $\{P^{\mathrm{gen}}_i:i\in [n]\}$ is \emph{critically dense} at $\alpha \in (0,1)$, if there exists $0<t_0(\alpha)<t_1(\alpha)<\alpha$ such that
\(
\smash{\liminf_{n \rightarrow \infty}\inf_{t \in [t_0,t_1]}\left\{(nt)^{-1}~\sum_{i=1}^n\mathbb P\!\left[P^{\mathrm{gen}}_i \le t
\right]\right\}~>~\alpha^{-1}.}
\)
\end{definition}
This criticality assumption is closely related to the standard asymptotic setup of~\citet{storey_2004}. Under weak dependence on the $P^{\mathrm{gen}}_i$, it implies that BH asymptotically rejects all p-values below a data-driven threshold $\wh{t} \geq t_0$; also see~\citet[Section 3.3.1]{ignatiadis2025empirical}.

In our setting, we apply BH at level $\alpha$ to $\wh P_1^{\trd},\ldots,\wh{P}_n^{\trd}$, which are only approximate p-values. We denote the number of false discoveries by $V_n^{\trd}$, by
the total number of discoveries $R_n^{\trd}$, and the FDR by $\fdr^\trd_n := \mathbb E\left[V^\trd_n/(R^\trd_n \vee 1)\right].$ Our main FDR control results is the following.

\begin{theorem}
\label{thm:final_rate_cs}
 Fix $\alpha \in (0,1)$. Suppose Assumptions~\ref{assu:limma_trnd_var} and~\ref{assum:trend_estimation} hold and the oracle partially Bayes p-value sequence $\{P^\trd_i\}$ is critically dense at $\alpha$ in the sense of Definition~\ref{asm:1d_limma_trend_bh}. Let $n_0:= \#\{i \in [n]: \theta_i=0\}$. Then, for any sequence $\vartheta_n > 0$ with 
\(
\smash{\limsup_{n \rightarrow \infty}\mathfrak L_n(\Delta_n) \vartheta_n =0,}
\)
we have
\[
\lim_{n \to \infty}
\left\{\vartheta_n\left(\fdr_n^{\trd}-\frac{n_0}{n}\alpha\right)_+\right\}=0.
\]
\end{theorem}
We emphasize that we impose the criticality condition on the oracle p-values $P^\trd_i$; the procedure we study uses the estimated p-values $\wh P^{\trd}_i$.

\subsection{Misspecified \reglitrd{}}
\label{sec:reg_ltrd_int}

In this section, we study the robustness of Algorithm~\ref{algo:limma_trend_reg} when the model~\eqref{eq:well-specified-limma-reg} is misspecified or when Assumption~\ref{assum:trend_estimation} fails and $\wh \xi$ is not even consistent for $\xi_0$. The reason is two-fold: first,~\eqref{eq:well-specified-limma-reg} can be violated in some settings (e.g., in our proteomics dataset from Section~\ref{sec:seq_peptide}), yet \limmatrd{} is still used in practice, raising the question of whether it retains any inferential guarantees. Second, it is common to impose effectively parametric assumptions when estimating trends (e.g., the trend estimation in Supplement~\bluelink{subsec:reg_trend} is effectively parametric) and such parametric trends are likely to be misspecified.
In particular, we continue to assume~\eqref{eq:bayesian_cs_model_2} and replace~\eqref{eq:well-specified-limma-reg} by the following:
\begin{equation}
M_i \sim \mathbb P^{M},\;\;\sigma_i^2 \mid  M_i \sim G_{M_i}(\cdot \mid M_i), \text{ and } \theta_i \text{ is deterministic}.
\label{eq:misspecified-limma-reg}
\end{equation}
Thus, we now allow for the conditional distribution of $\sigma_i^2$ to depend in an arbitrarily complicated way on $M_i$. (We explicitly define the distribution $\mathbb P^M$ of $M_i$ here because it will be needed below.)

Meanwhile, we replace Assumption~\ref{assum:trend_estimation} by the following.
\begin{assumption}[Misspecified trend estimation]
\label{assum:missp_trend_estimation}
There exists a deterministic function $\xi_\mis$ such that 
$\|\wh \xi-\xi_\mis\|_\infty \le \Delta_n$ where $\Delta_n=o(1)$. Moreover, the estimated trend satisfies the regularity conditions of Assumption~\ref{assum:trend_estimation} (that is, formally replacing $\xi_0$ by $\xi_\mis$ in the statement therein).
\end{assumption}
Even when~\eqref{eq:well-specified-limma-reg} holds, Assumption~\ref{assum:missp_trend_estimation} allows for $\xi_\mis \neq \xi_0$, which implies that $\wh \xi$ is an inconsistent estimator of $\xi_0$. More generally, $\xi_\mis$ is not tied to the data-generating process (beyond the fact that $\wh \xi$ converges to $\xi_\mis$). Given $\xi_\mis$, we also define:
\begin{equation}
\tau_{i,\mis}^2 := \sigma_i^2 / \xi_{\mis}^2(M_i) ,\;\; V_{i,\mis}^2 := S_i^2 / \xi_{\mis}^2(M_i).
\label{eq:tau_mis_v_mis}
\end{equation}
With these definition in place, we can replace Assumption~\ref{assu:limma_trnd_var} by the following.
\begin{assumption}
\label{assu:misspec_limma_trnd_var}
The tuples $(\tau_{i,\mis}^2, M_i, Z_i, V_{i,\mis}^2)$ are generated according to~\eqref{eq:bayesian_cs_model_2},~\eqref{eq:misspecified-limma-reg} and~\eqref{eq:tau_mis_v_mis} for $K -p\ge 2$ and are jointly independendent across $i \in [n]$.
Furthermore, the random variables $M_1,\ldots,M_n$ satisfy
$\mathbb P\!\left(\max_{i \in [n]}|M_i| > \mathrm W\sqrt{\log n}\right)\le n^{-3}$ for some $W>0$.
Moreover, we optimize $G$ in~\eqref{eq:def_small_ell} over 
$\mathcal{G}_{\trd} = [\underline L_\trd,\wb U_\trd]$
for absolute constants $\underline L_\trd,\wb U_\trd>0$ that are such that $\tau_{i,\mis}^2 \in [\underline L_\trd,\wb U_\trd]$ almost surely.
\end{assumption}
Now let us call $P_{i,\mis}^{\trd}$ to be the p-values in~\eqref{eq:popn_emp_bayes_p_val} with $(G,\xi_0)$ replaced by \smash{$(G_{\mis},\xi_\mis)$}, where  $G_{\mis}$ is the unconditional distribution of  \smash{$\tau_{i,\mis}^2$} (implied by the joint distribution of $(M_i, \sigma_i^2)$ in~\eqref{eq:misspecified-limma-reg}). 
Let \smash{$\fdr^\mis_n$} be the FDR of BH at level $\alpha$ applied to \smash{$\widehat{P}_i^{\trd}$} in the present misspecified setting. 
Then, similarly to the well-specified case (Theorem~\ref{thm:final_rate_cs}), FDR is asymptotically controlled. 
\begin{theorem}
\label{thm:final_rate_mis}
Fix $\alpha \in (0,1)$. Consider Assumptions~\ref{assum:missp_trend_estimation} and~\ref{assu:misspec_limma_trnd_var} for $G_\mis,\xi_\mis$ and $\wh \xi$. Assume that $n_0:= \#\{i \in [n]: \theta_i=0\}$ satisfies $n_0/n \rightarrow \pi_0 \in (0,1)$ as $n \rightarrow \infty$.
If the sequence $\{\misP{i}:i \in [n]\}$ is critically dense at $\alpha$,
then
\(
\limsup_{n \rightarrow \infty}\fdr^\mis_n \le \pi_0\,\alpha.
\)
\end{theorem}
The above theorem is the main robustness result of this section. The intuition behind it is that the (oracle) misspecified p-values $P_{i,\mis}^{\trd}$ retain a partially Bayes interpretation analogous as in Lemma~\ref{lem:bayes_trd_p_val_cs}.
\begin{lemm}
\label{lem:val_p_mis}
Suppose that~\eqref{eq:bayesian_cs_model_2} and~\eqref{eq:misspecified-limma-reg} hold. Fix $\xi_\mis$ and let $V_{i,\mis}$ be defined as in~\eqref{eq:tau_mis_v_mis}. Then, for all $i \in [n]$ such that $\theta_i=0$, we have that $
    \mathbb P[P^\trd_{i,\mis} \le t \mid V^2_{i,\mis}]=t$ for all $t \in [0,1]$ almost surely.
\end{lemm} 
In particular, whatever the misspecified trend $\xi_{\mis}$ is,\footnote{
The power of the p-values, however, depends critically on the misspecified trend $\xi_{\mis}$.
} $G_\mis$ ``buffers'' the error so that uniformity holds conditional on oracle trend-adjusted variances $V_{i,\mis}$. While stronger than unconditional uniformity and sufficient for asymptotic FDR control, Lemma~\ref{lem:val_p_mis} is not as strong as Lemma~\ref{lem:bayes_trd_p_val_cs} and does not allow us to argue, e.g., in the application of Section~\ref{sec:seq_peptide} that p-values are calibrated for proteins with small peptide count.
Finally, we note that Proposition~\bluelink{lem:approx_unif_mis_strong} in the supplement, establishes an asymptotic version of Lemma~\ref{lem:val_p_mis} for the estimated p-values $\wh P^\trd_i$.

\section{\jtlitrd{}: \limmatrd{} with bivariate nuisance priors}
\label{sec:joint_limma_trend}

Algorithm~\ref{algo:limma_trend_reg} (\reglitrd{}) in its oracle form ensures the partially Bayes property of Lemma~\ref{lem:bayes_trd_p_val_cs} only when~\eqref{eq:well-specified-limma-reg} is well-specified. However, that model takes a shortcut: it captures the conditional distribution of $\sigma_i^2 \mid M_i$ through only a trend and a trend-adjusted variance prior, rather than modeling the full conditional prior. To go beyond this, we must learn that full conditional prior directly. In Supplement~\bluelink{sec:discrete_MI} and Algorithm~\bluelink{algo:limma_trend_jt_discrete} therein, we show how this can be accomplished when $M_i$ is discrete (as in our proteomics application). 

In this section, we provide a construction for the average intensity, $M_i=A_i$. 
Throughout, we assume data are generated according to~\eqref{eq:def_z_i_dis_2d} and \eqref{eq:limma_trend_2d_prior} (which implicity also impose Assumption~\ref{assu:design}). The main crux is that instead of directly learning the conditional distribution given $M_i$, 
it suffices to learn the bivariate nuisance prior $H$.
We rewrite the oracle  p-values $P_i^{\ltrdp}$ in~\eqref{eq:oracle_limma_trend} in terms of $H$,
\begin{equation}
\Por_i := \frac{\int_{\mathbb R \times \mathbb R_{\ge 0}}2\,\Phi\!\left(-\frac{|Z_i|}{\nu \,\sigma}\right)
\,p_{K-p}(S^2_i,A_i \mid \mu,\sigma^2)\, H(\dd\mu,\dd\sigma^2)}{\int_{\mathbb R \times \mathbb R_{\ge 0}}
p_{K-p}(S^2_i,A_i \mid \mu,\sigma^2)\, H(\dd\mu,\dd\sigma^2)}, \qquad \mbox{for $i \in [n]$,}
\label{eq:2d_pb_pvalue_general_case}
\end{equation}
where $p_{\kappa}(s^2,a \mid \mu,\sigma^2)
:= p_{\chi^2}(s^2 \mid \kappa,\sigma^2)\,
\phi\!\left(\sqrt{K}(a-\mu)/\sigma\right)$
and $\phi(\cdot)$ is the standard Gaussian density.
The following proposition establishes conditional uniformity for $\Por_i$ analogous to that for $P_i^{\trd}$ in Lemma~\ref{lem:bayes_trd_p_val_cs}.
\begin{lemm}
\label{lem:valid_oracle_p_vals}
For all $i \in [n]$ such that $\theta_i=0$, we have $\mathbb P_H[P^\jt_i \le t \mid S^2_i,A_i]=t, \quad \mbox{for all $t \in [0,1]$.}$
\end{lemm}
We propose to estimate $H$ from the data by NPMLE using the marginal distribution of the summary statistics $\{(S^2_i,A_i):i \in [n]\}$ over a nonparametric class $\mathcal G_{H}$,
\begin{align}
    \label{eq:joint_npmle}
    \wh H \in \argmax_{\rmH \in \mathcal G_{H}} \frac{1}{n}\sum_{i=1}^n\log f_{\rmH,K-p}(S^2_i,A_i).
\end{align}
Here, for any bivariate mixing measure $\rmH$, we write
\begin{align}
\label{eq:2_d_marginal}
f_{\rmH,\kappa}(s^2,a)
:=\int_{\mathbb{R} \times \mathbb{R}_{\ge 0}}
p_\kappa(s^2,a \mid \mu,\sigma^2)\,\rmH(\dd\mu,\dd\sigma^2), \quad \mbox{for $\kappa \in \{K-p,K-p+1\}$}.
\end{align}
We solve~\eqref{eq:joint_npmle} by approximating $H$ through a finite sieve, which yields a finite-dimensional convex
program, and applying the interior point solver MOSEK \citep{aps2020mosek}; see Supplement~\bluelink{subsec:discretization_twod} for details. (Our theory below ignores the approximation error and data-driven choice of the sieve.) 

We obtain plug-in estimates $\wh P^\jt_i$ of the oracle p-values $P^\jt_i$ by replacing $H$ with $\wh H$ in~\eqref{eq:2d_pb_pvalue_general_case}. Our proposal is summarized in Algorithm~\ref{algo:limma_trend_jt}.
\begin{algorithm}[H]
\caption{\limmatrd{} with bivariate nuisances (\jtlitrd{})}
\begin{algorithmic}[1] 
\State Estimate $H$ using $\wh H$ via the NPMLE of~\eqref{eq:joint_npmle}.
\State Compute the p-values $\wh P^\jt_i$ using \eqref{eq:2d_pb_pvalue_general_case} while replacing $H$ by $\wh H$.
\end{algorithmic}
\label{algo:limma_trend_jt}
\end{algorithm}
Our theory crucially builds on the following result. The oracle p-values $\{P^\jt_i\}$ satisfy an Eddington-Tweedie type representation in terms of the marginal densities of the summary statistics. (Analogous results are a cornerstone of the empirical Bayes literature, see e.g.,~\cite{dyson1926method, cressie1982useful, efron2011tweedie, ignatiadis2025empirical}).
\begin{theorem}
\label{thm:tweedie_2_d}
Recall $f_{H,K-p}$ and $f_{H,K-p+1}$ defined in \eqref{eq:2_d_marginal}. There exists $C_{K,p}>0$ such that
\begin{align}
\label{eq:def_tweedie_2d}
P^\jt_i = C_{K,p} \frac{(S^2_i)^{\frac{K-p}{2}-1}}{f_{H,K-p}(S^2_i,A_i)} \int_{\frac{(K-p)S^2_i + (Z^2_i/\nu^2)}{K-p+1}}^{\infty} \frac{(t^2)^{-\frac{K-p-1}{2}}f_{H,K-p+1}(t^2,A_i)}{\sqrt{(K\!-\!p\!+\!1)t^2 - (K\!-\!p)S^2_i}} \,  \, \dd t^2.
\end{align}
\end{theorem}
To prove asymptotic validity of the $\{\wh P^\jt_i\}$ to $\{P^\jt_i\}$s, we make the following assumption.
\begin{assumption}
The tuples $(\mu_i, \sigma^2_i, Z_i, S^2_i,A_i)$ are generated according to~\eqref{eq:def_z_i_dis_2d} and~\eqref{eq:limma_trend_2d_prior} for $K -p\ge 2$ and are jointly independendent across $i \in [n]$.
We use the nonparametric class $\mathcal G_H:= \left\{\rmH: \rmH([-M,M] \times [\underline L,\wb U])=1\right\}$ for absolute constants $M,\underline L,\wb U>0$. It holds that the data-generating $H$ in~\eqref{eq:limma_trend_2d_prior} satisfies $H \in \mathcal G_H$. 
\label{assu:compact_2d}
\end{assumption}
Under the above assumption, the Eddington-Tweedie formula from Theorem~\ref{thm:tweedie_2_d} can be utilized in conjunction with the convergence of $f_{\widehat H,K-p}$ to $f_{H,K-p}$ in Hellinger distance and of $f_{\widehat H,K-p+1}$ to $f_{H,K-p+1}$ in the $L^2$ distance (both shown in the supplement), to obtain the following theorem on the convergence of the estimated p-values to the oracle p-values.
\begin{theorem}
\label{prop:asymp_p_val}
Suppose~Assumption~\ref{assu:compact_2d} hold. Then, for all $\zeta \in \left(\frac{1}{2}, 1\right)$, there exists a constant $C > 0$ (depending only on $\underline{L}, \overline{U}, M, K, p, \nu, \zeta$) such that for all $n \in \mathbb{N}_{\ge 1}$,
\[
\max_{1 \le i \le n} \mathbb{E}\left[
\left| P^\jt_i \wedge \zeta - \wh P^\jt_i \wedge \zeta \right|
\right] \leq C \cdot \frac{\log^{13/4} n}{\sqrt{n}}.
\]
\end{theorem}
The rate is parametric, even though the procedure implicitly learns the conditional prior of $\sigma_i^2 \mid A_i$. 
The following results are established under the same assumptions as Theorem \ref{prop:asymp_p_val}. First, we extend the conditional uniformity of Lemma~\ref{lem:valid_oracle_p_vals} to approximately hold for the estimated p-values. 
\begin{proposition}
\label{prop:valid_p_values}
Suppose Assumption~\ref{assu:compact_2d} hold. Fix $\zeta\in(1/2,1)$. There exists a constant $C>0$, depending only on $\underline L,\overline U,M,K,p,\nu$ and $\zeta$, such that for all $n\ge1$,
\begin{align}
\label{eq:jt_cond_valid_p_val}
\max_{i\in [n]\,:\,\theta_i=0}\mathbb E_H\!\left[\sup_{t \in [0,\zeta]}\left|\mathbb P_H\!\left(\widehat P_i^{\mathrm{jt}} \le t \mid \{(S_j^2,A_j)\}_{j=1}^n\right) - t\right|\right]\le C \frac{\log^{13/4} n}{\sqrt n}.
\end{align}
\end{proposition}

Next, let $\fdr_n^{\jt}$ denote the FDR of the BH procedure at level $\alpha$ applied to $\widehat P_1^{\jt}, \ldots, \widehat P_n^{\jt}$  to
to test $H_i:\theta_i=0$ for $i \in [n]$ with the nominal level of significance $\alpha \in (0,1)$. The following theorem shows asymptotic FDR control under a criticality condition as in Theorem~\ref{thm:final_rate_cs}.
\begin{theorem}
\label{thm:final_rate}
Suppose Assumption~\ref{assu:compact_2d} hold and the oracle p-values $\{P_i^{\jt}\}$ are critically dense at $\alpha$ as defined in Definition~\ref{asm:1d_limma_trend_bh}.  Let $n_0:= \#\{i \in [n]: \theta_i=0\}$.
Then, for $\varrho > 13/4$,
\[
\lim_{n \to \infty}
\left\{
n^{1/2}(\log n)^{-\varrho}
\left(\fdr_n^{\jt}-\frac{n_0}{n}\alpha\right)_+
\right\}
= 0.\]
\end{theorem}

Under slightly stronger assumptions on the number of null hypotheses and the mixing distribution, one can also study the asymptotic power of BH applied to $\widehat P_1^{\jt}, \ldots, \widehat P_n^{\jt}$. We discuss such power properties in Section \bluelink{sec:power_jt_bh}.

\section{Numerical Experiments}
\label{sec:numerical_exp}

We assess our proposals, \reglitrd{} (Algorithm~\ref{algo:limma_trend_reg}) and \jtlitrd{} (Algorithm~\ref{algo:limma_trend_jt}), against seven alternatives: (a)~the $t$-test; (b)~untrended \limma{}~\citep{smyth2004linear}, denoted Untrended-Inv$\chi^2$; (c)~its NPMLE variant~\citep{ignatiadis2025empirical}, Untrended-NPMLE; (d)~parametric \reginvc{} (Algorithm~\ref{algo:limma_trend}); (e)~MAnorm2~\citep{tu2021manorm2}; (f)~MAP~\citep{li2019map}; and (g)~an \texttt{Oracle} that computes partially Bayes p-values using the true data-generating distribution. Supplement~\bluelink{subsec:implementation_baseline} provides more details on the baselines.

All four settings follow the two-sample comparison of Example~\ref{ex:two_sample} with $n=10000$ independent units and $K=K_A+K_B$ samples. Of the $n$ units, $n_0=9000$ are null ($\theta_i=0$); the remaining have \smash{$\theta_i \sim \mathrm{N}(0,\,16\sigma_i^2)$}, where $\sigma_i^2 = \exp\{m(\mu_i)\}\,\tau_i^2$ and \smash{$\tau_i^2 \sim G$}. The choices of $(K_A,K_B)$, trend $m(\cdot)$, residual prior $G$, and the distribution of $\mu_i$ are specified for each setting below. In Settings 1--3, the trended methods use $M_i=A_i$ (and \smash{$M_i=\tilde{A}_i^{\mathrm{MA2}}$} for MAnorm2); in Setting 4, we provide $\mu_i$ directly as external side-information $M_i=\mu_i$. The latter is not feasible in practice ($\mu_i$ is unknown), but we use it to construct an idealized setting in which the trend perfectly determines the variance. P-values are passed through BH at $\alpha=0.05$, and FDR and power are estimated over $100$ Monte Carlo replicates. Results are reported in Table~\ref{tbl:simulation_result}.

\begin{table}[t]
\caption{FDR and power (in \%) at nominal level $\alpha=0.05$ across the four simulation settings of Section~\ref{sec:numerical_exp}. FDR entries exceeding the nominal level are highlighted in \textcolor{red}{red}. Settings 1--3 use $M_i=A_i$ (and \smash{$\tilde{A}_i^{\mathrm{MA2}}$} for MAnorm2); Setting 4 treats $M_i=\mu_i$ as external side-information, where MAnorm2 and \jtlitrd{} do not apply (``--'').}
\label{tbl:simulation_result}
\centering
{\footnotesize
\setlength{\tabcolsep}{4pt}
\renewcommand{\arraystretch}{1.15}
\begin{tabular*}{\textwidth}{@{\extracolsep{\fill}}@{}
p{0.5cm}
*{9}{S[table-format=3.1, table-space-text-post=\,]<{\,\%}}
@{}}
\toprule
& \multicolumn{1}{c}{}  
& \multicolumn{1}{c}{\textbf{Classical}}
& \multicolumn{2}{c}{\textbf{Untrended}}
& \multicolumn{2}{c}{\textbf{Regression}}
& \multicolumn{1}{c}{\textbf{Joint}}
& \multicolumn{1}{c}{}
& \multicolumn{1}{c}{} \\
\cmidrule(lr){3-3}
\cmidrule(lr){4-5}
\cmidrule(lr){6-7}
\cmidrule(lr){8-8}
\textbf{} 
& \multicolumn{1}{c}{\textbf{Oracle}}
& \multicolumn{1}{c}{\textbf{t-test}}
& \multicolumn{1}{c}{\textbf{Inv$\chi^2$}}
& \multicolumn{1}{c}{\textbf{NPMLE}}
& \multicolumn{1}{c}{\textbf{Inv$\chi^2$}}
& \multicolumn{1}{c}{\textbf{NPMLE}}
& \multicolumn{1}{c}{\textbf{NPMLE}}
& \multicolumn{1}{c}{\textbf{MAP}}
& \multicolumn{1}{c}{\textbf{MAnorm2}} \\
\midrule
\simhead{Setting 1 (parametric prior misspecified, no trend)}
FDR   & 4.5  & 4.6  & \redS{5.8}  & 4.5  & \redS{5.8}  & 4.6  & 4.5  & \redS{50.2} & \redS{5.8} \\
Power & 37.8 & 12.9 & 33.5 & 37.9 & 33.0 & 36.3 & 37.8 & 53.5 & 33.3 \\
\addlinespace[0.5em]
\simhead{Setting 2 (parametric prior misspecified, trend)}
FDR   & 4.5  & 4.6  & \redS{5.4}  & 4.5  & \redS{5.8}  & 4.3  & 4.4  &  \redS{51.3} & \redS{5.8} \\
Power & 37.7 & 12.9 & 31.7 & 32.8 & 32.0 & 33.2 & 34.0 & 52.5 & 32.4 \\ \addlinespace[0.5em]
\simhead{Setting 3 (unbalanced design, parametric prior, trend)}
FDR   & 4.6  & 4.5  & 4.5  & 4.4  & 4.0  & 4.2  & 4.3  & \redS{32.5} & \redS{12.6} \\
Power & 52.4 & 46.3 & 47.2 & 47.7 & 51.2 & 51.3 & 51.8 & 62.6 & 44.8 \\
\addlinespace[0.5em]
\simhead{Setting 4 (no residual heterogeneity, external side-information, trend)}
FDR   & 4.4  & 4.3  & 5.0  & 4.4  & 4.3  & 4.3  & \multicolumn{1}{c}{--} & 4.4  & \multicolumn{1}{c}{--} \\
Power & 59.3 & 38.2 & 43.2 & 44.1 & 59.2 & 59.3 & \multicolumn{1}{c}{--} & 59.3 &  \multicolumn{1}{c}{--}\\
\bottomrule
\end{tabular*}
}
\end{table}

Across all four settings, our procedures match or beat the feasible competitors in power while remaining the only methods (besides \texttt{Oracle}) to control FDR throughout. The trended methods (\reglitrd{}, \jtlitrd{}) outperform their untrended counterparts whenever a non-trivial trend is present (Settings 2--4). Additional simulations are reported in Supplement~\bluelink{sec:numeric_additional}.

\begin{itemize}[leftmargin=*,nosep,wide] 
\item \emph{Setting 1 (parametric prior misspecified, no trend).} We take $K_A=K_B=3$, $G=0.5\delta_1+0.5\delta_{10}$, $m\equiv 0$, $\mu_i\sim\mathrm{N}(20,3)$. The parametric prior is misspecified, so Untrended-Inv$\chi^2$, \reginvc{}, and MAnorm2 inflate FDR; MAP inflates much more severely (exceeding $50\%$) because it ignores residual heterogeneity. The NPMLE-based methods control FDR while gaining substantial power over the $t$-test. Since $K_A=K_B$, we have that \smash{$A_i=\tilde A_i^{\mathrm{MA2}}$}, and so MAnorm2 satisfies Assumption~\ref{assu:design}. Its lack of FDR control is due to the parametric model.
\item \emph{Setting 2 (parametric prior misspecified, trend).} Same as Setting 1 but with the choice $m(\mu_i)=-4\,\mathrm{logistic}((\mu_i-16)/4)+12$. The parametric methods continue to inflate FDR. Among feasible methods, \jtlitrd{} attains the highest power.
\item \emph{Setting 3 (unbalanced design, parametric prior, trend).} We take $K_A=2$, $K_B=10$, $G=10\cdot\mathrm{Inv}\chi^2_{10}$, $m(\mu_i)=-6\,\mathrm{logistic}((\mu_i-20)/0.15)$, $\mu_i\sim\mathrm{N}(20,0.2)$. MAnorm2 and MAP inflate FDR substantially, and for different reasons: MAnorm2's \smash{$\tilde A_i^{\mathrm{MA2}}$} violates orthogonality in Assumption~\ref{assu:design}, and MAP's plug-in \smash{$\widehat\sigma_i^2=\widehat\xi^2(A_i)$} ignores residual heterogeneity. 
\item \emph{Setting 4 (no residual heterogeneity, external side-information, trend).} We take $K_A=3$, $K_B=5$, $G=\delta_1$, $m$ as in Setting 3, $\mu_i\sim\mathrm{N}(20,0.2)$, with $M_i=\mu_i$ supplied externally. The trend perfectly determines $\sigma_i^2$, MAP's ideal scenario. MAnorm2 and \jtlitrd{} do not apply. We see that \reglitrd{} matches the power of MAP and \texttt{Oracle} while controlling FDR: the nonparametric residual prior costs essentially nothing, even when not needed.
\end{itemize}

\section{Applications to bulk RNA-seq, ChIP-seq and proteomics}
\label{sec:case_studies}

We demonstrate the methodology on three datasets spanning bulk RNA-seq, ChIP-seq, and quantitative proteomics: the melanoma RNA-seq data of~\citet{goswami2018modulation}, the H3K4me3 ChIP-seq data of~\citet{tu2021manorm2}, and the A431 proteomics data of~\citet{zhu2020deqms}. In each case, we work with standard normalized and transformed measurements for which unit-wise linear modeling and the associated Gaussian approximation are routinely used in practice.\footnote{For the RNA-seq and ChIP-seq datasets, the $Y_{ij}$ are $\log_2$-transformed counts, so Gaussianity cannot hold exactly. Treating them as Gaussian is standard practice in \limma{}, \limmatrd{}, and \texttt{voom}~\citep{law2014voom}; we intend to study the impact of this approximation on downstream inference in future work.} We compare our proposals, \reglitrd{} and \jtlitrd{}, against four benchmarks from Section~\ref{sec:numerical_exp}: the $t$-test, Untrended-Inv$\chi^2$, Untrended-NPMLE, and \reginvc{}.  We omit MAnorm2 and MAP: MAnorm2's restriction to two-sample comparisons makes it inapplicable to two of our three datasets, and MAP's FDR exceeded 30\% throughout our Section~\ref{sec:numerical_exp} simulations whenever residual heterogeneity was present. Table~\ref{tbl:realdata_significance} reports the number of discoveries per dataset and method.

The mean-variance trends for each dataset are visible in Figure~\ref{fig:intro_trend} and motivate adjusting for the trend. For each dataset, we also report diagnostic plots assessing the fit of each method's variance model; these are introduced in detail in Section~\ref{sec:bulk_rna_1} and reused with minor variations in the subsequent subsections.

\begin{table}
\caption{Number of discoveries (BH at target FDR $\alpha$) across the  datasets of Section~\ref{sec:case_studies} with different methods. Our proposals are \reglitrd{} (Regression-NPMLE) and \jtlitrd{} (Joint-NPMLE).}
\label{tbl:realdata_significance}
\centering
{\small
\setlength{\tabcolsep}{4pt}
\renewcommand{\arraystretch}{1.15}
\begin{tabular*}{\textwidth}{@{\extracolsep{\fill}}@{}p{6.1cm}
                            *{6}{S[table-format=4.0]}@{}}
\toprule
& \multicolumn{1}{c}{\textbf{Classical}}
& \multicolumn{2}{c}{\textbf{Untrended}}
& \multicolumn{2}{c}{\textbf{Regression}}
& \multicolumn{1}{c}{\textbf{Joint}} \\
\cmidrule(lr){2-2}\cmidrule(lr){3-4}\cmidrule(lr){5-6}\cmidrule(lr){7-7}
\textbf{Contrast}
& \multicolumn{1}{c}{\textbf{t-test}}
& \multicolumn{1}{c}{\textbf{Inv$\chi^2$}}
& \multicolumn{1}{c}{\textbf{NPMLE}}
& \multicolumn{1}{c}{\textbf{Inv$\chi^2$}}
& \multicolumn{1}{c}{\textbf{NPMLE}}
& \multicolumn{1}{c}{\textbf{NPMLE}} \\
\midrule
\dshead{\textsc{RNA-seq}: CD4$^{+}$ T cells in melanoma {\hypersetup{hidelinks}\citep{goswami2018modulation}}}{$n=15{,}735$ genes, $K=8$, $p=5$, $\alpha=0.05$}
\ct{Post vs Pre Ipilimumab} & 2 & 0 & 0 & 133 & 76 & 106 \\
\addlinespace[0.5em]
\dshead{\textsc{ChIP-seq}: H3K4me3 in lymphoblastoid cells {\hypersetup{hidelinks}\citep{tu2021manorm2}}}{$n=51{,}128$ genomic intervals, $K=6$, $p=2$, $\alpha=0.0001$}
\ct{GM12891 vs GM12892} & 0 & 197 & 2525 & 4638 & 4620 & 5674 \\
\addlinespace[0.5em]
\dshead{\textsc{Proteomics}: miRNA-mimic treatment in A431 cells {\hypersetup{hidelinks}\citep{zhu2020deqms}}}{$n=8{,}633$ proteins, $K=10$, $p=4$, $\alpha=0.05$}
\ct{miR372 vs Ctrl}      & 322  & 503  & 526  & 570  & 629  & 644 \\
\ct{miR519 vs Ctrl}      & 573  & 820  & 863  & 984  & 1029 & 1050 \\
\ct{miR191 vs Ctrl}      & 98   & 160  & 169  & 184  & 198  & 207 \\
\ct{miR372 vs miR519}    & 876  & 1098 & 1122 & 1238 & 1305 & 1333 \\
\ct{miR372 vs miR191}    & 661  & 921  & 943  & 1008 & 1079 & 1113 \\
\ct{miR519 vs miR191}    & 505  & 699  & 732  & 816  & 871  & 907 \\
\bottomrule
\end{tabular*}
}
\end{table}

\subsection{Differentially expressed genes after treatment of Ipilimumab}
\label{sec:bulk_rna_1}

\begin{figure}
\centering
\setlength{\tabcolsep}{0pt}

\begin{tabular}{@{}l@{\hspace{0.03\textwidth}}l@{\hspace{0.03\textwidth}}l@{}}
(a) & (b) & (c) \\[0.3em]
\includegraphics[width=0.31\textwidth]{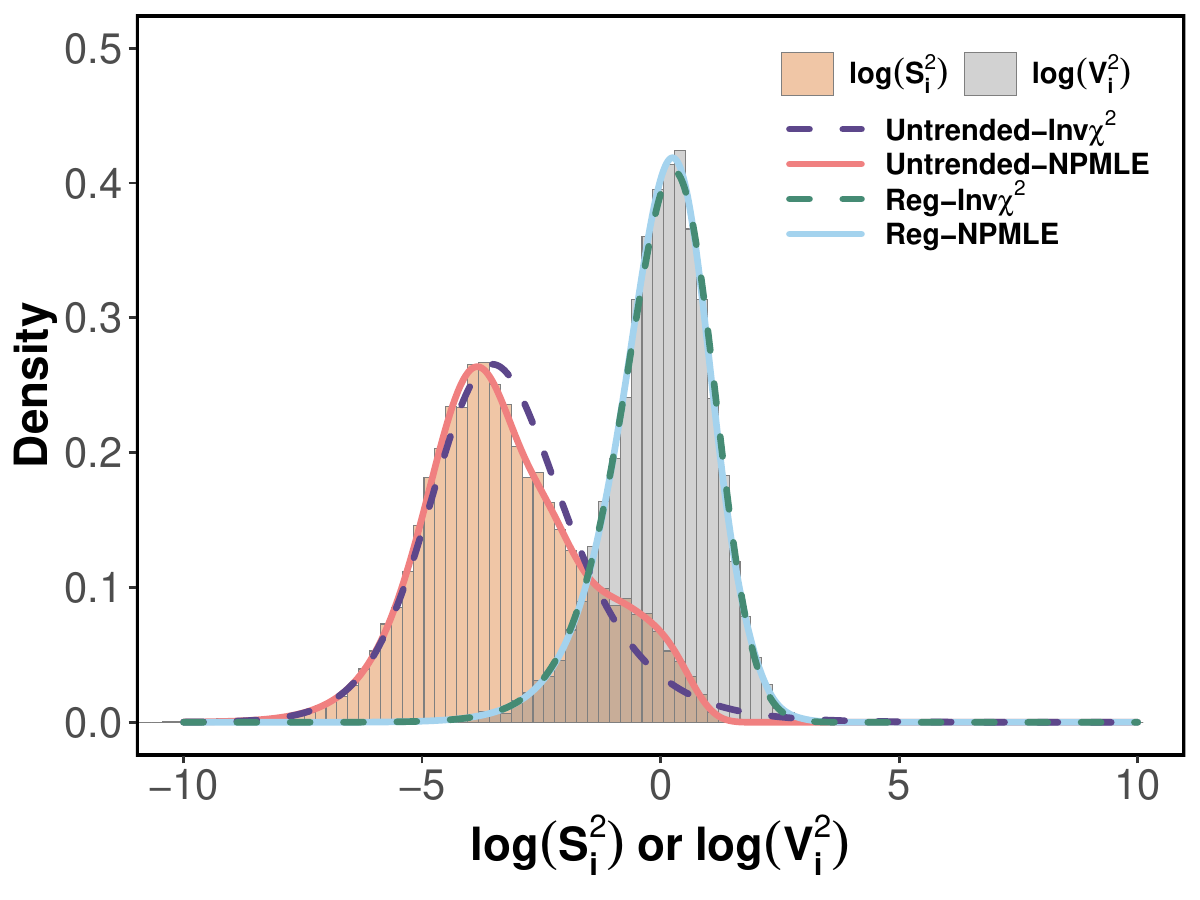} &
\includegraphics[width=0.31\textwidth]{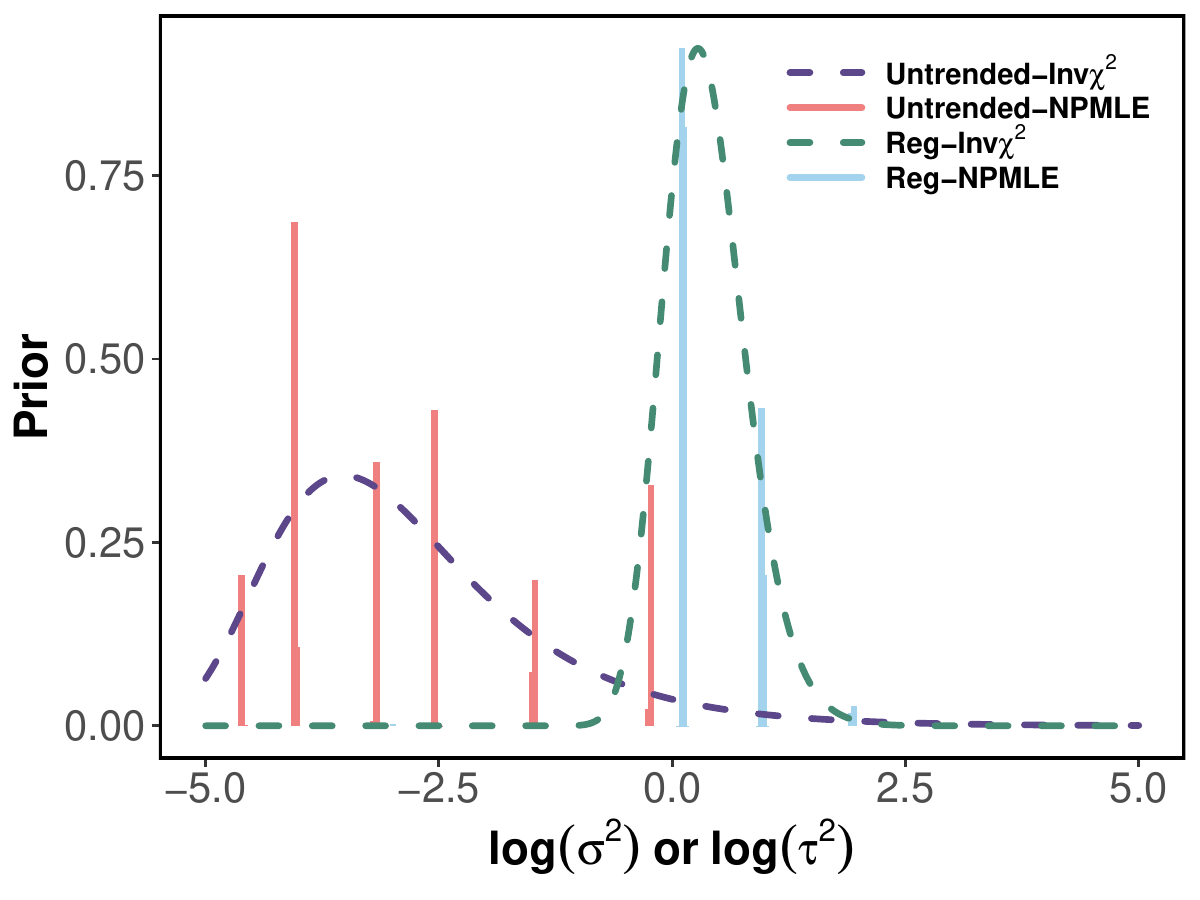} &
\includegraphics[width=0.31\textwidth]{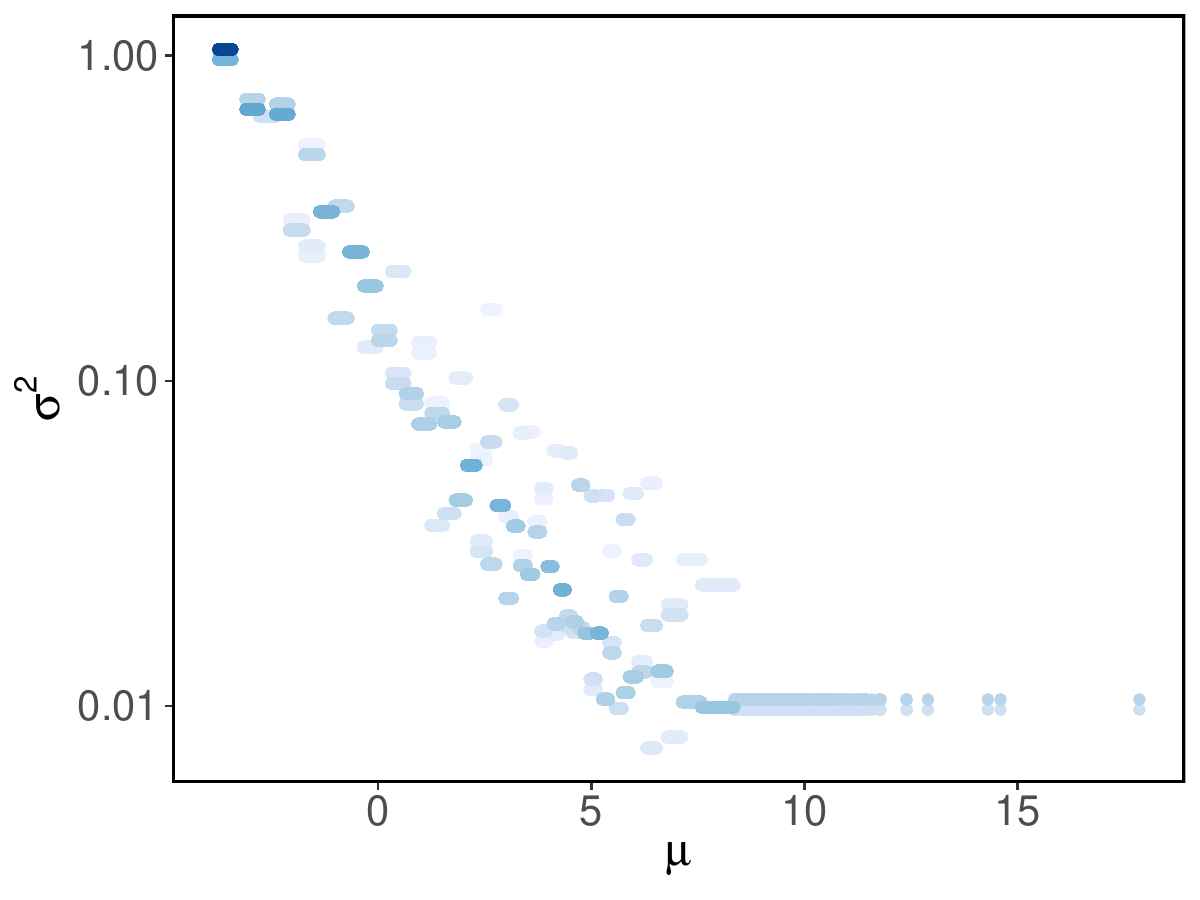}
\end{tabular}

\caption{Trend adjustment substantially concentrates both the marginal of log sample variance (panel (a)) and the variance prior (panel (b)) for the RNA-Seq data of~\cite{goswami2018modulation} ($n=15{,}735$ genes in bulk RNA-seq of CD$4^{+}$ T cells); this reflects that the mean-variance trend of Figure~\ref{fig:intro_trend}(a) absorbs a large part of the variance heterogeneity, driving the power gain of the trended methods (Table~\ref{tbl:realdata_significance}). \textbf{(a)} Histograms of the log sample variances $\log(S_i^2)$ and the trend-adjusted log sample variances $\log(V_i^2) = \log(S_i^2) - \widehat m(A_i)$, overlaid with the corresponding fitted marginal densities. The untrended curves correspond to Untrended-Inv$\chi^2$ (\limma{}) and Untrended-NPMLE; the trend-adjusted curves correspond to \reginvc{} (in \textcolor{ggplotGreen}{green}) and \reglitrd{} (in \textcolor{ggplotBlue}{blue}). \textbf{(b)} Estimated priors on the log-variance component scale: the untrended priors on the $\log(\sigma^2)$ scale, the trend-adjusted priors on the $\log(\tau^2)$ scale. Smooth dashed curves represent the fitted scaled Inv$\chi^2$ densities transformed to the log scale; vertical segments represent the NPMLE support points, with probability masses rescaled for visualization. \textbf{(c)} Estimated bivariate prior \smash{$\widehat H$} from \jtlitrd{}; color intensity is proportional to the probability mass assigned to each support point.}
\label{fig:rnaseq1}
\end{figure}

\citet{goswami2018modulation} profile transcriptomic changes in circulating CD$4^+$ T cells following CTLA-4 blockade, sorting CD$4^+$ T cells from peripheral blood mononuclear cells (PBMCs) of 4 metastatic melanoma patients at baseline and after three doses of ipilimumab ($K=8$, paired pre/post design). After preprocessing (Supplement~\bluelink{sec:preprocessing}), we obtain $n=15{,}735$ genes with a design that includes both treatment and patient group ($p=5$). The primary contrast is pre- vs.\ post-treatment, and Assumption~\ref{assu:design} holds for the average intensity $A_i$. We control FDR at $\alpha=0.05$.

Table~\ref{tbl:realdata_significance} shows that the untrended methods make no discoveries at all on this dataset, the t-test makes 2 discoveries, whereas the trended methods make on the order of 100 each. To check whether each method's variance model fits the data, Figure~\ref{fig:rnaseq1}(a) overlays each method's model-implied marginal density on a histogram of $\log (S_i^2)$ (untrended) or the trend-adjusted $\log (V_i^2) = \log (S_i^2) - \widehat m(A_i)$ (trended). This marginal fit (rather than the prior itself, which is latent) is what our convergence guarantees for \reglitrd{} p-values go through (Remark~\ref{rem:marginal}). In Figure~\ref{fig:rnaseq1}(a), \limma{}'s parametric Inv-$\chi^2$ prior fits the untrended empirical marginal of $\log(S_i^2)$ poorly. The Untrended-NPMLE prior, in contrast, recovers the empirical marginal closely. After trend adjustment, both \reginvc{} and \reglitrd{} fit the empirical marginal of $\log(V_i^2)$.

While untrended-NPMLE fits the distribution of $\log(S_i^2)$ reasonably well, the main advantage of subtracting the fitted trend is that it produces a markedly tighter empirical marginal: $\log(V_i^2)$ is far more concentrated than $\log(S_i^2)$ in Figure~\ref{fig:rnaseq1}(a). This is consistent with the strong mean-variance trend in Figure~\ref{fig:intro_trend}(a) and is the source of the power gain of the trended methods. The corresponding effect on the prior side is visible in Figure~\ref{fig:rnaseq1}(b): the priors on $\tau_i^2$ are visibly tighter than those on $\sigma_i^2$. Panel (c) displays the estimated bivariate prior $\widehat H$ which also shows a strong trend.

An additional RNA-seq dataset is analyzed in Supplement~\bluelink{sec:suppl_case_studies}. There, Untrended-Inv$\chi^2$ makes the most discoveries despite a clearly misspecified parametric prior (Figure~\bluelink{fig:rnaseq2}(a)), casting doubt on those results; after trend adjustment, the parametric assumption becomes reasonable, and \reglitrd{} and \reginvc{} agree.

\subsection{Differential expression in ChIP-seq data from lymphoblastoid cells}

\begin{figure}
\centering
\begin{tabular}{@{}l@{\hspace{0.03\textwidth}}l@{\hspace{0.03\textwidth}}l@{}}
(a) & (b) & (c) \\[0.3em]
\includegraphics[width=0.31\textwidth]{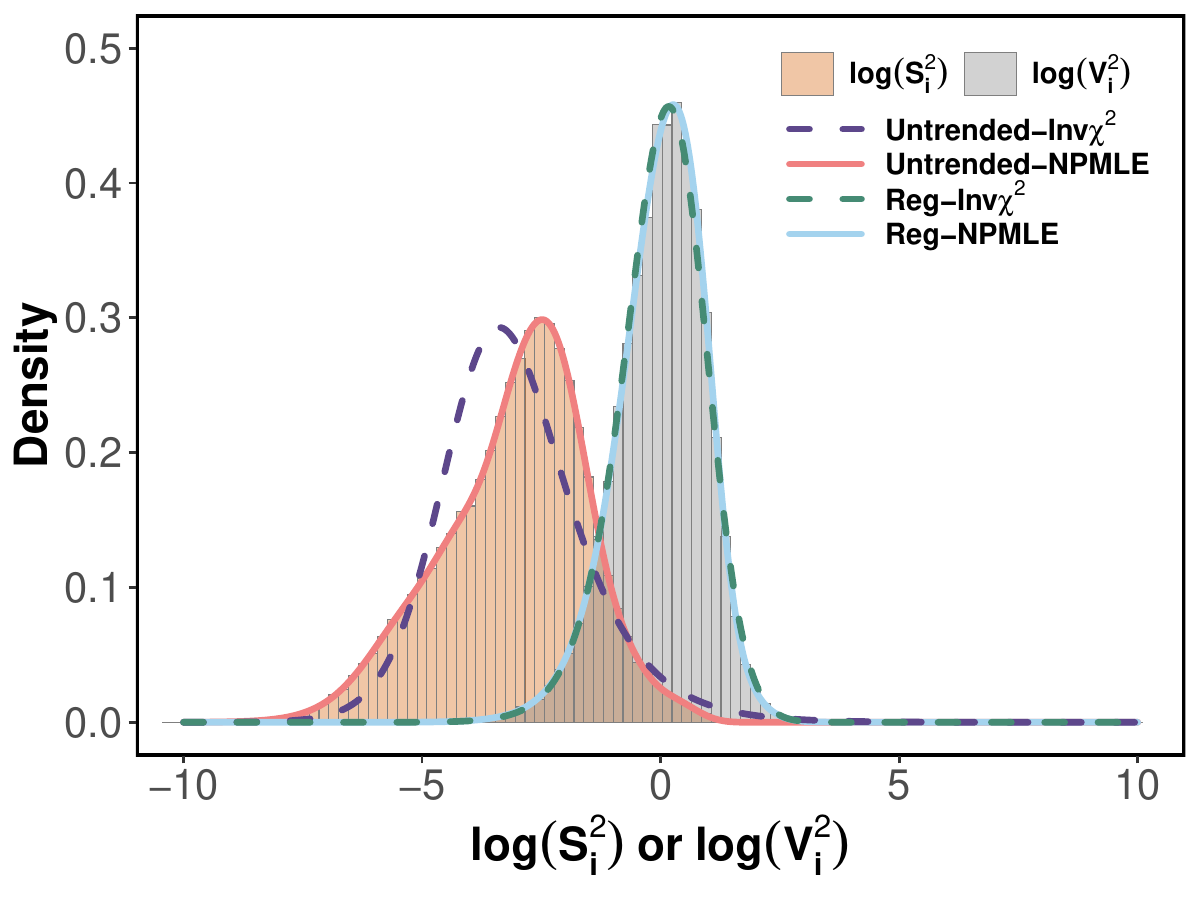} &
\includegraphics[width=0.31\textwidth]{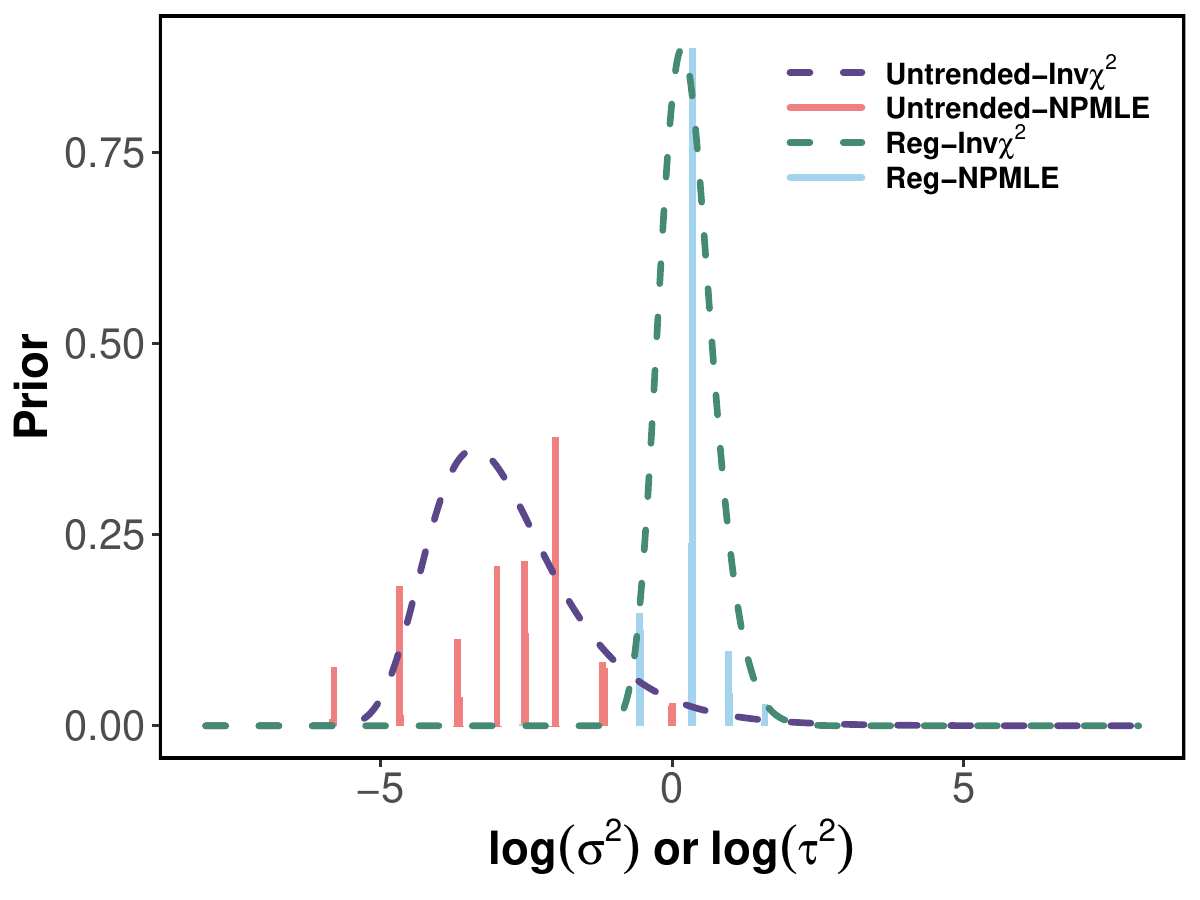} &
\includegraphics[width=0.31\textwidth]{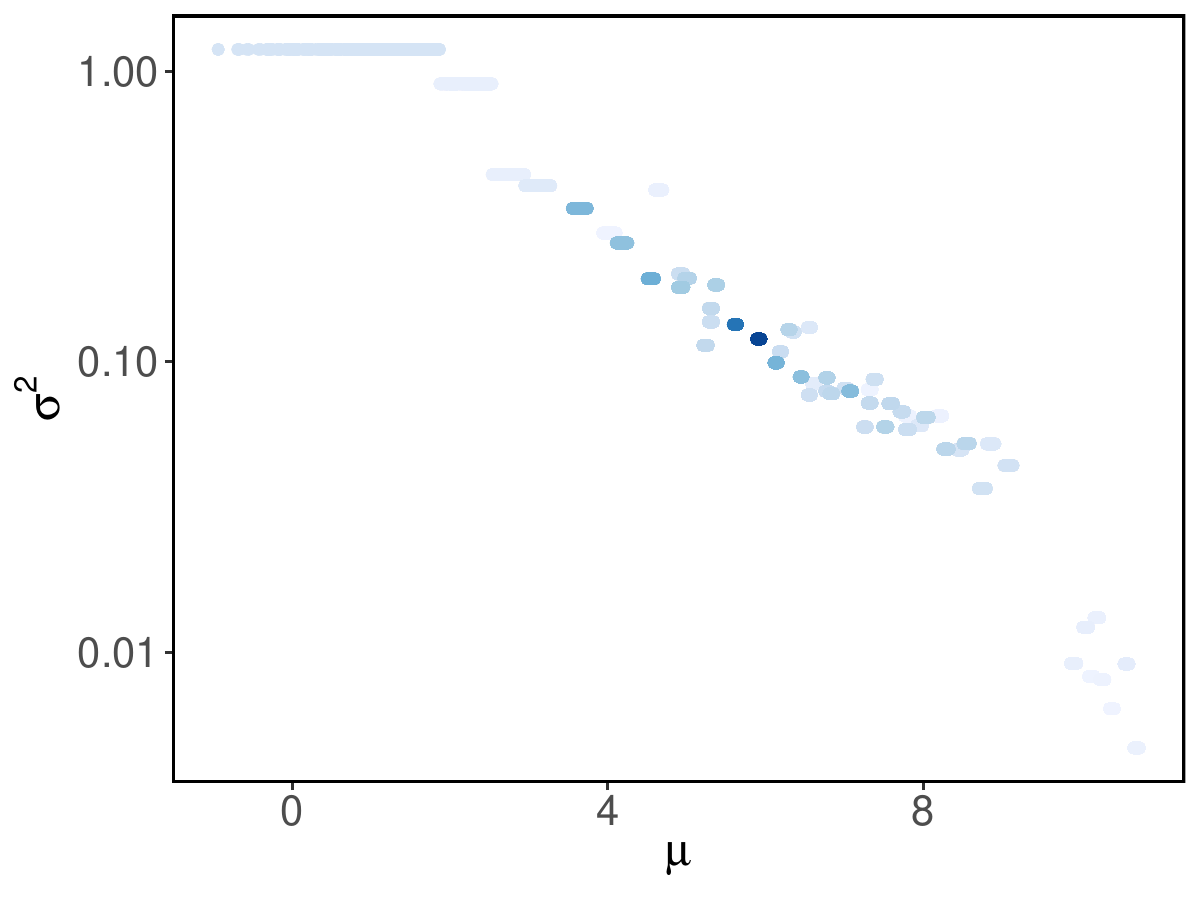} 
\end{tabular}

\caption{H3K4me3 ChIP-seq data of~\cite{tu2021manorm2}: each unit $i$ is a distinct genomic interval ($n=51{,}128$); $K=6$, $p=2$.  The three panels are analogous to Figure~\ref{fig:rnaseq1}.}
\label{fig:chipseq}
\end{figure}

We analyze H3K4me3 ChIP-seq data from~\citet{tu2021manorm2}, comparing GM12891 and GM12892 lymphoblastoid cell lines across $K=6$ samples. ChIP-seq measures the genomic enrichment of a chromatin mark (here, H3K4me3) by counting sequencing reads mapped to each genomic region; $n=51{,}128$ such intervals are observed. The design includes a group indicator ($p=2$) and the primary contrast is between cell lines; Assumption~\ref{assu:design} holds for the average intensity $A_i$, which we use as $M_i$. We use the normalized $\log_2$ read counts from the source paper and apply BH at $\alpha=0.0001$.

The pattern in Table~\ref{tbl:realdata_significance} and Figure~\ref{fig:chipseq} echoes the RNA-seq case but is more pronounced. Untrended-inv$\chi^2$ misspecifies the untrended marginal of $\log(S_i^2)$ (panel (a)), which in this case is reflected in lower power (197 discoveries) versus Untrended-NPMLE (2{,}525).\footnote{We note that in general such misspecification can lead to decreased power or to inflated false discoveries.} Trend adjustment concentrates the marginal (cf. Figure~\ref{fig:chipseq} (a)), driven by the strong trend visible in Figure~\ref{fig:intro_trend}(b), and after trend adjustment both \reginvc{} and \reglitrd{} fit the empirical marginal of $\log(V_i^2)$ closely, making over 4{,}500 discoveries each. \jtlitrd{} makes somewhat more discoveries (5{,}674); its bivariate prior $\widehat H$ is shown in panel (c).

\subsection{Differential protein expression among miRNA-mimic treatment groups}
\label{sec:seq_peptide}

In this application, we analyze a TMT10-plex–labelled quantitative proteomics dataset on A431 human epidermoid carcinoma cells sourced from \cite{zhu2020deqms}. The cells were treated with three different miRNA mimics and compared to a control condition. After peptide-spectrum matching and protein-level filtering, the data are organized as a protein-by-sample abundance matrix, where each entry represents the observed abundance of a protein in a sample. These abundance measurements are continuous-valued and are analyzed on the $\log_2$ scale, with $n=8,633$ proteins and $K=10$ samples. The design contains four groups (ctrl, miR191, miR372, miR519) yielding $p=4$ covariates, and we test all six pairwise contrasts. Unlike previous settings where the average intensity $A_i$ is a natural continuous summary, proteomics variance is strongly driven by quantification depth: proteins supported by more identified spectra or peptides tend to have more stable abundance estimates~\citep{zhu2020deqms}. Hence, we take $M_i$ as the number of peptides. Since $M_i$ is discrete and external, \jtlitrd{} from Algorithm~\ref{algo:limma_trend_jt} does not apply directly; we use the discrete variant from Algorithm~\bluelink{algo:limma_trend_jt_discrete} (Supplement~\bluelink{sec:discrete_MI}), which fits a separate NPMLE within each $M_i$ bin (binning details in Supplement~\bluelink{sec:bins_protein}). To unify notation, we still refer to this variant as \jtlitrd{}. For all these methods, p-values are computed, and BH correction is applied to control FDR at $\alpha = 0.05$.

\begin{figure}
\centering

\begin{tabular}{@{}l@{\hspace{0.03\textwidth}}l@{\hspace{0.03\textwidth}}l@{}}
(a) & (b) & (c) \\[0.3em]
\includegraphics[width=0.31\textwidth]{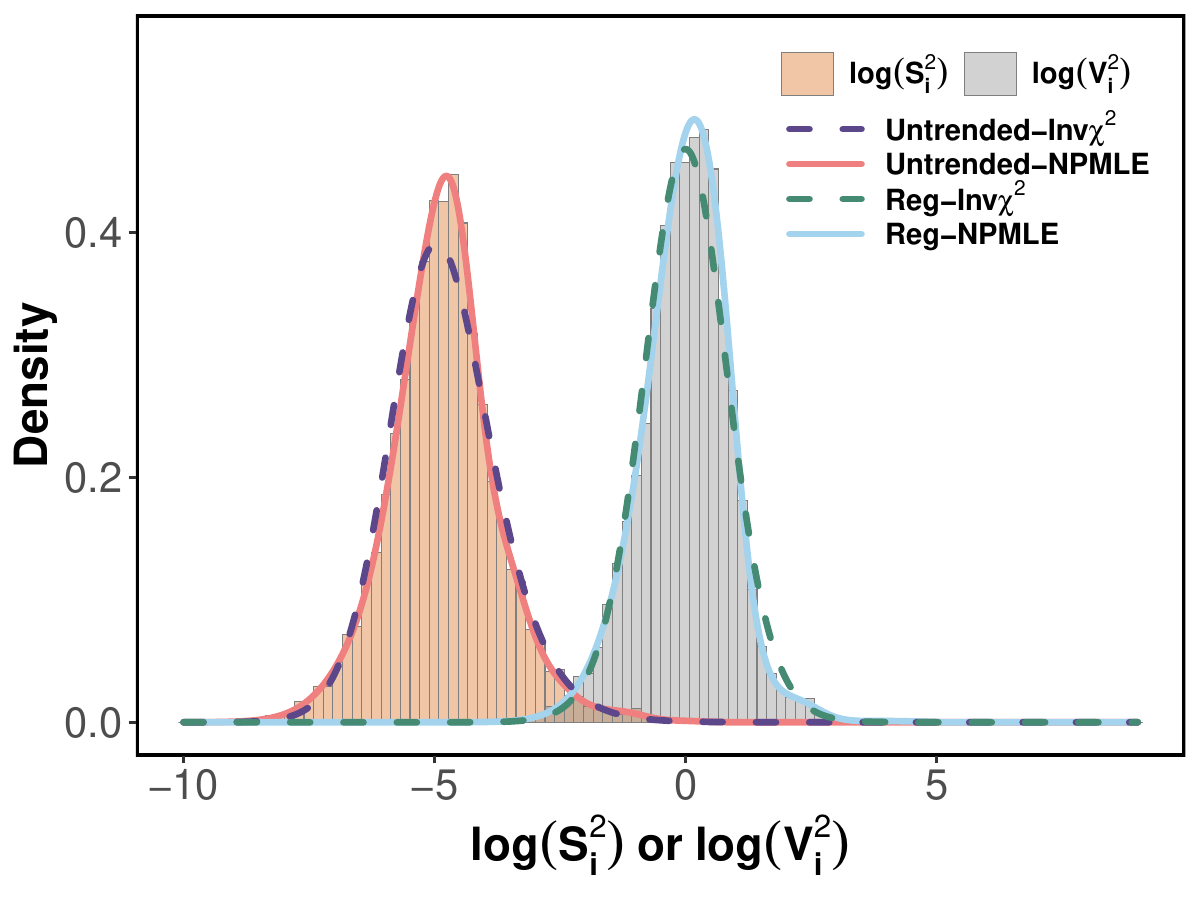} &
\includegraphics[width=0.31\textwidth]{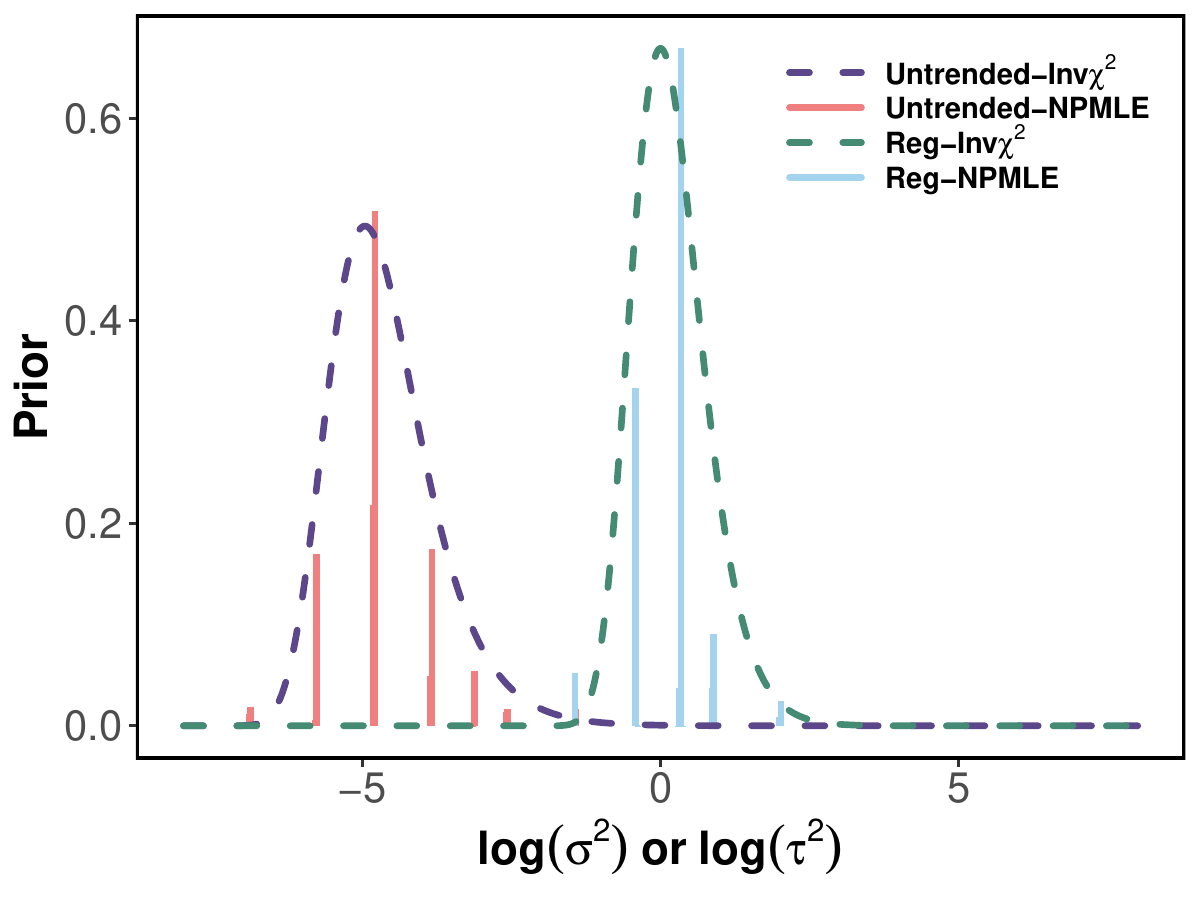} &
\includegraphics[width=0.31\textwidth]{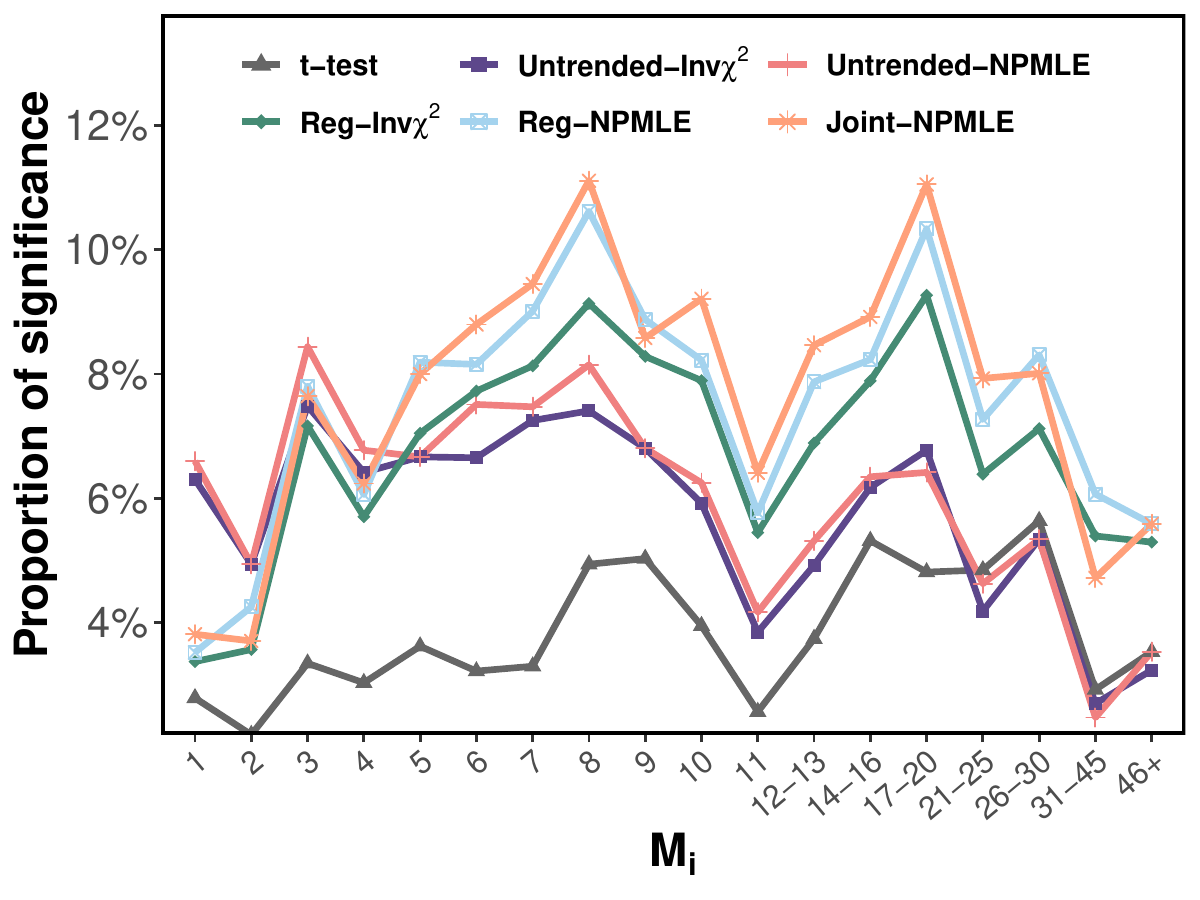}  \\
\multicolumn{3}{@{}l@{}}{(d)} \\
\multicolumn{3}{@{}l@{}}{
  \hspace*{-22pt}\includegraphics[width=1.05\textwidth]{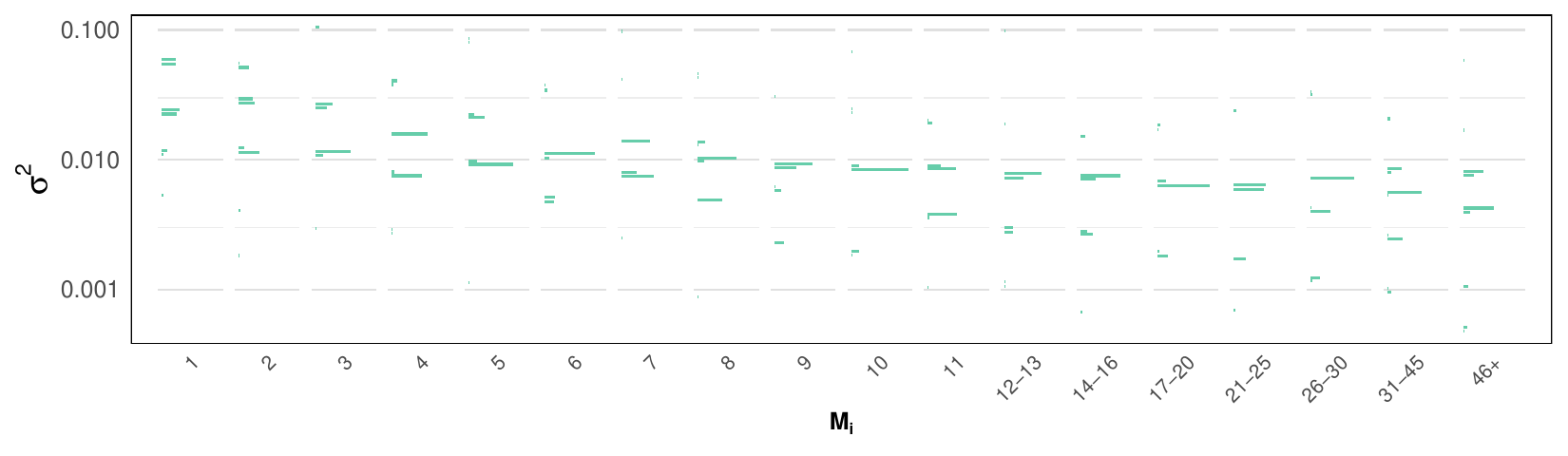}
}
\end{tabular}

\caption{TMT-based quantitative proteomics data of~\cite{zhu2020deqms}: each
unit $i$ corresponds to a distinct protein ($n=8{,}633$) quantified from bulk
TMT reporter-ion intensities, with $K=10$ samples, $p=4$, and $M_i$ equal to
the number of peptides supporting protein $i$. Panels (a) and (b) are
analogous to panels (a) and (b) of Figure~\ref{fig:rnaseq1}, respectively.
Because $M_i$ is discrete and external, \jtlitrd{} is implemented through the
discrete variant of Algorithm~\bluelink{algo:limma_trend_jt_discrete} in the
Supplement, which fits a separate NPMLE within each $M_i$ bin. \textbf{(c)}
Proportion of discoveries within each $M_i$ bin for the miR372 vs.\ control contrast, for each of the six methods in Table~\ref{tbl:realdata_significance}.
\textbf{(d)} Bin-wise estimated priors on $\sigma^2$ produced by discrete \jtlitrd{}, one panel per $M_i$ bin (binning details in Supplement~\bluelink{sec:bins_protein}). 
The priors visibly shift and reshape across bins; evidence of the heterogeneity in $\sigma_i^2 \mid M_i$ that motivates the discrete variant over \reglitrd{}.
}
\label{fig:proteomics}
\end{figure}

Table~\ref{tbl:realdata_significance} shows a consistent ordering across all
six pairwise contrasts: trended methods uniformly outperform their untrended
counterparts, and discrete \jtlitrd{} attains the most discoveries throughout.
Unlike the RNA-seq and ChIP-seq applications, here the untrended NPMLE makes
only marginally more discoveries than Untrended-Inv$\chi^2$, consistent with
Figure~\ref{fig:proteomics}(a), where the parametric Inv$\chi^2$ prior already
fits the empirical marginal of $\log(S_i^2)$ reasonably well. After trend
adjustment, both \reginvc{} and \reglitrd{} continue to fit the empirical
marginal of $\log(V_i^2)$ closely. The corresponding estimated priors are
shown in Figure~\ref{fig:proteomics}(b).

To interpret these results, we return to the mean-variance diagnostic in
Figure~\ref{fig:intro_trend}(c): $S_i^2$ is systematically larger for proteins
quantified from few peptides (small $M_i$). This was precisely the concern
of~\citet{zhu2020deqms}, who worried that methods ignoring the trend might
inflate false positives among low-peptide-count proteins.
Figure~\ref{fig:proteomics}(c) supports this concern: the untrended partially
Bayes methods make a disproportionately larger fraction of their discoveries
at small $M_i$, a pattern that the trended methods do not exhibit. The
well-specified model in~\eqref{eq:well-specified-limma-reg} does not appear to
hold for this dataset, however: proteins with smaller $M_i$ exhibit not only a
larger trend but also greater residual variability, as we document in
Figure~\bluelink{fig:proteomics_dignostic} of the Supplement. Under this
misspecification, type-I error guarantees for \reglitrd{} follow from
Section~\ref{sec:reg_ltrd_int}, but calibration conditional on $M_i$ is no
longer guaranteed. Discrete \jtlitrd{} addresses this by allowing the full
conditional distribution of $\sigma_i^2 \mid M_i$ to vary across $M_i$ bins.
Figure~\ref{fig:proteomics}(d) shows that the bin-wise estimated priors differ across $M_i$ levels in both location and shape, indicating
heterogeneity that a single trend with a shared residual prior cannot absorb.
Consistent with this, Figure~\bluelink{fig:proteomics_dignostic} in the
Supplement shows that the discrete \jtlitrd{} fit aligns more closely with
the data than \reglitrd{}.

\section{Frequentist guarantees}
\label{sec:compound_decision_theory}

We have throughout invoked a partially Bayes specification that treats the primary parameters $\theta_i$ as deterministic (as in a frequentist analysis), while the nuisance parameters follow an unknown distribution. For instance, for \jtlitrd{}, we assumed that 
\smash{$(\mu_i,\sigma_i^2) \sim H$} in~\eqref{eq:limma_trend_2d_prior}. We close by considering whether our results hold up when all parameters (including the nuisance parameters $\mu_i, \sigma_i^2$) are treated as deterministic. 
We conduct such a frequentist analysis in Supplement~\bluelink{sec:cmp_bayes_ltrd} and refer to~\cite{ignatiadis2025empirical} for more intuition on this setting, building on compound decision theory~\citep{zhang2003compound}. We show that FDR control is robust to the fully frequentist setting, while the conditional interpretation relies on the validity of the partially Bayes specification. More precisely:
\begin{itemize}[leftmargin=*,nosep, wide]
    \item Asymptotic FDR control (under critically dense signals) continues to hold when p-values are computed according to either Algorithm~\ref{algo:limma_trend_reg} or Algorithm~\ref{algo:limma_trend_jt}. However, FDR is controlled at level $\alpha$ and not at $\alpha n_0/n$. This means that one is not allowed to replace BH by a null-proportion adaptive method such as Storey's procedure~\citep{storey_2004}. 
    \item Under the partially Bayes framework, our proposed p-values were shown to be asymptotically uniform conditional on $(M_i, S_i^2)$, or on $V_i^2:= S_i^2/\xi^2_{\text{mis}}(M_i)$ in the misspecified setting. In the frequentist setting, the p-values do not satisfy such conditional uniformity, nor do they satisfy unconditional uniformity. Instead, they satisfy a weaker condition referred to as asymptotic compound p-values in~\citet{ignatiadis_ramdas} (also see~\citet{armstrong2022false}), that is, $\limsup_{n \to \infty} n^{-1} \sum_{i \in [n]: \theta_i=0} \mathbb P\big[\widehat{P}_i \leq t\big] \leq t$ for all $t\in [0,1]$, where \smash{$\widehat{P}_i$} can be either \smash{$\widehat{P}_i^{\trd}$} or  \smash{$\widehat{P}_i^{\jt}$}. 
\end{itemize}

\paragraph{Reproducibility.} All results in this paper are fully third-party reproducible with code under the following Github repository: \url{https://github.com/wyling01/limma-trend-partially-bayes}.

\paragraph{Acknowledgments.}
We thank Jeungju Kim and Jierui Zhu for helpful feedback on an earlier version of this manuscript.
N.I.\ gratefully acknowledges support from NSF (DMS 2443410).

\bibliography{limma_trend}

@article{weber2019diffcyt,
  title={diffcyt: Differential discovery in high-dimensional cytometry via high-resolution clustering},
  author={Weber, Lukas M and Nowicka, Malgorzata and Soneson, Charlotte and Robinson, Mark D},
  journal={Communications biology},
  volume={2},
  number={1},
  pages={183},
  year={2019},
  publisher={Nature Publishing Group UK London}
}

@article{goswami2018modulation,
  title={Modulation of EZH2 expression in T cells improves efficacy of anti--CTLA-4 therapy},
  author={Goswami, Sangeeta and Apostolou, Irina and Zhang, Jan and Skepner, Jill and Anandhan, Swetha and Zhang, Xuejun and Xiong, Liangwen and Trojer, Patrick and Aparicio, Ana and Subudhi, Sumit K and others},
  journal={The Journal of clinical investigation},
  volume={128},
  number={9},
  pages={3813--3818},
  year={2018},
  publisher={American Society for Clinical Investigation}
}

@article{tonkin2018plasmodium,
  title={The Plasmodium falciparum transcriptome in severe malaria reveals altered expression of genes involved in important processes including surface antigen--encoding var genes},
  author={Tonkin-Hill, Gerry Q and Trianty, Leily and Noviyanti, Rintis and Nguyen, Hanh HT and Sebayang, Boni F and Lampah, Daniel A and Marfurt, Jutta and Cobbold, Simon A and Rambhatla, Janavi S and McConville, Malcolm J and others},
  journal={PLoS Biology},
  volume={16},
  number={3},
  pages={e2004328},
  year={2018},
  publisher={Public Library of Science San Francisco, CA USA}
}

@article{tu2021manorm2,
  title={MAnorm2 for quantitatively comparing groups of ChIP-seq samples},
  author={Tu, Shiqi and Li, Mushan and Chen, Haojie and Tan, Fengxiang and Xu, Jian and Waxman, David J and Zhang, Yijing and Shao, Zhen},
  journal={Genome research},
  volume={31},
  number={1},
  pages={131--145},
  year={2021},
  publisher={Cold Spring Harbor Lab}
}

@article{zhu2020deqms,
  title={DEqMS: a method for accurate variance estimation in differential protein expression analysis},
  author={Zhu, Yafeng and Orre, Lukas M and Tran, Yan Zhou and Mermelekas, Georgios and Johansson, Henrik J and Malyutina, Alina and Anders, Simon and Lehti{\"o}, Janne},
  journal={Molecular \& Cellular Proteomics},
  volume={19},
  number={6},
  pages={1047--1057},
  year={2020},
  publisher={Elsevier}
}

@book{bracewell1999fourier,
  title     = {The Fourier Transform and Its Applications},
  author    = {Bracewell, Ronald N.},
  year      = {1999},
  edition   = {3rd},
  publisher = {McGraw-Hill},
  address   = {New York}
}

@article{ignatiadis_ramdas,
  title={Asymptotic and compound e-values: multiple testing and empirical {B}ayes},
  author={Ignatiadis, Nikolaos and Wang, Ruodu and Ramdas, Aaditya},
  journal={arXiv preprint arXiv:2409.19812},
  year={2024}
}

@article{zhang2003compound,
  title={Compound decision theory and empirical {B}ayes methods},
  author={Zhang, Cun-Hui},
  journal={Annals of Statistics},
  volume  = {31},
  number  = {2},
  pages={379--390},
  year={2003},
}

@article{armstrong2022false,
  title={False discovery rate adjustments for average significance level controlling tests},
  author={Armstrong, Timothy B},
  journal={arXiv preprint arXiv:2209.13686},
  year={2022}
}

@book{vershynin2025high,
  title={High-dimensional probability},
  author={Vershynin, Roman},
  year={2025},
  publisher={Cambridge University Press}
}

@article{DuZhang2014SIM,
  author  = {Du, Lilun and Zhang, Chunming},
  title   = {Single-Index Modulated Multiple Testing},
  journal = {The Annals of Statistics},
  year    = {2014},
  volume  = {42},
  number  = {4},
  pages   = {1262-1311},
}

@article{ferreira_zwinderman,
author = {J. A. Ferreira and A. H. Zwinderman},
title = {{On the Benjamini–Hochberg method}},
volume = {34},
journal = {The Annals of Statistics},
number = {4},
publisher = {Institute of Mathematical Statistics},
pages = {1827 -- 1849},
year = {2006},
}

@article{roquain_verzelen,
author = {Etienne Roquain and Nicolas Verzelen},
title = {{False discovery rate control with unknown null distribution: Is it possible to mimic the oracle?}},
volume = {50},
journal = {The Annals of Statistics},
number = {2},
publisher = {Institute of Mathematical Statistics},
pages = {1095 -- 1123},
year = {2022},
}

@article{storey_2004,
author = {Storey, John D. and Taylor, Jonathan E. and Siegmund, David},
title = {Strong control, conservative point estimation and simultaneous conservative consistency of false discovery rates: a unified approach},
journal = {Journal of the Royal Statistical Society: Series B},
volume = {66},
number = {1},
pages = {187-205},
year = {2004}
}

@book{billingsley1995probability,
  title={Probability and Measure},
  author={Billingsley, Patrick},
  year={1995},
  edition={3rd},
  publisher={John Wiley \& Sons},
  address={New York},
}

@book{10.1093/acprof:oso/9780199535255.001.0001,
    author = {Boucheron, Stéphane and Lugosi, Gábor and Massart, Pascal},
    title = {Concentration Inequalities: A Nonasymptotic Theory of Independence},
    publisher = {Oxford University Press},
    year = {2013},
}

@article{saha_guntu,
author = {Sujayam Saha and Adityanand Guntuboyina},
title = {{On the nonparametric maximum likelihood estimator for Gaussian location mixture densities with application to Gaussian denoising}},
volume = {48},
journal = {The Annals of Statistics},
number = {2},
publisher = {Institute of Mathematical Statistics},
pages = {738 -- 762},
year = {2020},
}

@article{keifer_npmle,
author = {J. Kiefer and J. Wolfowitz},
title = {{Consistency of the Maximum Likelihood Estimator in the Presence of Infinitely Many Incidental Parameters}},
volume = {27},
journal = {The Annals of Mathematical Statistics},
number = {4},
publisher = {Institute of Mathematical Statistics},
pages = {887 -- 906},
year = {1956},
}

@article{cressie1982useful,
  author = {Noel Cressie},
title = {{A Useful Empirical Bayes Identity}},
volume = {10},
journal = {The Annals of Statistics},
number = {2},
publisher = {Institute of Mathematical Statistics},
pages = {625 -- 629},
year = {1982},
}

@article{dyson1926method,
  author  = {Dyson, Frank W.},
  title   = {A Method for Correcting Series of Parallax Observations},
  journal = {Monthly Notices of the Royal Astronomical Society},
  volume  = {86},
  pages   = {686--706},
  year    = {1926},
}

@article{efron2011tweedie,
  title     = {Tweedie's Formula and Selection Bias},
  author    = {Efron, Bradley},
  journal   = {Journal of the American Statistical Association},
  volume    = {106},
  number    = {496},
  pages     = {1602--1614},
  year      = {2011},
  publisher = {Taylor \& Francis},
}

@article{law2014voom,
  title = {Voom: Precision Weights Unlock Linear Model Analysis Tools for {{RNA-seq}} Read Counts},
  shorttitle = {Voom},
  author = {Law, Charity W and Chen, Yunshun and Shi, Wei and Smyth, Gordon K},
  year = {2014},
  journal = {Genome Biology},
  volume = {15},
  number = {2},
  pages = {R29},
}

@article{GenoveseWasserman2002FDR,
  author  = {Genovese, Christopher and Wasserman, Larry},
  title   = {Operating characteristics and extensions of the false discovery rate procedure},
  journal = {Journal of the Royal Statistical Society: Series B},
  volume  = {64},
  number  = {3},
  pages   = {499--517},
  year    = {2002},
}

@article{chen2024empiricalbayesestimationprecision,
  title = {Empirical {{Bayes}} When Estimation Precision Predicts Parameters},
  author = {Chen, Jiafeng},
  year = {2026},
  journal = {Econometrica},
  volume = {94},
  number = {2},
  pages = {305--340},
}

@article{Zhang2009GMLE,
  author  = {Cun-Hui Zhang},
  title   = {Generalized Maximum Likelihood Estimation of Normal Mixture Densities},
  journal = {Statistica Sinica},
  year    = {2009},
  volume  = {19},
  number  = {3},
  pages   = {1297--1318}
}

@article{JiangZhang2009,
  author    = {Wenhua Jiang and Cun-Hui Zhang},
  title     = {General Maximum Likelihood Empirical Bayes Estimation of Normal Means},
  journal   = {The Annals of Statistics},
  volume    = {37},
  number    = {4},
  pages     = {1647--1684},
  year      = {2009},
  publisher = {Institute of Mathematical Statistics},
}

@article{sartor2006intensitybased,
  title = {Intensity-Based Hierarchical {{Bayes}} Method Improves Testing for Differentially Expressed Genes in Microarray Experiments},
  author = {Sartor, Maureen A and Tomlinson, Craig R and Wesselkamper, Scott C and Sivaganesan, Siva and Leikauf, George D and Medvedovic, Mario},
  year = {2006},
  journal = {BMC Bioinformatics},
  volume = {7},
  number = {1},
  pages = {538},
}

@article{messner2023proteomic,
  title = {The Proteomic Landscape of Genome-Wide Genetic Perturbations},
  author = {Messner, Christoph B. and Demichev, Vadim and Muenzner, Julia and Aulakh, Simran K. and Barthel, Natalie and R{\"o}hl, Annika and {Herrera-Dom{\'i}nguez}, Luc{\'i}a and Egger, Anna-Sophia and Kamrad, Stephan and Hou, Jing and Tan, Guihong and Lemke, Oliver and Calvani, Enrica and Szyrwiel, Lukasz and M{\"u}lleder, Michael and Lilley, Kathryn S. and Boone, Charles and Kustatscher, Georg and Ralser, Markus},
  year = {2023},
  journal = {Cell},
  volume = {186},
  number = {9},
  pages = {2018-2034.e21},
}

@article{phipson2014diffvar,
  title = {{{DiffVar}}: A New Method for Detecting Differential Variability with Application to Methylation in Cancer and Aging},
  shorttitle = {{{DiffVar}}},
  author = {Phipson, Belinda and Oshlack, Alicia},
  year = {2014},
  journal = {Genome Biology},
  volume = {15},
  number = {9},
  pages = {465},
}

@article{smyth2004linear,
  title = {Linear Models and Empirical {{Bayes}} Methods for Assessing Differential Expression in Microarray Experiments},
  author = {Smyth, Gordon K},
  journal = {Statistical Applications in Genetics and Molecular Biology},
  year    = {2004},
  volume  = {3},
  number  = {1},
  pages   = {Article 3},
}

@article{ignatiadis2025empirical,
  title = {Empirical Partially {{Bayes}} Multiple Testing and Compound {$\chi^2$} Decisions},
  author = {Ignatiadis, Nikolaos and Sen, Bodhisattva},
  year = {2025},
  journal = {The Annals of Statistics},
  volume = {53},
  number = {1},
  pages = {1--36},
}

@article{benjamini1995controlling,
  title = {Controlling the False Discovery Rate: A Practical and Powerful Approach to Multiple Testing},
  shorttitle = {Controlling the {{False Discovery Rate}}},
  author = {Benjamini, Yoav and Hochberg, Yosef},
  year = {1995},
  journal = {Journal of the Royal Statistical Society: Series B},
  volume = {57},
  number = {1},
  pages = {289--300},
}

@article{fan2010nonparametric,
  title = {Nonparametric Estimation of Genewise Variance for Microarray Data},
  author = {Fan, Jianqing and Feng, Yang and Niu, Yue S.},
  year = {2010},
  journal = {The Annals of Statistics},
  volume = {38},
  number = {5},
  pages = {2723--2750},
}

@article{carroll2008nonparametric,
  title = {Nonparametric Variance Estimation in the Analysis of Microarray Data: A Measurement Error Approach},
  shorttitle = {Nonparametric Variance Estimation in the Analysis of Microarray Data},
  author = {Carroll, R. J. and Wang, Y.},
  year = {2008},
  journal = {Biometrika},
  volume = {95},
  number = {2},
  pages = {437--449},
}

@article{wang2009variance,
  title = {Variance Estimation in the Analysis of Microarray Data},
  author = {Wang, Yuedong and Ma, Yanyuan and Carroll, Raymond J.},
  year = {2009},
  journal = {Journal of the Royal Statistical Society: Series B},
  volume = {71},
  number = {2},
  pages = {425--445},
}

@article{li2024robust,
  author = {Li, Mushan and Ma, Yanyuan},
title = {Robust estimation of mean-variance relation},
journal = {Statistics in Medicine},
volume = {43},
number = {2},
pages = {419-434},
year = {2024},
}

@article{mandel2013variance,
  title = {Variance Function Estimation in Quantitative Mass Spectrometry with Application to {{iTRAQ}} Labeling},
  author = {Mandel, Micha and Askenazi, Manor and Zhang, Yi and Marto, Jarrod A.},
  year = {2013},
  journal = {The Annals of Applied Statistics},
  volume = {7},
  number = {1},
  pages = {1--24},
}

@article{love2014moderated,
  title = {Moderated Estimation of Fold Change and Dispersion for {{RNA-seq}} Data with {{DESeq2}}},
  author = {Love, Michael I and Huber, Wolfgang and Anders, Simon},
  year = {2014},
  journal = {Genome Biology},
  volume = {15},
  number = {12},
  pages = {550},
}

@article{mccarthy2012differential,
  title = {Differential Expression Analysis of Multifactor {{RNA-Seq}} Experiments with Respect to Biological Variation},
  author = {McCarthy, Davis J. and Chen, Yunshun and Smyth, Gordon K.},
  year = {2012},
  journal = {Nucleic Acids Research},
  volume = {40},
  number = {10},
  pages = {4288--4297},
}

@article{townsend2025establishing,
  title = {Establishing an \emph{Ex Vivo} Porcine Skin Model to Investigate the Effects of Broad-Spectrum Antiseptic on Viable Skin Microbial Communities},
  author = {Townsend, E. C. and Xu, K. and De La Cruz, K. and Huang, L. and Sandstrom, S. and Meudt, J. J. and Shanmuganayagam, D. and Huttenlocher, A. and Gibson, A. L. F. and Kalan, L. R.},
  editor = {Young, Vincent B.},
  year = {2025},
  journal = {mSphere},
  volume = {10},
  number = {9},
  pages = {e00441-25},
}

@article{ritchie2015limma,
  title = {Limma Powers Differential Expression Analyses for {{RNA-sequencing}} and Microarray Studies},
  author = {Ritchie, Matthew E. and Phipson, Belinda and Wu, Di and Hu, Yifang and Law, Charity W. and Shi, Wei and Smyth, Gordon K.},
  year = {2015},
  journal = {Nucleic Acids Research},
  volume = {43},
  number = {7},
  pages = {e47},
}

@article{li2019map,
  title = {{{MAP}}: Model-Based Analysis of Proteomic Data to Detect Proteins with Significant Abundance Changes},
  shorttitle = {{{MAP}}},
  author = {Li, Mushan and Tu, Shiqi and Li, Zijia and Tan, Fengxiang and Liu, Jian and Wang, Qian and Zhang, Yuannyu and Xu, Jian and Zhang, Yijing and Zhou, Feng and Shao, Zhen},
  year = {2019},
  journal = {Cell Discovery},
  volume = {5},
  number = {1},
  pages = {40},
}

@article{cox1975note,
  title = {A Note on Partially {{Bayes}} Inference and the Linear Model},
  author = {Cox, D. R.},
  year = {1975},
  journal = {Biometrika},
  volume = {62},
  number = {3},
  pages = {651--654},
}

@manual{aps2020mosek,
  type = {Manual},
  title = {The {{MOSEK}} Optimization Suite Manual, Version 11.0},
  author = {{MOSEK ApS}},
  year = {2024},
}

@article{koenker2014convex,
  title = {Convex Optimization, Shape Constraints, Compound Decisions, and Empirical {{Bayes}} Rules},
  author = {Koenker, Roger and Mizera, Ivan},
  year = {2014},
  journal = {Journal of the American Statistical Association},
  volume = {109},
  number = {506},
  pages = {674--685}
}

@article{dicker2016highdimensional,
  title = {High-Dimensional Classification via Nonparametric Empirical {{Bayes}} and Maximum Likelihood Inference},
  author = {Dicker, Lee H. and Zhao, Sihai D.},
  year = {2016},
  journal = {Biometrika},
  volume = {103},
  number = {1},
  pages = {21--34}
}

@article{soloff2024multivariate,
  author  = {Soloff, Jake A. and Guntuboyina, Adityanand and Sen, Bodhisattva},
  title   = {Multivariate, heteroscedastic empirical {B}ayes via nonparametric maximum likelihood},
  journal = {Journal of the Royal Statistical Society: Series B},
  year    = {2025},
  volume  = {87},
  number  = {1},
  pages   = {1--32},
}

@article{fayiii1979estimates,
  title = {Estimates of Income for Small Places: An Application of {{James-Stein}} Procedures to Census Data},
  author = {Fay III, Robert E and Herriot, Roger A},
  year = {1979},
  journal = {Journal of the American Statistical Association},
  volume = {74},
  number = {366a},
  pages = {269--277},
}

@inproceedings{ignatiadis2019covariatepowered,
  title = {Covariate-Powered Empirical {{Bayes}} Estimation},
  booktitle = {Advances in Neural Information Processing Systems},
  author = {Ignatiadis, Nikolaos and Wager, Stefan},
  year = {2019},
  volume = {32},
}

@techreport{mccullagh1990note,
  title = {A Note on Partially {{Bayes}} Inference for Generalized Linear Models},
  author = {McCullagh, Peter},
  year = {1990},
  number = {284},
  address = {Chicago, Illinois, USA},
  institution = {Department of Statistics, University of Chicago}
}

@article{phipson2016robust,
  title = {Robust Hyperparameter Estimation Protects against Hypervariable Genes and Improves Power to Detect Differential Expression},
  author = {Phipson, Belinda and Lee, Stanley and Majewski, Ian J. and Alexander, Warren S. and Smyth, Gordon K.},
  year = {2016},
  journal = {The Annals of Applied Statistics},
  volume = {10},
  number = {2},
  pages = {946--963},
}

@article{lu2016variance,
  title = {Variance Adaptive Shrinkage ({\emph{Vash}}): Flexible Empirical {{Bayes}} Estimation of Variances},
  author = {Lu, Mengyin and Stephens, Matthew},
  journal = {Bioinformatics},
  year    = {2016},
  volume  = {32},
  number  = {22},
  pages   = {3428--3434},
}

@article{chen2025edger,
  title = {{{edgeR}} v4: Powerful Differential Analysis of Sequencing Data with Expanded Functionality and Improved Support for Small Counts and Larger Datasets},
  shorttitle = {{{edgeR}} V4},
  author = {Chen, Yunshun and Chen, Lizhong and Lun, Aaron~T L and Baldoni, Pedro~L and Smyth, Gordon~K},
  year = {2025},
  journal = {Nucleic Acids Research},
  volume = {53},
  number = {2},
  pages = {gkaf018},
}

@article{robbins1950generalization,
  title = {A Generalization of the Method of Maximum Likelihood: Estimating a Mixing Distribution (Abstract)},
  author = {Robbins, Herbert},
  year = {1950},
  journal = {The Annals of Mathematical Statistics},
  volume = {21},
  pages = {314--315}
}

@article{ho2025largescale,
  title = {Large-Scale Estimation under Unknown Heteroskedasticity},
  author = {Ho, Sheng Chao},
  year = {2025},
  journal = {arXiv preprint},
  volume = {arXiv:2507.02293}
}

@inproceedings{robbins1956empirical,
  title = {An Empirical {{Bayes}} Approach to Statistics},
  booktitle = {Proceedings of the {{Third Berkeley Symposium}} on {{Mathematical Statistics}} and {{Probability}}, {{Volume}} 1: {{Contributions}} to the {{Theory}} of {{Statistics}}},
  author = {Robbins, Herbert},
  year = {1956},
  pages = {157--163},
  publisher = {The Regents of the University of California}
}

@book{efron2010largescale,
  title = {Large-Scale Inference: Empirical {{Bayes}} Methods for Estimation, Testing, and Prediction},
  author = {Efron, Bradley},
  year = {2010},
  series = {Institute of {{Mathematical Statistics Monographs}}},
  publisher = {Cambridge University Press},
  address = {Cambridge},
}
\bibliographystyle{abbrvnat}

\appendix

\setcounter{equation}{0}
\setcounter{figure}{0}
\setcounter{table}{0}
\setcounter{proposition}{0}
\setcounter{algorithm}{0}

\renewcommand{\theequation}{S\arabic{equation}}
\renewcommand{\thefigure}{S\arabic{figure}}
\renewcommand{\thetable}{S\arabic{table}}
\renewcommand{\theproposition}{S\arabic{proposition}}
\renewcommand{\thetheorem}{S\arabic{theorem}}
\renewcommand{\thelemm}{S\arabic{lemm}}
\renewcommand{\thealgorithm}{S\arabic{algorithm}}

\section{Methodological extension: \reglitrd{} for discrete $M_i$}
\label{sec:discrete_MI}

A natural generalization beyond \reglitrd{} for discrete $M_i$ is to consider a bin-based prior estimate for $\sigma_i^2\mid M_i$. In particular, we assume that those units with similar $M_i$ value would also share similar $\sigma_i^2$. Instead of positing a prior on bivariate nuisances, we group units by the value of $M_i$, and estimate a separate prior for $\sigma_i^2$ within each group. Let $b(i)\in\{1,...,B\}$ denote the index of the bin to which unit $i$ is assigned, where the binning is induced by $M_i$. In addition to \eqref{eq:bayesian_cs_model_2}, we further assume that within each bin $b$,
$$
\sigma_i^2 \mid b(i)=b \sim H_b,\; \text{ and } \theta_i \text{ are deterministic}.
$$
The corresponding oracle partially Bayes p-values that use the exact knowledge of $\{H_b\}$ are defined as, 
\begin{align}
\label{eq:2d_pb_pvalue_discrete_case}
P_i^{\djt} = \frac{\int_{0}^{\infty}2\Phi(-\frac{|Z_i|}{\nu\sigma_i})p_{\chi^2}(S_i^2\mid K-p,\sigma_i^2)H_{b(i)}(\dd\sigma_i^2)}{\int_{0}^{\infty}p_{\chi^2}(S_i^2\mid K-p,\sigma_i^2)H_{b(i)}(\dd\sigma_i^2)} \qquad \mbox{for $i \in [n]$.}
\end{align}
For each bin $b$, we estimate $H_b$ from the data within that bin by NPMLE using the marginal distribution of the statistics $\{S_i^2:b(i)=b\}$. In that direction, for any mixing measure $H^{\prime}_b$, define 
\[f_{H_b^{\prime},K-p}(s^2)
    = \int_0^\infty 
        p_{\chi^2}(s^2\mid K-p,\sigma^2)\, H_b^{\prime}(\dd\sigma^2).\]
We can then estimate $H_b$ by solving the following optimization problem
\begin{align}
    \label{eq:joint_npmle_discrete}
    \wh H_b \in \argmax_{\rmH_b \in \mathcal G_{H_B}} \frac{1}{\#\{i:b(i)=b\}}\sum_{i:b(i)=b}\log f_{\rmH_b,K-p}(S^2_i),
\end{align}
where $\mathcal{G}_{H_B} :=\{\rmH_b: \rmH_b(\underline{L},\wb{U})=1\}$ for the same absolute constants $\underline{L},\wb{U}>0$ used in \eqref{eq:joint_npmle}. This problem can be solved in essentially the same way as in \eqref{eq:def_small_ell}. Then a plug-in estimate $\wh P^{\djt} _i$ of the oracle p-values $P^{\djt} _i$ are constructed by replacing each $H_b$ with $\wh H_b$ in \eqref{eq:2d_pb_pvalue_discrete_case}. The full procedure is summarized in Algorithm \ref{algo:limma_trend_jt_discrete}.
\begin{algorithm}[H]
\caption{\limmatrd{} generlization with discrete $M_i$ (\jtlitrd{})}
\begin{algorithmic}[1]
\State Partition the units into $B$ bins according to the value of $M_i$, and construct bin index $b(i)$.
\State Estimate $H_b$ using $\wh H_b$ through NPMLE (cf. \eqref{eq:joint_npmle_discrete}) for each $b=1,...,B$.
\State Compute the p-values $\wh P^\djt_i$ using \eqref{eq:2d_pb_pvalue_discrete_case} by replacing $H_{b(i)}$ with $\wh H_{b(i)}$. 
\end{algorithmic}
\label{algo:limma_trend_jt_discrete}
\end{algorithm}

\section{Implementation details}
\label{sec:det_npmle}
\subsection{Estimation of $\wh m(\cdot)$ in Algorithms~\ref{algo:limma_trend} and~\ref{algo:limma_trend_reg}}
\label{subsec:reg_trend}
We adopt the trend estimation procedure from the implementation of \limmatrd{} in \texttt{R} \citep{ritchie2015limma}. We model $y=\log{(S_i^2)}$ as a smooth function of $A_i$ using a natural cubic spline basis. The spline degrees of freedom are chosen adaptively according to the number of observations and capped by the number of distinct $A_i$ values. In particular, let $n$ be the number of units in $y$ and $U_A$ be the number of distinct observed $A_i$ values. The spline degrees of freedom are chosen as
\[
\nu=\min\left\{1+\mathbf{1}(n\ge 3)+\mathbf{1}(n\ge 6)+\mathbf{1}(n\ge 30),\ U_A\right\}.
\]
If $\nu<2$, we use the constant fit $\wh m(a)\equiv \wb y$; otherwise, we fit a natural spline with $\nu$ degrees of freedom by least squares and use the resulting fitted curve as $\wh m(\cdot)$.
\subsection{Discretization for NPMLE computation in \reglitrd{}}
\label{sec:modek_reg_npmle}
As observed before, the optimization program in \eqref{eq:def_small_ell} is convex. In our implementation, we use the interior point convex programming solver MOSEK \citep{aps2020mosek} with the discretization technique proposed by \citet{koenker2014convex} to solve it. In particular, we choose $B=300$ grid points that are logarithmically spaced between the $1\%$ quantile and the largest value of $\{\wh V_i^2,...,\wh V_n^2\}$. We optimize \eqref{eq:def_small_ell} over all distributions supported on this finite grid. This is a conic programming problem and the solution is computed using MOSEK. 

\subsection{Discretization for NPMLE computation in \jtlitrd{}}
\label{subsec:discretization_twod}
To compute \eqref{eq:joint_npmle}, we restrict attention to distributions supported on a finite two-dimensional grid and optimize only over the corresponding probability masses. We adopt the following discretization scheme. We partition $A_i$ into $B=50$ bins, which are constructed deterministically from the order statistics of $A_i$, yielding $B$ consecutive intervals with roughly equal counts.

For each bin $b$, let $I_b$ be the set of indices in that bin and let $\mathcal{M}_b:=\{A_i:i\in I_b\}$ be the set of observed $A_i$-values in the same bin. We use $\mathcal{M}_b$ as the grid points for $\mu$. Next, letting $\wh m(\cdot)$ be the trend fitted in Algorithm~\ref{algo:limma_trend}, we define the residuals $r_i:=\log S_i^2 - \wh m (A_i)$ for $i=1,...,n$ and construct a common residual grid $\mathcal{R}_{p_v}=\{r^{(1)},...,r^{(p_v)}\}$ using $p_v=50$ equally spaced points between $1\%$ quantile and the maximum of $\{r_1,...,r_n\}$. For each bin, let $u_b$ be the median of $A_i$'s within the bin $b$. We define the bin-specific trend-centered grid for $\sigma^2$ given by $\mathcal{V}_b=\{\exp{\left(\wh m(u_b)+r\right)}:r\in \mathcal{R}_{p_v}\}$. 

For each pair $(b,v)$ with $b\in[B]$ and $v\in \mathcal{V}_b$, the bin-level likelihood contribution for each data unit $i$ is defined as
$$
L_{i,b,v} = p_{\chi^2}(S_i^2 \mid K-p,\sigma^2=v)\cdot \frac{1}{|\mathcal{M}_b|}\sum_{\mu\in\mathcal{M}_b}\phi\left(A_i;\mu,\frac{v}{K}\right),
$$
where $\phi(x;\mu,\sigma^2)$ denotes the density of a univariate Gaussian distribution with mean $\mu$ and variance $\sigma^2$ evaluated at $x$, and optimize over probability masses $\{f_{b,v}\}$ by solving
$$
\max_{\{f_{b,v}\}}\frac{1}{n}\sum_{i=1}^n\log\left(\sum_{b=1}^B\sum_{v\in\mathcal{V}_b}f_{b,v}L_{i,b,v}\right), \quad f_{b,v}\geq0,\quad \sum_{b=1}^B\sum_{v\in\mathcal{V}_b}f_{b,v}=1.$$
This is a finite-dimensional convex optimization problem that can be solved numerically using MOSEK. After optimization, each fitted bin-level mass $\wh{f}_{b,v}$ is mapped back to point masses on $(\mu,\sigma^2)$ by distributing each $\wh{f}_{b,v}$ equally over the atoms in $\mathcal{M}_b\times \{v\}$. That is, for every $\mu \in \mathcal{M}_b$, we assign mass $\wh{f}_{b,v}/|\mathcal{M}_b|$ to the point $(\mu,v)$. Therefore, the final discretized estimator is supported on 
$$
\mathcal S_\jt:=\bigcup_{b=1}^B\left\{(\mu,v):\,\mu \in \mathcal{M}_b\;\text{and}\; v \in \mathcal V_b\right\}.
$$
Note that the choice of $B$, the size of the residual grid $p_v$, and the truncation quantiles for constructing the common residual grid $\mathcal{R}$ are tuning parameters and can be adjusted to accommodate different datasets and computational budgets. We ignore the discretization error in the theoretical development.

\subsection{Implementation for baseline methods in Section~\ref{sec:numerical_exp}}
\label{subsec:implementation_baseline}

%\fixme{call methods according to their name from main text/tables?}
\begin{itemize}[leftmargin=*,nosep,wide] 
\item \emph{Parametric \texttt{Limma} (Untrended-Inv$\chi^2$) and \reglitrd{}.} 
As described in Section~\ref{subsec:background} and Section~\ref{subsec:limma_trd_interpretation}, we will need to estimate $\kappa_0$ and $s_0^2$ in \eqref{eq:exchangeable_sigma} and \eqref{eq:reg-model-trnd} for Untrended-Inv$\chi^2$ and \reglitrd{} respectively. Indeed, we use the methods of moments that is provided in \cite{smyth2004linear}, and also the build in method in the \texttt{limma} package to estimate $(\kappa_0, s_0^2)$ for both Untrended-Inv$\chi^2$ and \reglitrd{}. In particular, for Untrended-Inv$\chi^2$, we estimate it by the following procedure (denote $d=K-p$):
$$
  \begin{aligned}
    z_i &= \log{S_i^2}, \\
    e_i &= z_i - \psi(d/2) + \log{(d/2)}, \\
    \psi'(\wh{\kappa}_0/2) &= \operatorname{mean}\{(e_i-\bar e)^2n/(n-1) - \psi^{'}(d/2)\}, \\
    \wh{s}_0^2 &= \exp{(\text{mean}(e)+\psi(\wh{\kappa}_0/2)-\log(\wh{\kappa}_0/2))},
    \end{aligned}
$$
where $\psi(\cdot)$ and $\psi'(\cdot)$ is the diagamma function and trigamma function respectively. Similarly for \reglitrd{}, we estimate it as follows:
$$
\begin{aligned}
z_i &= \log S_i^2,\\
r_i &= z_i-\widehat m(A_i),\\
\psi'(\wh\kappa_0/2)
&= \operatorname{mean}\{(r_i-\bar r)^2\}n/(n-\nu)
   - \psi'(d/2),\\
\wh s_0^2
&=
\exp\left\{
\psi(\wh\kappa_0/2)-\log(\wh\kappa_0/2)
-\psi(d/2)+\log(d/2)
\right\},
\end{aligned}
$$
where $\wh{m}$ is the trend function estimated in Supplement~\ref{subsec:reg_trend}, and $\nu$ is the degrees of freedom of the spline that is used for estimating $\wh m(\cdot)$.

\item \emph{Untrended-NPMLE.} We refer to \cite{ignatiadis2025empirical} for the description of the method. The prior $G$ is estimated with NPMLE using similar techniques in Supplement~\ref{sec:modek_reg_npmle}, with the optimization function being replaced as 
$$
\widehat{G}_{\limma{}}
   \in \argmax_{G  \in \mathcal G_{\limma{}}}
    \frac{1}{n}\sum_{i=1}^n
        \log \left\{
\int_0^\infty
p_{\chi^2}\!\left(S_i^2 \mid K-p,\sigma^2\right)
G(\dd \sigma^2)
\right\},
$$
where $\mathcal G_{\limma{}}$ is the nonparametric class of distributions considered in \cite{ignatiadis2025empirical}, and we choose $B=300$ optimization grid points that are logarithmically spaced between the $1\%$ quantile and the largest value of $\{S_1^2,...,S_n^2\}$.
\item \emph{MAnorm2~\citep{tu2021manorm2}.} We use the built-in function provided by the R package \texttt{MAnorm2}~\citep{tu2021manorm2}.  For the two-sample comparison setting of Example~\ref{ex:two_sample}, we first construct \smash{$\wb Y_{iA}$}, \smash{$\wh{S}_{iA}^2$}, \smash{$\wb Y_{iB}$}, \smash{$\wh{S}_{iB}^2$}, where $\wb Y_{iA}$ and $\wb Y_{iB}$ denote the treated/control group means, \smash{$\wh{S}_{iA}^2$} and \smash{$\wh{S}_{iB}^2$} denote the treated/control group sample variances for unit $i$. The function \texttt{fitMeanVarCurve} estimates a mean-variance curve 
\smash{$\widehat\xi_{\mathrm{MA}}(\cdot)$} from the pooled collection of above pairs \smash{$\{(\wb Y_{iA},\widehat S_{iA}^2)\}_{i=1}^n$}, \smash{$\{(\wb Y_{iB},\widehat S_{iB}^2)\}_{i=1}^n$} using robust gamma family local regression. MAnorm2 then estimates the prior degrees of freedom $d_0$ by matching the empirical variability of the normalized within-group variances to its theoretical counterpart. Specifically, for $g\in\{A,B\}$, define
$$
z_{ig}
=
\log\left\{
\frac{\widehat S_{ig}^2}
{\widehat\xi_{\mathrm{MA}}(\wb Y_{ig})}
\right\},
\quad
\nu_g=K_g-1, \quad \text{and assume working prior }  \frac{1}{\sigma_{ig}^2} \sim \frac{\chi_{d_0}^2}{\wh\xi_{\mathrm{MA}}(\mu_{ig})d_0}.
$$
Under the working scaled Inv$\chi^2$ model, $z_{ig}$ is approximately distributed as $\log F_{\nu_g,d_0}$, so that
$$
\mathrm{Var}(z_{ig})
\approx
\psi_1\left(\frac{\nu_g}{2}\right)
+
\psi_1\left(\frac{d_0}{2}\right),
$$
where $\psi_1$ is the trigamma function. Thus \texttt{estimatePriorDf} estimates $d_0$ by matching the observed sample variances of $z_{iA}$ and $z_{iB}$ to this expression. After this, MAnorm2 forms the moderated variance and the corresponding statistics
$$
\widetilde S_i^2
=
\frac{
\wh d_0\,\wh\xi_{\mathrm{MA}}(\wb Y_{i,\mathrm{avg}})
+
(K_A+K_B-2)\wh S_{i,\mathrm{pool}}^2
}{
\wh d_0+K_A+K_B-2
},\quad T_i^{\mathrm{MA}}
=
\frac{\wb Y_{iB}-\wb Y_{iA}}
{
\sqrt{
\left(K_A^{-1}+K_B^{-1}\right)\widetilde S_i^2
}
},
$$
with $Y_{i,\mathrm{avg}}=(\wb Y_{iA}+\wb Y_{iB})/2$ and $\widehat S_{i,\mathrm{pool}}^2=\{(K_A-1)\widehat S_{iA}^2
+
(K_B-1)\widehat S_{iB}^2\}/(K_A+K_B-2)$, and then the p-value is constructed as 
$
2\wb F_{t_{\widehat d_0+K_A+K_B-2}}(
|T_i^{\mathrm{MA}}|
)$. %\fixme{Use notation from introduction with $F_{t...}$}

\item \emph{MAP~\citep{li2019map}.} To facilitate the comparison, we replace the built-in trend estimation procedure of MAP with the procedure used in the other methods (described in Supplement~\ref{subsec:reg_trend}). With the estimated $\wh m(\cdot)$, we calculate the p-value as \smash{$2\Phi(-|Z_i|/\{\nu \exp(\widehat{m}(M_i))\})$}. 
\end{itemize}

\section{Additional numerical results in Section~\ref{sec:numerical_exp}}
\label{sec:numeric_additional}
This section provides additional numerical results that extend the first experiment in Section~\ref{sec:numerical_exp}. We use a balanced design with $K_A=K_B$, and take $K\in\{4,6,10,18\}$. We consider three choices for $G$: (i) $G=\delta_1$ (referred to as \texttt{Dirac}), (ii) $G=10\times\mathrm{Inv}\chi^2_{10}$ (\texttt{Scaled Inv$\chi^2$}), and (iii) $G=0.5\,\delta_1+0.5\,\delta_{10}$ (\texttt{Two-point}).
We also consider two choices of the mean–variance trend function $m(\cdot)$: (1)  Type 1 (without mean–variance trend): $\exp{\left(m(\mu_i)\right)}=1$, and (2)  Type 2 (with mean–variance trend): $m(\mu_{i})=-4\,\mathrm{logistic}\,((\mu_i-16)/4)+12$, where $\mathrm{logistic}(y)=(1+e^{-y})^{-1}$. 

Results are summarized in Figure~\ref{fig:simulation_result}. We omit the MAP results from this figure, as it does not control the FDR in any of the simulation settings considered. Since the simulations are based on a balanced design, MAnorm2 performs very similarly to Reg-Inv$\chi^2$ throughout these settings. The upper panel corresponds to $m$ coming from Type 1 and the lower panel from Type 2 (these types are defined above). For $m$ belonging to Type 1, for \texttt{Dirac} and \texttt{Scaled Inv$\chi^2$} priors, we observe that the FDR is controlled at the correct nominal level across all degrees of freedom. However, for the \texttt{Two-point} prior, the prior specification in \reginvc{} and \limma{} is wrong, resulting in inflated FDR for $K-p=4$, as shown in Section~\ref{sec:numerical_exp}. Across all the settings, we observe that the partially Bayes methods outperform t-tests in terms of power, demonstrating the benefit of pooling information across units. The impact of prior misspecification on power is evident in the third column of the upper panel, where the NPMLE-based procedures outperform those depending on parametric priors. The best performance is given by \jtlitrd{} and NPMLE-based \limma{} from \citep{ignatiadis2025empirical}, whereas \reglitrd{} slightly suffers due to the trend estimation error. Since there is no trend in the oracle model, \jtlitrd{} does not exhibit an advantage. As the degrees of freedom increase, the performance of all the methods becomes comparable. Next, when $m$ comes from Type 2 with non-constant mean-variance trend, the FDR of all the procedures behaves similarly as the previous setting. However, in this framework, \reglitrd{} and \jtlitrd{} outperform all methods ignoring the mean variance trend. Furthermore, the performance of \reglitrd{} is slightly worse compared to \jtlitrd{}. Again, the performance gap disappears as $K$ increases.

\begin{figure}[!htbp]
\centering
\setlength{\tabcolsep}{0pt}

\begin{tabular}{@{}c@{\hspace{0.03\textwidth}}c@{\hspace{0.03\textwidth}}c@{}}
\multicolumn{3}{@{}l@{}}{(a) Type 1 (without mean-variance trend):}\\[0.8em]

\includegraphics[width=0.28\textwidth]{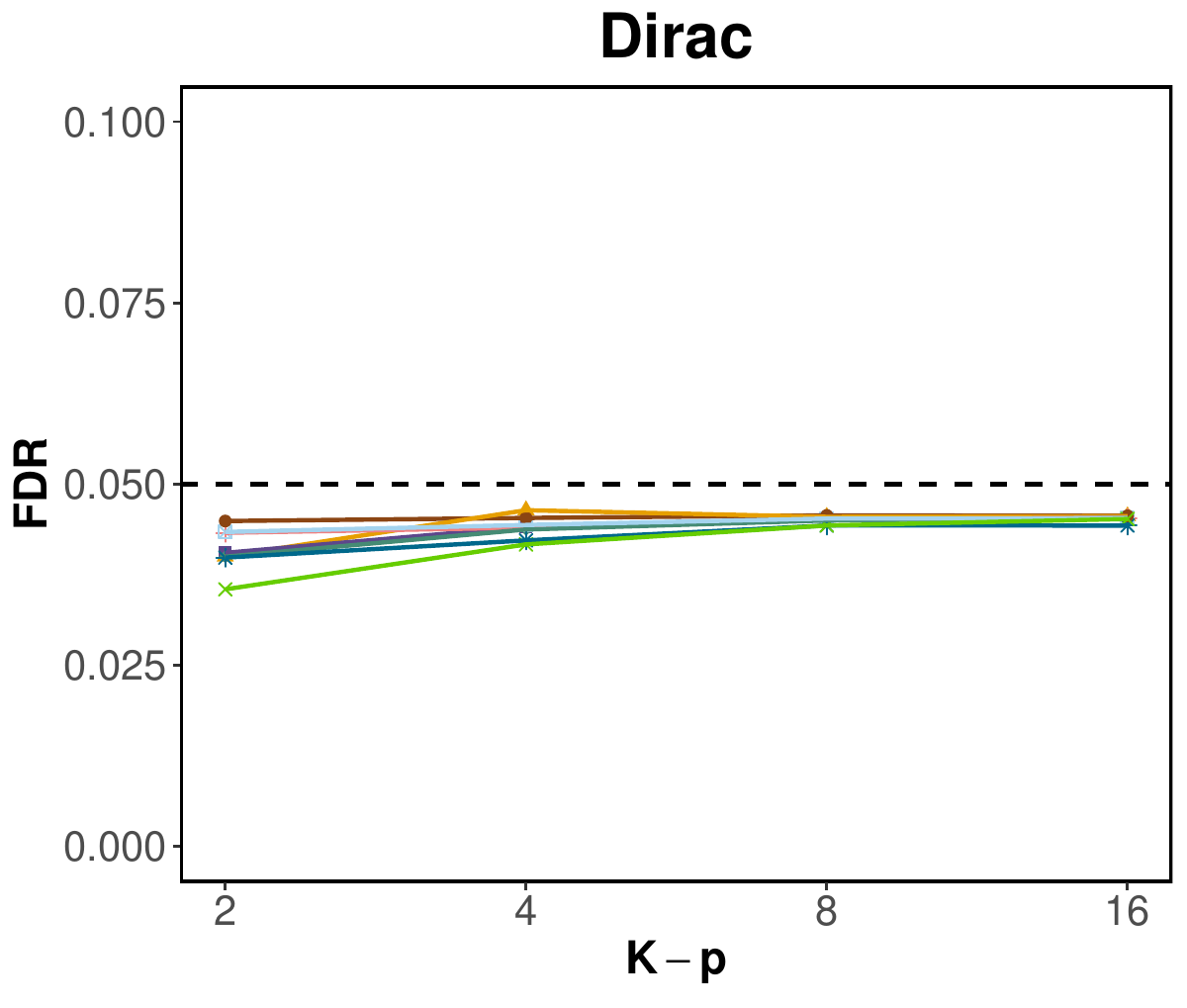} &
\includegraphics[width=0.28\textwidth]{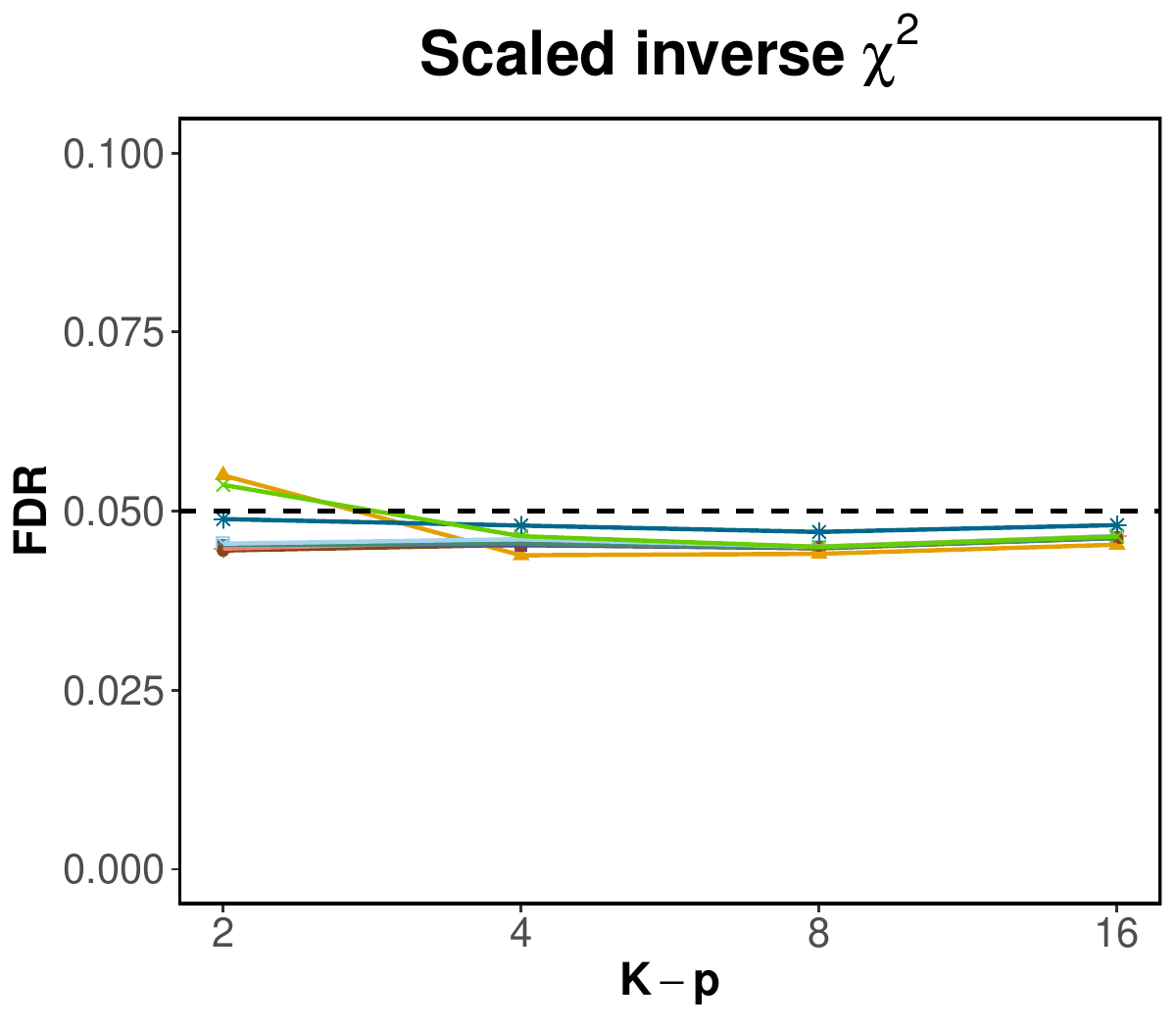} &
\includegraphics[width=0.28\textwidth]{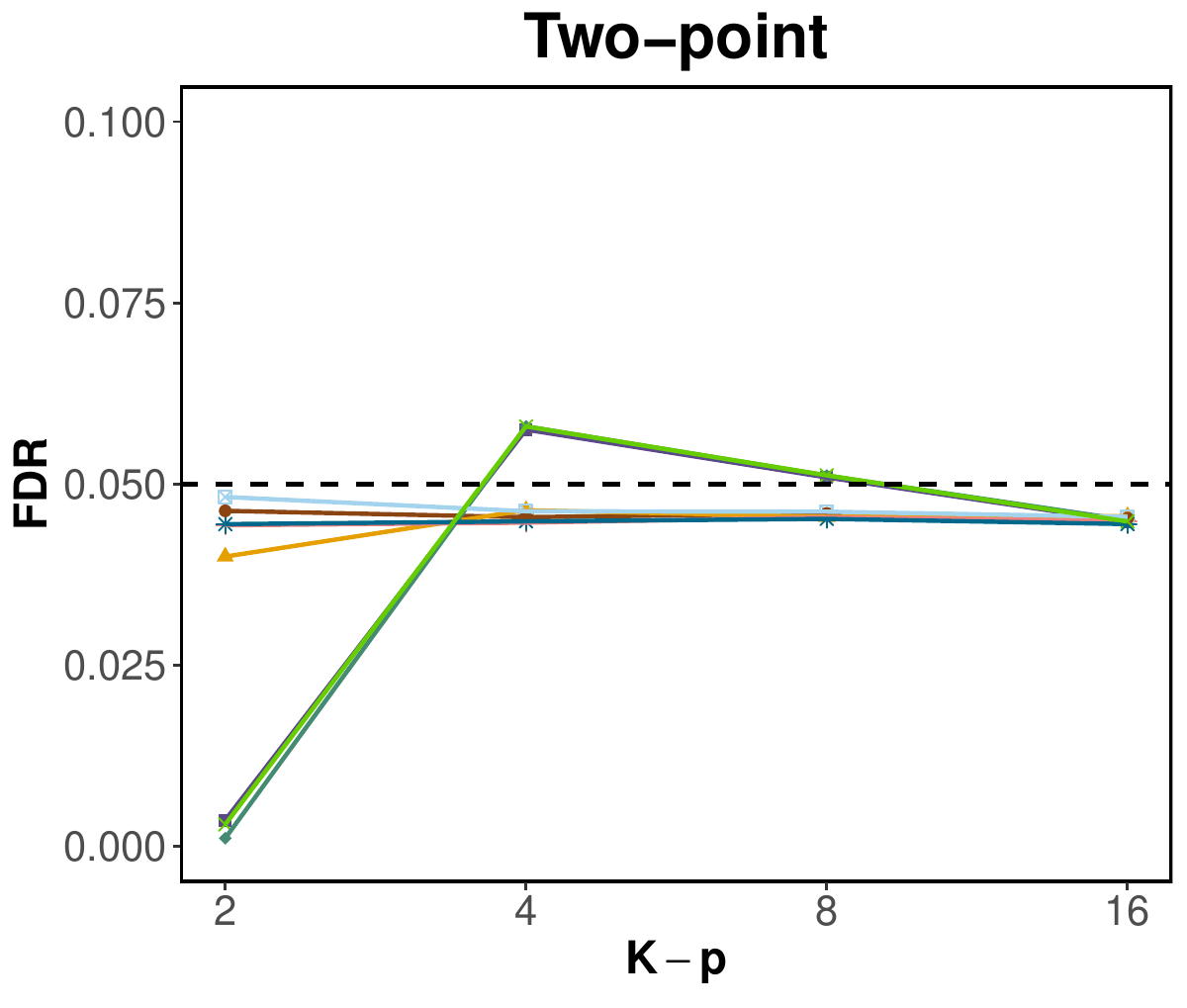} \\[0.4em]

\includegraphics[width=0.28\textwidth]{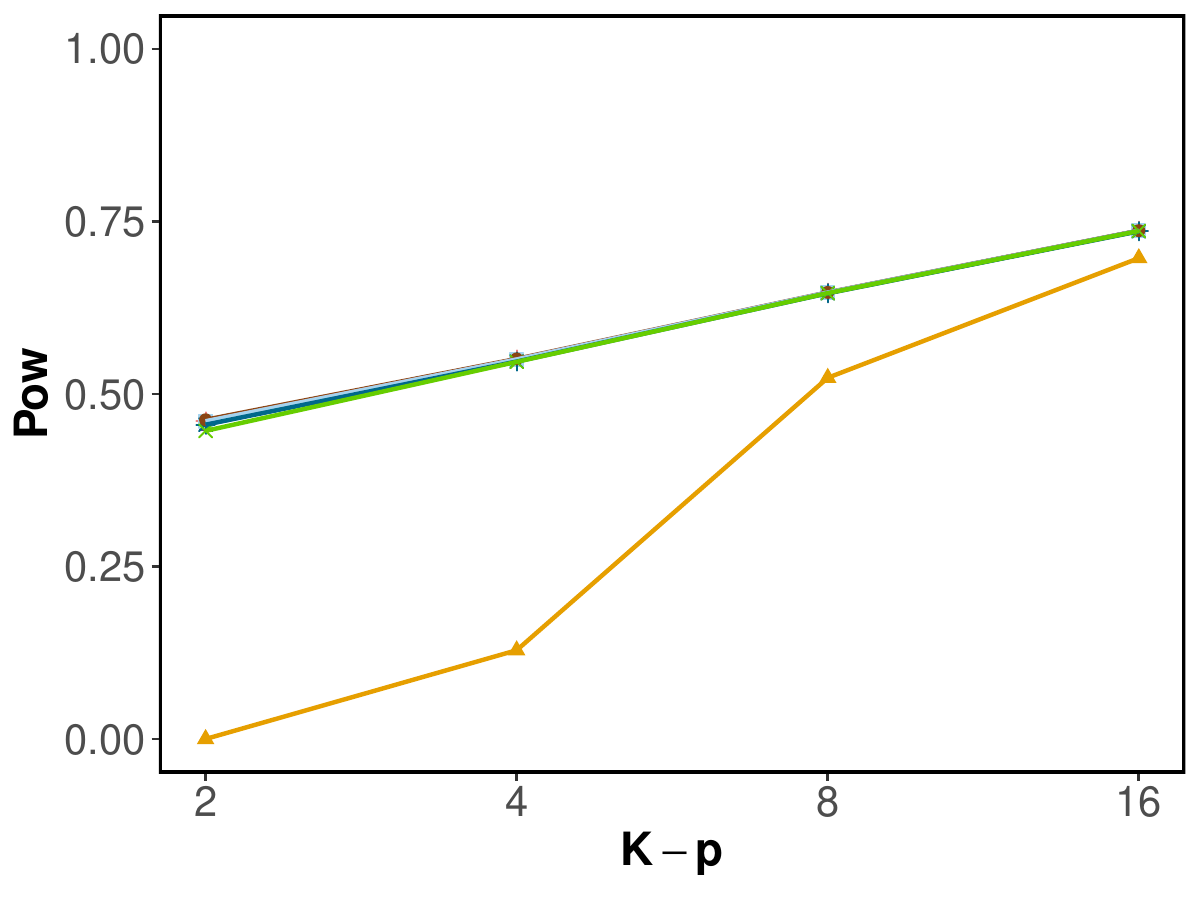} &
\includegraphics[width=0.28\textwidth]{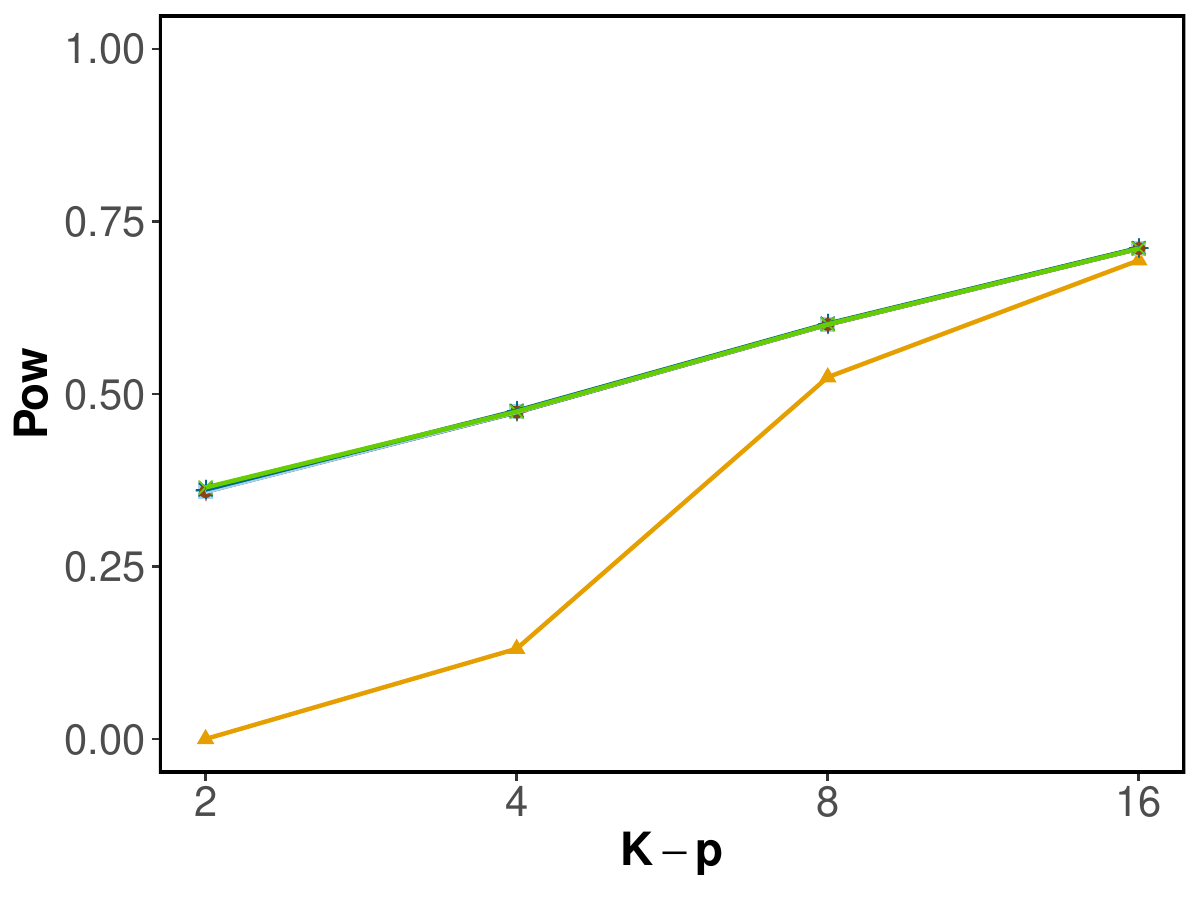} &
\includegraphics[width=0.28\textwidth]{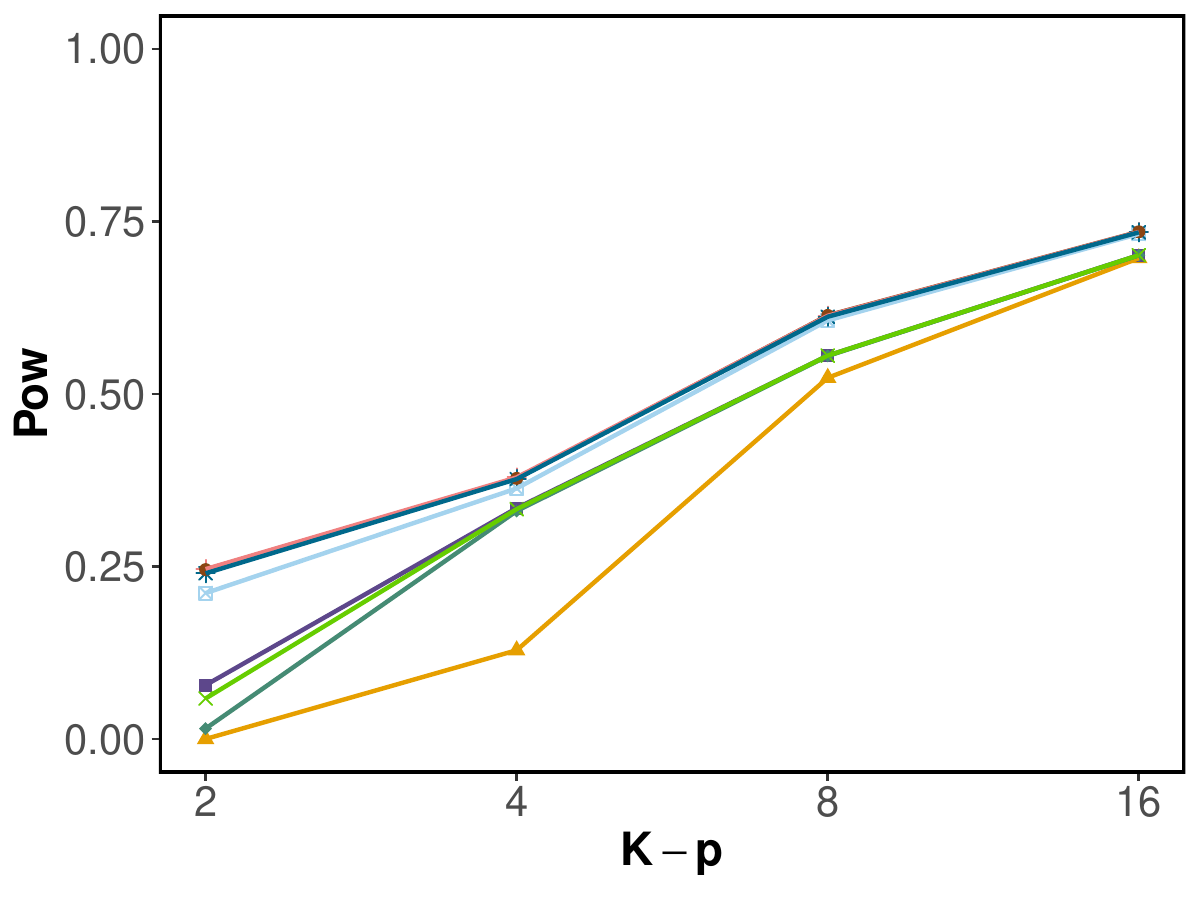} \\[1.0em]

\multicolumn{3}{@{}l@{}}{(b) Type 2 (with mean-variance trend):}\\[0.8em]

\includegraphics[width=0.28\textwidth]{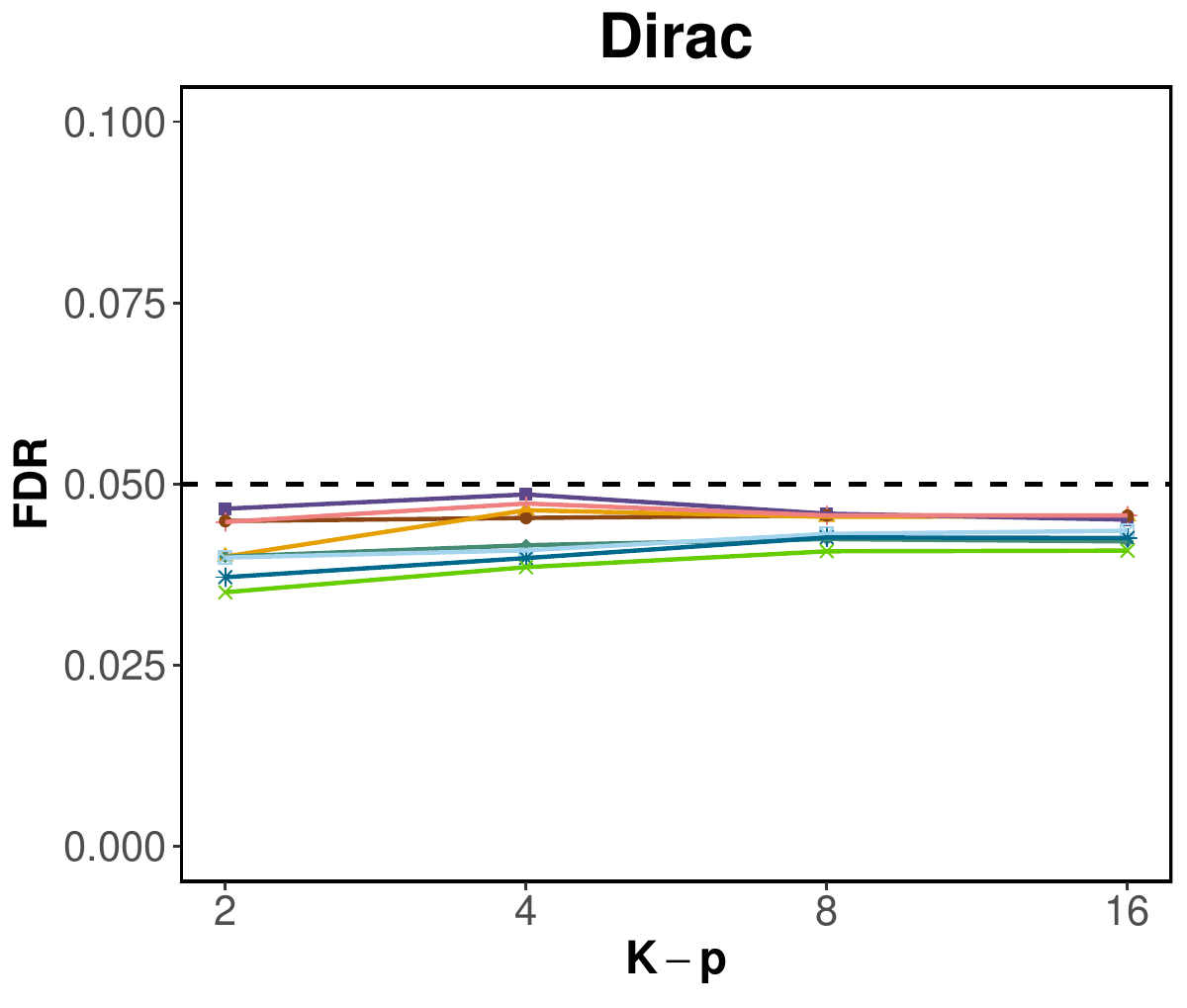} &
\includegraphics[width=0.28\textwidth]{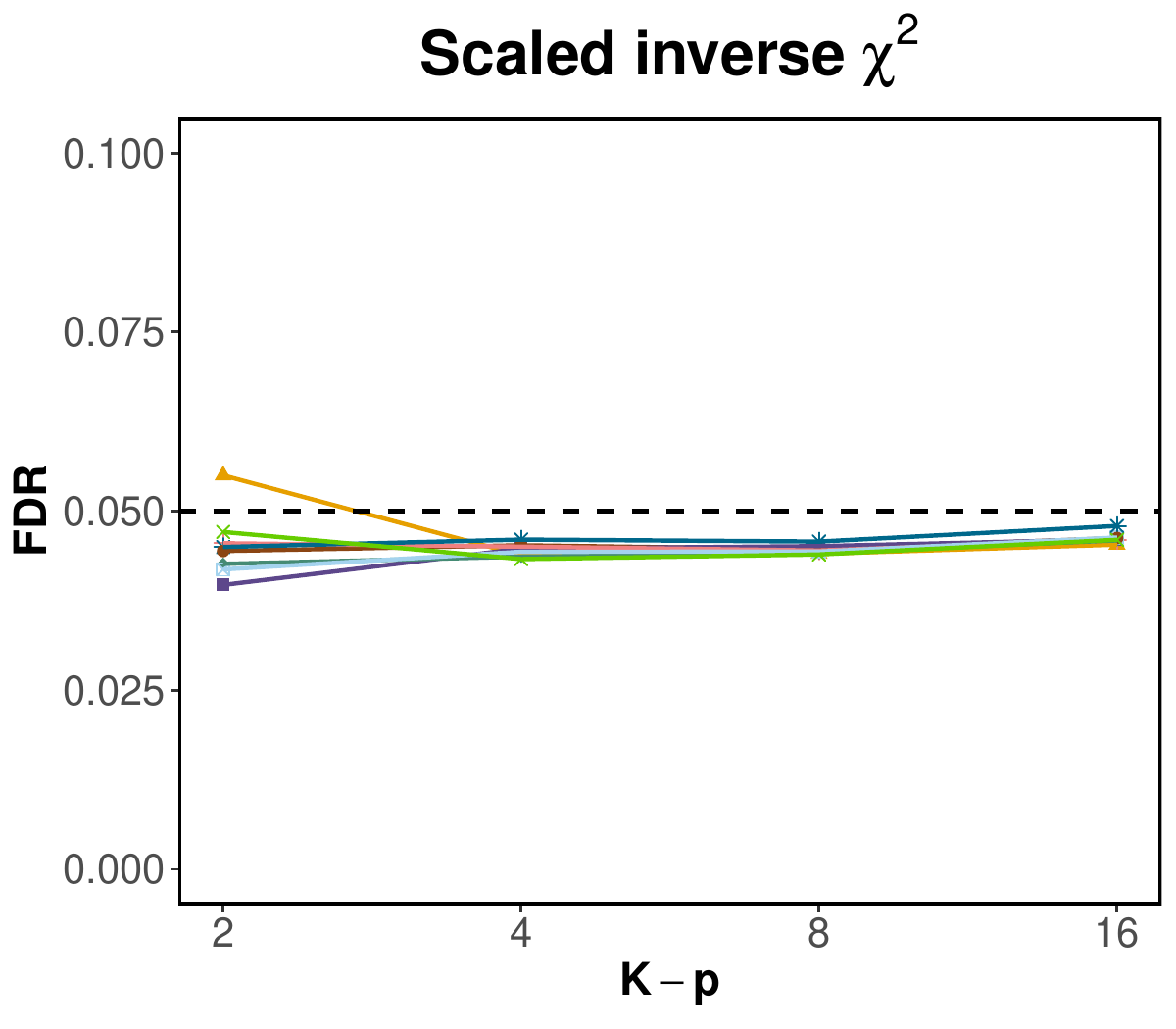} &
\includegraphics[width=0.28\textwidth]{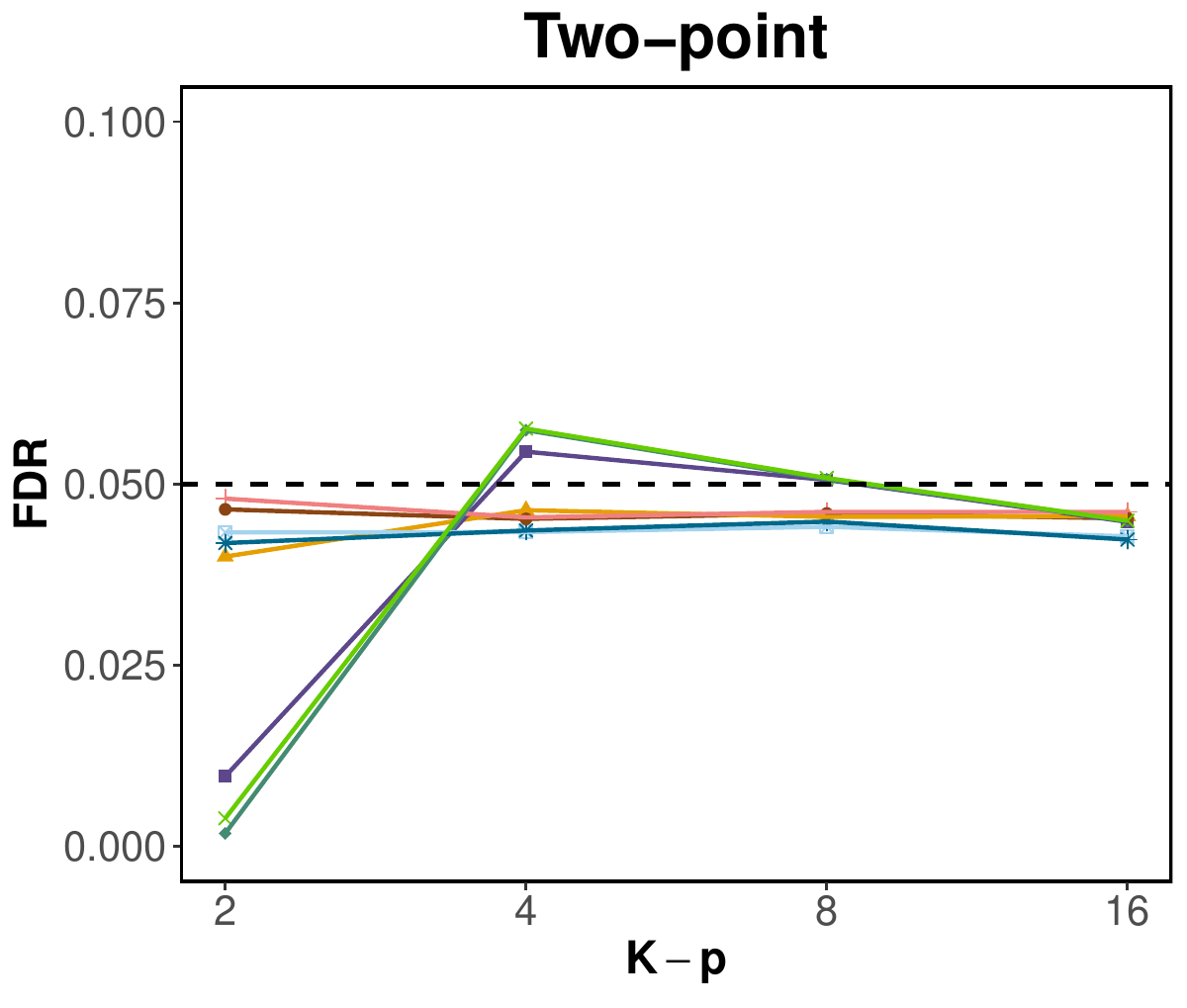} \\[0.4em]

\includegraphics[width=0.28\textwidth]{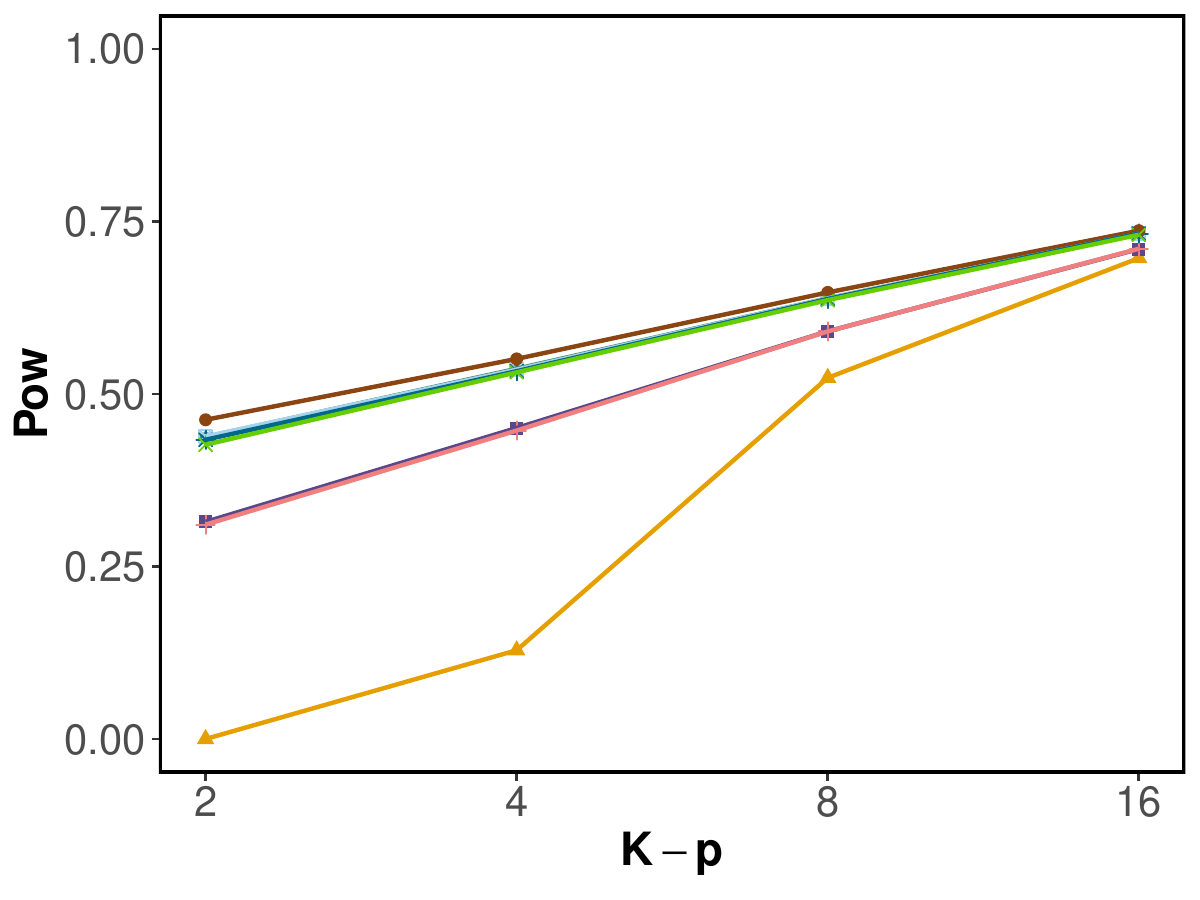} &
\includegraphics[width=0.28\textwidth]{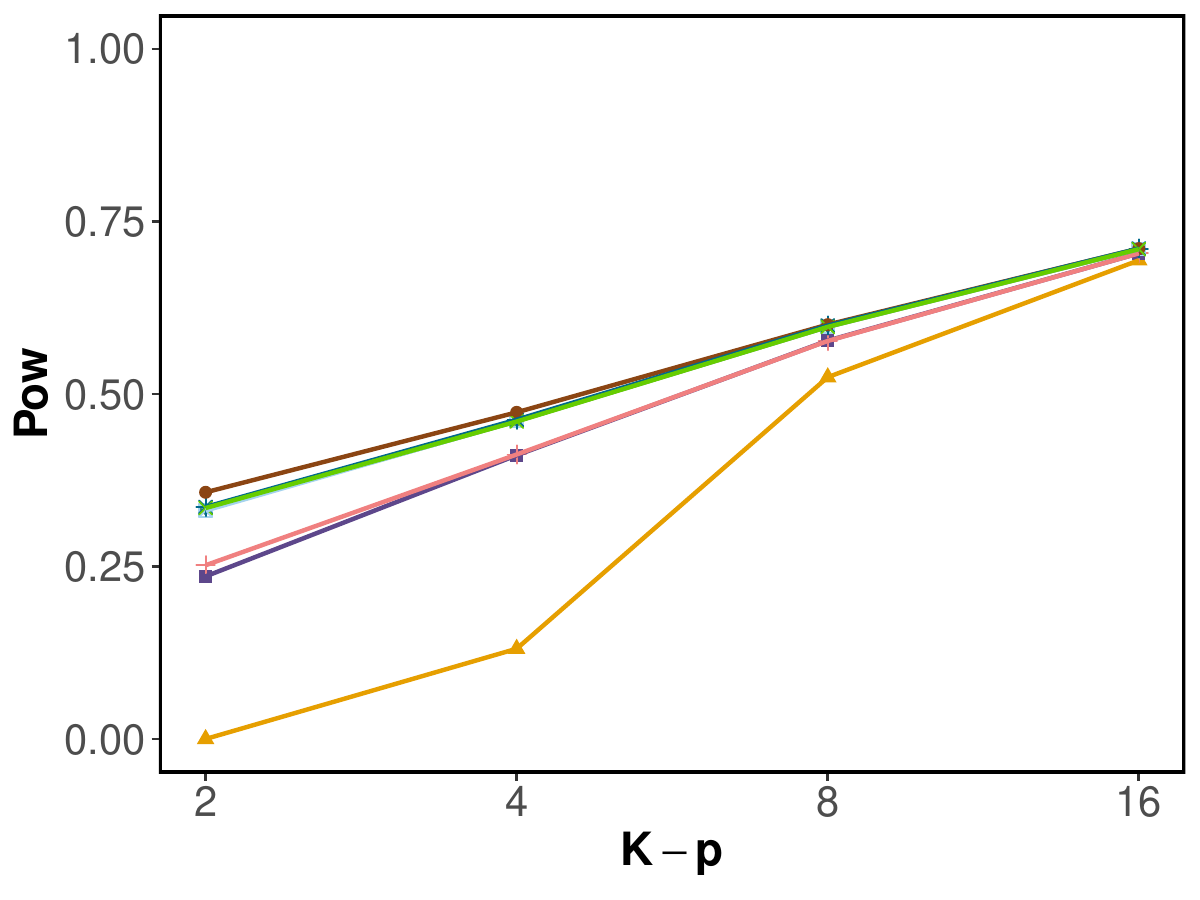} &
\includegraphics[width=0.28\textwidth]{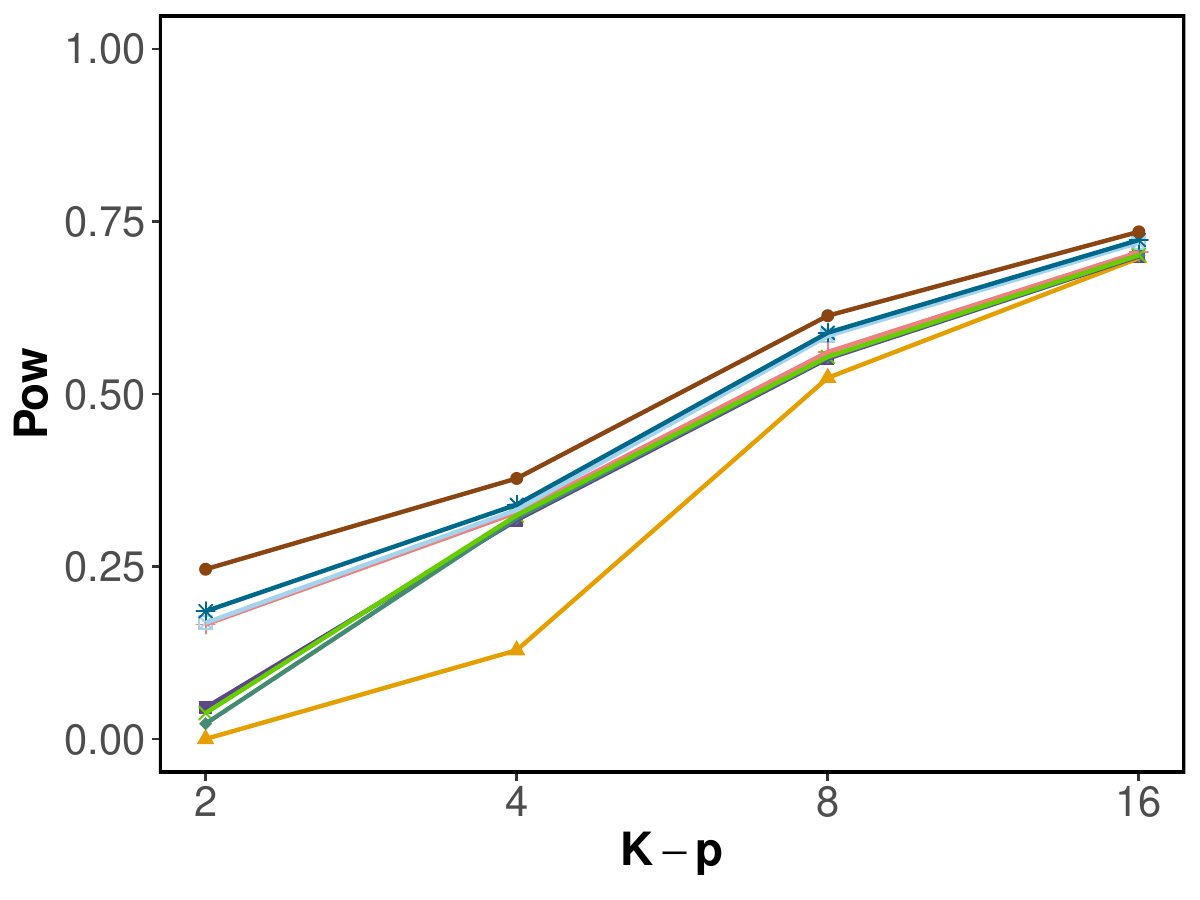} \\[0.5em]

\multicolumn{3}{c}{\includegraphics[width=0.9\textwidth]{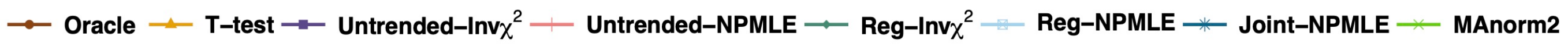}}
\end{tabular}

\caption{
The plot of FDR and power as a function of the degrees of freedom $K-p \in \{2,4,8,16\}$. 
The upper panel corresponds to $m$ from a constant mean-variance trend, and the lower panel corresponds to $m$ from a logistic mean-variance trend (described in the main text). 
The three columns represent the three choices of $G$: \texttt{Dirac} (left), \texttt{Scaled Inv$\chi^2$} (middle), and \texttt{Two-point} (right). 
Within each panel, the top row reports the false discovery rate, and the bottom row reports power.
}
\label{fig:simulation_result}
\end{figure}

\section{Additional details and results for real data examples}

\subsection{Preprocessing of the count data in Section~\ref{sec:bulk_rna_1}}
\label{sec:preprocessing}
We downloaded the raw data from the Gene Expression Omnibus (GEO, \url{http://www.ncbi.nlm.nih.gov/geo/}) with accession number \texttt{GSE114716}. The raw data consist of a feature-by-sample count matrix, where each entry records the number of sequencing reads mapped to a given feature in a given sample. We conduct the following preprocessing steps using the edgeR \citep{chen2025edger} pipeline: (1) removed lowly expressed genes in a design-aware manner, (2) normalized for differences in library composition across samples using the trimmed mean of M-values (TMM) procedure, and (3) transformed the normalized counts to $\log_2$ counts per million with a prior count of $3$ for downstream analysis.

\subsection{Details and diagnostic figure for proteomics application of Section~\ref{sec:seq_peptide}}
\label{sec:bins_protein}
For the proteomics data in Section~\ref{sec:seq_peptide}, \jtlitrd{} is not directly applicable since the summary statistics are discrete. However, we can estimate the prior for $\sigma_i^2\mid M_i$ using the bin based procedure in Supplement~\ref{sec:discrete_MI}.  
In that direction, we stratify the proteins by peptide counts $M_i$, using exact strata for $M_i\in\{1,2,...,11\}$, and pooled $7$ bins for larger $M_i$ to ensure stable estimation of the variance prior. 
Observe that the binning already takes care of the stratification, and hence, trend estimation is not required here.
The reversal across $M$ in the number of discoveries made by the trended and the untrended methods, as observed in Figure~\ref{fig:proteomics}(c) can be understood as a consequence of variance misspecification under untrended modeling. Untrended methods shrink all units toward a common global variance trend, which may be smaller than the local variance level in low-$M$ bins, leading to small estimated variances and therefore more significant findings. This phenomenon reverses as $M$ increases. However, trend-based methods gain advantages from modeling the mean-variance relationship and thus avoid the corresponding misspecification. Figure~\ref{fig:proteomics}(d) displays the estimated prior on $\sigma_i^2$ within each bin of $M$. The estimated discrete priors exhibit clear variation across bins in both support and mass allocation, indicating substantial heterogeneity in the variance distribution. To further assess model fit, we examine bin-wise marginal diagnostic plots (cf. Figure~\ref{fig:proteomics_dignostic}) for $S_i^2$, with particular focus on bins corresponding to small $M_i$. Within each bin, we compare the empirical distribution of the sample variances $S_i^2$ with the model-implied marginal distribution obtained by $\reglitrd{}$ and discrete $\jtlitrd{}$. In bins with small $M_i$, the marginal distribution implied by \reglitrd{} does not adequately match the observed distribution of $S_i^2$, with noticeable discrepancies in both shape and spread. In contrast, the marginal distribution implied by discrete \jtlitrd{} is more closely aligned. This supports the interpretation in the main text that, when $M_i$ is small, a single smooth trend may be insufficient to capture the full conditional distribution of variances, whereas the bin-specific prior used by discrete \jtlitrd{} provides a more flexible and better-calibrated fit.
\begin{figure}
    \centering
    \includegraphics[width=0.95\linewidth]{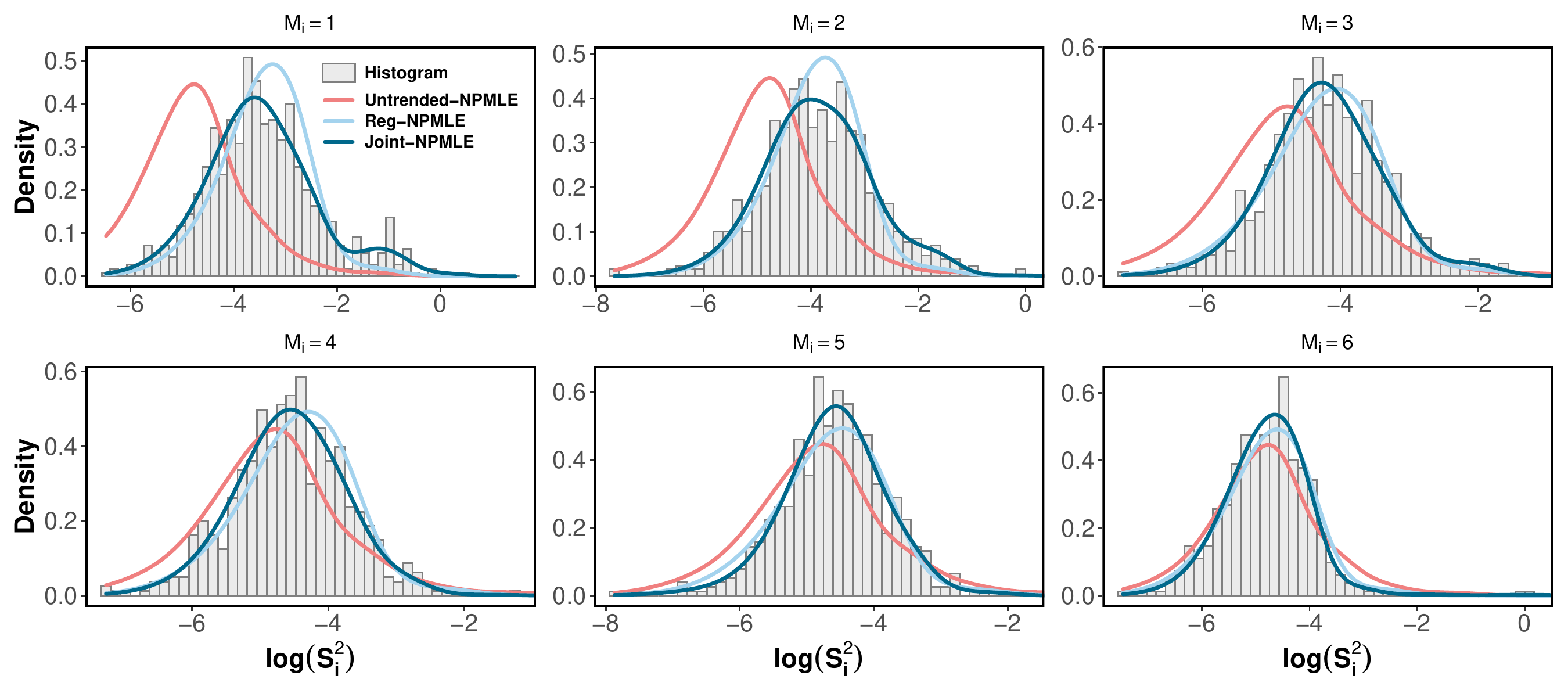}
    \caption{Marginal plots of the log sample variances $\log(S_i^2)$ for the proteomics data~\cite{zhu2020deqms}, shown separately for subsets with $M_i \in \{1,\dots,6\}$. In each panel, the histograms of the log sample variance $\log(S_i^2)$, is overlaid with the corresponding fitted marginal densities. We use 3 choices of estimated $G=\wh G$ for marginalization corresponding to Untrended-Inv$\chi^2$  (in \textcolor{ggplotRed}{red}), \reglitrd{} (in \textcolor{ggplotBlue}{blue}), and \jtlitrd{} obtained by the bin-based procedure (in \textcolor{ggplotBlue2}{blue}).}
    \label{fig:proteomics_dignostic}
\end{figure}

\subsection{Differential expression of \emph{Plasmodium falciparum} genes between severe and
non-severe malaria infections}
\label{sec:suppl_case_studies}

Here, we analyze a bulk RNA-seq dataset sourced from \cite{tonkin2018plasmodium} consisting of dual RNA-seq samples in which sequencing reads were aligned to both the human genome and the Plasmodium falciparum 3D7 reference genome. The primary unit of analysis is the parasite gene, and reads were allocated to parasite gene features using the P. falciparum 3D7 genome annotation, so that each count represents transcriptional evidence for a particular parasite gene in a given blood sample. We follow the pre-processing steps in \cite{tonkin2018plasmodium}, and we obtain $n=4,144$ genes for differential analysis. The main contrast of interest is the difference in parasite gene expression between severe and non-severe infections among $K=35$ samples. The design matrix includes an additional covariate for the estimated parasite life-cycle stage proportion (e.g., ring stage) to adjust for between-sample differences in stage composition, which otherwise can confound severity-associated differential expression, therefore resulting in $p=3$ covariates. Assumption~\ref{assu:design} holds in this data even when the ring stage covariate is continuous. We apply Benjamini-Hochberg to control the FDR at $\alpha = 0.05$ after computing p-values with all $6$ methods from Section~\ref{sec:bulk_rna_1}.

The results of significance are shown in Table \ref{tbl:realdata_significance_mal}, and Figure~\ref{fig:rnaseq2} presents the same diagnostic plots as in Section~\ref{sec:bulk_rna_1}. Untrended-Inv$\chi^2$ produces substantially more discoveries than the Untrended-NPMLE and the methods that incorporate a mean–variance trend. This inflation is plausibly driven by model misspecification. In particular, Figure \ref{fig:rnaseq2}(a) indicates that the marginal distribution of $S_i^2$ implied by the parametric prior of Untrended-Inv$\chi^2$ is noticeably misaligned with the empirical histogram, suggesting that this parametric fit does not adequately represent the variance distribution in the data, and may lead to more false discoveries.

\begin{table}[t]
\caption{Number of discoveries (BH at target FDR $\alpha$) in the \emph{P. falciparum} RNASeq study of Supplement~\bluelink{sec:suppl_case_studies}.}
\label{tbl:realdata_significance_mal}
\centering
{\small
\setlength{\tabcolsep}{4pt}
\renewcommand{\arraystretch}{1.15}
\begin{tabular*}{\textwidth}{@{\extracolsep{\fill}}@{}p{6.1cm}
                            *{6}{S[table-format=4.0]}@{}}
\toprule
& \multicolumn{1}{c}{\textbf{Classical}}
& \multicolumn{2}{c}{\textbf{Untrended}}
& \multicolumn{2}{c}{\textbf{Regression}}
& \multicolumn{1}{c}{\textbf{Joint}} \\
\cmidrule(lr){2-2}\cmidrule(lr){3-4}\cmidrule(lr){5-6}\cmidrule(lr){7-7}
\textbf{Contrast}
& \multicolumn{1}{c}{\textbf{t-test}}
& \multicolumn{1}{c}{\textbf{Inv$\chi^2$}}
& \multicolumn{1}{c}{\textbf{NPMLE}}
& \multicolumn{1}{c}{\textbf{Inv$\chi^2$}}
& \multicolumn{1}{c}{\textbf{NPMLE}}
& \multicolumn{1}{c}{\textbf{NPMLE}} \\
\midrule
\dshead{\textsc{RNA-seq}: \emph{P.\ falciparum} malaria {\hypersetup{hidelinks}\citep{tonkin2018plasmodium}}}{$n=4{,}144$ genes, $K=35$, $p=3$, $\alpha=0.05$}
\ct{Severe vs Non-severe} & 13 & 77 & 11 & 6 & 6 & 10 \\
\bottomrule
\end{tabular*}
}
\end{table}

\begin{figure}
\centering
\begin{tabular}{@{}l@{\hspace{0.03\textwidth}}l@{\hspace{0.03\textwidth}}l@{}}
(a) & (b) & (c) \\[0.3em]
\includegraphics[width=0.31\textwidth]{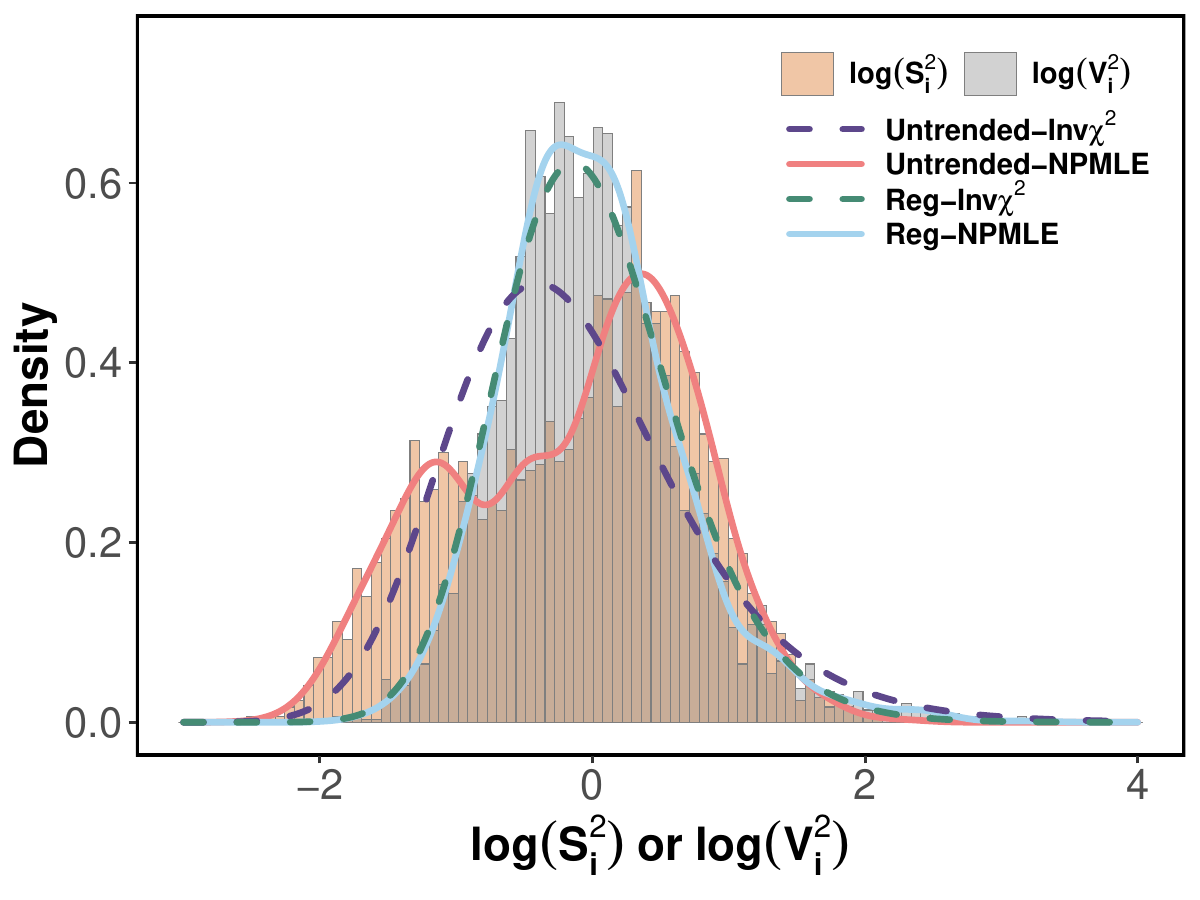} &
\includegraphics[width=0.31\textwidth]{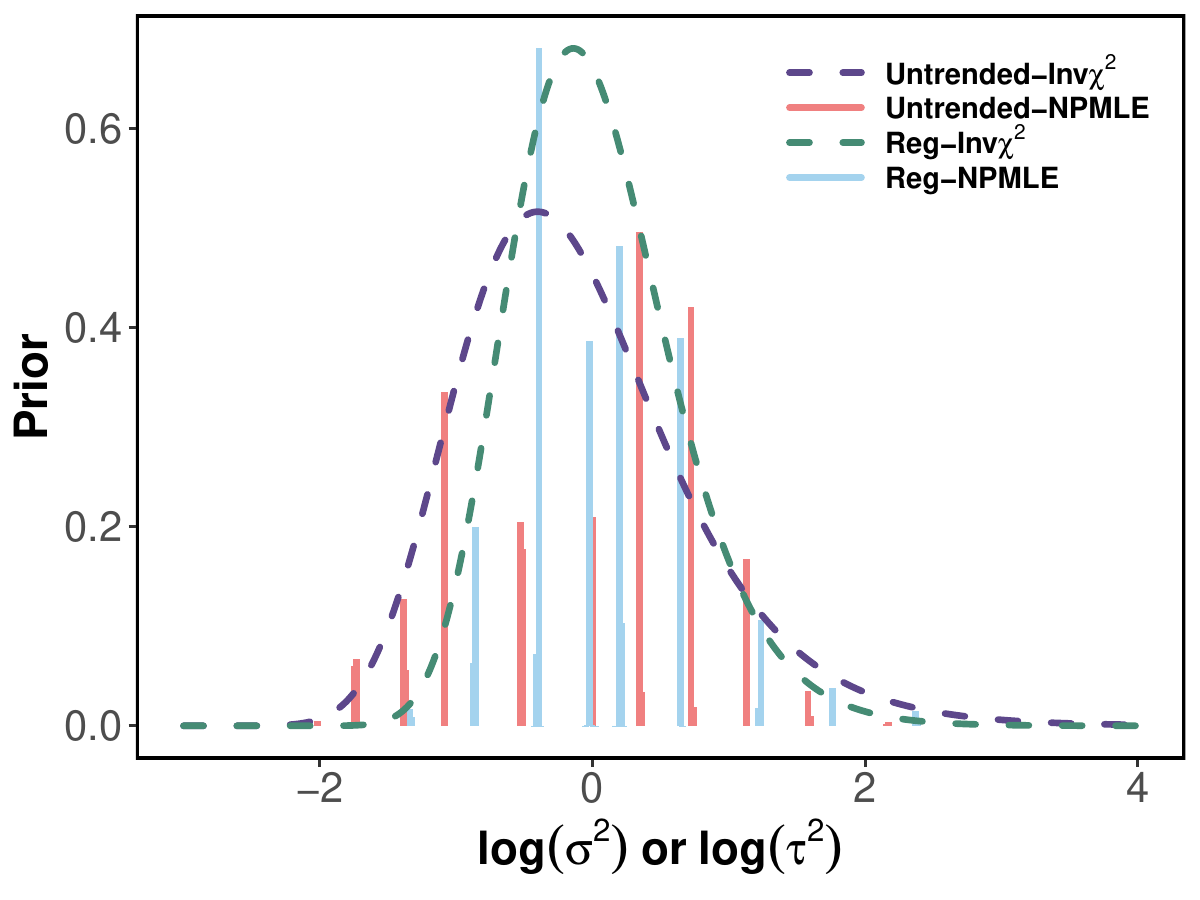} &
\includegraphics[width=0.31\textwidth]{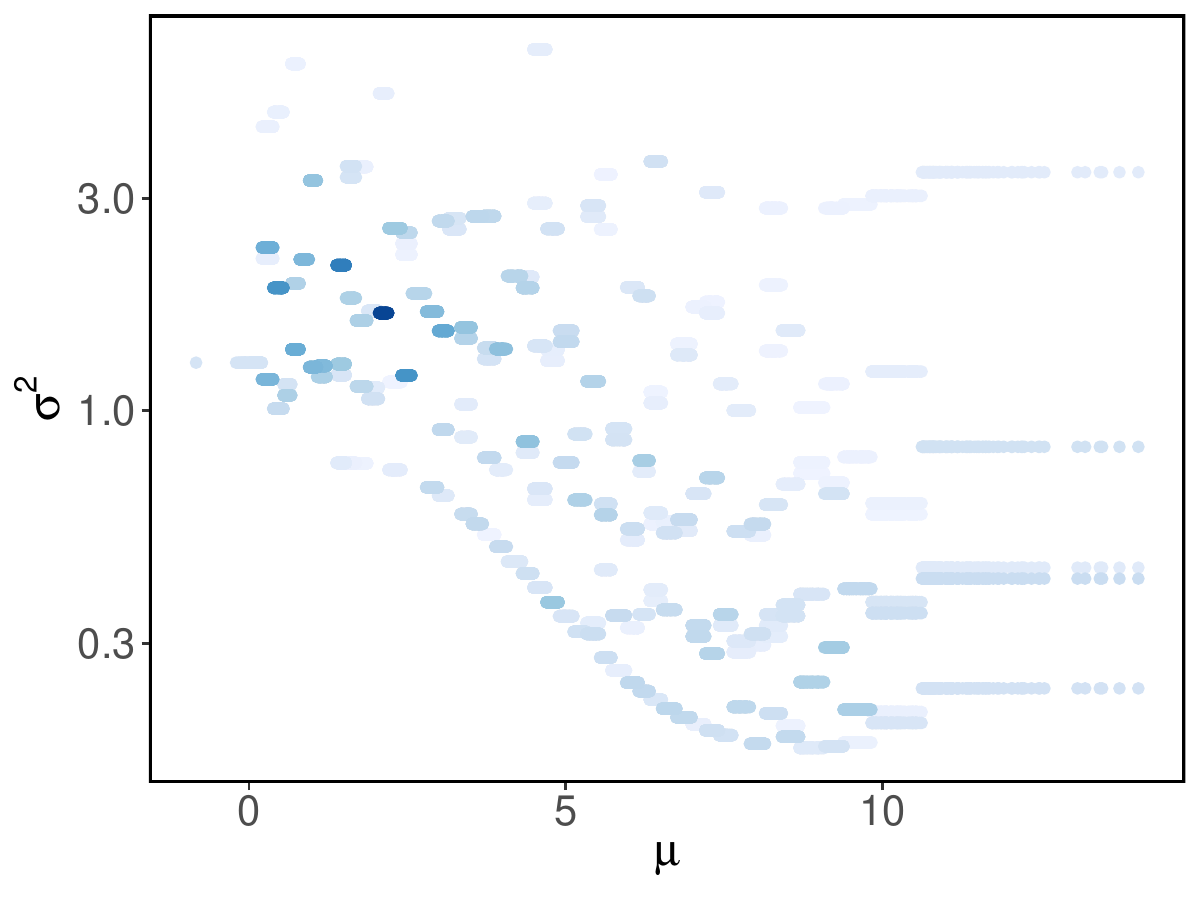} 
\end{tabular}

\caption{Empirical partially Bayes analysis without and with trend for RNA-Seq data of \cite{tonkin2018plasmodium}: each unit $i$ corresponds to a distinct Plasmodium falciparum gene $(n=4,144)$ quantified from bulk RNA-seq of infected blood samples, and its data is summarized as in \eqref{eq:limma_trend_summary} with $K=35$ and $p=3$. The three panels are analogous to the three panels of Figure \ref{fig:rnaseq1}.}
\label{fig:rnaseq2}
\end{figure}

\section{Mathematical notation}
We shall adopt the following notations to simplify the presentation of our mathematical arguments in the subsequent discussion. For two sequences $\{a_n\}$ and $\{b_n\}$, we shall say that $a_n \lesssim_{\square} b_n$ if there exists constants $C_1>0$ (depending on the quantities expressed in $\square$ but independent of $n$) such that $a_n \le C_1b_n$, for all $n \in \mathbb N$. Similarly, the notation $a_n \asymp_\square b_n$ shall mean that there exists constants $C_2,C_3>0$ (depending on the quantities expressed in $\square$ but independent of $n$) such that $C_2b_n \le a_n \le C_3b_n$, for all $n \in \mathbb N$. If there are no squares, then we just mean that the constants involved in the inequalities are independent of $n$ but do not stress the dependence of such constants on other problem parameters. The sequence $a_n =o(b_n)$, if $a_n/b_n \rightarrow 0$, as $n \rightarrow \infty$ and $a_n=\omega(b_n)$ is $a_n/b_n \rightarrow \infty$, as $n \rightarrow \infty$. Furthermore, we shall denote the set of null indices by $\mathcal H_0:=\{i:\theta_i=0\}$.

\section{Proofs of results in Section~\ref{sec:stat_setting}}
\subsection{Proof of Proposition \ref{prop:orthog}}
From the properties of ordinary least squares estimators, we can express
\[
Z_i = c^\top_\theta(X^\top X)^{-1}X^\top Y_i, \quad \wt A_i=c^\top_{\wt A}(X^\top X)^{-1}X^\top Y_i, \quad S^2_i=Y^\top_i(I-P_X)Y_i/(K-p),
\]
where $Y_i=(Y_{i1},\ldots,Y_{iK})^\top \in \mathbb R^K$ and $P_X$ is the orthogonal projection matrix onto $\mathcal C(X)$, the column space of $X$. Since, conditioned on $(\beta_i,\sigma^2_i)$, the random vector $Y_i$ is Gaussian for all $i \in [n]$, to conclude the proposition, it enough to conclude that
\begin{align}
    \mathrm{Cov}(Z_i,\wt A_i)=0, \quad \mathrm{Cov}(Z_i,(I-P_X)Y_i)=0, \quad \mbox{and} \quad \mathrm{Cov}(\wt A_i,(I-P_X)Y_i)=0, \quad \mbox{for all $i \in [n]$.}
\end{align}
It is easy to observe that the above conditions translate to
\begin{align}
    c^\top_\theta(X^\top X)^{-1}c_{\wt A}=0, \quad c^\top_\theta(X^\top X)^{-1}X^\top(I-P_X)=0, \quad \mbox{and} \quad c^\top_{\wt A}(X^\top X)^{-1}X^\top(I-P_X)=0.
\end{align}
The last two assertions always hold and the first condition is implied by Assumption~\ref{assu:design}. Hence the proposition follows.

\subsection{Proof of Example~\ref{ex:treatment_effect}}
Observe that since intercept is included in each of the fitted regression, therefore $\mathds{1}_K \in \mathcal C(X)$. Consequently, we have
\[
(X^\top X)^{-1}c_A=\frac{1}{K}(X^\top X)^{-1}X^\top \mathds{1}_K=\frac{1}{K}(X^\top X)^{-1}X^\top Xe_1=\frac{1}{K}e_1,
\]
where $e_1=(1,0,\ldots,0) \in \mathbb R^{q+2}$. Furthermore, since $\tau_i$ is not the intercept coefficient, $c^\top e_1=0$, which concludes the assertion in the example.

\section{Properties of chi-square mixtures}
Consider the following properties of mixtures of $\chi^2_{K-p}/(K-p)$ densities defined in \eqref{eq:lik_1d_limma_trend}.
\begin{lemm}\label{lem:prop_chi_sq}
For any mixing measure $\rmG\in\mathcal G_\trd$, the densities $f_{\rmG,K-p}$ defined in \eqref{eq:lik_1d_limma_trend} satisfy the following properties.
\begin{enumerate}

\item The derivative of $f_{\rmG,K-p}(s^2/\eta^2)$ with respect to $\eta \in (\underline M_\trd/2,2\wb M_\trd)$ is given by
\begin{align}
\frac{\partial f_{\rmG,K-p}(s^2/\eta^2)}{\partial \eta}
&= \int C_{K,p}\,\exp\left(-\frac{(K-p)s^2}{2\eta^2\tau^2}\right)\left(\frac{s^2}{\eta^2\tau^2}\right)^{\frac{K-p}{2}-1}\\
&\qquad\qquad\times\frac{1}{\tau^2}\left(\frac{2-K+p}{\eta}+\frac{(K-p)s^2}{\tau^2\eta^3}\right)\,\dd\rmG(\tau^2),
\end{align}
where $C_{K,p}=(K-p)^{(K-p)/2}/(2^{(K-p)/2}\Gamma((K-p)/2))$. Hence
\begin{align}
\label{eq:deriv_log_density}
\frac{\partial}{\partial\eta}\log f_{\rmG,K-p}(s^2/\eta^2)
= \mathbb E_{\tau^2 \sim \rmG}\left[\frac{2-K+p}{\eta}+\frac{(K-p)X}{\tau^2\eta}\,\Big|\,X=\frac{s^2}{\eta^2}\right],
\end{align}
where $X \sim \tau^2\chi^2_{K-p}/(K-p)$ (we shall use the same notation in the rest of the lemma).

\item The second derivative of $f_{\rmG,K-p}(s^2/\eta^2)$ with respect to $\eta$ satisfies 
\begin{align}
\frac{\partial^2}{\partial\eta^2}\log f_{\rmG,K-p}(s^2/\eta^2)
&= \mathrm{Var}_{\tau^2 \sim \rmG}\left[\frac{2-K+p}{\eta}+\frac{(K-p)X}{\tau^2\eta}\,\Big|\,X=\frac{s^2}{\eta^2}\right]\\
&\hskip 4em- \mathbb E_{\tau^2 \sim \rmG}\left[\frac{2-K+p}{\eta^2}+\frac{3(K-p)X}{\tau^2\eta^2}\,\Big|\,X=\frac{s^2}{\eta^2}\right],
\end{align}
for any $\eta \in (\underline M_\trd/2,2\wb M_\trd)$.

\item  Recall $V^2_i$ for $i=1,\ldots,n$ defined as $V_i^2=S^2_i/\xi_0^2(M_i) \sim \tau_i^2\chi^2_{K-p}/(K-p)$. If $K>p+2$ and the parameters $\tau_1^2,\ldots,\tau_n^2 \in [\underline L_\trd,\overline U_\trd]$, then there exists $\underline\kappa>0$ such that
\begin{align}
\label{eq:prop_4_chi_sq}
\max_{i\le n}\mathbb P\left(V_i^2/\tau_i^2\le \underline\kappa\,n^{-6/(K-p)}\mid \tau^2_i\right)\le n^{-3}.
\end{align}
Similarly, there exists $\overline\kappa>0$ such that
\begin{align}
\label{eq:prop_5_chi_sq}
\max_{i\le n}\mathbb P\left(V_i^2/\tau_i^2\ge 3\overline\kappa\log n\mid \tau^2_i\right)\le n^{-3}.
\end{align}
Furthermore, the constants $\underline\kappa$ and $\overline\kappa$ depend only on $K,p,\underline M_\trd,\wb M_\trd,\underline L_\trd,\overline U_\trd$. Similar results also hold for $S^2_i$ with possibly different constants.

\item There exists $\widetilde\kappa>0$, depending only on $K,p,\underline M_\trd,\wb M_\trd,\underline L_\trd,\wb L_\trd$, such that with $C_n=\widetilde\kappa\log n$,
\begin{align}
\label{eq:prop_score_tail_chi_sq_simple}
\sup_{\eta\in(\underline M_\trd/2,\,2\wb M_\trd)}
\left|
\int_{C_n}^{\infty}
\frac{\partial}{\partial \eta}
f_{\rmG,K-p}\!\left(\frac{s^2}{\eta^2}\right)\,\dd s^2
\right|
\lesssim_{K,p,\underline M_\trd,\wb M_\trd,\underline L_\trd,\overline U_\trd}
\frac{1}{n^2}.
\end{align}
\end{enumerate}
\end{lemm}
While we state the above lemma for scale mixtures of $\chi^2_{K-p}$ random variables, similar results continue to hold for scale mixtures of $\chi^2_r$ random variables for any degrees of freedom $r>2$. Furthermore, (3) also holds for $\misV{i}$ for $i=1,\ldots,n$ defined in Section~\ref{sec:reg_ltrd_int}, if Assumptions~\ref{assum:missp_trend_estimation} and~\ref{assu:misspec_limma_trnd_var} holds.

\begin{proof}
Recall
\[
p_{\chi^2}(x\mid K-p,\tau^2)
=
C_{K,p}\exp\!\left(-\frac{(K-p)x}{2\tau^2}\right)\left(\frac{x}{\tau^2}\right)^{\frac{K-p}{2}-1}\frac{1}{\tau^2},
\qquad
C_{K,p}=\frac{(K-p)^{(K-p)/2}}{2^{(K-p)/2}\Gamma((K-p)/2)},
\]
and for any $x>0$, the mixture density is
\[
f_{G,K-p}(x)=\int_0^\infty p_{\chi^2}(x\mid K-p,\tau^2)\,G(\dd\tau^2).
\]
Fix $s^2>0$ and $\eta\in(\underline M_\trd,\wb M_\trd)$ and set $x(\eta)=s^2/\eta^2$.

\paragraph{(1) First derivative.} Define $x(\eta)=s^2/\eta^2$. We compute $\partial_\eta f_{G,K-p}(x(\eta))$ by differentiating under the integral sign. For fixed $\tau^2$,
\[
\log p_{\chi^2}(x(\eta)\mid K-p,\tau^2)
=
\log C_{K,p}
-\frac{(K-p)x(\eta)}{2\tau^2}
+\Big(\frac{K-p}{2}-1\Big)\log x(\eta)
-\frac{K-p}{2}\log\tau^2
\]
By definition, we have $x'(\eta)=-2s^2/\eta^3$ and
\[
\frac{d}{d\eta}\log x(\eta)=-\frac{2}{\eta}.
\]
Hence
\[
\frac{\partial}{\partial\eta}\log p_{\chi^2}(x(\eta)\mid K-p,\tau^2)
=
\frac{(K-p)s^2}{\tau^2\eta^3}+\frac{2-K+p}{\eta}.
\]
Therefore
\[
\frac{\partial}{\partial\eta}p_{\chi^2}(x(\eta)\mid K-p,\tau^2)
=
p_{\chi^2}(x(\eta)\mid K-p,\tau^2)\left(\frac{2-K+p}{\eta}+\frac{(K-p)s^2}{\tau^2\eta^3}\right).
\]
Assumption~\ref{assu:limma_trnd_var} implies $\tau^2\in[\underline L_\trd,\overline U_\trd]$ almost surely under $G$, so the derivative is dominated by an integrable function of $\tau^2$. Differentiation under the integral sign is therefore justified, yielding
\[
\frac{\partial}{\partial\eta}f_{G,K-p}(x(\eta))
=
\int p_{\chi^2}(x(\eta)\mid K-p,\tau^2)\left(\frac{2-K+p}{\eta}+\frac{(K-p)s^2}{\tau^2\eta^3}\right)\,G(\dd\tau^2).
\]
Dividing by $f_{G,K-p}(x(\eta))$ yields the expression in \eqref{eq:deriv_log_density}.

\paragraph{(2) Second derivative.}
From \eqref{eq:deriv_log_density},
\[
\frac{\partial}{\partial\eta}\log f_{G,K-p}(s^2/\eta^2)
=
\mathbb E\left[\frac{2-K+p}{\eta}+\frac{(K-p)X}{\tau^2\eta}\,\Big|\,X=\frac{s^2}{\eta^2}\right].
\]
Differentiating once more and using dominated convergence yields
\begin{align*}
\frac{\partial^2}{\partial\eta^2}\log f_{G,K-p}(s^2/\eta^2)
&=
\mathrm{Var}\left(\frac{2-K+p}{\eta}+\frac{(K-p)X}{\tau^2\eta}\,\Big|\,X=\frac{s^2}{\eta^2}\right)\\
&\quad-\mathbb E\left[\frac{2-K+p}{\eta^2}+\frac{3(K-p)X}{\tau^2\eta^2}\,\Big|\,X=\frac{s^2}{\eta^2}\right],
\end{align*}
which is exactly the identity stated in the lemma.

\paragraph{(3) Tail bounds.}
For each $i$,
\[
(V_i^2/\tau_i^2) \mid \tau^2_i\sim\frac{\chi^2_{K-p}}{K-p}.
\]
Thus it suffices to bound the lower and upper tails of $\chi^2_{K-p}$. For the lower tail observe that, for $0<t\le1$,
\[
\mathbb P\left(\frac{\chi^2_{K-p}}{K-p}\le t\right)
\le
\frac{c_{K,p}}{2^{(K-p)/2}\Gamma((K-p)/2)}\int_0^t u^{(K-p)/2-1}\,\dd u
=
c_{K,p}\,t^{(K-p)/2},
\]
for some constant $c_{K,p}>0$.
Choosing $\underline\kappa>0$ sufficiently small (depending on $K,p,\underline L_\trd,\overline L_\trd$) such that $A_n=\underline\kappa n^{-6/(K-p)}\le1$ and $c_{K,p}A_n^{(K-p)/2}\le n^{-3}$ yields
\[
\max_{i\le n}\mathbb P\left(\frac{V_i^2}{\tau_i^2}\le A_n \mid \tau^2_i\right)\le n^{-3}.
\]
Similarly, for the upper tail, for $t\ge K-p$, a Chernoff bound gives
\[
\mathbb P\left(\frac{\chi^2_{K-p}}{K-p}\ge t\right)\le \widetilde c_{K,p}\exp\!\left(-\frac{t}{4}+\frac{K-p}{2}\log2\right),
\]
for some constant $\widetilde c_{K,p}>0$.
Choosing $\overline\kappa>0$ large enough (depending on $K,p,\underline L_\trd,\overline U_\trd$) such that $B_n=3\overline\kappa\log n\ge K-p$ and 
the above bound is at most $n^{-3}$ yields
\[
\max_{i\le n}\mathbb P\left(\frac{V_i^2}{\tau_i^2}\ge B_n \mid \tau^2_i\right)\le n^{-3}.
\]
Analogous results follow for $S^2_i$ by noting that $V^2_i=S^2_i/\xi^2_0(A_i)$ and $\xi_0 \in \mathcal X$ defined in Assumption~\ref{assu:limma_trnd_var}.
\paragraph{(5) Tail bound for the integral of the score.}
Using the derivative expression from part (1),
\begin{align}
\frac{\partial}{\partial\eta}f_{G,K-p}\!\left(\frac{s^2}{\eta^2}\right)
&=
\int C_{K,p}\exp\!\left(-\frac{(K-p)s^2}{2\eta^2\tau^2}\right)\left(\frac{s^2}{\eta^2\tau^2}\right)^{(K-p)/2-1}\\
& \qquad \times \frac{1}{\tau^2}\left(\frac{2-K+p}{\eta}+\frac{(K-p)s^2}{\tau^2\eta^3}\right)\,G(\dd \tau^2).
\end{align}
Since $\eta\in(\underline M_\trd/2,2\wb M_\trd)$ and $\tau^2\in[\underline L_\trd,\overline U_\trd]$, there exist constants $c,C>0$ depending only on $K,p,\underline M_\trd,\wb M_\trd,\underline L_\trd,\overline U_\trd$ such that
\[
\sup_{\eta}\left|\frac{\partial}{\partial\eta}f_{G,K-p}\!\left(\frac{s^2}{\eta^2}\right)\right|\le C e^{-cs^2}\{(s^2)^{(K-p)/2-1}+(s^2)^{(K-p)/2}\}.
\]
Integrating over $[C_n,\infty)$ gives
\[
\sup_{\eta}\left|\int_{C_n}^\infty\frac{\partial}{\partial\eta}f_{G,K-p}\!\left(\frac{s^2}{\eta^2}\right)\dd s^2\right|
\le C\int_{C_n}^\infty e^{-cu}(u^{(K-p)/2-1}+u^{(K-p)/2})\dd u.
\]
Using the standard incomplete gamma bound,
\[
\int_t^\infty u^\alpha e^{-cu}\dd u\lesssim_{\alpha,c}e^{-ct}(1+t^\alpha),
\]
and taking $C_n=\widetilde\kappa\log n$ yields
\[
\sup_{\eta}\left|\int_{C_n}^\infty\frac{\partial}{\partial\eta}f_{G,K-p}\!\left(\frac{s^2}{\eta^2}\right)\dd s^2\right|
\lesssim e^{-c\widetilde\kappa\log n}(1+(\widetilde\kappa\log n)^{(K-p)/2}).
\]
Choosing $\widetilde\kappa$ sufficiently large gives the bound $O(n^{-2})$, proving \eqref{eq:prop_score_tail_chi_sq_simple}.
\end{proof}

Next, consider the following result.

\begin{lemm}\label{lem:entropy_joint_density_deriv}
Let 
\begin{align}
\label{eq:cal_f_bar_trd}
\wb{\mathcal F}_{\trd}=\Big\{f_{G,K-p}:\operatorname{supp}(G)\subset[\underline L_\trd,\overline U_\trd]\Big\},
\end{align}
where $f_{G,K-p}$\footnote{In this lemma and the subsequent proof, we shall occasionally abuse the notation slightly to use $f_G$ for $f_{G,K-p}$.} is defined in \eqref{eq:lik_1d_limma_trend}. For $\eta\in[\underline M_\trd/2,\,2\overline M_\trd]$, define the density and derivative semi-metrics on $(0,\infty)$ by
\[
\|f_{G_1}-f_{G_2}\|_{\infty,\eta}:=\sup_{s^2>0}\sup_{\eta\in[\underline M_\trd/2,\,2\overline M_\trd]}\left|f_{G_1,K-p}(s^2/\eta^2)-f_{G_2,K-p}(s^2/\eta^2)\right|,
\]
\[
\|f_{G_1}-f_{G_2}\|_{\infty,\partial,\eta}:=\sup_{s^2>0}\sup_{\eta\in[\underline M_\trd/2,\,2\overline M_\trd]}\left|\frac{\partial}{\partial\eta}f_{G_1,K-p}(s^2/\eta^2)-\frac{\partial}{\partial\eta}f_{G_2,K-p}(s^2/\eta^2)\right|.
\]
Let $\|\cdot\|_{\mathrm{joint}}:=\|\cdot\|_{\infty,\eta}\vee \|\cdot\|_{\infty,\partial,\eta}$. Then there exists a constant $C>0$, depending only on $K,p,\underline L_\trd,\overline U_\trd,\underline M_\trd,\overline M_\trd$, such that for all $\varepsilon\in(0,1)$,
\begin{align}
\label{eq:entropy_joint_density_deriv_global}
\log N\!\left(\varepsilon,\wb{\mathcal F}_{\trd},\|\cdot\|_{\mathrm{joint}}\right)\le C\bigl(\log(1/\varepsilon)\bigr)^2.
\end{align}
Moreover, the same bound holds if the suprema over $s^2>0$ in the definitions of $\|\cdot\|_{\infty,\eta}$ and $\|\cdot\|_{\infty,\partial,\eta}$ are replaced by suprema over any interval $[A_n,B_n]$ with $0<A_n<B_n<\infty$.
\end{lemm}
\begin{proof}
For any semi-metrics $d_1,d_2$ on a class $\mathcal F$, the elementary product-cover bound yields
\[
N(\varepsilon,\mathcal F,d_1\vee d_2)\le N(\varepsilon,\mathcal F,d_1)\,N(\varepsilon,\mathcal F,d_2),
\]
and hence
\[
\log N(\varepsilon,\mathcal F,d_1\vee d_2)\le \log N(\varepsilon,\mathcal F,d_1)+\log N(\varepsilon,\mathcal F,d_2).
\]
We apply this with $\mathcal F=\wb{\mathcal F}_{\trd}$, $d_1=\|\cdot\|_{\infty,\eta}$, and $d_2=\|\cdot\|_{\infty,\partial,\eta}$, which yields
\[
\log N\!\left(\varepsilon,\wb{\mathcal F}_\trd,\|\cdot\|_{\mathrm{joint}}\right)\le \log N\!\left(\varepsilon,\wb{\mathcal F}_\trd,\|\cdot\|_{\infty,\eta}\right)+\log N\!\left(\varepsilon,\wb{\mathcal F}_\trd,\|\cdot\|_{\infty,\partial,\eta}\right).
\]
By Lemma~S.4 of \citet{ignatiadis2025empirical}, there exists $C_1>0$ depending only on the problem parameters $K,p,\underline L_\trd,\overline U_\trd,\underline M_\trd,\overline M_\trd$ such that for all $\varepsilon\in(0,1)$,
\[
\log N\!\left(\varepsilon,\wb{\mathcal F}_\trd,\|\cdot\|_{\infty,\eta}\right)\le C_1\bigl(\log(1/\varepsilon)\bigr)^2.
\]
Next, we consider the derivative term. For $\tau^2\in[\underline L_\trd,\overline U_\trd]$, differentiating under the integral sign gives
\[
\frac{\partial}{\partial\eta}f_G(s^2/\eta^2)=\int \frac{\partial}{\partial\eta}p_{\chi^2}(s^2/\eta^2\mid K-p,\tau^2)\,G(\dd \tau^2).
\]
From Lemma~\ref{lem:prop_chi_sq}, we can conclude that there exists a constant $C_0>0$ depending only on the problem parameters $K,p,\underline L_\trd,\overline U_\trd,\underline M_\trd,\overline M_\trd$ such that for all $s^2>0$, $\eta\in[\underline M_\trd/2,2\overline M_\trd]$, and $\tau^2\in[\underline L_\trd,\overline U_\trd]$,
\[
\left|\frac{\partial}{\partial\eta}p_{\chi^2}(s^2/\eta^2\mid K-p,\tau^2)\right|\le C_0(1+s^2)(s^2)^{(K-p)/2-1}\exp(-C_0^{-1}s^2).
\]
Consequently,
\[
\sup_{G\in\wb{\mathcal F}_\trd}\sup_{\eta\in[\underline M_\trd/2,2\overline M_\trd]}\left|\frac{\partial}{\partial\eta}f_G(s^2/\eta^2)\right|\le C_0(1+s^2)(s^2)^{(K-p)/2-1}\exp(-C_0^{-1}s^2).
\]
Choose $B(\varepsilon):=C_2\log(1/\varepsilon)$, for an appropriate constant $C_2>0$ so that we have 
\[
\sup_{s^2\ge B(\varepsilon)}\sup_{G}\sup_{\eta}\left|\partial_\eta f_G(s^2/\eta^2)\right|\le \varepsilon/2.\]
Then any $\varepsilon$-cover of $\{f_G:G\in\wb{\mathcal F}_\trd\}$ in the semi-metric $\|\cdot\|_{\infty,\partial,\eta}$ may be constructed by covering only the truncated domain $s^2\in[0,B(\varepsilon)]$. By the above truncation,
\[
\log N\!\left(\varepsilon,\wb{\mathcal F}_\trd,\|\cdot\|_{\infty,\partial,\eta}\right)\le \log N\!\left(\varepsilon,\wb{\mathcal F}_\trd,\|\cdot\|_{\infty,\partial,[0,B(\varepsilon)],\eta}\right),
\]
where
\[
\|f_{G_1}-f_{G_2}\|_{\infty,\partial,[0,B(\varepsilon)],\eta}:=\sup_{s^2\in[0,B(\varepsilon)]}\sup_{\eta\in[\underline M_\trd/2,\,2\overline M_\trd]}\left|\frac{\partial}{\partial\eta}f_{G_1}(s^2/\eta^2)-\frac{\partial}{\partial\eta}f_{G_2}(s^2/\eta^2)\right|.
\]
On $s^2\in[0,B(\varepsilon)]$, $\eta\in[\underline M_\trd/2,2\overline M_\trd]$, and $\tau^2\in[\underline L_\trd,\overline U_\trd]$, the derivative kernel $(s^2,\eta,\tau^2)\mapsto \partial_\eta f_{\tau^2}(s^2/\eta^2)$ is smooth and uniformly bounded. Therefore, the same discretization-and-moment-matching construction used to prove Lemma~S.4 of \citet{ignatiadis2025empirical} applies to this derivative kernel on $[0,B(\varepsilon)]$, yielding a constant $C_2>0$ depending only on $\nu,\underline L_\trd,\overline U_\trd,\underline M_\trd,\overline M_\trd$ such that
\[
\log N\!\left(\varepsilon,\wb{\mathcal F}_\trd,\|\cdot\|_{\infty,\partial,\eta}\right)\le C_2\bigl(\log(1/\varepsilon)\bigr)^2.
\]

Combining the density and derivative bounds completes the proof with $C:=C_1+C_2$.
\end{proof}

\begin{lemm}
\label{lem:derivative_g_trd}
Fix $s^2>0$ and $|z|\ge \underline z$. Define
\begin{align}
\wh g_{\trd}(\eta;z,s^2)
&= C_{K,p}
\int_{0}^{\infty}
\frac{(t^2)^{-\frac{K-p+1}{2}} \left(s^2/\eta^2\right)^{\frac{K-p}{2}-1}}
{\sqrt{(K-p+1)t^2-(K-p)\left(s^2/\eta^2\right)}}
f_{\widehat{G}_{\trd},K-p+1}(t^2)\\
&~~~~~~~~~~~~~~~~~~~~\times \mathds{1}\!\left\{t^2 \ge \frac{(K-p)\left(s^2/\eta^2\right)+z^2/(\nu^2\eta^2)}{K-p+1}\right\}\, \dd t^2,
\end{align}
for $\eta\in(0,\infty)$. Then
\begin{align}
\frac{d}{d\eta}\wh g_{\trd}(\eta;z,s^2)
&= -C_{K,p}(K-p-2)\frac{s^2}{\eta^3}
\int_{0}^{\infty}
\frac{(t^2)^{-\frac{K-p+1}{2}} \left(s^2/\eta^2\right)^{\frac{K-p}{2}-2}}
{\sqrt{(K-p+1)t^2-(K-p)\left(s^2/\eta^2\right)}}
f_{\widehat{G}_{\trd},K-p+1}(t^2)\\
&~~~~~~~~~~~~~~~~~~~~~~~~~~~~~~~~~~~~~~~~~~~~~~~\times\mathds{1}\!\left\{t^2 \ge \frac{(K-p)\left(s^2/\eta^2\right)+z^2/(\nu^2\eta^2)}{K-p+1}\right\}
\, \dd t^2 \notag\\
&\quad -C_{K,p}(K-p)\frac{s^2}{\eta^3}
\int_{0}^{\infty}
\frac{(t^2)^{-\frac{K-p+1}{2}} \left(s^2/\eta^2\right)^{\frac{K-p}{2}-1}}
{\left((K-p+1)t^2-(K-p)\left(s^2/\eta^2\right)\right)^{3/2}}
f_{\widehat{G}_{\trd},K-p+1}(t^2)\\
&~~~~~~~~~~~~~~~~~~~~~~~~~~~~~~~~~~~~~~~~~~~~~~~\times\mathds{1}\!\left\{t^2 \ge \frac{(K-p)\left(s^2/\eta^2\right)+z^2/(\nu^2\eta^2)}{K-p+1}\right\}
\, \dd t^2 \notag\\
&\quad + C_{K,p}\,
f_{\widehat{G}_{\trd},K-p+1}\!\left(\frac{(K-p)\left(s^2/\eta^2\right)+z^2/(\nu^2\eta^2)}{K-p+1}\right)\times
\frac{2\big((K-p)s^2+z^2/\nu^2\big)}{(K-p+1)\eta^3}\\
&~~~~~~~~~~~~~~~~~~~~
\times\left(\frac{(K-p)\left(s^2/\eta^2\right)+z^2/(\nu^2\eta^2)}{K-p+1}\right)^{-\frac{K-p+1}{2}}
\frac{\left(s^2/\eta^2\right)^{\frac{K-p}{2}-1}}{\sqrt{z^2/(\nu^2\eta^2)}}
\end{align}
\end{lemm}

\begin{proof}
Write
\begin{align}
\wh g_{\trd}(\eta;z,s^2)
&=
C_{K,p}
\int_{\frac{ (K-p)\left(s^2/\eta^2\right) + z^2/(\nu^2\eta^2)}{K-p+1}}^{\infty}
\frac{ (t^2)^{-\,\frac{K-p+1}{2}} 
       \left(s^2/\eta^2\right)^{\frac{K-p}{2}-1} }
     { \sqrt{ (K-p+1)t^2 - (K-p)\left(s^2/\eta^2\right) } }
f_{\widehat{G}_{\trd},K-p+1}(t^2)\,\dd t^2 .
\end{align}
We differentiate this integral with respect to $\eta$ using Leibniz' rule.

Set
\begin{align}
\Phi(\eta,t^2)
:=
\frac{ (t^2)^{-\,\frac{K-p+1}{2}} 
       \left(s^2/\eta^2\right)^{\frac{K-p}{2}-1}}
     { \sqrt{ (K-p+1)t^2 - (K-p)\left(s^2/\eta^2\right) } }
f_{\widehat{G}_{\trd},K-p+1}(t^2),
\end{align}
and
\begin{align}
a(\eta):=
\frac{ (K-p)\left(s^2/\eta^2\right) + z^2/(\nu^2\eta^2)}
     { K-p+1 }.
\end{align}
Then
\begin{align}
\wh g_{\trd}(\eta;z,s^2)
=
C_{K,p}\int_{a(\eta)}^\infty \Phi(\eta,t^2)\,dt^2.
\end{align}
Hence, by Leibniz' rule,
\begin{align}
\frac{d}{d\eta}\wh g_{\trd}(\eta;z,s^2)
=
C_{K,p}\left\{
\int_{a(\eta)}^\infty \frac{\partial}{\partial \eta}\Phi(\eta,t^2)\,dt^2
-
a'(\eta)\Phi(\eta,a(\eta))
\right\}.
\end{align}
We now compute the two terms separately.

First, since
\begin{align}
\frac{d}{d\eta}\left(\frac{s^2}{\eta^2}\right)
=
-\frac{2s^2}{\eta^3},
\end{align}
we have
\begin{align}
\frac{\partial}{\partial \eta}
\left(s^2/\eta^2\right)^{\frac{K-p}{2}-1}
=
-\,(K-p-2)\frac{s^2}{\eta^3}
\left(s^2/\eta^2\right)^{\frac{K-p}{2}-2},
\end{align}
and
\begin{align}
\frac{\partial}{\partial \eta}
\left((K-p+1)t^2-(K-p)\left(s^2/\eta^2\right)\right)^{-1/2}
=
-\,(K-p)\frac{s^2}{\eta^3}
\left((K-p+1)t^2-(K-p)\left(s^2/\eta^2\right)\right)^{-3/2}.
\end{align}
Therefore
\begin{align}
\frac{\partial}{\partial \eta}\Phi(\eta,t^2)
&=
-\,(K-p-2)\frac{s^2}{\eta^3}
\frac{ (t^2)^{-\,\frac{K-p+1}{2}}
       \left(s^2/\eta^2\right)^{\frac{K-p}{2}-2} }
     { \sqrt{ (K-p+1)t^2-(K-p)\left(s^2/\eta^2\right) } }
f_{\widehat{G}_{\trd},K-p+1}(t^2)\notag\\
&\quad
-\,(K-p)\frac{s^2}{\eta^3}
\frac{ (t^2)^{-\,\frac{K-p+1}{2}}
       \left(s^2/\eta^2\right)^{\frac{K-p}{2}-1} }
     { \left((K-p+1)t^2-(K-p)\left(s^2/\eta^2\right)\right)^{3/2} }
f_{\widehat{G}_{\trd},K-p+1}(t^2).
\end{align}
Substituting this into the integral term and rewriting the lower limit using the indicator gives the first two terms in the claimed formula.

Next, for the boundary term,
\begin{align}
a'(\eta)
&=
\frac{d}{d\eta}
\left[
\frac{ (K-p)\left(s^2/\eta^2\right) + z^2/(\nu^2\eta^2)}
     { K-p+1 }
\right]\notag\\
&=
-\frac{2\big((K-p)s^2+z^2/\nu^2\big)}{(K-p+1)\eta^3}.
\end{align}
Moreover, by direct substitution,
\begin{align}
(K-p+1)a(\eta)-(K-p)\left(s^2/\eta^2\right)
=
\frac{z^2}{\nu^2\eta^2},
\end{align}
so that
\begin{align}
\Phi(\eta,a(\eta))
&=
\left(
\frac{ (K-p)\left(s^2/\eta^2\right) + z^2/(\nu^2\eta^2)}
     { K-p+1 }
\right)^{-\,\frac{K-p+1}{2}}
\frac{\left(s^2/\eta^2\right)^{\frac{K-p}{2}-1}}
     {\sqrt{z^2/(\nu^2\eta^2)}}\notag\\
&\hspace{3em}\times
f_{\widehat{G}_{\trd},K-p+1}\!\left(
\frac{ (K-p)\left(s^2/\eta^2\right) + z^2/(\nu^2\eta^2)}
     { K-p+1 }
\right).
\end{align}
Hence
\begin{align}
&-\,a'(\eta)\Phi(\eta,a(\eta))\\
&=
f_{\widehat{G}_{\trd},K-p+1}\!\left(
\frac{ (K-p)\left(s^2/\eta^2\right) + z^2/(\nu^2\eta^2)}
     { K-p+1 }
\right)\notag\\
&\hspace{1em}\times
\frac{2\big((K-p)s^2+z^2/\nu^2\big)}{(K-p+1)\eta^3}
\left(
\frac{ (K-p)\left(s^2/\eta^2\right) + z^2/(\nu^2\eta^2)}
     { K-p+1 }
\right)^{-\,\frac{K-p+1}{2}}
\frac{\left(s^2/\eta^2\right)^{\frac{K-p}{2}-1}}
     {\sqrt{z^2/(\nu^2\eta^2)}}.
\end{align}
Combining the integral term and the boundary term proves the result.
\end{proof}

\section{Proofs of validity for oracle p-values}
\subsection{Proof of Lemma~\ref{lem:bayes_trd_p_val_cs}}
Notice that if $\theta_i=0$, for any $t \ge 0$, by \eqref{eq:bayesian_cs_model_2}
\begin{align}
\label{eq:validity_seq_1}
    \mathbb P_G\left[|Z_i| \ge t \mid S^2_i,M_i\right] &= \mathbb E_G\left[\mathbb P_G\left[|Z_i| \ge t \mid \theta_i=0,\tau^2_i,S^2_i,M_i\right]\mid S^2_i,M_i\right]\\
    &=\mathbb E_G\left[2\Phi(-t/\{\nu\xi_0(M_i)\tau_i\}) \mid S^2_i, M_i\right]:=\PregFun(t,S^2_i,M_i;G).
\end{align}
Also observe that, by \eqref{eq:popn_emp_bayes_p_val} for all $i \in [n]$
\begin{align}
\label{eq:def_preg_func}
P^\trd_i=\PregFun(Z_i,S^2_i,M_i;G).
\end{align}
Therefore, the lemma follows using the probability integral transform conditioned on $S^2_i$ and $M_i$.

\subsection{Proof of Lemma~\ref{lem:val_p_mis}}
Under \eqref{eq:misspecified-limma-reg} and \eqref{eq:tau_mis_v_mis}, the penultimate equality in \eqref{eq:validity_seq_1} does not hold. Nevertheless, one can directly condition on $\misV{i}$ and conclude that under the modified working model, for all $\theta_i=0$, or equivalently, $\omega_i=\theta_i/\xi_\mis(M_i)=0$, we define $O_i:=Z_i/\xi_\mis(M_i)$. Clearly,
\[
O_i \mid \misT{i} \sim \dnorm(\omega_i, \nu^2\misT{i}), \qquad \mbox{for all $i \in [n]$.}
\]
Therefore, we have for any $t \ge 0$,
\begin{align}
\label{eq:validity_seq_2}
    \mathbb P_{G_\mis}\left[|O_i| \ge t \mid \misV{i}\right] &= \mathbb E_{G_\mis}\left[\mathbb P_{G_\mis}\left[|O_i| \ge t \mid \omega_i=0,\misT{i},\misV{i}\right]\mid \misV{i}\right]\\
    &=\mathbb E_{G_\mis}\left[2\Phi(-t/\{\nu\tau_{i,\mis}\}) \mid \misV{i}\right]:=\PregFun_\mis(t,\misV{i};G_\mis).
\end{align}
Again, by definition in \eqref{eq:popn_emp_bayes_p_val}, we have
\[
P^\trd_{\mis,i}:=\PregFun_\mis(O_i,\misV{i};G_\mis),
\]
we can conclude using the probability integral transform that under \eqref{eq:misspecified-limma-reg} and \eqref{eq:tau_mis_v_mis}
\[
P^\trd_{\mis,i} \mid \omega_i=0,\misV{i} \sim \mathrm{Unif}(0,1), \qquad \mbox{almost surely, under the mixing measure $G_\mis$ on $\misT{i}$.}
\]

\subsection{Proof of Lemma~\ref{lem:valid_oracle_p_vals}}
Observe that under Assumption~\ref{assu:design}, the triple $(Z_i,S^2_i,A_i)$ satisfy \eqref{eq:def_z_i_dis_2d} and hence we can conclude that 
\begin{align}
\label{eq:validity_seq_3}
    \mathbb P_H\left[|Z_i| \ge t \mid S^2_i,A_i\right] &= \mathbb E_H\left[\mathbb P_H\left[|Z_i| \ge t \mid \theta_i=0,\sigma^2_i,\mu_i,S^2_i,A_i\right]\mid S^2_i,A_i\right]\\
    &=\mathbb E_H\left[2\Phi(-t/\{\nu\sigma_i\}) \mid S^2_i,A_i\right]:=\PjtFun(t,S^2_i,A_i;H).
\end{align}
Finally, using
\[
P^\jt_i=\PjtFun(Z_i,S^2_i,A_i;H),
\]
and the probability integral transform, the lemma follows.

\section{Results for parametric prior specification}
\label{sec:param_prior}
In this section, we collect results when $G$ in \eqref{eq:np-model-trnd} or $H$ in \eqref{eq:limma_trend_2d_prior} belong to parametric families. This connects our results to the established results in \cite{smyth2004linear} or \cite{law2014voom}. 

\subsection{Parametric priors in \limmatrd{}}
Define the following version of pooled variances corresponding to each unit under consideration.
    \begin{align}
        \label{eq:param_pooled_var}
        \wt S^2_i:=\frac{(K-p)S^2_i+\kappa_0\xi^2_0(M_i)s^2_0}{(K-p)+\kappa_0}, \qquad \mbox{for $i \in [n]$.}
    \end{align}
Next, consider the following proposition.
\begin{proposition}
\label{prop:param_limtrd}
    Let the distribution of $(Z,S^2,M)$ be distributed as \eqref{eq:bayesian_cs_model_2} and \eqref{eq:reg-model-trnd}. Then 
    \begin{align}
    \label{eq:joint_t_stat}
    \wt T \mid S^2,M,\theta \sim t_{(K-p)+\kappa_0},
    \end{align}
    where $\wt T:=(Z-\theta)/(\nu \wt S)$ and 
    \begin{align}
    \label{eq:generic_wt_S}
        \wt S^2:=\frac{(K-p)S^2+\kappa_0\xi^2_0(M)s^2_0}{(K-p)+\kappa_0}
    \end{align}
\end{proposition}
\begin{proof}
    Let us first re-parametrize $\sigma^2$ by the precision $\tau^2:=\sigma^{-2}$. Observe that, using \eqref{eq:bayesian_cs_model_1} and \eqref{eq:reg-model-trnd}, the joint density of $(Z,S^2,\tau^2)$ given $M$ can be written as
    \begin{align}
        &p(z,s^2,\tau^2 \mid M)\\
        &\asymp (\tau^2)^{\frac{(K-p)+\kappa_0+1}{2}-1}\cdot\exp\left(-\frac{1}{2}\left(\frac{(z-\theta)^2\tau^2}{\nu^2}+(K-p)s^2\tau^2+\kappa_0s^2_0\xi^2_0(M)\tau^2\right)\right)\cdot (s^2)^{\frac{K-p}{2}-1},
    \end{align}
    where the symbol $\asymp$ implies equality up to multiplication by absolute constants depending on $K,p,\nu$ and $\kappa_0$. Therefore, the joint distribution of $(Z,S^2)$ given $M$ after marginalizing $\tau^2$ is given by
    \begin{align}
        p(z,s^2 \mid M) &\asymp (s^2)^{\frac{K-p}{2}-1} \cdot \left(\frac{(z-\theta)^2}{\nu^2}+(K-p)s^2+\kappa_0s^2_0\xi^2_0(M)\right)^{-\frac{(\kappa_0+K-p+1)}{2}}.
    \end{align}
    By the change of variable 
    \[
    \wt t:=\frac{z-\theta}{\nu}\cdot\left(\frac{(K-p)s^2+\kappa_0s^2_0\xi^2_0(M)}{(K-p)+\kappa_0}\right)^{-1/2},
    \]
    we can get that
     \begin{align}
        &p(\wt t,s^2 \mid M)\\
        &\asymp (s^2)^{\frac{K-p}{2}-1} \times \left\{\left(\frac{\wt t^2}{(K-p)+\kappa_0}+1\right)\left((K-p)s^2+\kappa_0s^2_0\xi^2_0(M)\right)\right\}^{-\frac{(\kappa_0+K-p+1)}{2}}\\
        &~~~~~~~~~~~~~~~~~~~~~~~~~~~~~~~~~~~~~~~~~~~~~~~\times \left(\frac{(K-p)s^2+\kappa_0s^2_0\xi^2_0(M)}{(K-p)+\kappa_0}\right)^{1/2}\\
        & \asymp \left(\wt t^2+(K-p)+\kappa_0\right)^{-\frac{(\kappa_0+K-p+1)}{2}} \times (s^2)^{\frac{K-p}{2}-1} \times \left((K-p)s^2+\kappa_0s^2_0\xi^2_0(M)\right)^{-\frac{(\kappa_0+K-p)}{2}}.
    \end{align}
    This implies conditioned on $M$, $S^2 \indep \wt T$ and $\wt T \sim t_{(K-p)+\kappa_0}$.
\end{proof}
Next, consider the following lemma connecting the p-values $\{P^\trd_i\}$ defined in \eqref{eq:popn_emp_bayes_p_val} to the limma trend p-values from \cite{law2014voom}. 
\begin{lemm}
    \label{lem:limm_trnd_G}
    Assume that
    \[
    \frac{1}{\sigma^2_i} \mid M_i \sim \frac{\chi^2_{\kappa_0}}{\kappa_0s^2_0\xi^2_0(M_i)}.
    \]
    Then the p-values $\{P^\trd_i\}$ defined in \eqref{eq:popn_emp_bayes_p_val} (or, equivalently \eqref{eq:oracle_limma_trend} with the aforementioned prior) satisfy
    \(
       P^\trd_i=P^\ltrdp_i= 2\wb{F}_{t,(K-p)+\kappa_0}(|Z_i/(\nu\wt S_i)|),
   \)
    where $\wb{F}_{t,d_0}$ is the survival function of a $t$ distribution with $d_0$ degrees of freedom and $\wt S^2_i$ defined in \eqref{eq:param_pooled_var}. This aligns with the \limmatrd{} p-values from \cite{law2014voom}. 
\end{lemm}
\begin{proof}
    Observe that $P^\trd_i$ is obtained by plugging in $(Z_i,S^2_i,M_i)$ in the generic formula
    \begin{align}
        \PregFun(z,s^2,m;G)&:=\mathbb P_G\left[|Z|\ge |z| \mid S^2=s, M=m\right],
    \end{align}
    where $(Z,S^2)$ are distributed as specified in \eqref{eq:bayesian_cs_model_1} and \eqref{eq:reg-model-trnd} with $\theta=0$.
    Using Proposition~\ref{prop:param_limtrd}, we get that if $G$ is same as \eqref{eq:reg-model-trnd}, then 
    \begin{align}
        \PregFun(z,s^2,m;G)&:=\mathbb P_G\left[|Z|\ge |z| \mid \theta=0, S^2=s, M=m\right]\\
        &= \mathbb P_G\left[|Z/(\nu\wt S)|\ge |z/(\nu\wt s)|\mid \theta=0, S^2=s, M=m\right]\\
        & = \mathbb P_G\left[|\wt T_{K-p+\kappa_0}|\ge |z/\nu\wt s|\right]\\
        & = 2\wb{F}_{t,(K-p)+\kappa_0}(|z/(\nu\wt s)|),
    \end{align}
    where $\wt S^2:=\frac{(K-p)S^2+\kappa_0\xi^2_0(M)s^2_0}{(K-p)+\kappa_0}$ and $\wt s^2:=\frac{(K-p)s^2+\kappa_0\xi^2_0(M)s^2_0}{(K-p)+\kappa_0}$. The above formula aligns with the \limmatrd{} p-values from \cite{law2014voom}.
\end{proof}

\subsection{Connection between partially Bayes and fully Bayesian hypothesis testing}
Next, we draw a connection with the partially Bayes approach of \limmatrd{} with fully Bayes hypothesis testing based on posterior odds. In that direction, let us impose the spike-and-slab prior on the $\theta_1,\ldots,\theta_n$, as follows:
\begin{align}
\label{eq:theta_ss_prior}
\theta_i \mid \sigma^2_i \sim \pi\delta_0+(1-\pi)\dnorm(0,v^2_0\sigma^2_i),
\end{align}
in addition to \eqref{eq:bayesian_cs_model_1} and \eqref{eq:reg-model-trnd}. Consider the  test statistics $\wt T_j=Z_j/(\nu \wt S_j)$ for $\wt S_j$ defined in \eqref{eq:param_pooled_var}. Then the posterior odds of the null hypothesis for testing the $j$-th hypothesis is given by
\begin{align}
\label{eq:bayes_factor_def}
\mathcal O_j:=\frac{p(\theta_j=0 \mid \wt T_j, S^2_j,M_j)}{p(\theta_j\neq 0 \mid \wt T_j, S^2_j,M_j)}.
\end{align}
Consider the following proposition on the Bayes factors.
\begin{proposition}
    \label{prop:bayes_factor}
Consider the prior on $\theta_i$ given by \eqref{eq:theta_ss_prior} along with the model from \eqref{eq:bayesian_cs_model_1} and \eqref{eq:reg-model-trnd}. Then the posterior odds $\mathcal O_1,\ldots,\mathcal O_n$ from \eqref{eq:bayes_factor_def} satisfies
\begin{align}
    \mathcal O_j & = \frac{\pi}{(1-\pi)}\cdot \left(\frac{\nu^2}{v^2_0+\nu^2}\right)^{-1/2}\cdot\left(\frac{\wt T^2_j+(K-p)+\kappa_0}{\wt T^2_j(\nu^2/(\nu^2+v^2_0))+(K-p)+\kappa_0}\right)^{-\frac{(K-p)+\kappa_0+1}{2}}, \quad \mbox{for all $j \in [n]$.}
\end{align}
    Furthermore, 
    \[
     \mathcal O_j \rightarrow \frac{\pi}{(1-\pi)}\cdot \left(\frac{\nu^2}{v^2_0+\nu^2}\right)^{-1/2}\exp\left(-\frac{\wt T^2_j}{2}\frac{v^2_0}{\nu^2+v^2_0}\right), \qquad \mbox{if $(K-p)+\kappa_0 \rightarrow \infty$.}
    \]
\end{proposition}
\begin{proof}
    Observe that if $\theta_i=0$, then by Proposition~\ref{prop:param_limtrd}, $\wt T_j \indep (S^2_j,M_j)$ and $\wt T_j \sim t_{K-p+\kappa_0}$. Next, observe that if $\theta_j \neq 0$, the only modification to \eqref{eq:bayesian_cs_model_1} is that now
    \[
    Z_j \mid \sigma^2_j \sim \dnorm(0,(v^2_0+\nu^2)\sigma^2_j).
    \]
    Then the joint distribution of $(Z,S^2,\tau^2)$ where $Z \mid \sigma^2,M \sim \dnorm(0,(v^2_0+\nu^2)\tau^{-2})$ and $(S^2,\tau^2,M)$ distributed as \eqref{eq:bayesian_cs_model_1} and \eqref{eq:reg-model-trnd}, is given by
     \begin{align}
        &p(z,s^2,\tau^2 \mid M)\\
        &\asymp (\tau^2)^{\frac{(K-p)+\kappa_0+1}{2}-1}\cdot\exp\left(-\frac{1}{2}\left(\frac{z^2\tau^2}{(v^2_0+\nu^2)}+(K-p)s^2\tau^2+\kappa_0s^2_0\xi^2_0(M)\tau^2\right)\right)\cdot (s^2)^{\frac{K-p}{2}-1}.
    \end{align}
    After marginalizing over $\tau^2$, the joint distribution of $(Z,S^2)$ given $M$ after marginalizing $\tau^2$ is given by
    \begin{align}
        p(z,s^2 \mid M) &\asymp (s^2)^{\frac{K-p}{2}-1} \cdot \left(\frac{z^2}{2(v^2_0+\nu^2)}+\frac{(K-p)s^2}{2}+\frac{\kappa_0s^2_0\xi^2_0(M)}{2}\right)^{-\frac{(\kappa_0+K-p+1)}{2}}.
    \end{align}
    By the change of variable 
    \[
    \wb t:=\frac{z}{\sqrt{v^2_0+\nu^2}}\cdot\left(\frac{(K-p)s^2+\kappa_0s^2_0\xi^2_0(M)}{(K-p)+\kappa_0}\right)^{-1/2},
    \]
    we can get that
     \begin{align}
        &p(\wb t,s^2 \mid M)\\
        & \asymp \left(\wb t^2+(K-p)+\kappa_0\right)^{-\frac{(\kappa_0+K-p+1)}{2}} \times (s^2)^{\frac{K-p}{2}-1} \times \left((K-p)s^2+\kappa_0s^2_0\xi^2_0(M)\right)^{-\frac{(\kappa_0+K-p)}{2}}.
    \end{align}
    Therefore, $\wb T:=\frac{Z}{\sqrt{v^2_0+\nu^2}}\cdot\left(\frac{(K-p)S^2+\kappa_0s^2_0\xi^2_0(M)}{(K-p)+\kappa_0}\right)^{-1/2} \sim t_{(K-p)+\kappa_0}$. This implies
    \[
    \wt T_j \mid \theta_j \neq 0, M_j \sim \frac{\sqrt{v^2_0+\nu^2}}{\nu}t_{(K-p)+\kappa_0}.
    \]
    Also, the distribution of $S^2_j \mid \theta_j \neq 0, M_j$ is same as the distribution of $S^2 \mid \theta_j=0, M_j$ for all $j \in [n]$. Therefore, the Bayes factor of the $j$-th unit is given by
    \begin{align}
        \mathcal O_j & = \frac{\pi\,p(\wt T_j \mid \theta_j =0,M_j)}{(1-\pi)\,p(\wt T_j \mid \theta_j \neq 0,M_j)}\\
        & = \frac{\pi}{(1-\pi)}\cdot \left(\frac{\nu^2}{v^2_0+\nu^2}\right)^{-1/2}\cdot\left(\frac{\wt T^2_j+(K-p)+\kappa_0}{\wt T^2_j(\nu^2/(\nu^2+v^2_0))+(K-p)+\kappa_0}\right)^{-\frac{(K-p)+\kappa_0+1}{2}}.
    \end{align}
    Then
    \[
    \mathcal O_j \rightarrow \frac{\pi}{(1-\pi)}\cdot \left(\frac{\nu^2}{v^2_0+\nu^2}\right)^{-1/2}\exp\left(-\frac{\wt T^2_j}{2}\frac{v^2_0}{\nu^2+v^2_0}\right), \qquad \mbox{if $(K-p)+\kappa_0 \rightarrow \infty$.}
    \]
\end{proof}
In the Bayesian literature it is common to use the logarithm of the Bayes factor $\log \mathcal O_1,\ldots,\log O_n$ to rank the significance of the hypotheses. In contrast, in the frequentist literature it is common to rank the significance of the hypotheses by the p-values. This principle is adopted in the Benjamini-Hochberg procedure. The above result Proposition~\ref{prop:bayes_factor} and Lemma~\ref{lem:limm_trnd_G} shows that the partially Bayes p-values from \limmatrd{} results in the same ranking of the hypotheses (albeit, in the reverse order) as one would get for the fully Bayes framework.

\subsection{Parametric priors in \jtlitrd{}}
Let us consider the following proposition when $H$ in \eqref{eq:limma_trend_2d_prior} takes a parametric form.
\begin{proposition}
\label{prop:joint_parametric}
    Suppose the prior $H$ for $(\mu_i,\sigma_i^2)$ follows the following parametric hierarchical  model:
    \begin{align}
    \label{eq:2d_param_prior}
    \mu_i \mid\sigma_i^2 \sim \mathcal{N}(a_0,b_0\sigma_i^2),\quad\frac{1}{\sigma_i^2} \sim \frac{1}{\kappa_0s_0^2}\chi_{\kappa_0}^2,
    \end{align}
    where $a_0\in\mathbb{R}$ and $b_0,\kappa_0,s_0^2>0,$ are known constant. Then,
$$
\frac{1}{\sigma_i^2}\mid A_i=a
\sim
\left(\kappa_0s_0^2+\frac{(a-a_0)^2}{b_0+1/K}\right)^{-1}\chi^2_{\kappa_0+1}.
$$
Furthermore,
\[
\check T \mid S^2,A,\theta \sim  t_{(K-p)+\kappa_0+1},
\]
where $\check T:=(Z-\theta)/(\nu \check S)$, where $\check S^2:=\left\{\kappa_0s_0^2+(K-p)S^2+\frac{(A-a_0)^2}{(b_0+1/K)}\right\}/(K-p+\kappa_0+1)$.
Moreover, the p-value $P_i^{\jt}=2\wb{F}_{t,(K-p)+\kappa_0+1}(|Z_i/(\nu \check S_i)|)$, where
\[
\check{S}^2_i =\frac{\kappa_0s_0^2+(K-p)S^2_i+\frac{(A_i-a_0)^2}{(b_0+1/K)}}{(K-p)+\kappa_0+1},
\]
and for any $\wt{d}>0$, $\wb{F}_{t,\wt{d}}$ is the survival function of the t-distribution with $\wt{d}$ degrees of freedom.
\end{proposition}
\begin{proof}
    Let $\tau^2=\sigma^{-2}$. Consider $(Z,S^2,A,\mu,\sigma^2)$ distributed as \eqref{eq:def_z_i_dis_2d} and \eqref{eq:limma_trend_2d_prior} with $H$ taking the form in \eqref{eq:2d_param_prior}, we can conclude that 
    $$
    A\mid \tau^2 \sim \dnorm\left(a_0,\frac{b_0+1/K}{\tau^2}\right).
    $$
    Therefore, the joint likelihood of $(\tau^2,A)$ is given by
    \[
    p(\tau^2,a) \asymp \tau\exp\left\{-\frac{\tau^2}{2}\frac{(a-a_0)^2}{b_0+1/K}\right\}\times\exp\left(-\frac{\kappa_0s^2_0\tau^2}{2}\right)\times (\tau^2)^{\frac{\kappa_0}{2}-1}
    \]
    Here the symbol $\asymp$ implies the equality holds up to multiplication by an absolute constant depending $a_0,b_0,\kappa_0,K,p$ and $\nu$.
    This implies by conjugacy
    $$
    \frac{1}{\sigma^2}\mid A=a
    \sim
    \frac{1}{\,\kappa_0s_0^2+\frac{(a-a_0)^2}{b_0+1/K}\,}\chi^2_{\kappa_0+1}.
    $$
    Next, consider the joint density of $(Z,S^2,A,\mu,\sigma^2)$ given as follows.
    \begin{align*}
        &p(z,s^2,a,\mu,\tau^2)\\
        &\asymp \exp\left(-\frac{\tau^2}{2}\left\{\frac{(z-\theta)^2}{\nu^2}+(K-p)s^2+K(a-\mu)^2+\frac{(\mu-a_0)^2}{b_0}+\kappa_0s^2_0\right\}\right) \\
        &~~~~~~~~~~~\times (\tau^2)^{\frac{3+(K-p)+\kappa_0}{2}-1}(s^2)^{\frac{(K-p)}{2}-1}.
    \end{align*}
    We can marginalize the over $\mu$ and $\tau^2$ to obtain
    \begin{align}
    p(z,s^2,a) \asymp (s^2)^{\frac{K-p}{2}-1}\left(\frac{(z-\theta)^2}{\nu^2} + (K-p)s^2 + \kappa_0s_0^2 + \frac{(a - a_0)^2}{b_0 + 1/K}\right)^{-\frac{(K-p)+\kappa_0+2}{2}}.
    \end{align}
    Next, substituting 
    \[
    \check t:= \frac{z-\theta}{\nu} \cdot \left\{\frac{(K-p)s^2+\kappa_0s^2_0+\{(a-a_0)^2/(b_0+K^{-1})\}}{K-p+\kappa_0+1}\right\}^{-1/2},
    \]
    and proceeding as in the proof of Proposition~\ref{prop:param_limtrd}, we can show that $\check T \indep (S^2,A)$ and \eqref{eq:joint_t_stat} holds. The formula for p-values follows from
    \begin{align}
        \PjtFun(z,s^2,a;H):=\mathbb P_H\left[|Z| \ge |z| \mid S^2,A\right],
    \end{align}
    by repeating the arguments in the proof of Lemma~\ref{lem:limm_trnd_G}. 
\end{proof}

\section{Proof of asymptotic FDR control in Section~\ref{sec:reg_ltrd_ext}}
In this section, we prove the asymptotic FDR control result Theorem~\ref{thm:final_rate} under \eqref{eq:bayesian_cs_model_2} and \eqref{eq:well-specified-limma-reg}. We shall consistently assume Assumptions~\ref{assu:limma_trnd_var} and~\ref{assum:trend_estimation}.
\subsection{Approximate NPMLE property of $\wh G_\trd$}
\label{sec:approx_npmle_cs}
Observe that the solution $\wh G_\trd$ of \eqref{eq:def_small_ell} is not an exact NPMLE under \eqref{eq:bayesian_cs_model_2} and \eqref{eq:well-specified-limma-reg}. Therefore, we do not automatically have
\[
\frac{1}{n}\sum_{i=1}^n\log f_{\wh G_\trd,K-p}(S_i^2/\xi^2_0(M_i)) \ge \frac{1}{n}\sum_{i=1}^n\log f_{G,K-p}(S_i^2/\xi^2_0(M_i)).
\]
However, we shall show the following result.
\begin{lemm}
\label{lem:lik_ratio_limma_trnd}
There exist constants $C_\trd>0$ and $n_{\trd,1}\in\mathbb N_{\ge1}$, depending only on $h_1,h_2,K,p,\underline M_\trd,\wb M_\trd,\underline L_\trd,\wb U_\trd$, such that the following hold with probability at least $1-o(n^{-2})$
\begin{align}
\label{eq:log_lik_limma_trend_approx_cs}
\prod_{i=1}^n\frac{f_{\wh G_\trd,K-p}(S_i^2/\xi^2_0(M_i))}{f_{G,K-p}(S_i^2/\xi^2_0(M_i))}\ge e^{-n\,C_\trd\mathfrak R_{n,1}(\wh G_\trd)},
\end{align}
where for any $\rmG \in \mathcal G_\trd$
\[
\mathfrak R_{n,1}(\rmG)=\Delta_n^2(\log n)^2+\frac{(\log n)^2}{\sqrt{n}}\cdot \Delta_n+\frac{(\log n)}{\sqrt{n}}\cdot|\log \Delta_n|^{h_2/2}\,\Delta^{1-h_1/2}_n+\mathcal H\!\left(f_{\rmG,K-p},f_{G,K-p}\right)\Delta_n.
\]
\end{lemm}

\begin{proof}
    Recall that under the completely specified set-up \eqref{eq:bayesian_cs_model_2} and \eqref{eq:well-specified-limma-reg}, the conditional density of $S_i^2$ given $M_i=m$ under when the true trend is $\xi$ and the mixing measure is $\rmG$ is proportional to $\eta^{-2}f_{\rmG,K-p}(S_i^2/\eta^2)$, where $\eta=\xi(m)$. Let us define the log-likelihood (under $\xi$ and $\rmG$) as follows:
\[
\ell(S_i^2;\xi,\rmG,M_i):=\log f_{\rmG,K-p}(S_i^2/\xi^2(M_i))-2\log\xi(M_i).
\]
Observe that in this framework, the following optimization program (as in \eqref{eq:def_small_ell})
\[
\argmax_{\rmG \in \mathcal G_\trd}\frac{1}{n}\sum_{i=1}^n
        \log f_{\rmG,K-p}\!\left(S_i^2/\widehat{\xi}^2(M_i)\right),
\]
is equivalent to 
\[
\argmax_{\rmG \in \mathcal G_\trd}\frac{1}{n}\sum_{i=1}^n\ell(S_i^2;\wh \xi,\rmG,M_i).
\]
By the construction of the NPMLE $\wh G_\trd$, 
\[
\frac{1}{n}\sum_{i=1}^n\ell(S_i^2;\wh\xi,\wh G_\trd,M_i)-\frac{1}{n}\sum_{i=1}^n\ell(S_i^2;\wh\xi,G,M_i)\ge0, \qquad \mbox{for $G$ defined in \eqref{eq:well-specified-limma-reg}.}
\]
Our goal is to show that with probability greater than $1-o(n^{-2})$, there exists a constant $C_\trd>0$ such that
\begin{align}
\frac{1}{n}\sum_{i=1}^n\ell(S_i^2;\xi_0,\wh G_\trd,M_i)&\ge \frac{1}{n}\sum_{i=1}^n\ell(S_i^2;\xi_0,G,M_i)-C_\trd\,\mathfrak R_{n,1}(\wh G_\trd)
\end{align}
where the function $\mathfrak R_{n,1}(\wh G_\trd)$ is defined in the lemma statement. This is also the approximate NPMLE property under the correctly scaled likelihood as defined in \citet{JiangZhang2009}. Define
\begin{align}
\label{eq:def_sub_1}
\mathrm{Sub}_{n,1}=\Big|\frac{1}{n}\sum_{i=1}^n\ell(S_i^2;\xi_0,\wh G_\trd,M_i)-\frac{1}{n}\sum_{i=1}^n\ell(S_i^2;\wh\xi,\wh G_\trd,M_i)\Big|
\end{align}
and
\begin{align}
\label{eq:def_sub_2}
\mathrm{Sub}_{n,2}=\Big|\frac{1}{n}\sum_{i=1}^n\ell(S_i^2;\xi_0,G,M_i)-\frac{1}{n}\sum_{i=1}^n\ell(S_i^2;\wh\xi,G,M_i)\Big|.
\end{align}
Then, we also have
\[
\frac{1}{n}\sum_{i=1}^n\ell(S_i^2;\xi_0,\wh G_\trd,M_i)\ge \frac{1}{n}\sum_{i=1}^n\ell(S_i^2;\xi_0,G,M_i)-\mathrm{Sub}_{n,1}-\mathrm{Sub}_{n,2}.
\]
So it suffices to show that
\begin{align}
\mathrm{Sub}_{n,1}+\mathrm{Sub}_{n,2} &\le C_\trd\,\mathfrak R_{n,1}(\wh G_\trd),
\end{align}
with probability greater than $1-o(n^{-2})$.
Recall $V_i^2=S_i^2/\xi_0^2(M_i)$. Define
\begin{align}
\label{eq:cala_trd}
\mathcal A_\trd=\{|M_i|\le \mathrm W_n,\;V_i^2\in[\widetilde A_n,\widetilde B_n]\ \text{for all }i\in[n]\},
\end{align}
where $\widetilde A_n=A_n/(2\wb M_\trd)^2$ and $\widetilde B_n=4B_n/\underline M_\trd^2$, with $A_n=\underline\kappa n^{-6/(K-p)}$ and $B_n=(3\overline\kappa\vee\widetilde\kappa)\log n$ as in Lemma~\ref{lem:prop_chi_sq}. By Assumptions~\ref{assu:limma_trnd_var}-\ref{assum:trend_estimation}, \eqref{eq:prop_4_chi_sq}, \eqref{eq:prop_5_chi_sq}, and a union bound, $\mathbb P(\mathcal A_\trd^c)\le 2n^{-2}$. Henceforth we work on $\mathcal A_\trd$.

For any mixing distribution $\rmG\in\mathcal G_\trd$, define
\begin{align}
\label{eq:def_d_and_h}
D(\eta,s^2;\rmG)&:=\frac{\partial}{\partial\eta}\log f_{\rmG,K-p}(s^2/\eta^2)-\frac{2}{\eta}, \quad \mbox{and}\\
H(\eta,s^2;\rmG)&:=\frac{\partial^2}{\partial\eta^2}\log f_{\rmG,K-p}(s^2/\eta^2)+\frac{2}{\eta^2}.
\end{align}
If we assume $\eta=\xi(m)$, where $\xi \in \mathcal X$ (defined in Assumption~\ref{assu:limma_trnd_var}), then $\eta\in[\underline M_\trd/2,2\wb M_\trd]$, and therefore, if $s^2 \in [A_n,B_n]$, the additional terms $-2/\eta$ and $2/\eta^2$ only change constants in the bounds from Lemma~\ref{lem:prop_chi_sq}. In particular, Lemma~\ref{lem:prop_chi_sq}(1) and Lemma~\ref{lem:prop_chi_sq}(3) imply that if $s^2 \in \left[(A_n \underline M^2_\trd)/(4\wb M^2_\trd),(4B_n\wb M^2_\trd)/(\underline M^2_\trd)\right]$ and $\eta \in [\underline M_\trd/2,2\wb M_\trd]$, then
\[
|D(\eta,s^2;\rmG)|\lesssim\log n, \quad \mbox{and} \quad |H(\eta,s^2;\rmG)|\lesssim(\log n)^2,
\]
uniformly over $\eta\in[\underline M_\trd/2,2\wb M_\trd]$, $s^2\in[ A_n,B_n]$, and $\rmG\in\mathcal G_\trd$.

\subsubsection{Bounding $\mathrm{Sub}_{n,1}$}

Applying a second-order Taylor expansion of $\ell(S_i^2;\wh\xi,\wh G_\trd,M_i)$ around $\xi_0(M_i)$ gives
\begin{align}
\label{eq:dec_sub_1}
&\frac{1}{n}\sum_{i=1}^n\ell(S_i^2;\wh\xi,\wh G_\trd,M_i)-\frac{1}{n}\sum_{i=1}^n\ell(S_i^2;\xi_0,\wh G_\trd,M_i)\\
&=\frac{1}{n}\sum_{i=1}^n D(\xi_0(M_i),S_i^2;\wh G_\trd)\big(\wh\xi(M_i)-\xi_0(M_i)\big)+\frac{1}{2n}\sum_{i=1}^n H(\eta_{t,i},S_i^2;\wh G_\trd)\big(\wh\xi(M_i)-\xi_0(M_i)\big)^2
\end{align}
for some $\eta_{t,i}=t_i\xi_0(M_i)+(1-t_i)\wh\xi(M_i)$ with $t_i\in[0,1]$. Since $|H(\eta_{t,i},S_i^2;\wh G_\trd)|\lesssim (\log n)^2$ on $\mathcal A_\trd$, we obtain
\[
\left\lvert\frac{1}{2n}\sum_{i=1}^n H(\eta_{t,i},S_i^2;\wh G_\trd)\big(\wh\xi(M_i)-\xi_0(M_i)\big)^2\right\rvert\lesssim (\log n)^2\Delta_n^2.
\]
It remains to bound the gradient term. Let us define 
\[
\rmD_i(v^2,m;\xi_0,\rmG):=D(\xi_0(m),\xi_0^2(m)v^2;\rmG).
\]
Then by definition of $V^2_i$, for any $G' \in \mathcal G_\trd$
\[
D(\xi_0(M_i),S_i^2;\rmG)=\rmD_i(V_i^2,M_i;\xi_0,\rmG).
\]
In the subsequent analysis, we shall condition on $M_1,\ldots,M_n$. Observe that under \eqref{eq:well-specified-limma-reg}, the marginal density of $V^2=S^2/\xi^2_0(M)$ is given by $f_{G,K-p}(v^2)$ (defined in \eqref{eq:lik_1d_limma_trend}).
Therefore, define the truncated centering term by
\[
\overline{\rmD}_i(\xi_0,\rmG,M_i):=\int_{\widetilde A_n}^{\widetilde B_n}\rmD_i(v^2,M_i;\xi_0,\rmG)f_{G,K-p}(v^2)\,\dd v^2.
\]
We decompose
\[
\frac{1}{n}\sum_{i=1}^n \rmD_i(V_i^2,M_i;\xi_0,\wh G_\trd)\big(\wh\xi(M_i)-\xi_0(M_i)\big)=U_{1,n}+U_{2,n},
\]
where
\[
U_{1,n}:=\frac{1}{n}\sum_{i=1}^n\big(\rmD_i(V_i^2,M_i;\xi_0,\wh G_\trd)-\overline{\rmD}_i(\xi_0,\wh G_\trd,M_i)\big)\big(\wh\xi(M_i)-\xi_0(M_i)\big)
\]
and
\[
U_{2,n}:=\frac{1}{n}\sum_{i=1}^n\overline{\rmD}_i(\xi_0,\wh G_\trd,M_i)\big(\wh\xi(M_i)-\xi_0(M_i)\big).
\]
\paragraph{Bounding $U_{1,n}$.}
Recall the class of mixture densities $\wb{\mathcal F}_{\trd}$ from Lemma~\ref{lem:entropy_joint_density_deriv}. Observe that, by definition $f_{\wh G_\trd,K-p} \in \wb{\mathcal F}_{\trd}$. Also, recall the pseudo-metric $\|\cdot\|_{\mathrm{joint}}$ from the same lemma. Consider an $\varepsilon/4$-cover of the class $\wb{\mathcal F}_{\trd}$ in the pseudo-metric $\|\cdot\|_{\mathrm{joint}}$. Let the cover be $f_{G_1,K-p},\ldots,f_{G_N,K-p}$. By Lemma~\ref{lem:entropy_joint_density_deriv},
\[
\log N\!\left(\wb{\mathcal F}_{\trd},\|\cdot\|_{\mathrm{joint}},\varepsilon\right)\lesssim_{K,p,\underline M_\trd,\wb M_\trd,\underline L_\trd,\overline U_\trd}\bigl(\log(1/\varepsilon)\bigr)^2.
\]

Let us focus on bounding the difference of scores between two mixture densities $f_{G_1,K-p}$ and $f_{G_2,K-p}$, where $G_1,G_2 \in \mathcal G_\trd$. Observe that
\begin{align} 
\label{Eq:def_score_diff}
&\sup_{\eta\in[\underline M_\trd/2,\,2\overline M_\trd]}\sup_{s^2\in[A_n, B_n]}\left|\frac{\frac{d f_{G_1,K-p}}{d\eta}\!\left(s^2/\eta^2\right)}{f_{G_1,K-p}\!\left(s^2/\eta^2\right)}-\frac{\frac{d f_{G_2,K-p}}{d\eta}\!\left(s^2/\eta^2\right)}{f_{G_2,K-p}\!\left(s^2/\eta^2\right)}\right|\nonumber\\
&\leq \sup_{\eta\in[\underline M_\trd/2,\,2\overline M_\trd]}\sup_{s^2\in[ A_n,B_n]}\left|\frac{\frac{d f_{G_1,K-p}}{d\eta}\!\left(s^2/\eta^2\right)-\frac{d f_{G_2,K-p}}{d\eta}\!\left(s^2/\eta^2\right)}{f_{G_1,K-p}\!\left(s^2/\eta^2\right)}\right|\nonumber\\ 
&~~~+\sup_{\eta\in[\underline M_\trd/2,\,2\overline M_\trd]}\sup_{s^2\in[ A_n,B_n]}\left|\frac{\frac{d f_{G_2,K-p}}{d\eta}\!\left(s^2/\eta^2\right)}{f_{G_2,K-p}\!\left(s^2/\eta^2\right)}\right|\nonumber\\ 
&~~~~~~\times\sup_{\eta\in[\underline M_\trd/2,\,2\overline M_\trd]}\sup_{s^2\in[A_n,B_n]}\frac{\left|f_{G_2,K-p}\!\left(s^2/\eta^2\right)-f_{G_1,K-p}\!\left(s^2/\eta^2\right)\right|}{f_{G_1,K-p}\!\left(s^2/\eta^2\right)}. 
\end{align}
Observe that for any $\rmG \in \mathcal G_\trd$ we have the following deterministic lower bound for $f_{\rmG ,K-p}$ on $[\wt A_n,\wt B_n]$. Since $\operatorname{supp}(\rmG )\subset[\underline L_\trd,\overline U_\trd]$, for $x=s^2/\eta^2$, 
\[ 
f_{\rmG ,K-p}(x)=\int p_{\chi^2}(x \mid K-p,\tau^2)\,\rmG (\dd \tau^2)\ge \inf_{\tau^2\in[\underline L_\trd,\overline U_\trd]} p_{\chi^2}(x \mid K-p,\tau^2), 
\] 
where $p_{\chi^2}(x \mid K-p,\tau^2)$ is the $\tau^2\chi^2_{K-p}$ density. 
Using the explicit form of the scaled $\chi^2$ density, for all $x\in[A_n,B_n]$ we have 
\[ 
\inf_{\tau^2\in[\underline L_\trd,\overline U_\trd]} p_{\chi^2}(x \mid K-p,\tau^2)\gtrsim_{K,p,\underline L_\trd,\overline U_\trd} x^{\frac{K-p}{2}-1}\exp\!\left(-\frac{K-p}{2\overline L_\trd}x\right). 
\] 
By the definition of $B_n$, it follows that there exists an absolute constant $C>0$ (depending only on $K,p,\underline L_\trd,\overline U_\trd,\overline\kappa_n$) such that 
\[
\inf_{\eta\in[\underline M_\trd/2,\,2\overline M_\trd]}\inf_{s^2\in[A_n,B_n]} f_{\rmG,K-p}(s^2/\eta^2)\gtrsim n^{-C}.
\]
Combining this with the score bound $|D(\eta,s^2;\rmG)|\lesssim\log n$ for $s^2\in[A_n,B_n]$, we obtain from \eqref{Eq:def_score_diff} that if $\|f_{G_1,K-p}-f_{G_2,K-p}\|_{\mathrm{joint}}\le\varepsilon/4$, then
\[
\sup_{x\in[\widetilde A_n,\widetilde B_n]}\sup_{\eta\in[\underline M_\trd/2,2\wb M_\trd]}\left|D(\eta,\eta^2x;G_1)-D(\eta,\eta^2x;G_2)\right|\lesssim (\log n)n^{2C}\varepsilon.
\]
Taking $\varepsilon=n^{-2C-2}(\log n)^{-1}$, the same cover becomes an $n^{-2}$-cover for the score class. Furthermore
\begin{align}
\label{eq:bound_N}
\log N \lesssim (\log n)^2.
\end{align}
Since $\wh G_\trd\in\mathcal G_\trd$, on $\mathcal A_\trd$ there exists $J\in[N]$ such that uniformly in $i$
\[
\left|\rmD_i(V_i^2,M_i;\xi_0,\wh G_\trd)-\rmD_i(V_i^2,M_i;\xi_0,G_J)\right|\lesssim n^{-2}
\]
and
\[
\left|\overline{\rmD}_i(\xi_0,\wh G_\trd,M_i)-\overline{\rmD}_i(\xi_0,G_J,M_i)\right|\lesssim n^{-2}.
\]
Since $\|\wh \xi-\xi_0\|_{\mathrm W_n} \le \Delta_n$ (by Assumption~\ref{assum:trend_estimation}), therefore
\[
\Bigg|U_{1,n}-\frac{1}{n}\sum_{i=1}^n\big(\rmD_i(V_i^2,M_i;\xi_0,G_J)-\overline{\rmD}_i(\xi_0,G_J,M_i)\big)\big(\wh\xi(M_i)-\xi_0(M_i)\big)\Bigg|\lesssim \frac{\Delta_n}{n^2}.
\]
Since $\wh \xi \in \mathcal X$, the forgoing inequality along with Assumption~\ref{assum:trend_estimation} implies
\begin{align}
|U_{1,n}|\lesssim \mathfrak V_n+\frac{\Delta_n}{n^2}+\frac{\log n}{n^2},
\end{align}
where
\[
\mathfrak V_n:=\max_{j\in[N]}\max_{\xi\in \wb{\mathcal X}}\left|\frac{1}{n}\sum_{i=1}^n\big(\rmD_i(V_i^2,M_i;\xi_0,G_j)-\overline{\rmD}_i(\xi_0,G_j,M_i)\big)\big(\xi(M_i)-\xi_0(M_i)\big)\right|,
\]
and $\wb{\mathcal X}:= \left\{\xi \in \mathcal X: \|\xi-\xi_0\|_{\mathrm W_n} \le \Delta_n\right\}$.
Let us define the processes $V_{n,j}(\xi):\wb{\mathcal X} \rightarrow \mathbb R$ for $j \in [N]$, as follows:
\begin{align}
    \label{eq:v_process}
    V_{n,j}(\xi):=\frac{1}{n}\sum_{i=1}^n\big(\rmD_i(V_i^2,M_i;\xi_0,G_j)-\overline{\rmD}_i(\xi_0,G_j,M_i)\big)\big(\xi(M_i)-\xi_0(M_i)\big).
\end{align}
Observe that on $\mathcal A_\trd$, each increment term in the processes $V_{n,j}(\xi)$ is mean-zero conditioned on $M_1,\ldots,M_n$. Therefore, for $\xi_1,\xi_2 \in \wb{\mathcal X}$, using Lemma~\ref{lem:prop_chi_sq} (1), we have that on $\mathcal A_\trd$, conditioned on $M_1,\ldots,M_n$ and $\tau^2_1,\ldots,\tau^2_n$, 
\begin{align}
    \left|V_{n,j}(\xi_1)-V_{n,j}(\xi_2)\right| \lesssim \frac{(\log n)}{\sqrt{n}}\|\xi_1-\xi_2\|_{\mathrm W_n},
\end{align}
where the semi-norm and $\mathrm W_n$ are defined in Assumption~\ref{assum:trend_estimation}. Therefore, the processes $\{V_{n,j}(\xi)\}$ has sub-Gaussian increments in terms of the norm $\|\cdot\|_{\mathrm W_n}$ on $\mathcal A_\trd$. Furthermore, since $\mathcal X$ is separable under the supremum norm on $\RR$, it remains separable under $\|\cdot\|_{\mathrm W_n}$. This implies $\sup_{\xi \in \wb{\mathcal X}}|V_{n,j}(\xi)|$ is measurable. 
Since $\Delta_n=o(1)$, we assume that $\Delta_n \in (0,1)$ (for finitely many $n$ where $\Delta_n>1$, one can use an analogous argument with a courser bound to control $\mathfrak V_n$). Using Assumption~\ref{assum:trend_estimation} and (8.15)-(8.16) of \cite{vershynin2025high}, we can conclude that conditioned on $M_1,\ldots,M_n$ and $\tau^2_1,\ldots,\tau^2_n$, on $\mathcal A_\trd$ we have
\begin{align}
   &\max_{\xi\in \wb{\mathcal X}}\left|\frac{1}{n}\sum_{i=1}^n\big(\rmD_i(V_i^2,M_i;\xi_0,G_j)-\overline{\rmD}_i(\xi_0,G_j,M_i)\big)\big(\xi(M_i)-\xi_0(M_i)\big)\right| \\
   & \lesssim \frac{(\log n)}{\sqrt{n}}\left\{(1+u)\Delta_n+\int_{0}^{\Delta_n}\sqrt{\log N(\varepsilon,\wb{\mathcal X},\|\cdot\|_{\mathrm W_n})}\dd \varepsilon\right\}\\
   & \lesssim \frac{(\log n)}{\sqrt{n}}\left\{(1+u)\Delta_n+\int_{0}^{\Delta_n}(1/\varepsilon)^{h_1/2}\log^{h_2/2} (1/\varepsilon)\dd \varepsilon\right\}\\
   &\overset{(1)}{\lesssim} \frac{(\log n)}{\sqrt{n}}\left\{(1+u)\Delta_n+\int_{-\log \Delta_n}^{\infty}t^{h_2/2}e^{-t\,(1-h_1/2)}\dd t\right\}\\
   &\overset{(2)}{\lesssim} \frac{(\log n)}{\sqrt{n}}\left\{(1+u)\Delta_n+|\log \Delta_n|^{h_2/2}\cdot \Delta^{1-h_1/2}_n\right\},
\end{align}
for all $j \in [N]$, with probability greater than $1-2\exp(-u^2)$. In the above display, (1) follows by the change of variable $\varepsilon=e^{-t}$ and (2) follows using the property that the incomplete gamma function
$\Gamma(x):= \int_{x}^\infty e^{-t}t^{s-1}\dd t$ satisfies $\Gamma(x)/(e^{-x}x^{s-1}) = O(1)$ as $x \rightarrow \infty$.

For $u:=\sqrt{5\log n}+\sqrt{\log N}$, using \eqref{eq:bound_N}, the forgoing inequality reduces to
\begin{align}
\label{eq:chaining_bound}
    &\max_{\xi\in \wb{\mathcal X}}\left|\frac{1}{n}\sum_{i=1}^n\big(\rmD_i(V_i^2,M_i;\xi_0,G_j)-\overline{\rmD}_i(\xi_0,G_j,M_i)\big)\big(\xi(M_i)-\xi_0(M_i)\big)\right| \\
     & \lesssim \frac{(\log n)^2}{\sqrt{n}}\cdot \Delta_n+\frac{(\log n)}{\sqrt{n}}\cdot|\log \Delta_n|^{h_2/2}\,\Delta^{1-h_1/2}_n.
\end{align}
Applying Hoeffding's inequality along with the union bound on $\mathcal A_\trd$, one can conclude from \eqref{eq:chaining_bound}, conditioned on $M_1,\ldots,M_n,$ and $\tau^2_1,\ldots,\tau^2_n$,
\begin{align}
\label{eq:frak_u_bound}
\mathfrak V_n \le K_{\trd,1} \frac{(\log n)^2}{\sqrt{n}}\cdot \Delta_n+\frac{(\log n)}{\sqrt{n}}\cdot|\log \Delta_n|^{h_2/2}\,\Delta^{1-h_1/2}_n
\end{align}
with probability greater than $1-o(n^{-2})$. Here the absolute constant $K_{\trd,1}>0$ (independent of $\tau^2_1,\ldots,\tau^2_n$ and $M_1,\ldots,M_n$) is chosen large enough. Furthermore, since the constant $K_{\trd,1}$ is independent of $\tau^2_1,\ldots,\tau^2_n$ and $M_1,\ldots,M_n$, using $\mathbb P(\mathcal A_\trd^c)\le 2n^{-2}$, we can also conclude that the assertion holds unconditionally with probability greater than $1-o(n^{-2})$. In other words, with probability greater than $1-o(n^{-2})$, we have
\begin{align}
\label{eq:bound_u_1_n}
U_{1,n} \lesssim_{K,p,\underline M_\trd,\wb M_\trd,\underline L_\trd,\overline U_\trd} \frac{(\log n)^2}{\sqrt{n}}\cdot \Delta_n+\frac{(\log n)}{\sqrt{n}}\cdot|\log \Delta_n|^{h_2/2}\,\Delta^{1-h_1/2}_n,
\end{align}
both for fixed $M_1,\ldots,M_n$ satisfying $\mathcal A_\trd$ and unconditionally.

\paragraph{Bounding $U_{2,n}$.}
By definition
\[
\overline{\rmD}_i(\xi_0,\wh G_\trd,M_i)=\int_{\widetilde A_n}^{\widetilde B_n}\rmD_i(v^2,M_i;\xi_0,\wh G_\trd)f_{G,K-p}(v^2)\,\dd v^2.
\]
We can add and subtract $f_{\wh G_\trd,K-p}(v^2)$ to obtain
\begin{align}
\label{eq:support_skel_eq_bar_d}
\overline{\rmD}_i(\xi_0,\wh G_\trd,M_i)&=\int_{\widetilde A_n}^{\widetilde B_n}\rmD_i(v^2,M_i;\xi_0,\wh G_\trd)\bigl(f_{G,K-p}(v^2)-f_{\wh G_\trd,K-p}(v^2)\bigr)\,\dd v^2\\
&+\int_{\widetilde A_n}^{\widetilde B_n}\rmD_i(v^2,M_i;\xi_0,\wh G_\trd)f_{\wh G_\trd,K-p}(v^2)\,\dd v^2.
\end{align}
If \eqref{eq:bayesian_cs_model_2} holds, $\rmD_i(\cdot,M_i;\xi_0,\wh G_\trd)$ is the score of the conditional likelihood of $v^2$ and therefore
\begin{align}
\label{eq:bayesian_score_eq}
\int_0^\infty \rmD_i(v^2,M_i;\xi_0,\wh G_\trd)f_{\wh G_\trd,K-p}(v^2)\,\dd v^2=0.
\end{align}
Now, one can decompose the term in the foregoing expression as
\begin{align}
&\int_0^\infty \rmD_i(v^2,M_i;\xi_0,\wh G_\trd)f_{\wh G_\trd,K-p}(v^2)\,\dd v^2\\
&=\int_{\widetilde A_n}^{\widetilde B_n}\rmD_i(v^2,M_i;\xi_0,\wh G_\trd)f_{\wh G_\trd,K-p}(v^2)\,\dd v^2+\int_{[\widetilde A_n,\widetilde B_n]^c}\rmD_i(v^2,M_i;\xi_0,\wh G_\trd)f_{\wh G_\trd,K-p}(v^2)\,\dd v^2.
\end{align}
By Lemma~\ref{lem:prop_chi_sq}~(5), if we restrict to $\mathcal A_\trd$, this tail integral is $O(n^{-2})$ uniformly over $i=1,\ldots,n$. Plugging the above in \eqref{eq:support_skel_eq_bar_d}, we get using Cauchy-Schwarz inequality and the definition of Hellinger distance
\begin{align}
|\overline{\rmD}_i(\xi_0,\wh G_\trd,M_i)|&\lesssim \left(\int \rmD_i(v^2,M_i;\xi_0,\wh G_\trd)^2\bigl(f_{G,K-p}(v^2)+f_{\wh G_\trd,K-p}(v^2)\bigr)\,\dd v^2\right)^{1/2}\\
&~~~~~~~~~~~~~~~~~~~~~~~~~~~~~~~~~~~~~~~~~\times\mathcal H\!\left(f_{G,K-p},f_{\wh G_\trd,K-p}\right)
+\frac{1}{n^2}.
\end{align}
Using the score bound from Lemma~\ref{lem:prop_chi_sq}(1) together with the bounded support assumptions on $G$ and the finite moments of the scaled $\chi^2_{K-p}$ distribution, the integral factor can be shown to be uniformly bounded by a constant depending only on $K,p,\underline M_\trd,\wb M_\trd,\underline L_\trd,\overline U_\trd$ using Fubini's theorem. Hence
\[
|\overline{\rmD}_i(\xi_0,\wh G_\trd,M_i)|\lesssim \mathcal H\!\left(f_{G,K-p},f_{\wh G_\trd,K-p}\right)+\frac{1}{n^2}.
\]

Consequently, on $\mathcal A_\trd$ we have
\begin{align}
\label{eq:sharp_u_2_n}
|U_{2,n}|\lesssim \mathcal H\!\left(f_{G,K-p},f_{\wh G_\trd,K-p}\right)\Delta_n+\frac{\Delta_n}{n^2}.
\end{align}
Combining the bounds for the Taylor remainder, $U_{1,n}$, and $U_{2,n}$, we conclude that with probability greater than $1-o(n^{-2})$,
\begin{align}
\label{eq:final_bd_sub_1_cs}
\mathrm{Sub}_{n,1}\lesssim \Delta_n^2(\log n)^2+\frac{(\log n)^2}{\sqrt{n}}\cdot \Delta_n+\frac{(\log n)}{\sqrt{n}}\cdot|\log \Delta_n|^{h_2/2}\,\Delta^{1-h_1/2}_n+\mathcal H\!\left(f_{\wh G_\trd,K-p},f_{G,K-p}\right)\Delta_n\nonumber\\
\end{align}
when \eqref{eq:bayesian_cs_model_2} and \eqref{eq:well-specified-limma-reg} holds, both for fixed $M_1,\ldots,M_n$ satisfying $\mathcal A_\trd$ and unconditionally. 

\subsubsection{Bounding $\mathrm{Sub}_{n,2}$}

We again use a second-order Taylor expansion, now applied to $\ell(S_i^2;\wh\xi,G,M_i)$ around $\xi_0(M_i)$:
\begin{align}
\label{eq:taylor_sub_n_2}
&\frac{1}{n}\sum_{i=1}^n\ell(S_i^2;\wh\xi,G,M_i)-\frac{1}{n}\sum_{i=1}^n\ell(S_i^2;\xi_0,G,M_i)\\
&=\frac{1}{n}\sum_{i=1}^n D(\xi_0(M_i),S_i^2;G)\big(\wh\xi(M_i)-\xi_0(M_i)\big)\\
&+\frac{1}{2n}\sum_{i=1}^n H(\widetilde\eta_{t,i},S_i^2;G)\big(\wh\xi(M_i)-\xi_0(M_i)\big)^2
\end{align}
for some $\widetilde\eta_{t,i}$ between $\xi_0(M_i)$ and $\wh\xi(M_i)$. 
As before, the Hessian term is bounded by $(\log n)^2\Delta_n^2$ on $\mathcal A_\trd$. Thus it remains to control the linear term. If \eqref{eq:bayesian_cs_model_2} holds, then $D(\eta,s^2;G)$ is proportional to the score of the full conditional density of $S_i^2$ given $M_i$ and we have
\[
\mathbb E\!\left[D(\xi_0(M_i),S_i^2;G)\mid M_i\right]=0.
\]
Proceeding exactly as in the proof of $U_{1,n}$, using 
the chaining argument, we can conclude that
\[
\left|\frac{1}{n}\sum_{i=1}^n D(\xi_0(M_i),S_i^2;G)\big(\wh\xi(M_i)-\xi_0(M_i)\big)\right|\lesssim \frac{(\log n)^{3/2}}{\sqrt{n}}\cdot \Delta_n+\frac{(\log n)}{\sqrt{n}}\cdot|\log \Delta_n|^{h_2/2}\,\Delta^{1-h_1/2}_n,
\] 
with probability greater than $1-o(n^{-2})$.
Therefore on $\mathcal A_\trd$,
\begin{align}
\label{eq:sharp_bound_sub_n_2_cs}
\mathrm{Sub}_{n,2}\lesssim \Delta_n^2(\log n)^2+\frac{(\log n)^{3/2}}{\sqrt{n}}\cdot \Delta_n+\frac{(\log n)}{\sqrt{n}}\cdot|\log \Delta_n|^{h_2/2}\,\Delta^{1-h_1/2}_n
\end{align}
under \eqref{eq:bayesian_cs_model_2} and \eqref{eq:well-specified-limma-reg}, with probability greater than $1-o(n^{-2})$.

When \eqref{eq:bayesian_cs_model_2} and \eqref{eq:well-specified-limma-reg} holds, noting that \eqref{eq:final_bd_sub_1_cs} and \eqref{eq:sharp_bound_sub_n_2_cs} holds on $\mathcal A_\trd$ and $\mathbb P(\mathcal A^c_\trd)=o(n^{-2})$, we can show that there exists a constant $C_\trd>0$ such that
\[
\frac{1}{n}\sum_{i=1}^{n}\ell(S^2_i,\wh \xi;\wh G_\trd,M_i)-\frac{1}{n}\sum_{i=1}^{n}\ell(S^2_i,\wh \xi;G,M_i) \ge -C_\trd\mathfrak R_{n,1}(\wh G_\trd),
\]
with probability greater than $1-o(n^{-2})$. 
\end{proof}

\subsection{Hellinger large deviation}
The next step to prove the asymptotic FDR control is to prove a large deviation property for 
\[
\mathcal H^2(f_{\wh G_\trd,K-p},f_{G,K-p}),
\]
where the Hellinger distance between the densities of two mixtures of scaled $\chi^2_{K-p}$ random variables is defined in (13) of \cite{ignatiadis2025empirical}. 

\begin{lemm} 
\label{lem:hell_limm_trnd_cs} 
Assume that $\tau^2_i \mid M_i \simiid G$ for all $i \in [n]$ (cf. \eqref{eq:well-specified-limma-reg}), where $G \in \mathcal G_\trd$. Fix an absolute constant $c_0>0$. Then there exist constants $\widetilde D>0$ and $n_\trd \in \mathbb N$, depending only on $h_1,h_2,K,p,\underline M_\trd,\overline M_\trd,\underline L_\trd,$ and $\overline U_\trd$, such that for all $n \ge n_\trd$, 
\[ \mathbb P\left[\mathcal H^2\left(f_{\wh G_\trd,K-p},f_{G,K-p}\right) \ge \widetilde D~\lambda^{2}_n\right]\lesssim \frac{1}{n^2}+\left(1+\log\log n+\log |\log \Delta_n|\right)~e^{-c_0(\log n)^2}, 
\]
where 
\begin{align} 
\label{eq:def_lambda_trd_n}
\lambda^2_n=\max\left\{\Delta^2_n (\log n)^2,\frac{(\log n)^{2}}{\sqrt{n}}\Delta_n,\frac{(\log n)}{\sqrt{n}}\cdot|\log \Delta_n|^{h_2/2}\,\Delta^{1-h_1/2}_n,\frac{(\log n)^2}{n}\right\}. 
\end{align} 
\end{lemm}
The proof adapts the empirical Bayes arguments of \cite{chen2024empiricalbayesestimationprecision} to the present $\chi^2$ mixture setting.

\begin{proof}
    Define the event $\mathcal H_{\trd}$ as follows:
\begin{align}
    \mathcal H_{\trd}&:=\left\{\prod_{i=1}^n\frac{f_{\wh G_\trd,K-p}(S_i^2/\xi^2_0(M_i))}{f_{G,K-p}(S_i^2/\xi^2_0(M_i))} \ge e^{-n\, C_\trd\mathfrak R_{n,1}(\wh G_\trd)}\right\},
\end{align}
where for any $\rmG \in \mathcal G_\trd$, the misspecification cost $\mathfrak R_{n,1}(\rmG)$ is defined in Lemma~\ref{lem:lik_ratio_limma_trnd}. By Lemma~\ref{lem:lik_ratio_limma_trnd}, we can conclude that
\begin{align}
\label{eq:fundamental_eq_cs}
\mathbb P\!\left(\mathcal H^c_{\trd}\right)=o(n^{-2}).
\end{align}
We choose $n$ large enough so that $\lambda_n \in (0,1)$. Next define $K=\lceil |\log_2(1/\epsilon)|\rceil$. For each $k\in\{1,\ldots,K\}$, let

\[
\mathcal B_k=\left\{\rmG:\operatorname{supp}(\rmG)\subset[\underline L_\trd,\overline U_\trd],\;\mathcal H\!\left(f_{\rmG,K-p},f_{G,K-p}\right)\in[\mu_{n,k+1},\mu_{n,k})\right\},
\]
where $\mu_{n,k}=\widetilde C\,\lambda_n^{\,1-2^{-k+1}}$.
Then
\[
\left\{\mathcal H^2\!\left(f_{\wh G_\trd,K-p},f_{G,K-p}\right)\ge \widetilde C^2\,\lambda_n^{2(1-\epsilon)}\right\}\subseteq \bigcup_{k=1}^K\{\wh G_\trd\in \mathcal B_k\}.
\]
Combining with \eqref{eq:fundamental_eq_cs}, we can conclude that
\begin{align*}
\mathbb P\!\left[\mathcal H^2\!\left(f_{\wh G_\trd,K-p},f_{G,K-p}\right)\ge \widetilde C^2\,\lambda_n^{2(1-\epsilon)}\right]
&\le \mathbb P\!\left[\mathcal H^2\!\left(f_{\wh G_\trd,K-p},f_{G,K-p}\right)\ge \widetilde C^2\,\lambda_n^{2(1-\epsilon)},\;\mathcal H_{\trd}\right]+o(n^{-2})\\
&\le \mathbb P\!\left[\left\{\wh G_\trd\in \bigcup_{k=1}^K\mathcal B_k\right\}\cap \mathcal H_{\trd}\right]+o(n^{-2})\\
&\le \sum_{k=1}^K\mathbb P\left(\left\{\wh G_\trd\in\mathcal B_k\right\}\cap \mathcal H_{\trd}\right)+o(n^{-2}).
\end{align*}
Thus, it is enough to bound 
\(
\smash{\mathbb P(\{\wh G_\trd\in\mathcal B_k\}\cap \mathcal H_{\trd})},
\)
uniformly over $k\in[K]$.
Let $\mathcal S_{\trd,k}=\{f_{G_j,K-p}:j=1,\ldots,\mathcal J_k\}\subseteq \mathcal B_k$ be a minimal proper $n^{-2}$-cover of $\mathcal B_k$ under the metric $\|f_{G_1,K-p}-f_{G_2,K-p}\|_\infty=\sup_{x\ge0}|f_{G_1,K-p}(x)-f_{G_2,K-p}(x)|$. Since $\mathcal B_k$ is a subset of the class considered in Theorem~9 of \citet{ignatiadis2025empirical}, it follows from the proof therein that
\begin{align}
\label{eq:metric_entropy_limma_trd_corr}
\log\mathcal J_k\lesssim_{K,p,\underline M_\trd,\overline M_\trd,\underline L_\trd,\overline U_\trd}(\log n)^2.
\end{align}
Moreover, for every $\rmG\in\mathcal B_k$, there exists $f_{G_J,K-p}\in \mathcal S_{\trd,k}$ such that
\begin{align}
\label{eq:cover_density_corr}
f_{\rmG,K-p}(v^2)\le f_{G_J,K-p}(v^2)+\eta\le \max_{1\le j\le |\mathcal J_k|}f_{G_j,K-p}(v^2)+\eta \qquad \mbox{for all $v^2\ge0$,}
\end{align}
where here and below we take $\eta=n^{-2}$. Define the truncation function
\begin{align}
\label{eq:truncation_func}
\eta_\trd(z)=\eta\,\mathds 1_{\{|z|\le \overline B\}}+\eta\,\frac{\overline B^2}{z^2}\,\mathds{1}_{\{|z|>\overline B\}},
\end{align}
where
\begin{align}
\label{eq:b_limma_trd_corr}
\overline B\asymp \frac{\overline U_\trd}{\underline M_\trd}\log n.
\end{align}
For every $\rmG\in \mathcal B_k$ we have
\(
\mathcal H\!\left(f_{\rmG,K-p},f_{G,K-p}\right)\le \mu_{n,k}.
\)
Also define
\begin{align}
\label{eq:def_gamma_n_well}
\gamma_n^2=\Delta_n^2(\log n)^2+\frac{(\log n)^2}{\sqrt{n}}\cdot \Delta_n+\frac{(\log n)}{\sqrt{n}}\cdot|\log \Delta_n|^{h_2/2}\,\Delta^{1-h_1/2}_n.
\end{align}
Since
\[
\mathfrak R_{n,1}(\rmG)=\gamma_n^2+\mathcal H\!\left(f_{\rmG,K-p},f_{G,K-p}\right)\Delta_n,
\]
we obtain for every $\rmG\in\mathcal B_k$,
\begin{align}
\label{eq:rel_mu_n_k_corr}
\mathfrak R_{n,1}(\rmG)&\le \gamma_n^2+\mu_{n,k}\Delta_n\lesssim \gamma_n^2+\mu_{n,k}\gamma_n\lesssim \mu_{n,k+1}^2,
\end{align}
where the penultimate inequality follows from \eqref{eq:def_gamma_n_well}.
Indeed, $\gamma_n\lesssim \lambda_n$ by definition of $\lambda_n$, and $\mu_{n,k}\gamma_n\le \widetilde C\,\lambda_n^{2-2^{-k+1}}\asymp \mu_{n,k+1}^2$, while $\gamma_n^2\le \lambda_n^2\le \mu_{n,k+1}^2$ since $\lambda_n\in(0,1)$.

Consequently, if $\wh G_\trd\in \mathcal B_k$ and $\mathcal H_\trd$ holds, then
\[
\prod_{i=1}^n\frac{f_{\wh G_\trd,K-p}(S_i^2/\xi^2_0(M_i))}{f_{G,K-p}(S_i^2/\xi^2_0(M_i))}\ge e^{-nC_\trd\mathfrak R_{n,1}(\wh G_\trd)}\ge e^{-nC(\gamma_n^2+\mu_{n,k}\gamma_n)}
\]
for some constant $C>0$ depending only on the model parameters.

Let $V_i^2=S_i^2/\xi_0^2(M_i)$. Using \eqref{eq:cover_density_corr} and arguing exactly as in \citet[proof of Theorem~9]{ignatiadis2025empirical}, we obtain, for any $a>1$,
\begin{align*}
&\mathbb P\left[\left\{\wh G_\trd\in\mathcal B_k\right\}\cap \mathcal H_{\trd}\right]\\
&=\mathbb P\!\left[\mathcal H\!\left(f_{G,K-p},f_{\wh G_\trd,K-p}\right)\in[\mu_{n,k+1},\mu_{n,k}),\;\prod_{i=1}^n\frac{f_{\wh G_\trd,K-p}(V_i^2)}{f_{G,K-p}(V_i^2)}\ge e^{-nC(\gamma_n^2+\mu_{n,k}\gamma_n)}\right]\\
&\le \mathbb P\!\left[\sup_{j\in[\mathcal J_k]}\prod_{i=1}^n\frac{f_{G_j,K-p}(V_i^2)}{f_{G,K-p}(V_i^2)}\ge e^{-nCa(\gamma_n^2+\mu_{n,k}\gamma_n)}\right]+\mathbb P\!\left[\prod_{i=1}^n\frac{C_\sharp}{\eta_\trd(V_i^2)}\ge e^{nC(a-1)(\gamma_n^2+\mu_{n,k}\gamma_n)}\right],
\end{align*}
where $C_\sharp=C_\sharp(K,p,\underline M_\trd,\overline M_\trd,\underline L_\trd,\overline U_\trd)$ is the same constant that appears in the proof of Theorem~9 in \citet{ignatiadis2025empirical}.

For the first term, since \eqref{eq:bayesian_cs_model_2} holds, the variables $V_1^2,\ldots,V_n^2$ are i.i.d. with density $f_{G,K-p}$. Therefore, retracing the argument in \citet{ignatiadis2025empirical} and using \eqref{eq:metric_entropy_limma_trd_corr},
\begin{align*}
&\mathbb P\!\left[\sup_{j\in\mathcal J_k}\prod_{i=1}^n\frac{f_{G_j,K-p}(V_i^2)}{f_{G,K-p}(V_i^2)}\ge e^{-nCa(\gamma_n^2+\mu_{n,k}\gamma_n)}\right]\\
&\le \sum_{j\in\mathcal J_k} e^{\frac{nCa}{2}(\gamma_n^2+\mu_{n,k}\gamma_n)}\mathbb E\!\left[\left(\prod_{i=1}^n\frac{f_{G_j,K-p}(V_i^2)}{f_{G,K-p}(V_i^2)}\right)^{1/2}\right]\\
&\le \sum_{j\in\mathcal J_k}\exp\!\left(-n\mathcal H^2\!\left(f_{G_j,K-p},f_{G,K-p}\right)+n\sqrt{2\eta\overline B}+\frac{nCa}{2}(\gamma_n^2+\mu_{n,k}\gamma_n)\right).
\end{align*}
Since each center $f_{G_j,K-p}$ belongs to the proper cover $\mathcal S_{\trd,k}\subseteq \mathcal B_k$, we have
\[
\mathcal H^2\!\left(f_{G_j,K-p},f_{G,K-p}\right)\ge \mu_{n,k+1}^2=\widetilde C^2\lambda_n^{2(1-2^{-k})}.
\]
Hence
\begin{align}
&\mathbb P\!\left[\sup_{j\in\mathcal J_k}\prod_{i=1}^n\frac{f_{G_j,K-p}(V_i^2)}{f_{G,K-p}(V_i^2)}\ge e^{-nCa(\gamma_n^2+\mu_{n,k}\gamma_n)}\right]\\
&\le \exp\!\left(-n\widetilde C^2\lambda_n^{2(1-2^{-k})}+n\sqrt{2\eta\overline B}+\frac{nCa}{2}(\gamma_n^2+\mu_{n,k}\gamma_n)+\log|\mathcal J_k|\right)\\
&\le \exp\!\left(-n\widetilde C^2\lambda_n^{2(1-2^{-k})}+n\check C\lambda_n^{2(1-2^{-k})}\right),
\end{align}
where the last inequality follows from \eqref{eq:metric_entropy_limma_trd_corr} and \eqref{eq:rel_mu_n_k_corr}, and those inequalities also determine the absolute constant $\check C>0$.
Using \eqref{eq:metric_entropy_limma_trd_corr}, \eqref{eq:rel_mu_n_k_corr}, and the definition of $\lambda_n$, we may choose $\widetilde C>0$ sufficiently large so that
\[
\mathbb P\!\left[\sup_{j\in\mathcal J_k}\prod_{i=1}^n\frac{f_{G_j,K-p}(V_i^2)}{f_{G,K-p}(V_i^2)}\ge e^{-nCa(\gamma_n^2+\mu_{n,k}\gamma_n)}\right]\le \exp\!\left(-nc_0\lambda_n^{2(1-2^{-k})}\right)
\]
for some constant $c_0>0$. Since $\lambda_n\in(0,1)$ and $2(1-2^{-k})\le 2$, it follows that $\lambda_n^{2(1-2^{-k})}\ge \lambda_n^2$. Therefore
\[
\exp\!\left(-nc_0\lambda_n^{2(1-2^{-k})}\right)\le \exp(-nc_0\lambda_n^2)\le \exp(-c_0(\log n)^2),
\]
because $\lambda_n^2\ge (\log n)^2/n$ by \eqref{eq:def_lambda_trd_n}.

For the second term, Lemma~S5 of \citet{ignatiadis2025empirical} yields, after possibly enlarging $\widetilde C$,
\begin{align}
\label{eq:union_bound_2_corr}
\mathbb P\!\left[\prod_{i=1}^n\frac{C_\sharp}{\eta_\trd(V_i^2)}\ge e^{nC(a-1)(\gamma_n^2+\mu_{n,k}\gamma_n)}\right]\le \exp\!\left(-c_0(\log n)^2\right).
\end{align}
Combining the two bounds, we conclude that for every $k\in[K]$,
\[
\mathbb P\left[\left\{\wh G_\trd\in\mathcal B_k\right\}\cap \mathcal H_{\trd}\right]\le \exp\!\left(-c_0(\log n)^2\right).
\] 
Using the union bound over $k=1,\ldots,K$ and using \eqref{eq:fundamental_eq_cs},
\begin{align}
\label{eq:hell_bayes_cs}
\mathbb P\!\left[\mathcal H^2\!\left(f_{\wh G_\trd,K-p},f_{G,K-p}\right)\ge \widetilde C\,\lambda_n^{2(1-\epsilon)}\right]\le \lceil |\log(1/\epsilon)|\rceil e^{-c_0(\log n)^2}+\frac{1}{n^2}.
\end{align}
To finalize the proof, we choose
\begin{align}
\label{eq:def_eps_dada}
\epsilon=c\,\min\left\{\frac{1}{|\log \Delta_n|},\frac{1}{\log n}\right\},
\end{align}
for a sufficiently small absolute constant $c>0$ (possibly depending on $h_1$ and $h_2$), so that $\epsilon\in(0,1/2)$ for all $n$ large enough. 
Since $\lambda_n^2\ge (\log n)^2/n$ by \eqref{eq:def_lambda_trd_n}, we have
\[
\lambda_n^{-2\epsilon}\le
\left(\frac{n}{(\log n)^2}\right)^{\epsilon}=\exp\!\left(\epsilon\log n-2\epsilon\log\log n\right)\le e
\]
for all sufficiently large $n$, and hence $\lambda_n^{2(1-\epsilon)}\lesssim \lambda_n^2$ for all sufficiently large $n$. Using \eqref{eq:hell_bayes_cs} with this choice of $\epsilon$ gives
\begin{align}
\mathbb P\Bigg[\mathcal H^2\!\left(f_{\wh G_\trd,K-p},f_{G,K-p}\right)\ge \widetilde C\,\lambda_n^{2(1-\epsilon)}
\Bigg]\lesssim\left(1+\log\log n+\log|\log\Delta_n|\right)e^{-c_0(\log n)^2}+\frac{1}{n^2}.
\end{align}
Since $\lambda_n^{2(1-\epsilon)}\lesssim \lambda_n^2$, there exists a constant $\widetilde D>0$ such that, for all sufficiently large $n$,
\[
\left\{
\mathcal H^2\!\left(f_{\wh G_\trd,K-p},f_{G,K-p}\right)\ge \widetilde D\,\lambda_n^2\right\}
\subseteq\left\{\mathcal H^2\!\left(f_{\wh G_\trd,K-p},f_{G,K-p}\right)\ge \widetilde C\,\lambda_n^{2(1-\epsilon)}\right\},
\]
and hence
\begin{align}
\mathbb P\!\left[\mathcal H^2\!\left(f_{\wh G_\trd,K-p},f_{G,K-p}\right)\ge \widetilde D\,\lambda_n^2\right]\lesssim
\left(1+\log\log n+\log|\log\Delta_n|\right)e^{-c_0(\log n)^2}+\frac{1}{n^2},
\end{align}
which completes the proof.
\end{proof}

\subsection{Proof of Theorem \ref{thm:conv_p_trnd_orc}}
\label{sec:proof_trend_pval}
To prove Theorem \ref{thm:conv_p_trnd_orc}, recall that $V^2_i=S^2_i/\xi^2_0(M_i)$ for $i \in [n]$. For any $G' \in \mathcal G_\trd$ and $\xi \in \mathcal X$, define the function
\begin{align}
    \PFun(z,S^2_i,M_i;G',\xi):=\mathbb E_{ G'}\left[2\Phi(-z/\{\nu\xi(M_i)\tau\}) \mid S^2_i, M_i\right], \qquad \mbox{for $z \in \mathbb R$}
\end{align}
where $S^2_i,M_i$ are assumed to satisfy distributional specification in \eqref{eq:bayesian_cs_model_1} with $\tau^2_i \mid M_i \simiid G'$.
Next, following the proof of Proposition 11 of \cite[Proposition 11]{ignatiadis2025empirical}, we can show that 
\begin{align}
\label{eq:tweedie_trnd}
\PFun(z,S^2_i,M_i;G',\xi)&:= \frac{C_{K,p}}{f_{G',K-p}\!\left(S^2_i/\xi^2(M_i)\right)}\int_{0}^{\infty}\frac{ (t^2)^{-\,\frac{K-p+1}{2}} 
       \left(S_i^2 / \xi^2(M_i)\right)^{\frac{K-p}{2}-1}}
     { \sqrt{ (K-p+1)t^2 - (K-p)\left(S_i^2/\xi^2(M_i)\right) } } \\
&\quad\quad \times f_{G',K-p+1}(t^2)\,\mathbbm{1}\!\left\{
t^2 \ge \frac{ (K-p)\left(S_i^2/\xi^2(M_i)\right) + z^2/(\nu^2\xi^2(M_i))}
     { K-p+1 }\right\}\, \dd t^2,
\end{align}
where \(C_{K,p}\) is the positive normalising constant from \cite[Proposition 11]{ignatiadis2025empirical} with $\nu:=K-p$ and  
\(f_{\widehat{G}_{\trd},K-p+1}\) denoting the mixture density. 
Furthermore, for all $i \in [n]$, we can show that
\begin{align}
P^\trd_i&= \PregFun(Z_i,S^2_i,M_i;G)=\PFun(Z_i,S^2_i,M_i;G,\xi_0),\quad \wh P^\trd_i=\PFun(Z_i,S^2_i,M_i;\wh G_\trd,\wh\xi)
\end{align}
For any $z \in \mathbb R$, let us also define the following:
\begin{align}
\label{eq:def_n_trd_all}
\widehat{\calD}^{\trd}_i(S_i^2,\widehat{\xi}(M_i))
&:=f_{\wh G_\trd,K-p}(S^2_i/\wh \xi^2(M_i)), \\
\widehat{\calD}^{\trd}_i(S_i^2,\xi_0(M_i))
&:=f_{\wh G_\trd,K-p}(S^2_i/\xi^2_0(M_i)), \\
\calD^{\trd}_i(S_i^2,\xi_0(M_i))&:=
f_{G,K-p}\!\left(S^2_i/\xi^2_0(M_i)\right),\\
N^{\trd}_i(z,S_i^2,M_i,\xi_0,G)
&:=
\PFun(z,S^2_i,M_i;G,\xi_0)\,
\calD^{\trd}_i(S_i^2,\xi_0(M_i)),\\
N^{\trd}_i(z,S_i^2,M_i,\wh\xi,\wh G_\trd)
&:=
\PFun(z,S^2_i,M_i;\wh G_\trd,\wh \xi)\,
\widehat{\calD}^{\trd}_i(S_i^2,\widehat{\xi}(M_i)).
\end{align}
Also, introduce the quantity
\begin{align}
\label{eq:def_n_trd_1}
N^{\trd}_i(z,S_i^2,M_i,\xi_0,\wh G_\trd)
&:= \PFun(z,S^2_i,M_i;\wh G_\trd,\xi_0)\,\widehat{\calD}^{\trd}_i(S_i^2,\xi_0(M_i)).
\end{align}
Let us define the set
\begin{align}
\mathcal B_\trd &= \Bigg\{S^2_i \in [A_n,B_n]\,\mbox{and}\,M_i \in [-\mathrm W_n,\mathrm W_n]\,~\text{for all $i \in [n]$},\\
&~~~~~~~~~~~~~~~~~~~~~~~~~~~\|\wh \xi-\xi_0\|_{\infty,\mathrm W_n} \le \Delta_n,~\mathcal H^2\left(f_{\wh G_\trd,K-p},f_{G,K-p}\right)
    \le \widetilde D~\lambda^{2}_n\Bigg\},
\end{align}
where $A_n=\underline\kappa n^{-6/(K-p)}$ and $B_n=(3\overline\kappa\vee\widetilde\kappa)\log n$.
Observe that by Assumptions~\ref{assu:limma_trnd_var}-\ref{assum:trend_estimation}, Lemma~\ref{lem:prop_chi_sq}, and Lemma~\ref{lem:hell_limm_trnd_cs}, we can conclude that $\mathbb P(\mathcal B^c_\trd) \lesssim n^{-2}$.

Now observe that that using \eqref{eq:lik_1d_limma_trend}, we have
\begin{align}
    \min\{P^\trd_i,\wh P^\trd_i\} \ge 2 \left(1-\Phi(|z| / \nu \underline M^{1/2}\underline L^{1/2}_\trd)\right),
\end{align}
for all $i \in [n]$. Then for $|z| \le \underline z:=\nu\underline M^{1/2}\underline L^{1/2}_\trd z_{1 - \zeta/2}$, it holds that:
\[
\min\{P^\trd_i,\wh P^\trd_i\} \ge 2 \wb{\Phi}(z_{1 - \zeta/2}) = \zeta,
\]
for all $i \in [n]$.
Therefore, we can conclude that $P^\trd_i \wedge \zeta - \wh P^\trd_i \wedge \zeta = \zeta - \zeta = 0$ for $|Z_i| \le \underline z$. Furthermore, we also have
\begin{align}
\left|P^\trd_i \wedge \zeta-\wh P^\trd_i \wedge \zeta\right| & \le \left|P^\trd_i - \wh P^\trd_i\right|.
\end{align}
Furthermore, observe that
\begin{align}
\label{eq:master_conv_limm_trd}
    &\frac{1}{n}\sum_{i=1}^n\mathbb E\left[\left|P^\trd_i \wedge \zeta-\wh P^\trd_i \wedge \zeta\right|\right]\nonumber\\
    & \le \frac{1}{n}\sum_{i=1}^n\left\{\mathbb E\left[\left|P^\trd_i \wedge \zeta-\wh P^\trd_i \wedge \zeta\right|\mathds 1(\mathcal B_\trd)\right]+\mathbb E\left[\left|P^\trd_i \wedge \zeta-\wh P^\trd_i \wedge \zeta\right|\mathds 1(\mathcal B^c_\trd)\right]\right\}\nonumber\\
    & \lesssim_{K, p,\underline M_\trd,\overline M_\trd,\underline L_\trd,\overline U_\trd} \frac{1}{n^2}+\frac{1}{n}\sum_{i=1}^n\mathbb E\left[\left|P^\trd_i \wedge \zeta-\wh P^\trd_i \wedge \zeta\right|\mathds 1(\mathcal B_\trd)\right]\\
    & \le \frac{1}{n^2}+\frac{1}{n}\sum_{i=1}^n\mathbb E\left[\left|P^\trd_i \wedge \zeta-\wh P^\trd_i \wedge \zeta\right|\mathds 1(\mathcal B_\trd)\right].
\end{align}
Now, define $ \calD^\trd_{i,\star}(S^2_i,M_i):=(\widehat{\calD}^{\trd}_i(S_i^2,\widehat{\xi}(M_i))+\calD^{\trd}_i(S_i^2,\xi_0(M_i)))/2,$ for all $i \in [n]$. Proceeding as in the proof of Lemma S9 of \cite{ignatiadis2025empirical}, we can show using $P^\trd_i,\wh P^\trd_i \in [0,1]$ that
\begin{align}
    &\frac{1}{n}\sum_{i=1}^n\mathbb E\left[\left|P^\trd_i \wedge \zeta-\wh P^\trd_i \wedge \zeta\right|\mathds 1(\mathcal B_\trd)\right]\nonumber\\
    &\le\frac{1}{n}\sum_{i=1}^n\mathbb E\left[\left|P^\trd_i \wedge \zeta-\wh P^\trd_i \wedge \zeta\right|\mathds 1(\mathcal B_\trd) \cdot \mathds 1(|Z_i| \ge \underline z)\right]\nonumber\\
    & \le \frac{1}{n}\sum_{i=1}^n\mathbb E\left[\left|\frac{N^{\trd}_i(Z_i,S_i^2,M_i,\xi_0,G)-N^{\trd}_i(Z_i,S_i^2,M_i,\wh\xi,\wh G_\trd)}{\calD^\trd_{i,\star}(S^2_i,M_i)}\right|\mathds 1(\mathcal B_\trd) \cdot \mathds 1(|Z_i| \ge \underline z)\right]\nonumber\\
    & \le \frac{1}{n}\sum_{i=1}^n\mathbb E\left[\sup_{|z|\ge \underline z}\left|\frac{N^{\trd}_i(z,S_i^2,M_i,\xi_0,G)}{\calD^{\trd}_i(S_i^2,\xi_0(M_i))}-\frac{N^{\trd}_i(z,S_i^2,M_i,\wh\xi,\wh G_\trd)}{\widehat{\calD}^{\trd}_i(S_i^2,\widehat{\xi}(M_i))}\right|\mathds 1(\mathcal B_\trd)\right]\nonumber\\
    & \lesssim \frac{1}{n}\sum_{i=1}^n\mathbb E\left[\sup_{|z|\ge \underline z}\left|\frac{N^{\trd}_i(z,S_i^2,M_i,\xi_0,G)-N^{\trd}_i(z,S_i^2,M_i,\wh\xi,\wh G_\trd)}{\calD^\trd_{i,\star}(S^2_i,M_i)}\right|\mathds 1(\mathcal B_\trd)\right]\nonumber\\
    & \hskip 4em +\frac{1}{n}\sum_{i=1}^n\mathbb E\left[\left|\frac{\widehat{\calD}^{\trd}_i(S_i^2,\widehat{\xi}(M_i))-\calD^{\trd}_i(S_i^2,\xi_0(M_i))}{\calD^\trd_{i,\star}(S^2_i,M_i)}\right|\mathds 1(\mathcal B_\trd)\right].
\end{align}
Let us define
\begin{align}
\label{eq:def_frk_b}
\mathfrak B^\trd_1&=\frac{1}{n}\sum_{i=1}^n\mathbb E\left[\sup_{|z|\ge \underline z}\left|\frac{N^{\trd}_i(z,S_i^2,M_i,\xi_0,G)-N^{\trd}_i(z,S_i^2,M_i,\wh\xi,\wh G_\trd)}{\calD^\trd_{i,\star}(S^2_i,M_i)}\right|\mathds 1(\mathcal B_\trd)\right]\\
\mathfrak B^\trd_2&=\frac{1}{n}\sum_{i=1}^n\mathbb E\left[\left|\frac{\widehat{\calD}^{\trd}_i(S_i^2,\widehat{\xi}(M_i))-\calD^{\trd}_i(S_i^2,\xi_0(M_i))}{\calD^\trd_{i,\star}(S^2_i,M_i)}\right|\mathds 1(\mathcal B_\trd)\right].
\end{align}
\paragraph{Controlling $\mathfrak B^\trd_1$.}
We shall show that when we restrict ourselves to $\mathcal B_\trd$ and $|z|\ge \underline z$,
\[
\sup_{|z|\ge \underline z}\left|N^{\trd}_i(z,S_i^2,M_i,\xi_0,G)-N^{\trd}_i(z,S_i^2,M_i,\wh\xi,\wh G_\trd)\right| \le \zeta_n,
\]
almost surely, for some sequence $\zeta_n>0$ to be specified later. Then, proceeding as in the proof of Lemma S9 of \cite{ignatiadis2025empirical}, we can show that
\begin{align}
\label{eq:b_1_trnd_p_val}
    \mathfrak B^\trd_1 \lesssim \zeta_n\, \mathbb E\left[\frac{\mathds 1(\mathcal B_\trd)}{\calD^\trd_{1,\star}(S^2_1,M_1)}\right] &\le \zeta_n\, \mathbb E\left[\frac{\mathds 1\{S^2_1 \le B_n\}}{\calD^\trd_{1,\star}(S^2_1,M_1)}\right]\\
    &=\zeta_n\; \mathbb E\left[\mathbb E\left[\frac{\mathds 1\{S^2_1 \le B_n\}}{\calD^\trd_{1,\star}(S^2_1,M_1)} \Bigg\lvert\, M_1\right]\right]\\
    & \le 2\cdot \zeta_n\; \mathbb E\left[\mathbb E\left[\frac{\mathds 1\{S^2_1 \le B_n\}}{\calD^{\trd}_1(S_1^2,\xi_0(M_1))} \Bigg\lvert\, M_1\right]\right] \overset{(1)}{\lesssim} B_n \zeta_n =\zeta_n(\log n),
\end{align}
where inequality (1) follows since the density of $S^2_1 \mid M_1$ is proportional to $\calD^{\trd}_1(S_1^2,\xi_0(M_1))$ with the Jacobian term uniformly bounded in $[\underline M_\trd, \overline M_\trd]$. Furthermore, in the derivation, we have used the fact that the pairs $(S^2_1,M_1),\ldots,(S^2_n,M_n)$ are identically distributed.
Using the triangle inequality, we can show that for any $|z| \ge \underline z$
\begin{align}
\label{eq:decom_n_trd_1}
  &\left|N^{\trd}_i(z,S_i^2,M_i,\xi_0,G)-N^{\trd}_i(z,S_i^2,M_i,\wh\xi,\wh G_\trd)\right|  \\
  & \le \left|N^{\trd}_i(z,S_i^2,M_i,\xi_0,G)-N^{\trd}_i(z,S_i^2,M_i,\xi_0,\wh G_\trd)\right|+\left|N^{\trd}_i(z,S_i^2,M_i,\xi_0,\wh G_\trd)-N^{\trd}_i(z,S_i^2,M_i,\wh\xi,\wh G_\trd)\right|.
\end{align}
For the first term, we can use the techniques in Lemma S9 of \cite{ignatiadis2025empirical} to show that, if we restrict to $\mathcal B_\trd$, we have
\begin{align}
    \label{eq:part_1_term_b_1}
    &\left|N^{\trd}_i(z,S_i^2,M_i,\xi_0,G)-N^{\trd}_i(z,S_i^2,M_i,\xi_0,\wh G_\trd)\right|\\
    & \lesssim_{K, p,\underline M_\trd,\overline M_\trd,\underline L_\trd,\overline U_\trd,\zeta} \mathcal H(f_{\wh G_\trd,K-p},f_{G,K-p})|\log \mathcal H(f_{\wh G_\trd,K-p},f_{G,K-p})|^{1/2}\\
    &\lesssim \lambda_n \sqrt{|\log\lambda_n|}, \qquad \mbox{where $\lambda_n$ is defined in \eqref{eq:def_lambda_trd_n}}.
\end{align}
From the definition of $\lambda_n$ and using $\Delta_n \ll 1$, we can conclude that $|\log \lambda_n| \lesssim (\log n).$ Therefore, 
\begin{align}
    &\left|N^{\trd}_i(z,S_i^2,M_i,\xi_0,G)-N^{\trd}_i(z,S_i^2,M_i,\xi_0,\wh G_\trd)\right|\\
    &\lesssim \lambda_n \sqrt{|\log\lambda_n|}\lesssim \max\left\{\Delta_n (\log n)^{3/2},\frac{(\log n)^{3/2}}{n^{1/4}}\Delta^{1/2}_n,\frac{(\log n)}{n^{1/4}}|\log \Delta_n|^{h_2/4}\Delta^{\frac{1}{2}(1-h_1/2)}_n,\frac{(\log n)^{3/2}}{\sqrt{n}}\right\},\nonumber
\end{align}
for all $i \in [n]$. 

To control the second term, we consider the function
\begin{align}
    \label{eq:def_n_func}
    \wh g_{\trd}(\eta;z,s^2)
&:= 
C_{K,p}
\int_{0}^{\infty}
\frac{ (t^2)^{-\,\frac{K-p+1}{2}} 
       \left(s^2/\eta^2\right)^{\frac{K-p}{2}-1}}
     { \sqrt{ (K-p+1)t^2 - (K-p)\left(s^2/\eta^2\right) } } \\
&\quad\quad \times
f_{\widehat{G}_{\trd},K-p+1}(t^2)\,
\mathbbm{1}\!\left\{
t^2 \ge 
\frac{ (K-p)\left(s^2/\eta^2\right) + z^2/(\nu^2\eta^2)}
     { K-p+1 }
\right\}
\, dt^2 ,
\end{align}
for any $s^2>0$ and $|z| \ge \underline z$. Observe that
\begin{align}
N^{\trd}_i(z,S_i^2,M_i,\xi_0,\wh G_\trd)&:=\wh g_{\trd}(\xi_0(M_i);z,S^2_i)\\
N^{\trd}_i(z,S_i^2,M_i,\wh \xi,\wh G_\trd)&:=\wh g_{\trd}(\wh\xi(M_i);z,S^2_i).
\end{align}
Then using the expression of $\frac{d}{d\eta}\wh g_{\trd}(\eta;z,s^2)$ from Lemma~\ref{lem:derivative_g_trd} and Assumption~\ref{assum:trend_estimation}, we can conclude that for all $i \in [n]$,
\begin{align}
\label{eq:deriv_wh_g}
\left|\frac{d}{d\eta}\wh g_{\trd}(\check\eta_i;z,S^2_i)\right| \lesssim_{\overline M_\trd,\underline M_\trd,\underline L_\trd,\wb U_\trd,\underline z} (\log n),
\end{align}
if we restrict ourselves to $\mathcal B_\trd$, where
\[
\check\eta_i = t_i\xi_0(M_i)+(1-t_i)\wh \xi(M_i), \qquad \mbox{for any $t_1,\ldots,t_n \in [0,1]$.}
\]
Consequently using Assumption~\ref{assum:trend_estimation}, we can conclude that
\begin{align}
\label{eq:b_1_trd_new}
\sup_{|z|\ge \underline z}\left|N^{\trd}_i(z,S_i^2,M_i,\xi_0,\wh G_\trd)-N^{\trd}_i(z,S_i^2,M_i,\wh \xi,\wh G_\trd)\right| \lesssim_{K, p,\underline M_\trd,\overline M_\trd,\underline L_\trd,\overline U_\trd,\zeta} (\log n)~\Delta_n.
\end{align}
Therefore, we choose 
\[
\zeta_n = \max\left\{\Delta_n (\log n)^{3/2},\frac{(\log n)^{3/2}}{n^{1/4}}\Delta^{1/2}_n,\frac{(\log n)}{n^{1/4}}|\log \Delta_n|^{h_2/4}\Delta^{\frac{1}{2}(1-h_1/2)}_n,\frac{(\log n)^{3/2}}{\sqrt{n}}\right\},
\]
and conclude using \eqref{eq:b_1_trnd_p_val} and \eqref{eq:b_1_trd_new}, we have
\begin{align}
    \label{eq:bound_b_1_trd}
    \mathfrak B^\trd_1 &\lesssim \max\Big\{\Delta_n (\log n)^{5/2},\frac{(\log n)^{5/2}}{n^{1/4}}\Delta^{1/2}_n,\frac{(\log n)^2}{n^{1/4}}|\log \Delta_n|^{h_2/4}\Delta^{\frac{1}{2}(1-h_1/2)}_n,\frac{(\log n)^{5/2}}{\sqrt{n}}\Big\},
\end{align}
where the absolute constants in the foregoing inequality depend only on $h_1,h_2,K, p,\underline M_\trd,\overline M_\trd,\underline L_\trd,\overline U_\trd$ and $\zeta$.

\paragraph{Controlling $\mathfrak B^\trd_2$.}
To control $\mathfrak B^\trd_2$, observe that
\begin{align}
\label{eq:decomp_b_2_trd_1}
    \mathfrak B^\trd_2 & = \frac{1}{n}\sum_{i=1}^n\mathbb E\left[\left|\frac{\widehat{\calD}^{\trd}_i(S_i^2,\widehat{\xi}(M_i))-\calD^{\trd}_i(S_i^2,\xi_0(M_i))}{\calD^\trd_{i,\star}(S^2_i,M_i)}\right|\mathds 1(\mathcal B_\trd)\right]\\
    & \lesssim \frac{1}{n}\sum_{i=1}^n\mathbb E\left[\left|\frac{\widehat{\calD}^{\trd}_i(S_i^2,\widehat{\xi}(M_i))-\widehat{\calD}^{\trd}_i(S_i^2,\xi_0(M_i))}{\calD^\trd_{i,\star}(S^2_i,M_i)}\right|\mathds 1(\mathcal B_\trd)\right]\\
    &\qquad+\frac{1}{n}\sum_{i=1}^n\mathbb E\left[\left|\frac{\calD^{\trd}_i(S_i^2,\xi_0(M_i))-\widehat{\calD}^{\trd}_i(S_i^2,\xi_0(M_i))}{\calD^\trd_{i,\star}(S^2_i,M_i)}\right|\mathds 1(\mathcal B_\trd)\right]\\
    & \lesssim \frac{1}{n}\sum_{i=1}^n\mathbb E\left[\left|\frac{\widehat{\calD}^{\trd}_i(S_i^2,\widehat{\xi}(M_i))-\widehat{\calD}^{\trd}_i(S_i^2,\xi_0(M_i))}{\calD^\trd_{i}(S^2_i,\xi_0(M_i))}\right|\mathds 1(\mathcal B_\trd)\right]\\
    &\qquad+\frac{1}{n}\sum_{i=1}^n\mathbb E\left[\left|\frac{\calD^{\trd}_i(S_i^2,\xi_0(M_i))-\widehat{\calD}^{\trd}_i(S_i^2,\xi_0(M_i))}{\calD^\trd_{i,\star}(S^2_i,M_i)}\right|\mathds 1(\mathcal B_\trd)\right]
\end{align}
To analyze the second term, define 
\[
\wt \calD^\trd_{i,\star}(S^2_i,M_i):=\frac{\calD^{\trd}_i(S_i^2,\xi_0(M_i))+\widehat{\calD}^{\trd}_i(S_i^2,\xi_0(M_i))}{2}.
\]
Then we can write
\begin{align}
\label{eq:decomp_b_2_trd_2}
  &\frac{1}{n}\sum_{i=1}^n\mathbb E\left[\left|\frac{\calD^{\trd}_i(S_i^2,\xi_0(M_i))-\widehat{\calD}^{\trd}_i(S_i^2,\xi_0(M_i))}{\calD^\trd_{i,\star}(S^2_i,M_i)}\right|\mathds 1(\mathcal B_\trd)\right] \\
  & \lesssim \frac{1}{n}\sum_{i=1}^n\mathbb E\left[\left|\frac{\calD^{\trd}_i(S_i^2,\xi_0(M_i))-\widehat{\calD}^{\trd}_i(S_i^2,\xi_0(M_i))}{\wt \calD^\trd_{i,\star}(S^2_i,M_i)}\right|\mathds 1(\mathcal B_\trd)\right]\\
  &\qquad+\frac{1}{n}\sum_{i=1}^n\mathbb E\left[\left|\frac{\calD^{\trd}_i(S_i^2,\xi_0(M_i))-\widehat{\calD}^{\trd}_i(S_i^2,\xi_0(M_i))}{\wt \calD^\trd_{i,\star}(S^2_i,M_i)}\right|\times\left|\frac{\wt \calD^\trd_{i,\star}(S^2_i,M_i)- \calD^\trd_{i,\star}(S^2_i,M_i)}{ \calD^\trd_{i,\star}(S^2_i,M_i)}\right|\mathds 1(\mathcal B_\trd)\right].
\end{align}
Next, observe that
\begin{align}
\label{eq:decomp_corrected_thm_7}
    & \frac{1}{n}\sum_{i=1}^n\mathbb E\left[\left|\frac{\calD^{\trd}_i(S_i^2,\xi_0(M_i))-\widehat{\calD}^{\trd}_i(S_i^2,\xi_0(M_i))}{\wt \calD^\trd_{i,\star}(S^2_i,M_i)}\right|\times\left|\frac{\wt \calD^\trd_{i,\star}(S^2_i,M_i)- \calD^\trd_{i,\star}(S^2_i,M_i)}{ \calD^\trd_{i,\star}(S^2_i,M_i)}\right|\mathds 1(\mathcal B_\trd)\right]\\
    & \asymp \frac{1}{n}\sum_{i=1}^n\mathbb E\left[\left|\frac{\calD^{\trd}_i(S_i^2,\xi_0(M_i))-\widehat{\calD}^{\trd}_i(S_i^2,\xi_0(M_i))}{\calD^{\trd}_i(S_i^2,\xi_0(M_i))+\widehat{\calD}^{\trd}_i(S_i^2,\xi_0(M_i))}\right|\times\left|\frac{\wh\calD^\trd_{i}(S^2_i,\xi_0(M_i))- \wh\calD^\trd_{i}(S^2_i,\wh \xi(M_i))}{ \calD^\trd_{i,\star}(S^2_i,M_i)}\right|\mathds 1(\mathcal B_\trd)\right]\\
    & \lesssim \frac{1}{n}\sum_{i=1}^n\mathbb E\left[\left|\frac{\wh\calD^\trd_{i}(S^2_i,\xi_0(M_i))- \wh\calD^\trd_{i}(S^2_i,\wh \xi(M_i))}{ \calD^\trd_{i,\star}(S^2_i,M_i)}\right|\mathds 1(\mathcal B_\trd)\right]\\
    & \lesssim \frac{1}{n}\sum_{i=1}^n\mathbb E\left[\left|\frac{\wh\calD^\trd_{i}(S^2_i,\xi_0(M_i))- \wh\calD^\trd_{i}(S^2_i,\wh \xi(M_i))}{ \calD^\trd_{i}(S^2_i,\xi_0(M_i))}\right|\mathds 1(\mathcal B_\trd)\right].
\end{align}
The foregoing inequality along with \eqref{eq:decomp_b_2_trd_2} allows us to re-write \eqref{eq:decomp_b_2_trd_1} as follows:
\begin{align}
\label{eq:decomp_b_2_trd_final}
    \mathfrak B^\trd_2 & = \frac{1}{n}\sum_{i=1}^n\mathbb E\left[\left|\frac{\widehat{\calD}^{\trd}_i(S_i^2,\widehat{\xi}(M_i))-\calD^{\trd}_i(S_i^2,\xi_0(M_i))}{\calD^\trd_{i,\star}(S^2_i,M_i)}\right|\mathds 1(\mathcal B_\trd)\right]\\
    & \lesssim \frac{1}{n}\sum_{i=1}^n\mathbb E\left[\left|\frac{\calD^{\trd}_i(S_i^2,\xi_0(M_i))-\widehat{\calD}^{\trd}_i(S_i^2,\xi_0(M_i))}{\wt \calD^\trd_{i,\star}(S^2_i,M_i)}\right|\mathds 1(\mathcal B_\trd)\right]\\
    &\qquad+\frac{1}{n}\sum_{i=1}^n\mathbb E\left[\left|\frac{\wh\calD^\trd_{i}(S^2_i,\xi_0(M_i))- \wh\calD^\trd_{i}(S^2_i,\wh \xi(M_i))}{ \calD^\trd_{i}(S^2_i,\xi_0(M_i))}\right|\mathds 1(\mathcal B_\trd)\right].
\end{align}
Proceeding as in the proof of Lemma S.10 of \citet{ignatiadis2025empirical} and using Lemma \ref{lem:hell_limm_trnd_cs}, we can show that
\begin{align}
\label{eq:part_2_b_2_bound}
   &\frac{1}{n}\sum_{i=1}^n\mathbb E\left[\left|\frac{\calD^{\trd}_i(S_i^2,\xi_0(M_i))-\widehat{\calD}^{\trd}_i(S_i^2,\xi_0(M_i))}{\wt \calD^\trd_{i,\star}(S^2_i,M_i)}\right|\mathds 1(\mathcal B_\trd)\right]\\
   & \lesssim_{K,p,\underline M_\trd,\overline M_\trd,\underline L_\trd,\overline U_\trd} \frac{(\log n)}{\sqrt{n}}+\lambda_n\\
   &\lesssim \Delta_n(\log n)+\frac{(\log n)}{n^{1/4}}\Delta^{1/2}_n+\frac{(\log n)^{1/2}}{n^{1/4}}|\log \Delta_n|^{h_2/4}\Delta^{\frac{1}{2}(1-h_1/2)}_n+\frac{(\log n)}{\sqrt{n}}.
\end{align}
Next, using mean value theorem for the function $\eta \mapsto f_{\widehat{G}_{\trd},K-p}(s^2/\eta^2)$ and Lemma \ref{lem:prop_chi_sq} (1), we can conclude that for all $i \in \{1,\ldots,n\}$
\begin{align}
&\left|\widehat{\calD}^{\trd}_i(S_i^2,\widehat{\xi}(M_i))-\widehat{\calD}^{\trd}_i(S_i^2,\xi_0(M_i))\right|\\
& \le  
    \|\wh \xi-\xi_0\|_\infty~\left(\mathbb E_{\tau^2 \sim \wh G_\trd}\left[\frac{2-K+p}{\widetilde\eta_i}+\frac{(K-p)X}{\tau^2\widetilde\eta_i}\,\Big|\,X=S^2_i/\widetilde\eta^2_i\right]\right)\times f_{\widehat{G}_{\trd},K-p}(S^2_i/\wt\eta^2_i),
\end{align}
where $\widetilde \eta_i =t\xi_0(M_i)+(1-t)\wh \xi(M_i)$ for some $t \in (0,1)$. Since, $S^2_i \lesssim \log n$ for all $i \in [n]$, when we restrict ourselves to $\mathcal B_\trd$, using Assumption~\ref{assum:trend_estimation} and observing that $f_{\widehat{G}_{\trd},K-p}(S^2_i/\wt\eta^2_i)$ remains bounded by a constant independent of $n$ as $\wh G_\trd \in \mathcal G_\trd$, we can conclude that
\begin{align}
\left|\widehat{\calD}^{\trd}_i(S_i^2,\widehat{\xi}(M_i))-\widehat{\calD}^{\trd}_i(S_i^2,\xi_0(M_i))\right| &\lesssim_{K,p,\underline M_\trd,\overline M_\trd,\underline L_\trd,\overline U_\trd} \Delta_n(\log n).
\end{align}
Using the foregoing inequality and the techniques from \eqref{eq:b_1_trnd_p_val}, we can show that
\begin{align}
\label{eq:part_1_b_2_bound}
    \frac{1}{n}\sum_{i=1}^n\mathbb E\left[\left|\frac{\widehat{\calD}^{\trd}_i(S_i^2,\widehat{\xi}(M_i))-\widehat{\calD}^{\trd}_i(S_i^2,\xi_0(M_i))}{\calD^\trd_{i}(S^2_i,\xi_0(M_i))}\right|\mathds 1(\mathcal B_\trd)\right] &\lesssim_{K,p,\underline M_\trd,\overline M_\trd,\underline L_\trd,\overline U_\trd}  \Delta_n(\log n)^{2}. 
\end{align}
Combining \eqref{eq:part_2_b_2_bound} and \eqref{eq:part_1_b_2_bound}, we can conclude that
\begin{align}
    \label{eq:bound_b_2_trd}
    \mathfrak B^\trd_2 \lesssim_{K,p,\underline M_\trd,\overline M_\trd,\underline L_\trd,\overline U_\trd} \Delta_n(\log n)^{2}+\frac{(\log n)}{n^{1/4}}\Delta^{1/2}_n+\frac{(\log n)^{1/2}}{n^{1/4}}|\log \Delta_n|^{h_2/4}\Delta^{\frac{1}{2}(1-h_1/2)}_n+\frac{(\log n)}{\sqrt{n}}.\nonumber\\
\end{align}
Combining \eqref{eq:master_conv_limm_trd} along with \eqref{eq:bound_b_1_trd} and \eqref{eq:bound_b_2_trd}, we can conclude that
\begin{align}
    &\frac{1}{n}\sum_{i=1}^n\mathbb E\left[\left|P^\trd_i \wedge \zeta-\wh P^\trd_i \wedge \zeta\right|\right]\\
    &\lesssim_{K,p,\underline M_\trd,\overline M_\trd,\underline L_\trd,\overline U_\trd} \max\left\{\Delta_n(\log n)^{5/2},\frac{(\log n)^{5/2}}{n^{1/4}}\Delta^{1/2}_n,\frac{(\log n)^2}{n^{1/4}}|\log \Delta_n|^{h_2/4}\Delta^{\frac{1}{2}(1-h_1/2)}_n,\frac{(\log n)^{5/2}}{\sqrt{n}}\right\}.
\end{align}
This implies Theorem \ref{thm:conv_p_trnd_orc}.

\subsection{Proof of Proposition \ref{prop:avg_sign_limma_trnd}}
The proof of the proposition follows by exactly retracing the arguments of Proposition 15 of \cite{ignatiadis2025empirical} along with Theorem~\ref{thm:conv_p_trnd_orc}.

\subsection{Proof of Theorem~\ref{thm:final_rate_cs}}
The proof of the theorem proceeds by integrating Proposition~\ref{prop:avg_sign_limma_trnd} along with the fact that the oracle p-values $\{P^\trd_i\}$ are critically dense at $\alpha$ in the sense Assumption~\ref{asm:1d_limma_trend_bh}. In particular, this theorem can be proved by following the same strategy adopted for proving Theorem~\ref{thm:final_rate} for which we provide a detailed description in Supplement~\ref{sec:proof_fdr_joint}. 

\section{Proof of asymptotic FDR control in Section~\ref{sec:reg_ltrd_int}}
In this section, we shall consistently assume Assumptions~\ref{assum:missp_trend_estimation} and~\ref{assu:misspec_limma_trnd_var}.
\subsection{Approximate NPMLE property of $\wh G_\trd$}
\label{sec:approx_npmle_ms}
We begin by proving an approximate NPMLE property of $\wh G_\trd$ (analogous to Lemma~\ref{lem:lik_ratio_limma_trnd}) under the working model \eqref{eq:misspecified-limma-reg} and \eqref{eq:tau_mis_v_mis}.
\begin{lemm}
\label{lem:lik_ratio_limma_trnd_ms}
Consider the working model specified through \eqref{eq:bayesian_cs_model_2}, \eqref{eq:misspecified-limma-reg} and \eqref{eq:tau_mis_v_mis}. There exist constants $D_\trd>0$ and $n_{\trd,2}\in\mathbb N_{\ge1}$, depending only on $K,p,\underline M_\trd,\wb M_\trd,\underline L_\trd,\wb U_\trd$, such that the following hold with probability at least $1-o(n^{-2})$,
\begin{align}
\label{eq:log_lik_limma_trend_approx_ms}
\prod_{i=1}^n\frac{f_{\wh G_\trd,K-p}(\misV{i})}{f_{G_\mis,K-p}(\misV{i})}\ge e^{-n\,D_\trd\mathfrak R_{n,2}},
\end{align}
where
\[
\mathfrak R_{n,2}=\frac{(\log n)^2}{\sqrt{n}}\cdot \Delta_n+\frac{(\log n)}{\sqrt{n}}\cdot|\log \Delta_n|^{h_2/2}\,\Delta^{1-h_1/2}_n+\Delta_n(\log n)^2.
\]
\end{lemm}
However, due to the change in the working model, our proof requires some modification.

\begin{proof}
    Observe that under the working model defined by \eqref{eq:bayesian_cs_model_2}, \eqref{eq:misspecified-limma-reg} and \eqref{eq:tau_mis_v_mis}, the factor $\eta^{-2}$ disappears from the marginal density of $V^2$. Consequently, we redefine
    \[
\ell(S_i^2;\xi,\rmG,M_i)=\log f_{\rmG,K-p}(S_i^2/\xi^2(M_i)).
\]
Based on this definition, we aim to show that with probability greater than $1-o(n^{-2})$,
\begin{align}
    \frac{1}{n}\sum_{i=1}^n\ell(S^2_i;\xi_\mis,\wh G_\trd,M_i) \ge \frac{1}{n}\sum_{i=1}^n\ell(S^2_i;\xi_\mis,\misG,M_i)-D_\trd\,\mathfrak R_{n,2},
\end{align}
where $\mathfrak R_{n,2}$ is defined in the lemma statement. As in the proof of Lemma~\ref{lem:lik_ratio_limma_trnd}, we adopt the decomposition of 
\[
\frac{1}{n}\sum_{i=1}^n\ell(S^2_i;\xi_\mis,\wh G_\trd,M_i) - \frac{1}{n}\sum_{i=1}^n\ell(S^2_i;\xi_\mis,\misG,M_i),
\]
into $\mathrm{Sub}_{n,1}$ and $\mathrm{Sub}_{n,2}$ (defined in \eqref{eq:def_sub_1} and \eqref{eq:def_sub_2}) and show that 
\[
\mathrm{Sub}_{n,1}+\mathrm{Sub}_{n,2} \le D_\trd\,\mathfrak R_{n,2},
\]
with probability greater than $1-o(n^{-2})$. Then, arguing as in the proof of Lemma~\ref{lem:lik_ratio_limma_trnd}, the lemma follows.

To control the fluctuation of $\mathrm{Sub}_{n,1}$, we can restrict ourselves to
\[
\mathcal A_{\mis,\trd}=\{|M_i| \le \mathrm W_n,\misV{i}\in[\wt A_n,\wt B_n]\ \text{for all }i\in[n]\}, \qquad \mbox{as in \eqref{eq:cala_trd}.}
\]
By Assumptions~\ref{assum:missp_trend_estimation}-\ref{assu:misspec_limma_trnd_var}, \eqref{eq:bayesian_cs_model_2}, \eqref{eq:misspecified-limma-reg}, \eqref{eq:tau_mis_v_mis} and Lemma~\ref{lem:prop_chi_sq} (3), we have $\mathbb P(\mathcal A_{\mis,\trd}) \lesssim n^{-2}$. Next, we use Taylor expansion of $\ell(S^2_i;\wh\xi,\wh G_\trd,M_i)$ around $\ell(S^2_i;\xi_\mis,\wh G_\trd,M_i)$ as before. In that direction, we redefine the gradient and the Hessian functions from
\eqref{eq:def_d_and_h} as follows:
\begin{align}
D(\eta,s^2;\rmG)&:=\frac{\partial}{\partial\eta}\log f_{\rmG,K-p}(s^2/\eta^2), \quad \mbox{and}\\
H(\eta,s^2;\rmG)&:=\frac{\partial^2}{\partial\eta^2}\log f_{\rmG,K-p}(s^2/\eta^2).
\end{align}
Then, we control
\begin{align}
    \frac{1}{n}\sum_{i=1}^n D(\xi_\mis(M_i),S_i^2;\wh G_\trd)\big(\wh\xi(M_i)-\xi_\mis(M_i)\big)+\frac{1}{2n}\sum_{i=1}^n H(\eta_{t,i},S_i^2;\wh G_\trd)\big(\wh\xi(M_i)-\xi_\mis(M_i)\big)^2.
\end{align}
Within $\mathcal A_{\mis,\trd}$, the Hessian term can be handled using exactly the same technique as in the proof of Lemma~\ref{lem:lik_ratio_limma_trnd}. In particular, using Assumption~\ref{assu:misspec_limma_trnd_var} and Lemma~\ref{lem:prop_chi_sq} (2)
\[
\left\lvert\frac{1}{2n}\sum_{i=1}^n H(\eta_{t,i},S_i^2;\wh G_\trd)\big(\wh\xi(M_i)-\xi_\mis(M_i)\big)^2\right\rvert\lesssim (\log n)^2\,\|\wh \xi-\xi_\mis\|^2_\infty \lesssim (\log n)^2\Delta_n^2.
\]
For the gradient term, we need to redefine the centering term used to construct $U_{1,n}$ and $U_{2,n}$. In particular, we redefine
\[
\overline{\rmD}_i(\xi_\mis,\rmG,M_i):=\int_{\widetilde A_n}^{\widetilde B_n}\rmD_i(v^2,M_i;\xi_\mis,\rmG)f_{M_i,K-p}(v^2)\,\dd v^2,
\]
where
\[
f_{M_i,K-p}(v^2)=\int_0^{\infty}p_{\chi^2}(v^2\mid K-p,\tau^2)\,G_{M_i, \tau^2}(\dd \tau^2 \mid M_i),
\]
where $G_{m, \tau^2}(\,\cdot\,\mid m)$ is the conditional distribution of $\misT{i}=\sigma^2_i/\xi^2_\mis(M_i)$ given $M_i=m$ (these are same for all $i \in [n]$ by exchangeability).
We decompose
\[
\frac{1}{n}\sum_{i=1}^n \rmD_i(\misV{i},M_i;\xi_\mis,\wh G_\trd)\big(\wh\xi(M_i)-\xi_\mis(M_i)\big)=U_{1,n}+U_{2,n},
\]
where
\[
U_{1,n}:=\frac{1}{n}\sum_{i=1}^n\big(\rmD_i(\misV{i},M_i;\xi_\mis,\wh G_\trd)-\overline{\rmD}_i(\xi_\mis,\wh G_\trd,M_i)\big)\big(\wh\xi(M_i)-\xi_\mis(M_i)\big)
\]
and
\[
U_{2,n}:=\frac{1}{n}\sum_{i=1}^n\overline{\rmD}_i(\xi_\mis,\wh G_\trd,M_i)\big(\wh\xi(M_i)-\xi_\mis(M_i)\big).
\]
Using the same techniques as the proof of Lemma~\ref{lem:lik_ratio_limma_trnd}, we can conclude that
\begin{align}
\label{eq:bound_u_1_n_mis}
U_{1,n} \lesssim \frac{(\log n)^2}{\sqrt{n}}\cdot \Delta_n+\frac{(\log n)}{\sqrt{n}}\cdot|\log \Delta_n|^{h_2/2}\,\Delta^{1-h_1/2}_n,
\end{align}
with probability greater than $1-o(n^{-2})$. However, due to the different centering used in the construction of $U_{2,n}$, the same analysis will not go through for bounding the fluctuation of $U_{2,n}$. In particular, under \eqref{eq:misspecified-limma-reg}, the score identity \eqref{eq:bayesian_score_eq} does not hold. 
Consequently, we cannot get a sharp bound on $U_{2,n}$ akin to \eqref{eq:sharp_u_2_n}. Instead, we use the definition of $\overline{\rmD}_i(\xi_\mis,\wh G_\trd,M_i)$, along with the definition of $\mathcal A_{\mis,\trd}$ to conclude using Lemma~\ref{lem:prop_chi_sq} (1) that
\[
|\overline{\rmD}_i(\xi_\mis,\wh G_\trd,M_i)| \lesssim (\log n),
\]
where the constant depends on $K,p,\underline L_\trd,\overline U_\trd,\underline M_\trd$ and $\overline M_\trd$. Therefore on $\mathcal A_{\trd,\mis}$, we have
\[
|U_{2,n}|\lesssim \Delta_n(\log n).
\] 
This implies with probability greater than $1-o(n^{-2})$,
\begin{align}
    \label{eq:sub_n_1_mis}
    \mathrm{Sub}_{n,1} \lesssim \frac{(\log n)^2}{\sqrt{n}}\cdot \Delta_n+\frac{(\log n)}{\sqrt{n}}\cdot|\log \Delta_n|^{h_2/2}\,\Delta^{1-h_1/2}_n+\Delta_n(\log n).
\end{align}
Next, we focus on bounding $\mathrm{Sub}_{n,2}$. Once again, we use Taylor expansion to get an identity similar to \eqref{eq:taylor_sub_n_2}. The Hessian term can be bounded as in the proof of Lemma~\ref{lem:lik_ratio_limma_trnd}. However, under the working model in this framework, $D(\xi_\mis(M_i),S_i^2; G_\mis)$ is not the conditional score and hence the argument used to bound $U_{2,n}$ in Lemma~\ref{lem:lik_ratio_limma_trnd} cannot be used here. Instead, using Lemma~\ref{lem:prop_chi_sq} (1), we can show that on $\mathcal A_\trd$, we have
\[
|D(\xi_\mis(M_i),S_i^2; G_\mis)| \lesssim (\log n),
\]
where the constant in the inequality depends only on $K,p,\underline L_\trd,\overline U_\trd,\underline M_\trd$, and $\wb M_\trd$. Therefore, under this framework, we have
\begin{align}
\label{eq:sharp_bound_sub_n_2_ms}
    \mathrm{Sub}_{n,2}\lesssim \Delta_n^2(\log n)^2+\Delta_n(\log n).
\end{align}
Combining \eqref{eq:sub_n_1_mis} and $\eqref{eq:sharp_bound_sub_n_2_ms}$, we can conclude the lemma by noting that $\Delta_n \ll 1$ for large enough $n$.
\end{proof}

\subsection{Hellinger large deviation}
Next, we utilize the approximate NPMLE property of $\wh G_\trd$ to provide a result analogous to Lemma~\ref{lem:hell_limm_trnd_cs} for $\smash{\mathcal H^2(f_{\wh G_\trd,K-p},f_{G,K-p})}$ under the working model specified by \eqref{eq:bayesian_cs_model_2}, \eqref{eq:misspecified-limma-reg}, and \eqref{eq:tau_mis_v_mis}. Under this working model, we have the following result.

\begin{lemm}
    \label{lem:hell_limm_trnd_ms}
   Fix an absolute constant $c_0>0$. Then there exist constants $\widetilde D_2>0$ and $n_{\trd,2} \in \mathbb N$, depending only on $h_1,h_2,K,p,\underline M_\trd,\overline M_\trd,\underline L_\trd,$ and $\overline U_\trd$, such that for all $n \ge n_{\trd,2}$, 
   \begin{align}
       \mathbb P\left[\mathcal H^2\left(f_{\wh G_\trd,K-p},f_{G_\mis,K-p}\right) \ge \widetilde D_2\,\lambda^2_{n,2}\right]\lesssim \frac{1}{n^2}+e^{-c_0(\log n)^2},
   \end{align}
   where
   \begin{align}
   \label{eq:lamb_n_2_mis}
   \lambda^2_{n,2}:=\max\left\{\Delta_n (\log n)^2,\frac{(\log n)^2}{\sqrt{n}}\cdot \Delta_n,\frac{(\log n)}{\sqrt{n}}\cdot|\log \Delta_n|^{h_2/2}\,\Delta^{1-h_1/2}_n,\frac{(\log n)^2}{n}\right\}.
   \end{align}
\end{lemm}
\begin{proof}
Define event $\mathcal H_{\trd,\mis}$ as follows:
\begin{align}
    \mathcal H_{\trd,\mis}&:=\left\{\prod_{i=1}^n\frac{f_{\wh G_\trd,K-p}(S_i^2/\xi^2_\mis(M_i))}{f_{G_\mis,K-p}(S_i^2/\xi^2_\mis(M_i))} \ge e^{-n\, \mathfrak R_{n,2}}\right\},
\end{align}
where $\mathfrak R_{n,2}$ defined in Lemma~\ref{lem:lik_ratio_limma_trnd_ms}. By part (2) of Lemma~\ref{lem:lik_ratio_limma_trnd_ms}, any solution $\wh G_\trd$ of \eqref{eq:def_small_ell} satisfies
\begin{align}
\label{eq:fundamental_eq_ms}
\mathbb P\!\left(\mathcal H^c_{\trd,\mis}\right) \le o(n^{-2}),
\end{align}
under the working model.
Therefore, for all $\wt D_2>0$, we have
\begin{align}
    &\mathbb P\left[\mathcal H^2\left(f_{\wh G_\trd,K-p},f_{G_\mis,K-p}\right) \ge \widetilde D_2\,\lambda^2_{n,2}\right]\\
    & \le \mathbb P\left[\mathcal H^2\left(f_{\wh G_\trd,K-p},f_{G_\mis,K-p}\right) \ge \widetilde D_2\,\lambda^2_{n,2},\, \mathcal H_{\trd,\mis}\right]+o(n^{-2}).
\end{align}
Therefore, we focus on bounding
\[
\mathbb P\left[\mathcal H^2\left(f_{\wh G_\trd,K-p},f_{G_\mis,K-p}\right) \ge \widetilde D_2\,\lambda^2_{n,2},\, \mathcal H_{\trd,\mis}\right].
\]
We proceed as in the proof of Theorem 9 in \cite{ignatiadis2025empirical} by constructing a proper minimal $n^{-2}$ cover $\mathcal S_{\trd,\mis}=\{f_{G_j,K-p}:j=1,\ldots,\mathcal R_\mis\}\subseteq \mathcal B_\mis$ where
\[
\mathcal B_\mis:=\left\{f_{\rmG,K-p}: \rmG \in \mathcal G_\trd,\, \mathcal H^2\left(f_{\rmG,K-p},f_{G_\mis,K-p}\right) \ge \widetilde D_2\,\lambda^2_{n,2}\right\}.
\]
Once again, since $\mathcal B_\mis$ is a subset of the class of densities considered Theorem~9 of \citet{ignatiadis2025empirical}, therefore
\[
\log\mathcal R_\mis \lesssim (\log n)^2.
\]
Recall the truncation function $\eta_\trd(\cdot)$ introduced in \eqref{eq:truncation_func}. 
Proceeding as in the proof of Theorem~9 of \citet{ignatiadis2025empirical}, we can conclude that for any $\gamma_n>0$
\begin{align}
    &\mathbb P\left[\mathcal H^2\left(f_{\wh G_\trd,K-p},f_{G_\mis,K-p}\right) \ge \widetilde D_2\,\lambda^2_{n,2},\, \mathcal H_{\trd,\mis}\right]\\
    & \le \mathbb P\!\left[\sup_{j\in[\mathcal R_\mis]}\prod_{i=1}^n\frac{f_{G_j,K-p}(\misV{i})}{f_{G_\mis,K-p}(\misV{i})}\ge e^{-2\gamma_n-n\,\mathfrak R_{n,2}}\right]+\mathbb P\!\left[\prod_{i=1}^n\frac{C_\sharp}{\eta_\trd(\misV{i})}\ge e^{2\gamma_n}\right].
\end{align}
By choosing $\gamma_n \asymp (\log n)^2$, and noting that the marginal density of $\misV{i}$ is $f_{G_\mis,K-p}$, we can conclude that
\begin{align}
\label{eq:entropy_ms_bayesian}
\mathbb P\!\left[\sup_{j\in[\mathcal R_\mis]}\prod_{i=1}^n\frac{f_{G_j,K-p}(\misV{i})}{f_{G_\mis,K-p}(\misV{i})}\ge e^{-2\gamma_n-n\,\mathfrak R_{n,2}}\right] \le \exp(-c_0(\log n)^2),
\end{align}
for some absolute constant $c_0>0$, by proceeding as in the proof of Theorem~9 of \citet{ignatiadis2025empirical}. Using Lemma~S5 of \citet{ignatiadis2025empirical}, we can also conclude that
\begin{align}
\label{eq:truncate_ms_bayesian}
\mathbb P\!\left[\prod_{i=1}^n\frac{C_\sharp}{\eta_\trd(\misV{i})}\ge e^{2\gamma_n}\right] \le \exp(-c_0(\log n)^2).
\end{align}
The conclusion of the lemma follows by combining \eqref{eq:fundamental_eq_ms}, \eqref{eq:entropy_ms_bayesian} and \eqref{eq:truncate_ms_bayesian}.
\end{proof}

\subsection{Convergence of estimated p-values}
To prove Theorem~\ref{thm:final_rate_mis}, we begin by proving the following convergence result of the estimated p-values $\{\wh P^\trd_i\}$ analogous to Theorem~\ref{thm:conv_p_trnd_orc} in the misspecified setting.

\begin{proposition}
    \label{prop:avg_sign_limma_trnd_ms}
    Under the working model specified by \eqref{eq:bayesian_cs_model_2}, \eqref{eq:misspecified-limma-reg}, \eqref{eq:tau_mis_v_mis}, for any $\zeta \in \left(\frac12,1\right)$, we have a constant $\mathfrak C_{\mis}>0$ (possibly depending on $h_1,h_2,K,p,\underline M_\trd,\wb M_\trd,\underline L_\trd,\wb U_\trd$) such that
    \begin{align}
        \frac{1}{n}\sum_{i=1}^n 
\mathbb E\!\left[\left|\misP{i}\wedge \zeta - \wh P_i^{\trd} \wedge \zeta\right|\right] \le \mathfrak C_{\mis} \cdot\mathfrak L_{n,\mis},
    \end{align}
    where 
    \[
    \mathfrak L_{n,\mis}:= \max\Big\{\Delta^{1/2}_n (\log n)^{5/2},\frac{(\log n)^{5/2}}{n^{1/4}}\Delta^{1/2}_n,\frac{(\log n)^2}{n^{1/4}}|\log \Delta_n|^{h_2/4}\Delta^{\frac{1}{2}(1-h_1/2)}_n,\frac{(\log n)^{5/2}}{\sqrt{n}}\Big\}.
    \]
\end{proposition}

\begin{proof}
    To prove the above theorem, recall the Tweedie formula based representation of $\{P^\trd_i\}$ from \eqref{eq:tweedie_trnd}. One can recognize that $\{\misP{i}\}$ by replacing $\xi_0$ by $\xi_\mis$ and $G$ by $G_\mis$ in the formula.
    In particular, using the definition of the function $\PFun(\cdot)$, 
    we can conclude that
    \begin{align}
    \label{eq:tweedie_mis_p_val}
    \misP{i}:=\PFun(Z_i,S^2_i,M_i;G_\mis,\xi_\mis), \quad \mbox{and} \quad \wh P^\trd_i:=\PFun(Z_i,S^2_i,M_i;\wh G_\trd,\wh \xi).
    \end{align}
    One can further rewrite these as a function of only $ O_{\mis,i}:=Z_i/\xi_\mis(M_i), \wh O_{i}:=Z_i/\wh \xi(M_i), \misV{i}=S^2_i/\xi^2_\mis(M_i)$ and $\wh V^2_i=S^2_i/\wh\xi^2(M_i)$ suppressing the dependence on $M_i$. In particular, $\misP{i}:=\PregFun_\mis(O_{\mis,i},\misV{i};G_\mis)$ and $P^\trd_i:=\PregFun_\mis(\wh O_i,\wh V^2_i;\wh G_\trd)$,
    where
\begin{align}
\label{eq:misp_p_val_func}
\PregFun_\mis(x,v^2_\mis;G'):=\mathbb E_{\tau \sim G'}\left[2\Phi(-x/\{\nu\tau\}) \mid V^2_\mis=v^2_\mis\right], \quad \mbox{for $x \ge 0$ and $V^2_\mis \sim \frac{\tau^2\chi^2_{K-p}}{(K-p)}$.}
\end{align}
    However, it is important to observe that the formulas for the p-values remain the same in both the frameworks described in Sections~\ref{sec:reg_ltrd_ext} and Sections~\ref{sec:reg_ltrd_int}, only the distributional specification changes. 
    We modify the definition of the set $\mathcal B_\trd$ as follows
    \begin{align}
\mathcal B_\trd &= \Bigg\{|M_i| \le \mathrm W_n,\misV{i} \in [\wt A_n,\wt B_n],\,~\text{for all $i \in [n]$},~\|\wh \xi-\xi_\mis\|_{\infty} \le \Delta_n,\\
&~~~~~~~~~~~~~~~~~~~~~~~~~~~~\mathcal H^2\left(f_{\wh G_\trd,K-p},f_{G_\mis,K-p}\right)
    \le \widetilde D_2~\lambda^{2}_{n,2}\Bigg\},
\end{align}
where the quantities $\wt A_n, \wt B_n,$ and $\lambda_{n,2}$ are defined in the statement of Lemma~\ref{lem:hell_limm_trnd_ms}, Assumption~\ref{assu:limma_trnd_var}, and the proof of Theorem~\ref{thm:conv_p_trnd_orc}. For any $z \in \mathbb R$, let us define the terms corresponding to \eqref{eq:def_n_trd_all} for this working model. However, to avoid making the paper notation heavy, we shall recycle the notations from \eqref{eq:def_n_trd_all}. In this modified setting, the corresponding quantities are re-defined as follows:
\begin{align}
\label{eq:def_n_trd_all_1}
\widehat{\calD}^{\trd}_i(S_i^2,\widehat{\xi}(M_i))
&:=f_{\wh G_\trd,K-p}(S^2_i/\wh \xi^2(M_i)), \\
\widehat{\calD}^{\trd}_i(S_i^2,\xi_\mis(M_i))
&:=f_{\wh G_\trd,K-p}(S^2_i/\xi^2_\mis(M_i)), \\
\calD^{\trd}_i(S_i^2,\xi_\mis(M_i))&:=
f_{G_\mis,K-p}\!\left(S^2_i/\xi^2_\mis(M_i)\right),\\
N^{\trd}_i(z,S_i^2,M_i,\xi_\mis,G_\mis)
&:=
\PFun(z,S^2_i,M_i;G_\mis,\xi_\mis)\,
\calD^{\trd}_i(S_i^2,\xi_\mis(M_i)),\\
N^{\trd}_i(z,S_i^2,M_i,\wh\xi,\wh G_\trd)
&:=
\PFun(z,S^2_i,M_i;\wh G_\trd,\wh \xi)\,
\widehat{\calD}^{\trd}_i(S_i^2,\widehat{\xi}(M_i)).
\end{align}
Using Lemma~\ref{lem:hell_limm_trnd_ms}, one can show that even under the new working model in this framework
\begin{align}
    \frac{1}{n}\sum_{i=1}^n 
\mathbb E\!\left[\left|P^\trd_{\mis,i} \wedge \zeta - \wh P_i^{\trd} \wedge \zeta\right|\right] \lesssim \frac{1}{n^2}+\mathfrak B^\trd_1+\mathfrak B^\trd_2,
\end{align}
where $\mathfrak B^\trd_1$ and $\mathfrak B^\trd_2$ are defined analogous to \eqref{eq:def_frk_b} using the terms defined in \eqref{eq:def_n_trd_all_1}.

First, focus on bounding $\mathfrak B^\trd_1$ using the techniques developed in the proof of Theorem~\ref{thm:conv_p_trnd_orc}. However, in this proof we shall replace $\xi_0$ by $\xi_\mis$ in all definitions involving the former function. Observe that the arguments in \eqref{eq:b_1_trnd_p_val} do not replicate in this framework. Nevertheless, under the assumptions of the framework, we have the following
\begin{align}
    \mathcal B_\trd \subseteq \cap_{i=1}^n\{\misV{i} \in [0,\wt B_n]\},
\end{align}
where $\wt B_n$ is defined in Lemma~\ref{lem:prop_chi_sq}. This implies $\mathds 1(\mathcal B_\trd) \le \mathds 1\{\misV{i} \in [0,\wt B_n]\}$. Now, observe that the unconditional density of $\misV{i}$ under this framework is exactly $\mathcal D^\trd_i(S^2_i,\xi_\mis):=f_{G_\mis,K-p}(\misV{i})$. Since $\wt B_n \asymp \log n$, therefore if
\begin{align}
\label{eq:n_decomp_mis}
\sup_{|z|\ge \underline z}\left|N^{\trd}_i(z,S_i^2,M_i,\xi_\mis,G_\mis)-N^{\trd}_i(z,S_i^2,M_i,\wh\xi,\wh G_\trd)\right| \le \zeta_{n,\mis}, 
\end{align}
for some sequence $\zeta_{n,\mis}$ (independent of $M_1,\ldots,M_n$), then $\mathfrak B^\trd_1 \lesssim \zeta_{n,\mis} (\log n)$. We shall prove that $\zeta_{n,\mis}$ exists.
Recall the decomposition of the left hand side of \eqref{eq:n_decomp_mis} from \eqref{eq:decom_n_trd_1} as
\begin{align}
\label{eq:decom_n_trd_1_ms_hello}
  &\left|N^{\trd}_i(z,S_i^2,M_i,\xi_\mis,G_\mis)-N^{\trd}_i(z,S_i^2,M_i,\wh\xi,\wh G_\trd)\right|  \\
  & \le \left|N^{\trd}_i(z,S_i^2,M_i,\xi_\mis,G_\mis)-N^{\trd}_i(z,S_i^2,M_i,\xi_\mis,\wh G_\trd)\right|\\
  &~~~~~~~+\left|N^{\trd}_i(z,S_i^2,M_i,\xi_\mis,\wh G_\trd)-N^{\trd}_i(z,S_i^2,M_i,\wh\xi,\wh G_\trd)\right|.
\end{align}
Observe that bounding the first term in that decomposition is a Cauchy-Schwarz argument that results in
\[
\mathcal H(f_{\wh G_\trd,K-p},f_{G_\mis,K-p})|\log \mathcal H(f_{\wh G_\trd,K-p},f_{G_\mis,K-p})|^{1/2}.
\]
Under the modified working model, using Lemma~\ref{lem:hell_limm_trnd_ms} 
this implies
\begin{align}
    &\sup_{|z| \ge \underline z}\left|N^{\trd}_i(z,S_i^2,M_i,\xi_\mis,G_\mis)-N^{\trd}_i(z,S_i^2,M_i,\xi_\mis,\wh G_\trd)\right|\\
    &\lesssim \lambda_{n,2} \sqrt{|\log \lambda_{n,2}|}\\
    &\lesssim \max\left\{\Delta^{1/2}_n (\log n)^{3/2},\frac{(\log n)^{3/2}}{n^{1/4}}\Delta^{1/2}_n,\frac{(\log n)}{n^{1/4}}|\log \Delta_n|^{h_2/4}\Delta^{\frac{1}{2}(1-h_1/2)}_n,\frac{(\log n)^{3/2}}{\sqrt{n}}\right\},
\end{align}
where $\lambda_{n,2}$ is defined in \eqref{eq:lamb_n_2_mis}. Here we have used $\log |\lambda_{n,2}| \lesssim (\log n)$. Next, observe that if we restrict ourselves to $\mathcal B_\trd$, the gradient in \eqref{eq:deriv_wh_g} continues to be uniformly bounded by a constant factor of $\log n$ and by Assumption~\ref{assum:missp_trend_estimation}, $\|\wh \xi-\xi_\mis\|_\infty$ is uniformly bounded by $\Delta_n$. This implies, even under the working model \eqref{eq:bayesian_cs_model_2}, \eqref{eq:misspecified-limma-reg}, and \eqref{eq:tau_mis_v_mis}, the relation in \eqref{eq:b_1_trd_new} holds with $\xi_0$ replaced with $\xi_\mis$. Consequently
\begin{align}
\label{eq:b_1_trd_mis}
    \mathfrak B^\trd_1 \lesssim \max\Big\{\Delta^{1/2}_n (\log n)^{5/2},\frac{(\log n)^{5/2}}{n^{1/4}}\Delta^{1/2}_n,\frac{(\log n)^2}{n^{1/4}}|\log \Delta_n|^{h_2/4}\Delta^{\frac{1}{2}(1-h_1/2)}_n,\frac{(\log n)^{5/2}}{\sqrt{n}}\Big\}.
\end{align}
For bounding $\mathfrak B^\trd_2$, observe that
\begin{align}
    \mathfrak B^\trd_2 & \lesssim \frac{1}{n}\sum_{i=1}^n\mathbb E\left[\left|\frac{\widehat{\calD}^{\trd}_i(S_i^2,\widehat{\xi}(M_i))-\widehat{\calD}^{\trd}_i(S_i^2,\xi_\mis(M_i))}{\calD^\trd_{i,\star}(S^2_i,M_i)}\right|\mathds 1(\mathcal B_\trd)\right]\\
    &+\frac{1}{n}\sum_{i=1}^n\mathbb E\left[\left|\frac{\calD^{\trd}_i(S_i^2,\xi_\mis(M_i))-\widehat{\calD}^{\trd}_i(S_i^2,\xi_\mis(M_i))}{\calD^\trd_{i,\star}(S^2_i,M_i)}\right|\mathds 1(\mathcal B_\trd)\right].
\end{align}
Proceeding as in \eqref{eq:decomp_corrected_thm_7}, we get
\begin{align}
    \mathfrak B^\trd_2 & = \frac{1}{n}\sum_{i=1}^n\mathbb E\left[\left|\frac{\widehat{\calD}^{\trd}_i(S_i^2,\widehat{\xi}(M_i))-\calD^{\trd}_i(S_i^2,\xi_\mis(M_i))}{\calD^\trd_{i,\star}(S^2_i,M_i)}\right|\mathds 1(\mathcal B_\trd)\right]\\
    & \lesssim \frac{1}{n}\sum_{i=1}^n\mathbb E\left[\left|\frac{\calD^{\trd}_i(S_i^2,\xi_\mis(M_i))-\widehat{\calD}^{\trd}_i(S_i^2,\xi_\mis(M_i))}{\wt \calD^\trd_{i,\star}(S^2_i,M_i)}\right|\mathds 1(\mathcal B_\trd)\right]\\
    &\qquad+\frac{1}{n}\sum_{i=1}^n\mathbb E\left[\left|\frac{\wh\calD^\trd_{i}(S^2_i,\xi_\mis(M_i))- \wh\calD^\trd_{i}(S^2_i,\wh \xi(M_i))}{ \calD^\trd_{i}(S^2_i,\xi_\mis(M_i))}\right|\mathds 1(\mathcal B_\trd)\right],
\end{align}
where $\wt \calD^\trd_{i,\star}(S^2_i,M_i)$ is defined as 
\[
\wt \calD^\trd_{i,\star}(S^2_i,M_i):=\frac{\calD^{\trd}_i(S_i^2,\xi_\mis(M_i))+\widehat{\calD}^{\trd}_i(S_i^2,\xi_\mis(M_i))}{2}.
\]
Observing that the marginal density of $\misV{i}$ is $\calD^{\trd}_i(S_i^2,\xi_\mis(M_i))$ (a function of only $\misV{i}$), one can retrace the arguments from the proof of Theorem~\ref{thm:conv_p_trnd_orc} to show that
\begin{align}
&\frac{1}{n}\sum_{i=1}^n\mathbb E\left[\left|\frac{\calD^{\trd}_i(S_i^2,\xi_\mis(M_i))-\widehat{\calD}^{\trd}_i(S_i^2,\xi_\mis(M_i))}{\wt \calD^\trd_{i,\star}(S^2_i,M_i)}\right|\mathds 1(\mathcal B_\trd)\right]\\
& \lesssim \Delta^{1/2}_n(\log n)+\frac{(\log n)}{n^{1/4}}\Delta^{1/2}_n+\frac{(\log n)^{1/2}}{n^{1/4}}|\log \Delta_n|^{h_2/4}\Delta^{\frac{1}{2}(1-h_1/2)}_n+\frac{(\log n)}{\sqrt{n}}.
\end{align}
Furthermore, on $\mathcal B_\trd$, one can uniformly bound
\[
\left|\wh\calD^\trd_{i}(S^2_i,\xi_\mis(M_i))- \wh\calD^\trd_{i}(S^2_i,\wh \xi(M_i))\right| \lesssim (\log n)\Delta_n.
\]
Consequently, we again have
\[
\frac{1}{n}\sum_{i=1}^n\mathbb E\left[\left|\frac{\wh\calD^\trd_{i}(S^2_i,\xi_\mis(M_i))- \wh\calD^\trd_{i}(S^2_i,\wh \xi(M_i))}{ \calD^\trd_{i}(S^2_i,\xi_\mis(M_i))}\right|\mathds 1(\mathcal B_\trd)\right] \lesssim \Delta_n(\log n)^2.
\]
Combining the two foregoing inequalities, we can conclude that
\begin{align}
    \label{eq:mis_b_2_trd}
    \mathfrak B^\trd_2 &\lesssim \Delta^{1/2}_n(\log n)^2+\frac{(\log n)}{n^{1/4}}\Delta^{1/2}_n+\frac{(\log n)^{1/2}}{n^{1/4}}|\log \Delta_n|^{h_2/4}\Delta^{\frac{1}{2}(1-h_1/2)}_n+\frac{(\log n)}{\sqrt{n}}.
\end{align}
Finally, combining \eqref{eq:b_1_trd_mis} and \eqref{eq:mis_b_2_trd}, the lemma follows.
\end{proof}

Next, we consider the following lemma on approximate uniformity of $\{\wh P^\trd_i: \theta_i=0\}$ under \eqref{eq:bayesian_cs_model_2}, \eqref{eq:misspecified-limma-reg}, and \eqref{eq:tau_mis_v_mis}.
\begin{lemm}
    \label{lem:approx_unif_mis_strong}
    Fix $\zeta \in (1/2,1)$ and assume that $n_0/n \rightarrow \pi_0\in (0,1)$, as $n \rightarrow \infty$, where $n_0=|\mathcal H_0|$. Then, under the working model specified by \eqref{eq:bayesian_cs_model_2}, \eqref{eq:misspecified-limma-reg}, and \eqref{eq:tau_mis_v_mis}, for any $\varepsilon>0$
    \begin{align}
        \limsup_{n \rightarrow \infty}\mathbb P\left[\sup_{t \in [0,1]}\Bigg|\frac{1}{n_0}\sum_{i \in \mathcal H_0}\{\mathbb P[\wh P^\trd_i \le t\mid \misV{1},\ldots,\misV{n}]-t\}\Bigg| \ge \varepsilon\right]=0.
    \end{align}
    Furthermore, we also have
    \[
\limsup_{n \rightarrow \infty}\sup_{t \in [0,1]}\mathbb E\left[\Bigg|\frac{1}{n_0}\sum_{i \in \mathcal H_0}\{\mathbb P[\wh P^\trd_i \le t\mid \misV{1},\ldots,\misV{n}]-t\}\Bigg|\right] =0.
\]
\end{lemm}
\begin{proof}
    Let us fix $i \in \mathcal H_0$. Choose $\zeta'=(1+\zeta)/2 \in (0,1/2)$ and $\delta\in (0,\zeta')$. Then by Lemma S11 of \cite{ignatiadis2025empirical}, we have that for any $t \le \zeta'-\delta$
    \begin{align}
        \mathds 1\{\wh P^\trd_i \le t\} \le \mathds 1\{\misP{i} \le t+\delta\}+\frac{1}{\delta}\left|\wh P^\trd_i \wedge \zeta'-\misP{i} \wedge \zeta'\right|
    \end{align}
    This implies
    \begin{align}
        \mathds 1\{\wh P^\trd_i \le t\}-t &\le \mathds 1\{\misP{i} \le t+\delta\}-(t+\delta)+\delta+\frac{1}{\delta}\left|\wh P^\trd_i \wedge \zeta'-\misP{i} \wedge \zeta'\right|,
    \end{align}
    which in turn implies
    \begin{align}
        \frac{1}{n_0}\sum_{i \in \mathcal H_0}\{\mathds 1\{\wh P^\trd_i \le t\}-t\} & \le \frac{1}{n_0}\sum_{i \in \mathcal H_0}\left\{\mathds 1\{\misP{i} \le t+\delta\}-(t+\delta)\right\}+\delta+\frac{1}{n_0\delta}\sum_{i \in \mathcal H_0}\left|\wh P^\trd_i \wedge \zeta'-\misP{i} \wedge \zeta'\right|\\
        &  \le \frac{1}{n_0}\sum_{i \in \mathcal H_0}\left\{\mathds 1\{\misP{i} \le t+\delta\}-\mathbb P_G\left[\misP{i} \le t+\delta \mid \misV{i}\right]\right\}+\delta\\
        &~~~~~~~~~~+\frac{1}{n\delta}\sum_{i \in \mathcal H_0}\left|\wh P^\trd_i \wedge \zeta'-\misP{i} \wedge \zeta'\right|,
    \end{align}
    where we use 
    \[
    \mathbb P_G\left[\misP{i} \le t+\delta \mid \misV{i}\right] = t+\delta, \quad \mbox{which follows from Lemma~\ref{lem:val_p_mis}.}
    \]
    Taking the conditional expectation, we get
    \begin{align}
        &\left|\frac{1}{n_0}\sum_{i \in \mathcal H_0}\{\mathbb P[\wh P^\trd_i \le t\mid \misV{1},\ldots,\misV{n}]-t\}\right|\\
        & \le \mathbb E\left[\Bigg|\frac{1}{n_0}\sum_{i \in \mathcal H_0}\left\{\mathds 1\{\misP{i} \le t+\delta\}-\mathbb P_G\left[\misP{i} \le t+\delta \mid \misV{i}\right]\right\}\Bigg|\,\Big|\, \misV{1},\ldots,\misV{n}\right]+\delta\\
        &~~~~~~+\frac{1}{n_0\delta}\sum_{i \in \mathcal H_0}\mathbb E\left[\left|\wh P^\trd_i \wedge \zeta'-\misP{i} \wedge \zeta'\right|\mid \misV{1},\ldots,\misV{n}\right]\\
         & \le \mathbb E\left[\sup_{t \in [0,1]}\Bigg|\frac{1}{n_0}\sum_{i \in \mathcal H_0}\left\{\mathds 1\{\misP{i} \le t\}-\mathbb P_G\left[\misP{i} \le t \mid \misV{i}\right]\right\}\Bigg|\;\Big|\, \misV{1},\ldots,\misV{n}\right]+\delta\\
        &~~~~~~+\frac{1}{n_0\delta}\sum_{i \in \mathcal H_0}\mathbb E\left[\left|\wh P^\trd_i \wedge \zeta'-\misP{i} \wedge \zeta'\right|\mid \misV{1},\ldots,\misV{n}\right].
    \end{align}
This also implies
\begin{align}
    &\sup_{t \in [0,1]}\Bigg|\frac{1}{n_0}\sum_{i \in \mathcal H_0}\{\mathbb P[\wh P^\trd_i \le t\mid \misV{1},\ldots,\misV{n}]-t\}\Bigg|\\
    & \le \mathbb E\left[\sup_{t \in [0,1]}\Bigg|\frac{1}{n_0}\sum_{i \in \mathcal H_0}\left\{\mathds 1\{\misP{i} \le t\}-\mathbb P_G\left[\misP{i} \le t\mid \misV{i}\right]\right\}\Bigg|\;\Big|\, \misV{1},\ldots,\misV{n}\right]+\delta\\
        &~~~~~~+\frac{n}{n_0}\cdot\frac{1}{n\delta}\sum_{i \in \mathcal H_0}\mathbb E\left[\left|\wh P^\trd_i \wedge \zeta'-\misP{i} \wedge \zeta'\right|\mid \misV{1},\ldots,\misV{n}\right].
\end{align}
Recall $\mathfrak L_{n,\mis}$ from Proposition~\ref{prop:avg_sign_limma_trnd_ms}. By Lemma~\ref{lem:bdkw} and $n_0/n \rightarrow \pi_0 \in (0,1)$, we get that 
\[
\mathbb E\left[\sup_{t \in [0,1]}\Bigg|\frac{1}{n_0}\sum_{i \in \mathcal H_0}\left\{\mathds 1\{\misP{i} \le t\}-\mathbb P_G\left[\misP{i} \le t\mid \misV{i}\right]\right\}\Bigg|\;\Big|\, \misV{1},\ldots,\misV{n}\right] \le \frac{\mathfrak L^{1/4}_{n,\mis}}{3},
\]
almost surely.
This implies, as $n \rightarrow 0$, we have 
\[
\limsup_{n \rightarrow 0} \mathbb E\left[\sup_{t \in [0,1]}\Bigg|\frac{1}{n_0}\sum_{i \in \mathcal H_0}\left\{\mathds 1\{\misP{i} \le t\}-\mathbb P_G\left[\misP{i} \le t\mid \misV{i}\right]\right\}\Bigg|\;\Big|\, \misV{1},\ldots,\misV{n}\right]=0,\; \mbox{a.s.}
\]
Next, take $\delta=\mathfrak L^{1/4}_{n,\mis}/3$ and define
\[
A_{\mis,n}:=\left\{\frac{1}{n}\sum_{i \in \mathcal H_0}\mathbb E\left[\left|\wh P^\trd_i \wedge \zeta'-\misP{i} \wedge \zeta'\right|\mid \misV{1},\ldots,\misV{n}\right]<\frac{\mathfrak L^{1/2}_{n,\mis}}{3}\right\}.
\]
Using Proposition~\ref{prop:avg_sign_limma_trnd_ms}, $n_0/n \rightarrow \pi_0 \in (0,1)$ and Markov's inequality, we get
\begin{align}
    \mathbb P(A^c_{\mis,n}) \lesssim_{\pi_0} \mathfrak L^{1/2}_{n,\mis}.
\end{align}
Since $\mathfrak L_{n,\mis} \rightarrow 0$, this further implies
\begin{align}
    \frac{n}{n_0}\cdot\frac{1}{n\delta}\sum_{i \in \mathcal H_0}\mathbb E\left[\left|\wh P^\trd_i \wedge \zeta'-\misP{i} \wedge \zeta'\right|\mid \misV{1},\ldots,\misV{n}\right],
\end{align}
converges to 0 in probability \citep{billingsley1995probability}.
Therefore, by Slutsky's theorem for any fixed $\varepsilon>0$, we can conclude that
\begin{align}
    &\limsup_{n \rightarrow \infty}\mathbb P\left[\sup_{t \in [0,1]}\Bigg|\frac{1}{n_0}\sum_{i \in \mathcal H_0}\{\mathbb P[\wh P^\trd_i \le t\mid \misV{1},\ldots,\misV{n}]-t\}\Bigg| \ge \varepsilon\right]=0.
\end{align}

Since $\sup_{t \in [0,1]}|n^{-1}_0\sum_{i \in \mathcal H_0}\{\mathbb P[\wh P^\trd_i \le t\mid \misV{1},\ldots,\misV{n}]-t\}|$ is bounded, therefore the foregoing relation automatically implies
\[
\limsup_{n \rightarrow \infty}\mathbb E\left[\sup_{t \in [0,1]}\Bigg|\frac{1}{n_0}\sum_{i \in \mathcal H_0}\{\mathbb P[\wh P^\trd_i \le t\mid \misV{1},\ldots,\misV{n}]-t\}\Bigg|\right] =0,
\]
which further implies using the Jensen's inequality that 
\[
\limsup_{n \rightarrow \infty}\sup_{t \in [0,1]}\mathbb E\left[\Bigg|\frac{1}{n_0}\sum_{i \in \mathcal H_0}\{\mathbb P[\wh P^\trd_i \le t\mid \misV{1},\ldots,\misV{n}]-t\}\Bigg|\right] =0.
\]
This concludes the lemma.
\end{proof}
We also have the following (slightly restrictive) non-asymptotic version of the foregoing lemma, which is crucial to prove the FDR control.
\begin{lemm}
    \label{lem:approx_unif_mis}
    Fix $\zeta \in (1/2,1)$ and assume that $n_0/n \rightarrow \pi_0\in (0,1)$, as $n \rightarrow \infty$, where $n_0=|\mathcal H_0|$. Then, under the working model specified by \eqref{eq:bayesian_cs_model_2}, \eqref{eq:misspecified-limma-reg}, and \eqref{eq:tau_mis_v_mis}, there exists a constant $\mathfrak D_{3,\mis}>0$ such that 
    \begin{align}
        &\mathbb E\left[\sup_{t \in [0,\zeta]}\Bigg|\left\{\frac{1}{n}\sum_{i \in \mathcal H_0}\mathds 1\{\wh P^\trd_i \le t\}-\pi_0t\right\}\Bigg|\right]\\
        &\le \mathfrak D_{3,\mis}\cdot\max\left\{\Delta^{1/4}_n(\log n)^{5/4},\frac{(\log n)^{5/4}}{n^{1/8}}\Delta^{1/4}_n,\frac{(\log n)}{n^{1/8}}\cdot |\log \Delta_n|^{h_2/8}\Delta^{\frac{1}{4}(1-h_1/2)}_n,\frac{(\log n)^{5/4}}{n^{1/4}}\right\}.
    \end{align}
\end{lemm}

\begin{proof}
First, observe that under the working model, using Lemmas~\ref{lem:val_p_mis} and~\ref{lem:bdkw}, we can conclude that for any $t \in [0,1]$,
\[
\mathbb E\left[\sup_{t \in [0,1]}\Bigg|\frac{1}{n_0} \sum_{i \in \mathcal H_0}\left(\mathds 1\{\misP{i} \le t\}-t\right)\Bigg|\right] \lesssim \frac{1}{\sqrt{n_0}}.
\]
Since $n_0/n \rightarrow \pi_0 \in (0,1)$, therefore we can also conclude that
\[
\mathbb E\left[\sup_{t \in [0,1]}\Bigg|\frac{1}{n} \sum_{i \in \mathcal H_0}\mathds 1\{\misP{i} \le t\}-\pi_0t\Bigg|\right] \lesssim \frac{1}{\sqrt{n}}.
\]
    Using Lemma~\ref{lem:val_p_mis}, we can show by retracing the steps in the proof of Proposition~15* of \cite{ignatiadis2025empirical} that there exists a $n_\mis$ such that for all $n \ge n_\mis$, any fixed $\zeta \in (1/2,1)$ and $\delta \in [0,(1+\zeta)/2]$
    \begin{align}
    \label{eq:bad_bound_p_val_mis}
        \mathbb E\left[\sup_{t \in [0,\zeta]}\Bigg|\frac{1}{n}\sum_{i \in \mathcal H_0}\mathds 1\{\wh P^\trd_i \le t\}-\pi_0t\Bigg|\right]\lesssim \frac{1}{\sqrt{n}}+\delta+\frac{1}{n\delta}\sum_{i=1}^n 
\mathbb E\!\left[\left|\misP{i}\wedge \zeta' - \wh P_i^{\trd} \wedge \zeta'\right|\right].
    \end{align}
    By Proposition~\ref{prop:avg_sign_limma_trnd_ms}, we can take $\delta \asymp \mathfrak L^{1/2}_{n,\mis}$ (with the constant appropriately chosen to ensure $\delta  \in [0,(1+\zeta)/2]$) to get
    \begin{align}
        &\mathbb E\left[\sup_{t \in [0,\zeta]}\Bigg|\frac{1}{n}\sum_{i \in \mathcal H_0}\mathds 1\{\wh P^\trd_i \le t\}-\pi_0t\Bigg|\right]\\
        &\lesssim \frac{1}{\sqrt{n}}+\delta+\frac{1}{n\delta}\sum_{i=1}^n 
\mathbb E\!\left[\left|\misP{i}\wedge \zeta' - \wh P_i^{\trd} \wedge \zeta'\right|\right]\\
        &\le \max\left\{\Delta^{1/4}_n(\log n)^{5/4},\frac{(\log n)^{5/4}}{n^{1/8}}\Delta^{1/4}_n,\frac{(\log n)}{n^{1/8}}\cdot |\log \Delta_n|^{h_2/8}\Delta^{\frac{1}{4}(1-h_1/2)}_n,\frac{(\log n)^{5/4}}{n^{1/4}}\right\}.
    \end{align}
    This concludes the lemma.
\end{proof}

Next, we consider the proof of Theorem~\ref{thm:final_rate_mis}. However, a major part of the proof follows techniques to be described later in the proof of Theorem~\ref{thm:final_rate}. In the interest of brevity, we shall not repeat them in this proof and request the readers to refer to the proof of Theorem~\ref{thm:final_rate} for the details of the steps.

\begin{proof}[Proof of Theorem~\ref{thm:final_rate_mis}]
    Let $\zeta =\max\{3/4,(1+\alpha)/2\}$ and recall the Benjamini-Hochberg procedure employing $\{\wh P^\trd_i\}$.
    For any $t \in (0,1)$, define
    \[
    R^\mis_n(t):=\sum_{i=1}^n\mathds 1\{\wh P^\trd_i \le t\}, \qquad \mbox{and} \qquad V^\mis_n(t):=\sum_{i \in \mathcal H_0}\mathds 1\{\wh P^\trd_i \le t\}.
    \]
    One can use the self-characterization property of the Benjamini-Hochberg procedure to derive that the BH threshold $\wh t_\mis$ in this framework can be re-expressed as follows:
    \begin{align}
        \label{eq:trd_bh_ms}
        \wh t_\mis:=\sup\left\{t \in [0,1]:\frac{R^\mis_n(t)}{n} \ge \frac{t}{\alpha}\right\},
    \end{align}
    The BH procedure rejects all hypotheses $H_i:\theta_i=0$ such that $\wh P^\trd_i \le \wh t_\mis$. Since the oracle p-values $\{\misP{i}\}$ are critically dense at $\alpha$, proceeding as in the proof of Theorem~\ref{thm:final_rate}, one can show that if we take 
    \[
    \alpha_{0} \in \left(\liminf_{n\to\infty}\inf_{t \in [t_0,t_1]}\frac{1}{nt}\sum_{i=1}^n\mathbb P\!\left[\misP{i} \le t\right],\alpha\right),
    \] 
    and $\kappa_n=5\alpha/4(\alpha-\alpha_0)$,
    then 
    \begin{align}
    \label{eq:mis_nn_rej}
    \limsup_{n \rightarrow \infty}\mathbb P\left(R^\mis_n(\wh t_\mis) < n\kappa_n\right)=0.
    \end{align}
    Furthermore, by the definition of \emph{false discovery proportion}, we have
    \begin{align}
        \frac{V^\mis_n(\wh t_\mis)}{R^\mis_n(\wh t_\mis) \vee 1}& = \frac{n\pi_0\wh t_\mis}{R^\mis_n(\wh t_\mis)\vee 1} +\frac{n}{R^\mis_n(\wh t_\mis)\vee 1}\left\{\frac{1}{n}\sum_{i \in \mathcal H_0}\mathds 1\{\wh P^\trd_i \le \wh t_\mis\}-\pi_0\wh t_\mis\right\}.
    \end{align}
    Define
    \(
    \mathcal F_\mis := \{R^\mis_n(\wh t_\mis) \ge n\kappa_n \}.
    \)
    By \eqref{eq:mis_nn_rej}, we have $\smash{\limsup_{n \rightarrow \infty}\mathbb P(\mathcal F^c_\mis)=0.}$  The definition of $\wh t_\mis$ also implies
    \(
    \smash{(n/R^\mis_n(\wh t_\mis)) \le (\alpha/\wh t_\mis).}
    \) 
    In particular $\wh t_\mis \le \alpha \le \zeta$.
    By the definition of the \emph{false discovery proportion}, we have
    \begin{align}
    \frac{V^\mis_n(\wh t_\mis)}{R^\mis_n(\wh t_\mis) \vee 1}& \le \mathds 1(\mathcal F^c_\mis)+\frac{V^\mis_n(\wh t_\mis)}{R^\mis_n(\wh t_\mis) \vee 1}\cdot \mathds 1(\mathcal F_\mis)\\
    & \le \mathds 1(\mathcal F^c_\mis) + \frac{n\pi_0\wh t_\mis}{R^\mis_n(\wh t_\mis)}\cdot \mathds 1(\mathcal F_\mis)+\frac{n\mathds 1(\mathcal F_\mis)}{R^\mis_n(\wh t_\mis)}\left\{\frac{1}{n}\sum_{i \in \mathcal H_0}\mathds 1\{\wh P^\trd_i \le \wh t_\mis\}-\pi_0\wh t_\mis\right\}\\
    & \le \mathds 1(\mathcal F^c_\mis) + \frac{n\pi_0\wh t_\mis}{R^\mis_n(\wh t_\mis)}\cdot \mathds 1(\mathcal F_\mis)+\frac{n\mathds 1(\mathcal F_\mis)}{R^\mis_n(\wh t_\mis)}\sup_{t \in [0,\alpha]}\left\{\frac{1}{n}\sum_{i \in \mathcal H_0}\mathds 1\{\wh P^\trd_i \le t\}-\pi_0t\right\}\\
    &\le \mathds 1(\mathcal F^c_\mis) + \frac{n\pi_0\wh t_\mis}{R^\mis_n(\wh t_\mis)}\cdot \mathds 1(\mathcal F_\mis)+\frac{1}{\kappa_n}\sup_{t \in [0,\alpha]}\left\{\frac{1}{n}\sum_{i \in \mathcal H_0}\mathds 1\{\wh P^\trd_i \le t\}-\pi_0t\right\}\\
    &\le \mathds 1(\mathcal F^c_\mis) + \pi_0\alpha+\frac{1}{\kappa_n}\sup_{t \in [0,\zeta]}\left|\frac{1}{n}\sum_{i \in \mathcal H_0}\mathds 1\{\wh P^\trd_i \le t\}-\pi_0t\right|.
    \end{align}
    Taking the expectation on both sides 
    \begin{align}
        \fdr^\mis_n \le \mathbb P\left(R^\mis_n(\wh t_\mis) < n\kappa_n\right) + \pi_0\alpha+\frac{1}{\kappa_n}\mathbb E\left[\sup_{t \in [0,\zeta]}\Bigg|\frac{1}{n}\sum_{i \in \mathcal H_0}\mathds 1\left(\wh P^\trd_i \le t\right)-\pi_0t\Bigg|\right].
    \end{align}
    Using Lemma~\ref{lem:approx_unif_mis} and \eqref{eq:mis_nn_rej}, we conclude that
    \[
    \limsup_{n \rightarrow \infty}\fdr^\mis_n \le \pi_0\alpha.
    \]
\end{proof}

\section{Proofs of Results from Section \ref{sec:joint_limma_trend}}
In this section, we shall consistently assume Assumption~\ref{assu:compact_2d}.
\subsection{Proof of Theorem \ref{thm:tweedie_2_d}}
Recall the definition
\[
\PjtFun(z,s^2,a; H) := \mathbb{E}_{H}\left[2\left(1 - \Phi\left(\frac{|z|}{\nu \sigma}\right)\right) \bigg| S^2 = s^2, A=a\right],
\]
from \eqref{eq:validity_seq_3}.
Under Assumption~\ref{assu:design}, the pair $(S^2_i,A_i)$ satisfies \eqref{eq:def_z_i_dis_2d} and hence the conditional expectation in the definition of the p-value can be expressed as
\[
\PjtFun(z,s^2,a; H) = \frac{2}{f_{H,K-p}(s^2,a)} \int_{\mathbb{R} \times \mathbb{R}_{\ge 0}} \left\{1 - \Phi\left(\frac{|z|}{\nu\sigma}\right)\right\} p_{K-p}(s^2,a \mid \mu, \sigma^2) \, H(\dd\mu, \dd\sigma^2),
\]
where $p_{K-p}(s^2, a \mid \mu, \sigma^2)$ is the joint likelihood of $(S^2,A)$ given $(\mu, \sigma^2)$ defined in under \eqref{eq:def_z_i_dis_2d} and $H$ is the prior on $(\mu,\sigma^2)$ defined in \eqref{eq:limma_trend_2d_prior}.
The marginal density $f_{H,K-p}(s^2,a)$ of $(S^2,A)$ is defined in \eqref{eq:2_d_marginal}.
Define the constant
\[
\mathscr C_{K,p}(s^2) := \sqrt{\frac{2}{\pi}} \cdot \frac{(K-p)^{(K-p)/2} (s^2)^{\frac{K-p}{2} - 1}}{2^{(K-p)/2} \, \Gamma\left(\frac{K-p}{2}\right)}.
\]
Without loss of generality, assume $z \ge 0$. Using the integral representation of $1 - \Phi(x)$, we write:
\begin{align*}
& \PjtFun(z,s^2,a; H) \\
&= \frac{2}{f_{H,K-p}(s^2, a)} \int_{\mathbb{R} \times \mathbb{R}_{\ge 0}} \int_z^\infty \frac{1}{\sqrt{2\pi} \nu \sigma} \exp\left(-\frac{u^2}{2\nu^2 \sigma^2}\right) \, p(s^2, a \mid \mu, \sigma^2) \, \dd u \, H(\dd\mu, \dd\sigma^2) \\
&= \frac{\mathscr C_{K,p}(s^2)}{f_{H,K-p}(s^2, a)} \int_{\mathbb{R} \times \mathbb{R}_{\ge 0}} \int_{z/\nu}^\infty (\sigma^2)^{-(K-p+1)/2} \, \exp\left( -\frac{(K-p) s^2 + u^2}{2\sigma^2} \right) \\
& \hskip 12em \times \frac{\sqrt{K}}{\sigma} \, \phi\left(\frac{\sqrt{K}(a - \mu)}{\sigma}\right) \, \dd u \, H(\dd\mu, \dd\sigma^2).
\end{align*}
Using \eqref{eq:2_d_marginal}, the above display yields
\begin{align*}
\PjtFun(z,s^2,a; H) &= \frac{\mathscr C_{K,p}(s^2)}{f_{H,K-p}(s^2, a)} \int_{z/\nu}^\infty \Gamma\left(\frac{K - p + 1}{2}\right)\left(\frac{2}{K - p + 1}\right)^{(K - p + 1)/2} \\
& \quad \times \left(\frac{(K - p) s^2 + u^2}{K - p + 1}\right)^{-(K - p - 1)/2} f_{H,K-p+1}\left(\frac{(K - p) s^2 + u^2}{K - p + 1}, a\right) \dd u.
\end{align*}
Finally, by a change of variables $t^2 = \frac{(K - p) s^2 + u^2}{K - p + 1}$, one can match this to the integral representation stated in the theorem, by retracing the steps used to prove Proposition 11 of~\cite{ignatiadis2025empirical}. This yields
\begin{align*}
\PjtFun(z,s^2,a; H) = C_{K,p} \cdot \frac{(s^2)^{\frac{K-p}{2} - 1}}{f_{H,K-p}(s^2, a)} \int_0^{\infty} \frac{(t^2)^{-\frac{K-p-1}{2}}}{\sqrt{(K-p+1)t^2 - (K-p)s^2}} \\
\hskip 8em \times f_{H,K-p+1}(t^2, a) \cdot \mathds{1} \left\{ t^2 \ge \frac{(K-p)s^2 + (z^2/\nu^2)}{K-p + 1} \right\} \, \dd t^2,
\end{align*}
where 
\[
C_{K,p} = 
\frac{\bigl(1 + \tfrac{1}{K-p}\bigr)^{-(K-p)/2} \, \Gamma\!\left(\tfrac{K-p+1}{2}\right)}%
{\sqrt{\pi}\,(K-p+1)^{-1/2}\, \Gamma\!\left(\tfrac{K-p}{2}\right)}.
\]

\subsection{Hellinger Convergence of the Estimated Marginal Density}

To prove Theorem~\ref{prop:asymp_p_val}, we first establish a large deviation inequality for
\(
\smash{\mathcal H^2\big(f_{\wh H,K-p}, f_{H,K-p}\big),}
\)
where for two bivariate densities $g_1$ and $g_2$ supported on
$\mathbb R_{\ge 0} \times \mathbb R$, corresponding to the joint
distribution of the summary statistics $(S^2,A)$, the squared Hellinger
distance is defined as
\[
\mathcal H^2(g_1,g_2):=\int_{\mathbb R_{\ge 0}\times\mathbb R}\left(\sqrt{g_1(s^2,a)} - \sqrt{g_2(s^2,a)}\right)^2\, \dd s^2\, \dd a.
\]

In that direction, define the class of joint marginal densities
\begin{align}
\label{eq:def_f_joint}
\mathcal F_\jt:=\left\{\big(f_{H,K-p}, f_{H,K-p+1}\big):\; H \in\mathcal G_H\right\},
\end{align}
where $\mathcal G_H$ is defined in \eqref{eq:joint_npmle}. To control the metric entropy of $\mathcal F_\jt$, fix $\varepsilon>0$ and localize the densities to a compact rectangular set
\(
\mathcal K_\varepsilon:=(0,\wb B_\varepsilon]\times[-\wb A_\varepsilon,\wb A_\varepsilon],
\)
where the truncation levels $\wb A_\varepsilon$ and $\wb B_\varepsilon$ (depending on $\varepsilon$) are specified below in \eqref{eq:def_a_varepsilon}.

For two elements
$(f_{H,K-p}, f_{H,K-p+1})$ and $(f_{H',K-p}, f_{H',K-p+1})$ in $\mathcal F_\jt$, define the localized uniform semi-norm
\begin{align}
\label{eq:def_metric_cover}
d_{\mathcal K_\varepsilon,\infty}\big((f_{H,K-p}, f_{H,K-p+1}),
(f_{H',K-p}, f_{H',K-p+1})\big)&:=
\max\Big\{\|f_{H,K-p}-f_{H',K-p}\|_{\mathcal K_\varepsilon,\infty},\nonumber\\
&\hskip 3em \|f_{H,K-p+1}-f_{H',K-p+1}\|_{\mathcal K_\varepsilon,\infty}\Big\},
\end{align}
where, for $\kappa\in\{K-p,K-p+1\}$,
\[
\|f_{H,\kappa}-f_{H',\kappa}\|_{\mathcal K_\varepsilon,\infty}
:=\sup_{(s^2,a)\in\mathcal K_\varepsilon}\big|f_{H,\kappa}(s^2,a)-f_{H',\kappa}(s^2,a)\big|.
\]

We choose $\wb A_\varepsilon$ and $\wb B_\varepsilon$ so that the
contribution of the tails outside $\mathcal K_\varepsilon$
is uniformly negligible over $\mathcal F_\jt$.
Specifically, we require that
\begin{align}
\label{eq:character2d_rem_outside_cpt}
\sup_{H\in\mathcal G_H}\sup_{(s^2,a)\notin\mathcal K_\varepsilon}
\big|
f_{H,\kappa}(s^2,a)\big|\;\le\;\frac{\varepsilon}{2},\qquad
\kappa\in\{K-p,K-p+1\}.
\end{align}

To achieve this, define
\begin{align}
\label{eq:def_a_varepsilon}
\wb A_\varepsilon:=M+\sqrt{\frac{\wb U}{K}\log\!\left(\frac{8}{\pi\varepsilon^2}\right)},\quad \wb B_\varepsilon
:=\wb U\cdot\max\left\{1,\frac{2}{K-p}\log\!\left(\frac{2}{\varepsilon}\right)\right\}.
\end{align}
By Lemma S1C of \citet{ignatiadis2025empirical}
and the proof of Lemma D.2 in \citet{saha_guntu},
this choice ensures that \eqref{eq:character2d_rem_outside_cpt}
holds uniformly over $\mathcal F_\jt$.

It remains to bound the metric entropy of $\mathcal F_\jt$
under the semi-norm $d_{\mathcal K_\varepsilon,\infty}$.
To this end, we construct explicitly an $\varepsilon$-cover
of $\mathcal F_\jt$ under this semi-norm, which yields
a corresponding upper bound on the log covering number.

\begin{lemm}
\label{lem:metric_entropy_2d}
Let $\mathcal F_\jt$ be the class defined in~\eqref{eq:def_f_joint}, equipped with the semi-norm $d_{\mathcal K_\varepsilon,\infty}$.
Then
\[
\log \mathcal N(\mathcal F_\jt,d_{\mathcal K_\varepsilon,\infty},\varepsilon)\lesssim_{M,\underline L,\wb U,K,p}
\left(\log\left(\frac{1}{\varepsilon}\right)\right)^3 .
\]
\end{lemm}

\begin{proof}
Fix $\varepsilon>0$. We construct an $\varepsilon$-cover of $\mathcal F_\jt$ under $d_{\mathcal K_\varepsilon,\infty}$. We first approximate the Gaussian component of $f_{H,\kappa}(a,s^2)$,
for $\kappa\in\{K-p,K-p+1\}$, by Taylor expansion in $a$. Uniformly over $(s^2,a)\in\mathcal K_\varepsilon$, we get
\[
\frac{\sqrt K}{\sqrt{2\pi}\sigma}
\exp\!\left(-\frac{K}{2\sigma^2}(a-\mu)^2\right)=
\frac{\sqrt K}{\sqrt{2\pi}}\sum_{i=0}^{J}\frac{K^i}{2^i \sigma^{2i+1} i!}(-1)^i (a-\mu)^{2i}+R_J(a,\mu,\sigma^2),
\]
with
\[
|R_J(a,\mu,\sigma^2)|
\le\frac{2^{J+1}K^{J+3/2}}{\sqrt{2\pi}\,\underline L^{J+3/2}(J+1)!}\wb A_\varepsilon^{2J + 2}, \qquad \mbox{for $a \in [-\wb A_\varepsilon,\wb A_\varepsilon]$.}
\]
Similarly, expanding the $s^2$-dependent component yields
\[
\frac{\kappa^{\kappa/2}}{2^{\kappa/2}\Gamma(\kappa/2)}
\exp\!\left(-\frac{\kappa s^2}{2\sigma^2}\right)
\left(\frac{s^2}{\sigma^2}\right)^{\frac{\kappa}{2}-1}
\frac{1}{\sigma^2}
=\sum_{j=0}^{\wb J}
\frac{(-1)^j \kappa^{j+\kappa/2} (s^2)^{j+\frac{\kappa}{2}-1}}
{\Gamma(\kappa/2) 2^{j+\kappa/2} \sigma^{2j+\kappa} j!}+R_{\wb J}(s^2,\sigma^2),
\]
where, using Stirling's bound $(m!)\ge (m/e)^m$,
\[
|R_{\wb J}(s^2,\sigma^2)|
\lesssim_{\kappa,\underline L,\wb B_\varepsilon}(\wb B_\varepsilon)^{\kappa/2-1}\underline L^{-\kappa/2}\left(\frac{e\,\kappa\,\wb B_\varepsilon}{2\underline L(\wb J+1)}\right)^{\wb J+1},
\quad \mbox{uniformly over $(s^2,\sigma^2)\in(0,\wb B_\varepsilon]\times[\underline L,\wb U]$}.
\]
Therefore, if we approximate $f_{H,\kappa}(a,s^2)$ using
\[
\sum_{i=0}^J\sum_{j=0}^{\wb J}\frac{(a-\mu)^{2i}}{(\sigma^2)^{i+j+\frac{\kappa+1}{2}}},
\]
then the total remainder can be bounded by
\begin{align}
\sup_{\substack{(s^2,a) \in \mathcal K_\varepsilon\\(\sigma^2,\mu) \in [\underline L,\wb U] \times [-M,M]}}
&\left\{|p_{\chi^2}(s^2 \mid \kappa, \sigma^2)||R_J(a,\mu,\sigma^2)|
+ \left|\frac{\sqrt K}{\sqrt{2\pi}\sigma}\right||R_{\wb J}(s^2,\sigma^2)|\right\} \notag\\
&\le  
\frac{\sqrt{K}\kappa^{\kappa/2}}{\sqrt{2\pi}\,2^{\kappa/2}\Gamma(\kappa/2)}
(\wb B_\varepsilon)^{\kappa/2-1}\underline L^{-(\kappa+1)/2}
\left(\frac{e\,\kappa\,\wb B_\varepsilon}{2\underline L(\wb J+1)}\right)^{\wb J+1}
\notag\\
&\qquad
+
\frac{\kappa(\kappa-2)^{\kappa/2-1}e^{-(\kappa-2)/2}}{2^{\kappa/2}\Gamma(\kappa/2)}
\cdot
\frac{2^{J+1}K^{J+3/2}}{\sqrt{2\pi}\,\underline L^{J+5/2}(J+1)!}\wb A_\varepsilon^{2J + 2}.
\end{align}

We choose $J$ and $\wb J$ large enough to ensure the total remainder is bounded by $\varepsilon/8$, i.e.,
\begin{align}
\label{eq:2d_rem_bod}
&
\frac{\sqrt{K}\kappa^{\kappa/2}}{\sqrt{2\pi}\,2^{\kappa/2}\Gamma\!\left(\frac{\kappa}{2}\right)}
(\wb B_\varepsilon)^{\kappa/2-1}\underline L^{-(\kappa+1)/2}
\left( \frac{e\kappa\wb B_\varepsilon}{2\underline L(\wb J + 1)} \right)^{\wb J + 1}
\notag\\
&\qquad
+
\frac{\kappa(\kappa-2)^{\kappa/2-1}e^{-(\kappa-2)/2}}{2^{\kappa/2}\Gamma(\kappa/2)}
\cdot
\frac{2^{J+1}K^{J+3/2}}{\sqrt{2\pi}\,\underline L^{J+5/2}(J+1)!}\wb A_\varepsilon^{2J + 2}
\le \frac{\varepsilon}{8}.
\end{align}

In that direction, observe that if we choose
\[
J
=
4\,\max\left\{
\frac{2e^2 K \wb A_\varepsilon^2}{\underline L},\;
\log\!\left(
\frac{16}{\varepsilon}\cdot
\frac{\kappa(\kappa-2)^{\kappa/2-1}e^{-(\kappa-2)/2}}{2^{\kappa/2}\Gamma(\kappa/2)}
\cdot
\frac{\sqrt K}{\sqrt{2\pi}\,\underline L^{3/2}}
\right)
\right\},
\]
then using $(C/m)^m \le e^{-m}$ for all $m \ge eC$, we can conclude that the second term in \eqref{eq:2d_rem_bod} is bounded above by $\varepsilon/16$.

Similarly, if we choose
\[
\wb J =
4\max_{\kappa \in \{K-p,K-p+1\}}\left\{
\frac{e^2 \kappa \wb B_\varepsilon}{2\underline L},\;
\log\!\left(\frac{16}{\varepsilon}\cdot
\frac{\sqrt K\,\kappa^{\kappa/2}}{\sqrt{2\pi}\,2^{\kappa/2}\Gamma(\kappa/2)}
\cdot
\frac{(\wb B_\varepsilon)^{\kappa/2-1}}{\underline L^{(\kappa+1)/2}}
\right)\right\},
\]
then the first term in \eqref{eq:2d_rem_bod} is bounded by $\varepsilon/16$.

Thus each density $f_{H,\kappa}$ is uniformly approximated on $\mathcal K_\varepsilon$
by a polynomial mixture whose coefficients depend only on the moments
\[
\int\frac{(a-\mu)^{2i}}{(\sigma^2)^{r+\frac{\kappa+1}{2}}}\,H(\dd \mu,\dd\sigma^2),\quad \mbox{where $r=i+j$ and}\;0\le i\le J,\; 0\le j\le \wb J,
\]
for $\kappa\in\{K-p,K-p+1\}$. Expanding the polynomial $(a-\mu)^{2i}$ for all $i \in \{0,\ldots, J\}$, we get
\[
\sum_{\ell=0}^{2i}\binom{2 i}{\ell}(-1)^{\ell}a^{2i-\ell}\int\frac{\mu^\ell}{(\sigma^2)^{r+\frac{\kappa+1}{2}}}\,H(\dd \mu,\dd\sigma^2),\quad \mbox{where $r=i+j$ and}\;0\le i\le J,\; 0\le j\le \wb J.
\]
These distinct moment conditions are indexed by $(r,\ell)$ with $0\le r\le J+\wb J$ and $0\le \ell\le 2\min(r,J)$, giving
\[
\sum_{r=0}^{J+\wb J}\bigl(2\min(r,J)+1\bigr)
= \sum_{r=0}^{J}(2r+1) + \sum_{r=J+1}^{J+\wb J}(2J+1)
= (J+1)^2 + \wb J(2J+1)
\]
such linear constraints.
By Carath\'{e}odory's theorem
(see \citet{saha_guntu} and \citet{ignatiadis2025empirical}),
there exists a discrete measure $\wt H$ for the pair $(\mu,\sigma^2)$ supported on at most
\begin{align}
\label{eq:bound_j_2d}
\wt J := 2[(J+1)^2 + \wb J(2J+1)]+1
\end{align}
atoms in $[-M,M]\times[\underline L,\wb U]$ that matches all these moments.
Consequently,
\[
d_{\mathcal K_\varepsilon,\infty}
\big((f_{H,K-p},f_{H,K-p+1}),(f_{\wt H,K-p},f_{\wt H,K-p+1})\big)\le \frac{\varepsilon}{4}.
\]

It therefore suffices to construct a $\varepsilon/4$ cover for the class of discrete mixtures
with at most $\wt J$ support points in $[-M,M] \times [\underline L,\wb U]$. To achieve this, it is enough to construct  $\varepsilon/12$-covers of $[-M,M]$ and $[\underline L,\wb U]$ respectively, and an a
$\varepsilon/12$-cover of the probability simplex
$\Delta(\wt J)$ under $\ell_1$ distance. Let $\mathcal N_\mu$ and $\mathcal N_{\sigma^2}$ be
$\varepsilon/12$-covers of $[-M,M]$ and $[\underline L,\wb U]$ respectively. Since these intervals are compact,
\[
|\mathcal N_\mu| + |\mathcal N_{\sigma^2}|\lesssim_{M,\underline L,\wb U} \frac{1}{\varepsilon}.
\]
Let $\mathcal N(\Delta(\wt J))$ be a
$\varepsilon/12$-cover of the probability simplex
$\Delta(\wt J)$ under $\ell_1$ distance.
By Lemma~S.1 of \citet{ignatiadis2025empirical},
\[
\log |\mathcal N(\Delta(\wt J))|\lesssim\wt J \log\!\left(\frac{1}{\varepsilon}\right).
\]
One can combine these discretizations to show that $\mathcal F_\jt$ admits an $\varepsilon$-cover in terms of $d_{\mathcal K_\varepsilon,\infty}$ metric whose
cardinality is bounded by
\[
\log |\mathcal N(\mathcal F_\jt,d_{\mathcal K_\varepsilon,\infty},\varepsilon)|\lesssim \wt J \log\!\left(\frac{1}{\varepsilon}\right).
\]
Using~\eqref{eq:bound_j_2d},
\[
\wt J = 2[(J+1)^2 + \wb J(2J+1)]+1\lesssim_{K,p,M,\wb U,\underline L} \left(\log\left(\frac{1}{\varepsilon}\right)\right)^2,
\]
and therefore
\[
\log |\mathcal N| \lesssim_{K,p,M,\wb U,\underline L} \left(\log\left(\frac{1}{\varepsilon}\right)\right)^3.
\]
\end{proof}

With the above lemma, we can prove the following large deviation inequality for the Hellinger distance between the estimated and the true marginal densities of the pair $(S^2,A)$. We note that such a result was also shown recently by~\citet[Appendix C]{ho2025largescale}. For self-containedness, we record our independent proof here as well. 

\begin{theorem}
\label{thm:double_hellinger}
Let $\widehat{H}$ be the estimator defined by \eqref{eq:joint_npmle}. Fix an absolute constant $c_0>2$. Then there exists constants $C > 0$ and $n_0 \in \mathbb N_{\ge 1}$, depending only on $K, p, M, \underline{L}, \overline{U}$, such that 
\[
\mathbb{P}\left[ \mathcal{H}^2\left(f_{\widehat{H},K-p},f_{H,K-p}\right) \geq C \frac{(\log n)^3}{n} \right] \leq \exp\left(-c_0\,(\log n)\right), \quad \text{for all } n \ge n_0.
\]
\end{theorem}

\begin{proof}
Let us define the event
\[
A_{n,\varepsilon_n} = \left\{ \mathcal{H}^2\left(f_{\widehat{H},K-p},f_{H,K-p}\right) \ge \varepsilon_n \right\},
\]
for some $\varepsilon_n$ (to be specified later) and aim to characterize $\mathbb{P}[A_{n,\varepsilon_n}]$. Consider the subset of $\mathcal{F}_\jt$:
\[
\mathcal{F}_{\jt,\varepsilon_n} = \left\{ (f_{H',K-p},f_{H',K-p+1}) \in \mathcal{F}_\jt : \mathcal{H}^2\left(f_{H',K-p},f_{H,K-p}\right) \ge \varepsilon_n \right\}.
\]
Fix $\eta = 1/n^2$ and let $\mathcal{C}_{\mathcal{F}_{\jt,\varepsilon_n}} = \{ (f_{H_1,K-p},f_{H_1,K-p+1}), \ldots, (f_{H_{\wb N},K-p},f_{H_{\wb N},K-p+1}) \}$ be a proper $\eta$-cover of $\mathcal{F}_{\jt,\varepsilon_n}$ under the semi-norm $d_{\mathcal{K}_\eta,\infty}$ (defined in \eqref{eq:def_metric_cover}). Let the sub collection of the densities involving degrees of freedom $K-p$ be denoted by $\mathcal{C}^{(K-p)}_{\mathcal{F}_{\jt,\varepsilon_n}}$. In other words,
\[
\mathcal{C}^{(K-p)}_{\mathcal{F}_{\jt,\varepsilon_n}}=\left\{f_{H_1,K-p},\ldots,f_{H_{\wb N},K-p}\right\}.
\]
By the definition of the cover we have for $(f_{K-p},f_{K-p+1}) \in \mathcal{F}_{\jt,\varepsilon_n}$,
\[
\inf_{g \in \mathcal{C}^{(K-p)}_{\mathcal{F}_{\jt,\varepsilon_n}}} \left\|f_{K-p} - g\right\|_{\mathcal{K}_\eta, \infty} \le \eta.
\]
Although Lemma~\ref{lem:metric_entropy_2d} provides an upper bound on the covering number of $\mathcal{F}_\jt$, the resulting cover may be improper for the subset $\mathcal{F}_{\jt,\varepsilon_n}$. However, any such improper cover can be converted into a proper one at the expense of doubling the covering radius. Thus,
\[
\log(\wb N) \lesssim_{M, \underline{L}, \overline{U}} \left( \log \left( \frac{1}{\eta} \right) \right)^3.
\]

Let $\mathfrak{R}_\eta= (0,\wb B_\eta] \times [-\wb A_\eta,\wb A_\eta]$ (where $\wb A_\eta$ and $\wb B_\eta$ is defined as in \eqref{eq:def_a_varepsilon} with $\varepsilon=\eta$), and define, for $z = (z_1, z_2) \in \mathbb R_{\ge 0} \times \mathbb R$,
\begin{align}
\label{eq:def_eta}
\eta(z) &= \eta \cdot \mathds{1}\{ z \in \mathfrak{R}_\eta \} 
+ \frac{\eta A_\eta}{z_2^2} \cdot \mathds{1}\{  z_1 \in (0, B_\eta],\,|z_2| > A_\eta \} \nonumber \\
&\quad + \frac{\eta B_\eta}{z_1^2} \cdot \mathds{1}\{z_1 \notin (0, B_\eta] ,\, |z_2| \le A_\eta \}
+ \frac{\eta A_\eta B_\eta}{z_1^2 z_2^2} \cdot \mathds{1}\{  z_1 > B_\eta ,\,|z_2| > A_\eta\}.
\end{align}
Then,
\begin{align}
\label{eq:int_eta}
\int_{\mathbb{R}_{\ge 0} \times \mathbb{R}} \eta(z)\,\dd z \le 4(A_\eta+1)(B_\eta+1)\eta.
\end{align}
Now, on the event $A_{n,\varepsilon_n}$, there exists $\widehat{\imath} \in \{1, \ldots, \wb N\}$ such that
\[
\| f_{\widehat{H},K-p} - f_{H_{\widehat{\imath}},K-p} \|_{\mathfrak{R}_\eta, \infty} \le 2\eta.
\]
This implies
\[
f_{\widehat{H},K-p}(s^2,a) \le 
\begin{cases}
\max_{j \in [\wb N]} \left\{ f_{H_j,K-p}(s^2,a) + 2\eta \right\} & \text{if } (s^2,a) \in \mathfrak{R}_\eta, \\
\mathscr{R}_{K,p} & \text{otherwise},
\end{cases}
\]
where we define
\[
\mathscr{R}_{K,p} = \frac{\sqrt{K}}{\sqrt{2\pi}\underline L^{3/2}} \cdot \frac{(K-p)^{(K-p)/2}e^{-\frac{K-p-2}{2}}}{2^{(K-p)/2} \, \Gamma\left(\frac{K-p}{2}\right)}.
\]

Define the likelihood ratio
\[
L_n(f_{H',K-p}, f_{H,K-p}) = \prod_{i=1}^{n} \frac{f_{H',K-p}(S_i^2,A_i)}{f_{H,K-p}(S_i^2,A_i)}.
\]
Since $f_{\widehat{H},K-p}$ is the NPMLE, it satisfies $L_n(f_{\widehat{H},K-p}, f_{H,K-p}) \ge 1$. Following the proof of Theorem 9 in \cite{ignatiadis2025empirical}, we obtain on $A_{n,\varepsilon_n}$:
\begin{align*}
L_n(f_{\widehat{H},K-p}, f_{H,K-p}) & \le  \max_{j \in [\wb N]} \left\{ \prod_{i : (S_i^2,A_i) \in \mathfrak{R}_\eta} \frac{f_{H_j,K-p}(S_i^2,A_i) + 2\eta(S^2_i,A_i)}{f_{H,K-p}(S_i^2,A_i)} \times  \prod_{i : (S_i^2,A_i) \notin \mathfrak{R}_\eta}\frac{\mathscr{R}_{K,p}}{f_{H,K-p}(S_i^2,A_i)}\right\}\\
&= \max_{j \in [\wb N]} \Bigg\{ \prod_{i=1}^n \frac{f_{H_j,K-p}(S_i^2,A_i) + 2\,\eta(S^2_i,A_i)}{f_{H,K-p}(S_i^2,A_i)} \\
&~~~~~~~~~~~~~~~~~~~~~~~~~~~~~~~~~\times  \prod_{i : (S_i^2,A_i) \notin \mathfrak{R}_\eta}\frac{\mathscr{R}_{K,p}}{f_{H_j,K-p}(S_i^2,A_i) + 2\,\eta(S^2_i,A_i)}\Bigg\}\\
&\lesssim \max_{j \in [\wb N]} \left\{ \prod_{i=1}^n \frac{f_{H_j,K-p}(S_i^2,A_i) + 2\,\eta(S^2_i,A_i)}{f_{H,K-p}(S_i^2,A_i)} \times  \prod_{i : (S_i^2,A_i) \notin \mathfrak{R}_\eta}\frac{\mathscr{R}_{K,p}}{\eta(S^2_i,A_i)}\right\}\\
&\le \max_{j \in [\wb N]} \left\{ \prod_{i=1}^n \frac{f_{H_j,K-p}(S_i^2,A_i) + 2\,\eta(S^2_i,A_i)}{f_{H,K-p}(S_i^2,A_i)} \right\} \times \prod_{i : (S_i^2,A_i) \notin \mathfrak{R}_\eta} \frac{\mathscr{R}_{K,p}}{\eta(S_i^2,A_i)}.
\end{align*}

Fix $\gamma > 0$. Then,
\begin{align}
\mathbb{P}\left[ \mathcal{H}^2(f_{\widehat{H},K-p}, f_{H,K-p}) \ge \varepsilon_n \right]
&= \mathbb{P}[A_{n,\varepsilon_n}] \nonumber \\
&= \mathbb{P}\left[ A_{n,\varepsilon_n}\, L_n(f_{\widehat{H},K-p}, f_{H,K-p}) \ge 1 \right] \nonumber \\
&\le \mathbb{P}\left[\max_{j \in [\wb N]} \left\{ \prod_{i=1}^n \frac{f_{H_j,K-p}(S_i^2,A_i) + 2\,\eta(S^2_i,A_i)}{f_{H,K-p}(S_i^2,A_i)} \right\} \ge e^{-2\gamma} \right] \nonumber \\
&\quad + \mathbb{P}\left[ \prod_{i : (S_i^2,A_i) \notin \mathfrak{R}_\eta} \frac{\mathscr{R}_{K,p}}{\eta(S_i^2,A_i)} \ge e^{2\gamma} \right].
\end{align}
Following \cite{ignatiadis2025empirical}, using \eqref{eq:int_eta} we obtain 
\[
\mathbb{P}\left[ \max_{j \in [\wb N]} \left\{ \prod_{i=1}^n \frac{f_{H_j,K-p}(S_i^2,A_i) + 2\,\eta(S^2_i,A_i)}{f_{H,K-p}(S_i^2,A_i)} \right\} \ge e^{-2\gamma} \right] 
\le \exp\left( -n \varepsilon_n + n\sqrt{8\eta \wb A_\eta \wb B_\eta} + \gamma + \log \wb N \right),
\]
for all large $n \in \mathbb N_{\ge 0}$.
Define $\varepsilon_n = C (\log n)^3 / n$ and $\gamma = n t^2 \varepsilon_n$. Using the independence of $S_i^2$ and $A_i$, together with the arguments of \cite[Lemma~S.5]{ignatiadis2025empirical} and \cite[Theorem~1]{Zhang2009GMLE}, we can choose constants $C,t>0$ and $n_0 \in \mathbb N_{\ge 1}$ such that for all $n \ge n_0$,
\[
\mathbb{P}\left[ \prod_{i \in \mathcal{I}^c} \frac{\mathscr{R}_{K,p}}{\eta(S_i^2,A_i)} \ge e^{2\gamma} \right]
\le \exp\!\left(-c_0 \log n\right),
\]
for some constant $c_0>2$. Furthermore, using Lemma~\ref{lem:metric_entropy_2d}, we may choose $C>0$ sufficiently large so that
\[
\mathbb{P}\left[ \max_{j \in [\wb N]} \left\{ \prod_{i=1}^n \frac{f_{H_j,K-p}(S_i^2,A_i) + 2\,\eta(S^2_i,A_i)}{f_{H,K-p}(S_i^2,A_i)} \right\} \ge e^{-2\gamma} \right]
\le \exp\left(-c_0 (\log n)\right).
\]
Combining the two preceding bounds completes the proof.

\end{proof}

\subsection{Proof of Proposition \ref{prop:asymp_p_val}}
To control the difference between the true $p$-values $\{P^\jt_i\}$ (defined in \eqref{eq:2d_pb_pvalue_general_case}) and the NPMLE-based $p$-values $\{\wh P^\jt_i\}$ (obtained by replacing $H$ by the NPMLE $\wh H$), we shall combine the Eddington-Tweedie type representation of $\{P^\jt_i\}$ established in Theorem \ref{thm:tweedie_2_d}, along with the control on the estimation error of $f_{H,K-p}$ by $f_{\wh H,K-p}$ in terms of Hellinger distance quantified in Theorem \ref{thm:double_hellinger}. However, recall that the Eddington--Tweedie-type formula for the $p$-value involves both the densities $f_{H,K-p}$ and $f_{H,K-p+1}$. While the construction of $\widehat{H}$ based on the density $f_{\widehat{H},K-p}$ ensures that $f_{\widehat{H},K-p}$ and $f_{H,K-p}$ are close in Hellinger distance, and consequently also in the $L_2$
distance, this conclusion does not automatically extend to the pair
$(f_{\widehat H,K-p+1}, f_{H,K-p+1})$. Here, the $L_2$ distance between two densities
$f$ and $g$ is defined as
\[
\|f-g\|_{L_2}
:= \left[ \int_{\mathbb R_{\ge 0} \times \mathbb R}
\bigl(f(s^2,a)-g(s^2,a)\bigr)^2\,\dd s^2\,\dd a \right]^{1/2}.
\]
Nevertheless, under the assumptions of our model, it is possible to transfer the
Hellinger control at level $K-p$ to obtain an $L_2$ bound at level $K-p+1$. In
particular, the following lemma bounds the $L_2$ distance between
$f_{H_1,K-p+1}$ and $f_{H_2,K-p+1}$, for any priors $H_1,H_2\in\mathcal G_H$, in terms
of the Hellinger distance between the corresponding densities
$f_{H_1,K-p}$ and $f_{H_2,K-p}$.

\begin{lemm}
\label{lem:l_2_one_dof}
Consider the $L_2$ distance between the densities $f_{H_1,K-p+1}$ and $f_{H_2,K-p+1}$ where the priors $H_1,H_2 \in \mathcal G_H$. Then we have
\begin{align}
\label{eq:l_2_dist_density}
&\left\| f_{H_1,K-p+1} - f_{H_2,K-p+1} \right\|^2_{L_2}\nonumber\\
&\hskip 5em\lesssim_{K, p, M, \underline{L}, \overline{U}}
\left(1+\left| \log \mathcal{H} \left( f_{H_1,K-p}, f_{H_2,K-p} \right) \right|\right)
\cdot \mathcal{H}^2 \left( f_{H_1,K-p}, f_{H_2,K-p} \right).
\end{align}
\end{lemm}
Using the foregoing lemma, it suffices to characterize the convergence rate of 
\(
\mathcal{H}^2 \left( f_{\wh H,K-p}, f_{H,K-p} \right)
\)
which is established in Theorem~\ref{thm:double_hellinger}, in order to derive a corresponding convergence rate for $f_{\widehat H,K-p+1}$ toward $f_{H,K-p+1}$ in the $L_2$ metric defined in \eqref{eq:l_2_dist_density}. The proof of this lemma involves nontrivial arguments based on integral transforms. A similar result was established in Lemma 12 of \citet{ignatiadis2025empirical} in the context of univariate densities involving $S^2$, where the Mellin transform of the marginal density was employed. However, this approach does not directly apply in our setting, as the joint density of $(S^2,A)$ is not supported on the positive orthant. Instead, we introduce a new integral transform that combines the Mellin transform (for the $S^2$ coordinate) and the Fourier transform (for the $A$ coordinate), and leverage its properties to derive the desired inequality.

\begin{proof}[Proof of Proposition \ref{prop:asymp_p_val}]
Let us begin by observing that for any joint prior $H' \in \mathcal G_H$, we have
\begin{align}
    \PjtFun(z,S^2_i,A_i;H') \ge 2 \left(1-\Phi(|z| / \nu \underline L^{1/2})\right),
\end{align}
for $\PjtFun(\cdot)$ defined in the proof of Lemma~\ref{lem:valid_oracle_p_vals}
Then for $|z| \le z_{1 - \zeta/2}$, i.e., $|z| \le \underline z:=\nu\underline L^{1/2} z_{1 - \zeta/2}$, it holds that:
\[
\PjtFun(z,S^2_i,A_i;H') \ge 2 \wb{\Phi}(z_{1 - \zeta/2}) = \zeta, 
\]
where $\wb{\Phi}(\cdot)$ is the survival function of the standard Gaussian distribution. Since both the true mixing measure $H$ and the estimated mixing measure $\wh H$ belong to $\mathcal G_H$, therefore we can conclude that
\begin{align}
  \PjtFun(z,S^2_i,A_i;H) \ge \zeta, \qquad \mbox{and} \qquad   \PjtFun(z,S^2_i,A_i;\wh H) \ge \zeta.
\end{align}
Consequently, for all $|z| \le z_{1 - \zeta/2}$, we have
\[
\PjtFun(z,S^2_i,A_i;\wh H) \wedge \zeta - \PjtFun(z,S^2_i,A_i;H) \wedge \zeta = \zeta - \zeta = 0, \quad \mbox{for $|z| \le \nu \underline L^{1/2} z_{1 - \zeta/2}$.}
\]
Furthermore, since the projection of a number in $[0,1]$ to the set $[0,\zeta]$ we have
\begin{align}
\left|\PjtFun(z,S^2_i,A_i;\wh H) \wedge \zeta - \PjtFun(z,S^2_i,A_i;H) \wedge \zeta\right| & \le \left|\PjtFun(z,S^2_i,A_i;\wh H) - \PjtFun(z,S^2_i,A_i;H)\right|.
\end{align}
Combining the two above relations and letting $\widetilde{Z}_i := \operatorname{sign}(Z_i) \cdot (|Z_i| \vee \underline z)$, we find the following:
\begin{align}
    \left|\PjtFun(Z_i,S^2_i,A_i;\wh H) \wedge \zeta - \PjtFun(Z_i,S^2_i,A_i;H) \wedge \zeta\right| 
    &\le \left|\PjtFun(\widetilde{Z}_i,S^2_i,A_i;\wh H) - \PjtFun(\widetilde{Z}_i,S^2_i,A_i;H)\right| \\
    &\le \sup_{|z| \ge \underline z} \left|\PjtFun(z,S^2_i,A_i;\wh H) - \PjtFun(z,S^2_i,A_i;H)\right|.
\end{align}
Therefore, it is enough to control
\[
\frac{1}{n}\sum_{i=1}^n\mathbb E\left[\sup_{|z| \ge \underline z} \left|\PjtFun(z,S^2_i,A_i;\wh H) - \PjtFun(z,S^2_i,A_i;H)\right|\right].
\]

Let $\mathcal A$ be the event:
\begin{equation}
\label{eq:def_cal_A_2d}
\mathcal A := \left\{ \mathcal H^2\left(f_{\wh H,K-p}, f_{H,K-p}\right) < C \frac{(\log n)^3}{n} \right\}, \qquad \mbox{for $C>0$ defined in Theorem \ref{thm:double_hellinger}.}
\end{equation}
Consider the following decomposition.
\begin{align}
\label{eq:show_exp_later_comp}
&\mathbb E\left[\sup_{|z| \ge \underline z} \left|\PjtFun(z,S^2_i,A_i;\wh H) - \PjtFun(z,S^2_i,A_i;H)\right|\right]\\
&\le \mathbb{P}[\mathcal A^c] + \mathbb{E} \left[\sup_{|z| \ge \underline z} \left|\PjtFun(z,S^2_i,A_i;\wh H) - \PjtFun(z,S^2_i,A_i;H)\right| \cdot \mathds{1}(\mathcal A) \right].
\end{align}
In the above decomposition, we have used the fact that $\PjtFun(z,S^2_i,A_i;\wh H), \PjtFun(z,S^2_i,A_i;H) \in [0, 1]$. Next, for any distribution $H'$ supported on $\mathbb R \times \mathbb R_{\ge 0}$ and any $z \in \mathbb{R}$, $s^2 > 0$, $a \in \mathbb R$, let us write:
\begin{align}
\label{eq:def_n_i_z}
N(z, s^2, a; H') &:= \PjtFun(z,s^2,a;H') \cdot f_{H',K-p}(s^2,a)\\
N_i(z, H') &:= N(z, S^2_i, A_i; H'), \\
D_i(H') &:= f_{H',K-p}(S^2_i, A_i).
\end{align}
By these definitions, it holds that $\PjtFun(z,S^2_i,A_i;\wh H) = N_i(z, \widehat{H})/D_i(\widehat{H})$. Let us define $\widehat{H}_\star:=(\widehat{H} + H)/2$.
Then:
\[
\begin{aligned}
&\left|\PjtFun(z,S^2_i,A_i;\wh H) - \PjtFun(z,S^2_i,A_i;H)\right|\\
&= \left| \frac{N_i(z, {H})}{D_i({H})} - \frac{N_i(z, \widehat{H})}{D_i(\widehat{H})} \right| \\
&= \left| \frac{N_i(z, {H})}{D_i({H})}-\frac{N_i(z, {H})}{D_i(\widehat{H}_\star)}+\frac{N_i(z, {H})}{D_i(\widehat{H}_\star)} - \frac{N_i(z, \wh H)}{D_i(\wh H_\star)} + \frac{N_i(z, \wh H)}{D_i(\wh H_\star)} - \frac{N_i(z, \widehat{H})}{D_i(\widehat{H})} \right| \\
&\le \frac{N_i(z, {H})}{D_i({H})}\frac{|D_i(\widehat{H}_\star)-D_i({H})|}{D_i(\widehat{H}_\star)} + \frac{|N_i(z, {H})-N_i(z, \wh H)|}{D_i(\widehat{H}_\star)}+\frac{N_i(z, \wh{H})}{D_i(\wh{H})}\frac{|D_i(\widehat{H}_\star)-D_i(\wh{H})|}{D_i(\widehat{H}_\star)}\\
&\le \frac{|N_i(z, {H})-N_i(z, \wh H)|}{D_i(\widehat{H}_\star)}+\frac{|D_i(\widehat{H})-D_i({H})|}{D_i(\widehat{H}_\star)}
\end{aligned}
\]
In the last step we have used two facts: first, it holds that $N_i(z, H') / D_i(H') \in [0,1]$ for all $H' \in \mathcal G_H$ (since they correspond to conditional $p$-values), and second, the map $H' \mapsto D_i(H')$ is linear, which implies that:
\[
D_i(H) - D_i(\widehat{H}_\star) = \frac{D_i(H) - D_i(\widehat{H})}{2}, \quad
D_i(\widehat{H}_\star) - D_i(\widehat{H}) = \frac{D_i(\wh{H}_\star) - D_i(\widehat{H})}{2}.
\]
Combining the above results, we have
\begin{align}
\label{eq:final_master_eq_2d}
    &\frac{1}{n}\sum_{i=1}^n\mathbb E\left[\sup_{|z| \ge \underline z} \left|\PjtFun(z,S^2_i,A_i;\wh H) - \PjtFun(z,S^2_i,A_i;H)\right|\right]\\
&\le \mathbb{P}[\mathcal A^c] + \frac{1}{n}\sum_{i=1}^n\mathbb{E} \left[ \sup_{|z| \ge \underline z} \left|\PjtFun(z,S^2_i,A_i;\wh H) - \PjtFun(z,S^2_i,A_i;H)\right| \cdot \mathds{1}(\mathcal A) \right]\\
& \le  \mathbb{P}[\mathcal A^c] + \frac{1}{n}\sum_{i=1}^n\mathbb{E} \left[ \sup_{|z| \ge \underline z} \frac{|N_i(z, {H})-N_i(z, \wh H)|}{D_i(\widehat{H}_\star)} \cdot \mathds{1}(\mathcal A) \right]+\frac{1}{n}\sum_{i=1}^n\mathbb{E} \left[ \frac{|D_i(\widehat{H})-D_i({H})|}{D_i(\widehat{H}_\star)} \cdot \mathds{1}(\mathcal A) \right]\\
& \le \mathbb{P}[\mathcal A^c] + T_1 +T_2,
\end{align}
where
\begin{align}
&T_1=\frac{1}{n}\sum_{i=1}^n\mathbb{E} \left[ \sup_{|z| \ge \underline z} \frac{|N_i(z, {H})-N_i(z, \wh H)|}{D_i(\widehat{H}_\star)} \cdot \mathds{1}(\mathcal A) \right], \quad \mbox{and} \\
&\quad T_2=\frac{1}{n}\sum_{i=1}^n\mathbb{E} \left[ \frac{|D_i(\widehat{H})-D_i({H})|}{D_i(\widehat{H}_\star)} \cdot \mathds{1}(\mathcal A) \right].
\end{align}
By Theorem \ref{thm:double_hellinger}, we have
\(
\mathbb{P}[\mathcal A^c] \le  \exp\left(-c_0 (\log n)\right), 
\)
where $c_0>2$.
Now, we consider the following two lemmas about the terms $T_1$ and $T_2$.
\begin{lemm}
\label{lemm:lem_ratio_bound_2d}
    For all $n \in \mathbb N_{\ge 1}$, we have
    \begin{align}
    \label{eq:main_n_term}
        \frac{1}{n}\sum_{i=1}^n\mathbb{E} \left[ \sup_{|z| \ge \underline z} \frac{|N_i(z, {H})-N_i(z, \wh H)|}{D_i(\widehat{H}_\star)} \cdot \mathds{1}(\mathcal A) \right] \lesssim_{K,p,M,\wb U,\underline L,\underline z} \frac{(\log n)^{13/4}}{\sqrt{n}},
    \end{align}
    where the set $\mathcal A$ is defined in \eqref{eq:def_cal_A_2d}.
\end{lemm}

\begin{lemm}
\label{lemm:density_ratio_bound_2d}
    For all $n \in \mathbb N_{\ge 1}$, we have
    \[
    \frac{1}{n}\sum_{i=1}^n\mathbb{E} \left[ \frac{|D_i(\widehat{H})-D_i({H})|}{D_i(\widehat{H}_\star)} \cdot \mathds{1}(\mathcal A) \right]
     \lesssim_{K,p,M,\wb U,\underline L} \frac{(\log n)^{3/2}}{\sqrt{n}}
     \]
    where the set $\mathcal A$ is defined in \eqref{eq:def_cal_A_2d}.
\end{lemm}
Plugging in the conclusion of the above two lemmas in \eqref{eq:final_master_eq_2d}, the assertion of the proposition follows.
\end{proof}

\subsubsection{Proof of Lemma~\ref{lem:l_2_one_dof}}
To prove Lemma~\ref{lem:l_2_one_dof}, we begin by defining the following integral transform for the density $f_{H,K-p}$ in \eqref{eq:2_d_marginal}.

\begin{definition}[Mellin--Fourier transform]
\label{def:mellin_fourier_transform}
Let $(X_1,X_2)$ be a random vector with joint density $f_{H,K-p}(x,y)$
corresponding to a prior $H$ satisfying
\begin{align}
\label{Eq:def_mell_trans_assm}
\mathsf{supp}(H) \subseteq [-M,M] \times [\underline L,\overline U].
\end{align}
For $c \in \mathbb R$, the \emph{Mellin--Fourier transform} of $(X_1,X_2)$
is defined as
\begin{align}
\label{eq:mellin_f_tran}
\widehat f_{H,K-p}(t_1,t_2;c)
= \int_{0}^{\infty}\int_{-\infty}^{\infty}
x^{\,c-1+i t_1}\, e^{i t_2 y}\,
f_{H,K-p}(x,y)\, \dd y\, \dd x,
\end{align}
whenever the integral is well defined.
\end{definition}
Analogously, the same transform can be defined for the density $f_{H,K-p+1}$. Using the arguments from the proof of Lemma 12 in \cite{ignatiadis2025empirical}, we show that the integral transform is finite for both $f_{H,K-p}$ and $f_{H,K-p+1}$ whenever $c+(K-p)/2-1>0$.

With the change of variable $(x,y)\mapsto (e^u,y)$ we may write
\[
\widehat f_{H,K-p}(t_1,t_2;c)
=\int_{-\infty}^\infty\int_{-\infty}^\infty
e^{i t_1 u}\,e^{i t_2 y}\,F_{H,K-p}(u,y;c)\,du\,dy,
\]
where $F_{H,K-p}(u,y;c)=e^{uc}f_{H,K-p}(e^u,y)$.  
Thus $\widehat f_{H,K-p}(t_1,t_2;c)$ is the Fourier transform of the function $F_{H,K-p}(u,y;c)$.  
By the Plancherel isometry \citep{bracewell1999fourier},
\begin{align}
\label{eq:reduction_square_int_trn}
    &\frac{1}{(2\pi)^2}\int_{-\infty}^\infty\int_{-\infty}^\infty
    \Big|\widehat f_{H_1,K-p+1}(t_1,t_2;c)-\widehat f_{H_2,K-p+1}(t_1,t_2;c)\Big|^2\,\dd t_1\,\dd t_2 \nonumber\\
    &\quad = \int_{-\infty}^\infty\int_{-\infty}^\infty 
    \Big(F_{H_1,K-p+1}(u,y;c)-F_{H_2,K-p+1}(u,y;c)\Big)^2\,\dd u\,\dd y \nonumber\\
    &\quad = \int_{-\infty}^\infty\int_{-\infty}^\infty 
    e^{2uc}\,\Big(f_{H_1,K-p+1}(e^u,y)-f_{H_2,K-p+1}(e^u,y)\Big)^2\,\dd u\,\dd y \nonumber\\
    &\quad = \int_{0}^\infty\int_{-\infty}^\infty 
    x^{2c-1}\,\Big(f_{H_1,K-p+1}(x,y)-f_{H_2,K-p+1}(x,y)\Big)^2\,\dd x\,\dd y,
\end{align}
for any pair of densities $H_1,H_2$ satisfying \eqref{Eq:def_mell_trans_assm}.  
An analogous identity holds for $f_{H_1,K-p}$ and $f_{H_2,K-p}$. Observe that using the arguments from the proof of \cite[Lemma 12]{ignatiadis2025empirical}, we can show that the integral in \eqref{eq:reduction_square_int_trn} is finite.

In this proof, we consider the Mellin–Fourier transform of $(S^2,A)$ (denoted by $\mathcal M_c(t_1,t_2;H)$), where 
\[
A \sim \mathrm{N}(\mu, \sigma^2/K), 
\qquad (K-p)S^2 \sim \sigma^2 \chi^2_{K-p}, 
\qquad (\mu,\sigma^2) \sim H.
\]
We write
\[
\mathcal M_c(t_1,t_2;H)
:=\mathbb E_H\!\left[(S^2)^{c+i t_1-1}e^{i t_2 A}\right],
\]
Since $S^2 \sim \sigma^2 \chi^2_{K-p}/(K-p)$ and 
$A \sim \mu+\sigma Z/\sqrt{K}$ with $Z\sim \mathrm{N}(0,1)$ and the $\chi^2_{K-p}$ variable independent of $(\mu,\sigma^2)$. Then, using independence, we have
\begin{align}
    \mathcal M_{c,K-p}(t_1,t_2;H) 
    &= \mathbb E_H\left[(\sigma^2)^{c+i t_1-1}\left(\frac{\chi^2_{K-p}}{K-p}\right)^{c+i t_1-1}e^{i t_2\left(\mu+\frac{\sigma Z}{\sqrt{K}}\right)}\right] \\
    &= \mathbb E_H\left[(\sigma^2)^{c+i t_1-1}e^{i t_2\left(\mu+\frac{\sigma Z}{\sqrt{K}}\right)}\right]\,
       \mathbb E\left[\left(\frac{\chi^2_{K-p}}{K-p}\right)^{c+i t_1-1}\right],
\end{align}
where the expectation is over $\chi^2_{K-p}$, $Z$, and $H$.  
Note that only the second factor depends on $K$ and $p$. 

Fix $c$ satisfying $c+(K-p)/2-1>0$. By retracing the argument in the proof of Lemma 12 of \citet{ignatiadis2025empirical}, we conclude that 
\[
\mathcal M_{c,K-p+1}(t_1,t_2;H)
=\mathcal M_{c,K-p}(t_1,t_2;H)\,\cdot r_{K,p,c}(t_1),
\]
for some function $r_{K,p,c}:\mathbb R \rightarrow \mathbb C$ satisfying
\begin{align}
\label{eq:bound_r_low_t}
|r_{K,p,c}(t_1)| \le \mathfrak K_{K,p,c}|T_0|^{1/2},
\quad \text{for all $t\in \mathbb R$ with $|t_1|\le T_0$ where $T_0>\mathfrak r_{K,p,c}$},
\end{align}
and $\mathfrak K_{K,p,c}, \mathfrak r_{K,p,c}$ are absolute positive constants. Note that the function $r_{K,p,c}$ is the same as the $r(t)$ function defined in \cite[Lemma 12]{ignatiadis2025empirical}. Furthermore, since $\mathsf{supp}(H)\subseteq [-M,M] \times [\underline L,\overline U]$,
\[
\Bigg|\mathbb E_{Z}\Big[e^{i t_2\left(\mu+\sigma Z/\sqrt{K}\right)} \mid \sigma,\mu\Big]\Bigg|
\le \Bigg|\mathbb E_{Z}\Big[e^{i t_2\left(\sigma Z/\sqrt{K}\right)} \mid \sigma,\mu\Big]\Bigg| \le \exp\{-\underline L t^2_2/(2K)\}, 
\]
almost surely under $H$.
Combining the forgoing displays with Assumption~\ref{assu:compact_2d} and Lemma S.8 of \citet{ignatiadis2025empirical} gives a constant $\mathfrak L(K,p,c)>0$ such that
\begin{align}
\label{eq:huge_cound_mf_tenasform}
\big|\mathcal M_{c,K-p+1}(t_1,t_2;H)\big|
\lesssim_{K,p,c,M,\underline L,\wb U}
|t_1|^{\,c+(K-p)/2-1}\exp(-\pi|t_1|/2)\exp(-\underline L^2t^2_2/2K),
\end{align}
for all $|t_1|>\mathfrak L(K,p,c)$.

Observe that for any two priors $H_1$ and $H_2$ satisfying the assumptions of the lemma, we have
\begin{align}
    &\int_{0}^{\infty}\int_{-\infty}^\infty 
    x^{2c-1}\Big(f_{H_1,K-p+1}(x,y)-f_{H_2,K-p+1}(x,y)\Big)^2\,\dd y\,\dd x \nonumber\\
    &=\frac{1}{(2\pi)^2}\int_{-\infty}^\infty\int_{-\infty}^\infty
    \Big|\mathcal M_{c,K-p+1}(t_1,t_2;H_1)-\mathcal M_{c,K-p+1}(t_1,t_2;H_2)\Big|^2\,\dd t_2\,\dd t_1 \nonumber\\
    &=\frac{1}{(2\pi)^2}\int_{-\infty}^\infty\int_{-\infty}^\infty
    |r_{K,p,c}(t_1)|^2\,
    \Big|\mathcal M_{c,K-p}(t_1,t_2;H_1)-\mathcal M_{c,K-p}(t_1,t_2;H_2)\Big|^2\,\dd t_2\,\dd t_1.
\end{align}
For any $T>0$, we can decompose the integral into two regions:
\begin{align}
    \mathrm{I}&=\frac{1}{(2\pi)^2}\int_{|t_1|\ge T}\int_{-\infty}^\infty
    |r_{K,p,c}(t_1)|^2\,
    \Big|\mathcal M_{c,K-p}(t_1,t_2;H_1)-\mathcal M_{c,K-p}(t_1,t_2;H_2)\Big|^2\,\dd t_2\,\dd t_1,\\
    \mathrm{II}&=\frac{1}{(2\pi)^2}\int_{|t_1|<T}\int_{-\infty}^\infty
    |r_{K,p,c}(t_1)|^2\,
    \Big|\mathcal M_{c,K-p}(t_1,t_2;H_1)-\mathcal M_{c,K-p}(t_1,t_2;H_2)\Big|^2\,\dd t_2\,\dd t_1.
\end{align}
Noting that
\[
\int_{-\infty}^\infty \exp(-\underline L t^2_2/2)\,\dd t_2 \lesssim_{\underline L} 1,
\]
we may apply the same techniques as in the proof of Lemma 12 of \citet{ignatiadis2025empirical} along with \eqref{eq:bound_r_low_t} and \eqref{eq:huge_cound_mf_tenasform} to deduce that for any $T_0>\mathfrak r_{K,p,c}\vee \mathfrak L(K,p,c) \vee (2c+K-p)$ 
\begin{align}
    \mathrm{I}& \lesssim_{K,p,c,M,\underline L,\wb U} T^{2c+K-p-2}_0\exp(-\pi T_0),\\
    \mathrm{II}& \lesssim_{K,p,c,M,\underline L,\wb U}T_0\int_{0}^\infty\int_{-\infty}^\infty 
    x^{2c-1}\Big(f_{H_1,K-p}(x,y)-f_{H_2,K-p}(x,y)\Big)^2\,\dd x\,\dd y.
\end{align}
Combining the two aforementioned relations, we have
\begin{align}
    &\int_{0}^\infty\int_{-\infty}^\infty 
    x^{2c-1}\Big(f_{H_1,K-p+1}(x,y)-f_{H_2,K-p+1}(x,y)\Big)^2\,\dd x\,\dd y \nonumber\\
    &\quad\lesssim_{K,p,c,M,\underline L,\wb U} T^{2c+K-p-2}_0\exp(-\pi T_0)+T_0\int_{0}^\infty\int_{-\infty}^\infty 
    x^{2c-1}\Big(f_{H_1,K-p}(x,y)-f_{H_2,K-p}(x,y)\Big)^2\,\dd x\,\dd y\\
    &\quad\lesssim_{K,p,c,M,\underline L,\wb U} \exp(-\pi T_0)+T_0\int_{0}^\infty\int_{-\infty}^\infty 
    x^{2c-1}\Big(f_{H_1,K-p}(x,y)-f_{H_2,K-p}(x,y)\Big)^2\,\dd x\,\dd y,
\end{align}
where the final inequality follows from (S2) of \cite{ignatiadis2025empirical}.
Finally, plugging in $c=1/2$, taking $T_0=\mathfrak r_{K,p,c}\vee \mathfrak L(K,p,c) \vee (2c+K-p) \vee |\log \rho_{K,p}|$ in the foregoing calculations, where
\(
\smash{\rho_{K,p}:=\left\| f_{H_1,K-p} - f_{H_2,K-p} \right\|^2_{L_2},}
\)
we obtain the following by retracing the steps of the proof of Lemma 12 of \citet{ignatiadis2025empirical}:
\[
\left\| f_{H_1,K-p+1} - f_{H_2,K-p+1} \right\|^2_{L_2}
\;\le\; \Big(1+\big|\log \left\| f_{H_1,K-p} - f_{H_2,K-p} \right\|^2_{L_2}\big|\Big)
\cdot \left\| f_{H_1,K-p} - f_{H_2,K-p} \right\|^2_{L_2}.
\] 
Finally, since
\[
\left\| f_{H_1,K-p} - f_{H_2,K-p} \right\|^2_{L_2} \lesssim \mathcal H^2(f_{H_1,K-p},f_{H_2,K-p}),
\]
and the map $x \mapsto x(1+|\log x|)$ is monotone increasing, the assertion of Lemma~\ref{lem:l_2_one_dof} follows.

\subsubsection{Proof of Lemma \ref{lemm:lem_ratio_bound_2d}}
Define
\begin{align}
\label{eq:def_a_eps_here}
\mathtt A_n := M + \sqrt{\frac{4\wb U}{K}\log n}, 
\qquad
\mathtt B_n := \wb U \cdot \max \left\{ 1, \frac{8}{K-p}\log n \right\}.
\end{align}
Recall that
\[
D_i(\widehat{H}_\star) = \frac{D_i(H)}{2} + \frac{D_i(\wh H)}{2}\ge \frac{1}{2}\max\left\{D_i(H),D_i(\wh H)\right\}.
\]
Also recall the event $\mathcal A$ defined in \eqref{eq:def_cal_A_2d}.
Using the above inequality together with the fact that
\begin{align}
\label{eq:its_n}
\frac{N_i(z, H')}{D_i(H')} \in [0,1] \qquad \text{for any distribution } H',
\end{align}
we conclude that
\begin{align}
\label{eq:greater_than_t_2d_0}
&\mathbb E\left[\sup_{|z| \ge \underline z}\left|\frac{N_i(z,H)-N_i(z,\wh H)}{D_i(\widehat H_\star)}\right|\cdot\mathds{1}(\mathcal A)\right]\nonumber\\
&\le 4\,\mathbb E\left[\mathds{1}\!\left(S_i^2 \ge\Bt_n\right)\right]
+ 4\,\mathbb E\left[\mathds{1}\!\left(|A_i| > \At_n\right)\right]
\nonumber\\
&\quad +\mathbb E\left[\sup_{|z| \ge \underline z}\left|\frac{N_i(z,H)-N_i(z,\wh H)}{D_i(\widehat H_\star)}\right|
\mathds{1}\!\left(|A_i| \le \At_n,\,S_i^2 \le \Bt_n\right)\cdot\mathds{1}(\mathcal A)\right].
\end{align}
By definition of $\Bt_n$ and $\At_n$, we can conclude that
\[
\frac{1}{n}\sum_{i=1}^n\mathbb E\left[\mathds{1}(S^2_i \ge \Bt_n)\right] \le  \frac{1}{n}, \quad \mbox{and} \quad \frac{1}{n}\sum_{i=1}^n\mathbb E\left[\mathds{1}(|A_i|>\At_n)\right] \le  \frac{1}{n}.
\]
Therefore, it suffices to control
\[
\frac{1}{n}\sum_{i=1}^n\mathbb E\left[ \sup_{|z| \ge \underline z}\left|\frac{N_i(z, {H})-N_i(z, \wh H)}{D_i(\widehat{H}_\star)}\right|~\mathds{1}(|A_i| \le \At_n,S^2_i \le \Bt_n)\cdot\mathds{1}(\mathcal A)\right].
\]
By \eqref{eq:def_n_i_z} and Theorem \ref{thm:tweedie_2_d}, we have
\begin{align}
    &\sup_{|z| \ge \underline z}\left|N_i(z, {H})-N_i(z, \wh H)\right|\\
    & =\sup_{|z| \ge \underline z}\Bigg|\int_{0}^{\infty}\frac{C_{K,p}\,(s^2)^{(K-p)/2-1}(t)^{-(K-p-1)}}{\sqrt{(K-p+1)t^2 - (K-p)s^2}}\mathds{1} \left\{ t^2 \ge \frac{(K-p)s^2 + (z^2/\nu^2)}{K-p + 1} \right\}\\
    & \hskip 5em \times (f_{H,K-p+1}(t^2,a)-f_{\wh H,K-p+1}(t^2,a))\,\dd t^2\Bigg|,
\end{align}
where the absolute constant $C_{K,p}>0$ is defined in the referred theorem.
Let us take $T=16 \wb U\cdot \max \left\{ 1, (K-p)^{-1}\log n \right\}$ and decompose the integral in the foregoing expression as follows:
\begin{align}
     &\sup_{|z| \ge \underline z}\left|N_i(z, {H})-N_i(z, \wh H)\right|\\
    & \le \sup_{|z| \ge \underline z}\Bigg|\int_{0}^{T}\frac{C_{K,p}\,(S^2_i)^{(K-p)/2-1}(t)^{-(K-p-1)}}{\sqrt{(K-p+1)t^2 - (K-p)S^2_i}}\mathds{1} \left\{ t^2 \ge \frac{(K-p)S^2_i + (z^2/\nu^2)}{K-p + 1} \right\}\\
    & \hskip 5em \times (f_{H,K-p+1}(t^2,A_i)-f_{\wh H,K-p+1}(t^2,A_i))\,\dd t^2\Bigg|\\
    & + \sup_{|z| \ge \underline z}\Bigg|\int_{T}^{\infty}\frac{C_{K,p}\,(S^2_i)^{(K-p)/2-1}(t)^{-(K-p-1)}}{\sqrt{(K-p+1)t^2 - (K-p)S^2_i}}\mathds{1} \left\{ t^2 \ge \frac{(K-p)S^2_i + (z^2/\nu^2)}{K-p + 1} \right\}\\
    & \hskip 5em \times (f_{H,K-p+1}(t^2,A_i)-f_{\wh H,K-p+1}(t^2,A_i))\,\dd t^2\Bigg|.
\end{align}
 
Using Lemma S.1 of \cite{ignatiadis2025empirical}, we have that for all 
distributions $\wt H$ on $\mathbb R \times \mathbb R_{\ge 0}$ and all $t^2 \ge T$,
\[
|f_{\wt H,K-p+1}(t^2,a)| \lesssim_{K,p,M,\wb U,\underline L}\exp\left(-\frac{(K-p)t^2}{8\wb U}\right), \qquad \text{for all } a \in \mathbb R.
\]
If, $K-p>2$, if $S_i^2 \le \Bt_n$, retracing the proof of Lemma S.9 of 
\cite{ignatiadis2025empirical} yields
\begin{align}
&\int_{T}^{\infty}\frac{C_{K,p}^2\,(t^2)^{-(K-p-1)}}{(K-p+1)t^2 - (K-p)S_i^2}\mathds{1} \!\left\{t^2 \ge \frac{(K-p)S_i^2 + (z^2/\nu^2)}{K-p + 1}\right\}\,\dd t^2\\
&\lesssim_{K,p,\nu}\frac{1}{z^2}\min\left\{\left\{(K-p)S_i^2 + (z^2/\nu^2)\right\}^{K-p-2},T^{-(K-p-2)}\right\}.
\end{align}
If $K-p=2$, then retracing the arguments of \cite{ignatiadis2025empirical}, we get
\begin{align}
\int_{T}^{\infty}\frac{C_{K,p}^2\,(t^2)^{-1}}{3t^2 - 2S_i^2}\mathds{1} \!\left\{t^2 \ge \frac{2S_i^2 + (z^2/\nu^2)}{3}\right\}\,\dd t^2\lesssim_{K,p,\nu}\frac{1}{S^2_i}\log\left(\frac{2\nu^2S^2_i}{z^2}+1\right)\lesssim \frac{1}{z^2}.
\end{align}
This implies, for any $K-p \ge 2$, we have
\begin{align}
&\int_{T}^{\infty}\frac{C_{K,p}^2\,(t^2)^{-(K-p-1)}}{(K-p+1)t^2 - (K-p)S_i^2}\mathds{1} \!\left\{t^2 \ge \frac{(K-p)S_i^2 + (z^2/\nu^2)}{K-p + 1}\right\}\,\dd t^2\\
&\lesssim \frac{1}{z^2}\min\left\{\left\{(K-p)S_i^2 + (z^2/\nu^2)\right\}^{K-p-2},T^{-(K-p-2)}\right\}.
\end{align}

Using the Cauchy--Schwarz inequality together with the above bound and the exponential tail estimate, we obtain
\begin{align}
\label{eq:greater_than_t_2d}
&\sup_{|z| \ge \underline z}\Bigg|\int_{T}^{\infty}
\frac{C_{K,p}\,(S_i^2)^{(K-p)/2-1}(t^2)^{-(K-p-1)}}{\sqrt{(K-p+1)t^2 - (K-p)S_i^2}}\mathds{1} \!\left\{t^2 \ge \frac{(K-p)S_i^2 + (z^2/\nu^2)}{K-p + 1}\right\}\\
&\hskip 15em\times\bigl(f_{H,K-p+1}(t^2,A_i)-f_{\wh H,K-p+1}(t^2,A_i)\bigr)\,\dd t^2\Bigg|\nonumber\\
&\lesssim_{K,p,M,\wb U,\underline L,\nu,\underline z}
\left\{\left(\frac{S_i^2}{T}\right)^{(K-p)/2-1} \vee 1\right\} \times
\exp\left(-\frac{(K-p)T}{8\wb U}\right).
\end{align}
Since $S_i^2 \le \Bt_n$ and $T \ge \frac{16\wb U}{K-p}\log n$, the polynomial term is bounded and the exponential term satisfies
\[
\exp\left(-\frac{(K-p)T}{8\wb U}\right)\le\exp\left(-2\log n\right)=n^{-2}.
\]
Consequently,
\begin{align}
&\sup_{|z| \ge \underline z}\Bigg|\int_{T}^{\infty}
\frac{C_{K,p}\,(S_i^2)^{(K-p)/2-1}(t^2)^{-(K-p-1)}}{\sqrt{(K-p+1)t^2 - (K-p)S_i^2}}\mathds{1} \!\left\{t^2 \ge \frac{(K-p)S_i^2 + (z^2/\nu^2)}{K-p + 1}\right\}\\
&\hskip 15em\times\bigl(f_{H,K-p+1}(t^2,A_i)-f_{\wh H,K-p+1}(t^2,A_i)\bigr)\,\dd t^2\Bigg|\\
&\hskip 15em \lesssim_{K,p,M,\wb U,\underline L,\nu,\underline z} \frac{1}{n^2},\qquad\text{for all } S_i^2 \le \Bt_n.
\end{align}
This implies
\begin{align}
\label{eq:greater_than_t_2d_2}
    &\frac{1}{n}\sum_{i=1}^n\mathbb E\left[\frac{1}{D_i(\wh H_\star)} \times\sup_{|z| \ge \underline z}\Bigg|\int_{T}^{\infty}\frac{C_{K,p}\,(S^2_i)^{(K-p)/2-1}(t)^{-(K-p-1)}}{\sqrt{(K-p+1)t^2 - (K-p)S^2_i}}\mathds{1} \left\{ t^2 \ge \frac{(K-p)S^2_i + (z^2/\nu^2)}{K-p + 1} \right\}\right.\nonumber\\
    & \left.\hskip 10em \times (f_{H,K-p+1}(t^2,A_i)-f_{\wh H,K-p+1}(t^2,A_i))\,\dd t^2\Bigg|~\mathds{1}(|A_i| \le \At_n,S^2_i \le \Bt_n)\cdot\mathds{1}(\mathcal A)\right] \nonumber \\
    & \hskip 15em\lesssim_{K,p,M,\wb U,\underline L,\nu,\underline z}\frac{(\log n)^{3/2}}{n^2}.
\end{align}
Next, we consider the subclass of densities $\mathcal F_{\jt,\mathcal A} \subset \mathcal F_\jt$ where
\begin{align}
    \mathcal F_{\jt,\mathcal A}:=\left\{(f_{H',K-p},f_{H',K-p+1}) \in \mathcal{F}_\jt : \mathcal H^2\left(f_{H',K-p}, f_{H,K-p}\right) \le C \frac{(\log n)^3}{n} \right\},
\end{align}
for $C>0$ defined Theorem \ref{thm:double_hellinger}. Take $\mathfrak{R}= (0,\Bt_n] \times [-\At_n,\At_n]$ and consider a $\eta=1/n$ cover of the class of densities in $\mathcal F_{\jt,\mathcal A}$ under the semi-norm $d_{\mathfrak{R}}$ (defined in \eqref{eq:def_f_joint}). Let this cover be given by 
\begin{align}
\label{eq:cover_imp_N}
    \mathcal S=\{(f_{H_1,K-p},f_{H_1,K-p+1}),\ldots,(f_{H_{N},K-p},f_{H_{N},K-p+1})\}.
\end{align}
Then, if $\mathcal A$ holds, there exists $j \in \{1,\ldots,N\}$ such that
\begin{align}
\label{eq:ncover_cons}
&\sup_{(t^2,a) \in \mathfrak R}\left|f_{\wh H,K-p+1}(t^2,a)-f_{H_j,K-p+1}(t^2,a)\right| \le \frac{1}{n}, \quad \mbox{and}\\
&\sup_{(t^2,a) \in \mathfrak R}\left|f_{\wh H,K-p}(t^2,a)-f_{H_j,K-p}(t^2,a)\right| \le \frac{1}{n}. 
\end{align}
Define $\wh H_j=(H_j+H)/2$ and 
\begin{align}
\label{eq:n_i_z_t_def}
N_i(z,G;T)=&\int_{0}^{T}\frac{C_{K,p}\,(S^2_i)^{(K-p)/2-1}(t)^{-(K-p-1)}}{\sqrt{(K-p+1)t^2 - (K-p)S^2_i}}\\
&~~~~~~~~~~~~~~~\times\mathds{1} \left\{ t^2 \ge \frac{(K-p)S^2_i + (z^2/\nu^2)}{K-p + 1} \right\}f_{G,K-p+1}(t^2,A_i)\,\dd t^2,\nonumber\\
\end{align}
for any mixing measure $G$. 

Using \eqref{eq:greater_than_t_2d_2}, we can show that
\begin{align}
\label{eq:greater_than_t_2d_2_h}
    &\frac{1}{n}\sum_{i=1}^n\mathbb E\left[\frac{1}{D_i(\wh H_\star)} \times\sup_{|z| \ge \underline z}\Bigg|\int_{T}^{\infty}\frac{C_{K,p}\,(S^2_i)^{(K-p)/2-1}(t)^{-(K-p-1)}}{\sqrt{(K-p+1)t^2 - (K-p)S^2_i}}\mathds{1} \left\{ t^2 \ge \frac{(K-p)S^2_i + (z^2/\nu^2)}{K-p + 1} \right\}\right.\nonumber\\
    & \left.\hskip 10em \times (f_{H_j,K-p+1}(t^2,A_i)-f_{\wh H,K-p+1}(t^2,A_i))\,\dd t^2\Bigg|~\mathds{1}(|A_i| \le \At_n,S^2_i \le \Bt_n)\cdot\mathds{1}(\mathcal A)\right] \nonumber \\
    & \hskip 15em\lesssim_{K,p,M,\wb U,\underline L,\nu,\underline z}\frac{(\log n)^{3/2}}{n^2},
\end{align}
and
\begin{align}
\label{eq:greater_than_t_2d_2_l}
    &\frac{1}{n}\sum_{i=1}^n\mathbb E\left[\frac{1}{D_i(\wh H_\star)} \times\sup_{|z| \ge \underline z}\Bigg|\int_{T}^{\infty}\frac{C_{K,p}\,(S^2_i)^{(K-p)/2-1}(t)^{-(K-p-1)}}{\sqrt{(K-p+1)t^2 - (K-p)S^2_i}}\mathds{1} \left\{ t^2 \ge \frac{(K-p)S^2_i + (z^2/\nu^2)}{K-p + 1} \right\}\right.\nonumber\\
    & \left.\hskip 10em \times (f_{H_j,K-p+1}(t^2,A_i)-f_{H,K-p+1}(t^2,A_i))\,\dd t^2\Bigg|~\mathds{1}(|A_i| \le \At_n,S^2_i \le \Bt_n)\cdot\mathds{1}(\mathcal A)\right] \nonumber \\
    & \hskip 15em\lesssim_{K,p,M,\wb U,\underline L,\nu,\underline z}\frac{(\log n)^{3/2}}{n^2}.
\end{align}
Using the foregoing inequalities, we can decompose the left-hand side of \eqref{eq:main_n_term} as follows:
\begin{align}
\label{eq:reduction_2d_exp_max}
  & \frac{1}{n}\mathbb E\left[\sum_{i=1}^{n} \sup_{|z| \ge \underline z}\left|\frac{N_i(z, {H})-N_i(z, \wh H)}{D_i(\widehat{H}_\star)}\right|~\mathds{1}(|A_i| \le \At_n,S^2_i \le \Bt_n) \cdot \mathds{1}(\mathcal A)\right]\\
  & \lesssim \frac{1}{n}\sum_{i=1}^{n}\mathbb E\left[\sup_{|z| \ge \underline z}\left(\frac{N_i(z, \wh H)}{D_i(\widehat{H}_\star)}\right)\left|\frac{D_i(\widehat{H}_\star)-D_i(\widehat{H}_j)}{D_i(\widehat{H}_j)}\right|~\mathds{1}(|A_i| \le \At_n,S^2_i \le \Bt_n) \cdot \mathds{1}(\mathcal A)\right]\\
  & \hskip 1em + \frac{1}{n}\sum_{i=1}^{n}\mathbb E\left[ \sup_{|z| \ge \underline z}\left|\frac{N_i(z,H_j)-N_i(z,\wh H)}{D_i(\wh H_j)}\right|~\mathds{1}(|A_i| \le \At_n,S^2_i \le \Bt_n) \cdot \mathds{1}(\mathcal A)\right]\\
  & \hskip 1em + \frac{1}{n}\sum_{i=1}^{n}\mathbb E\left[ \sup_{|z| \ge \underline z}\left|\frac{N_i(z,H_j)-N_i(z,H)}{D_i(\wh H_j)}\right|~\mathds{1}(|A_i| \le \At_n,S^2_i \le \Bt_n) \cdot \mathds{1}(\mathcal A)\right]\\
  & \hskip 1em + \frac{1}{n}\sum_{i=1}^{n}\mathbb E\left[\sup_{|z| \ge \underline z}\left(\frac{N_i(z, H)}{D_i(\widehat{H}_\star)}\right)\left|\frac{D_i(\widehat{H}_\star)-D_i(\widehat{H}_j)}{D_i(\widehat{H}_j)}\right|~\mathds{1}(|A_i| \le \At_n,S^2_i \le \Bt_n) \cdot \mathds{1}(\mathcal A)\right]\\
  & \lesssim \frac{1}{n}\sum_{i=1}^{n}\mathbb E\left[\sup_{|z| \ge \underline z}\left(\frac{N_i(z, \wh H)}{D_i(\widehat{H}_\star)}\right)\left|\frac{D_i(\widehat{H}_\star)-D_i(\widehat{H}_j)}{D_i(\widehat{H}_j)}\right|~\mathds{1}(|A_i| \le \At_n,S^2_i \le \Bt_n) \cdot \mathds{1}(\mathcal A)\right]\\
  & \hskip 1em + \frac{1}{n}\sum_{i=1}^{n}\mathbb E\left[\sup_{|z| \ge \underline z}\left(\frac{N_i(z, H)}{D_i(\widehat{H}_\star)}\right)\left|\frac{D_i(\widehat{H}_\star)-D_i(\widehat{H}_j)}{D_i(\widehat{H}_j)}\right|~\mathds{1}(|A_i| \le \At_n,S^2_i \le \Bt_n) \cdot \mathds{1}(\mathcal A)\right]\\
  &\hskip 1em +\frac{1}{n}\sum_{i=1}^{n}\mathbb E\left[ \sup_{|z| \ge \underline z}\left|\frac{N_i(z, {H}_j;T)-N_i(z, \wh H;T)}{D_i(\widehat{H}_j)}\right|~\mathds{1}(|A_i| \le \At_n,S^2_i \le \Bt_n) \cdot \mathds{1}(\mathcal A)\right]\\
  &\hskip 1em +\mathbb E\left[\frac{1}{n}\sum_{i=1}^{n}\left\{\sup_{|z| \ge \underline z}\left|\frac{N_i(z, {H}_j;T)-N_i(z, H;T)}{D_i(\widehat{H}_j)}\right|~\mathds{1}(|A_i| \le \At_n,S^2_i \le \Bt_n)\right\}\right]+\frac{(\log n)^{3/2}}{n^2}\\
  & \lesssim \frac{1}{n}\sum_{i=1}^{n}\mathbb E\left[\sup_{|z| \ge \underline z}\left(\frac{N_i(z, \wh H)}{D_i(\widehat{H}_\star)}\right)\left|\frac{D_i(\widehat{H}_\star)-D_i(\widehat{H}_j)}{D_i(\widehat{H}_j)}\right|~\mathds{1}(|A_i| \le \At_n,S^2_i \le \Bt_n) \cdot \mathds{1}(\mathcal A)\right]\\
  & \hskip 1em + \frac{1}{n}\sum_{i=1}^{n}\mathbb E\left[\sup_{|z| \ge \underline z}\left(\frac{N_i(z, H)}{D_i(\widehat{H}_\star)}\right)\left|\frac{D_i(\widehat{H}_\star)-D_i(\widehat{H}_j)}{D_i(\widehat{H}_j)}\right|~\mathds{1}(|A_i| \le \At_n,S^2_i \le \Bt_n) \cdot \mathds{1}(\mathcal A)\right]\\
  &\hskip 1em +\frac{1}{n}\sum_{i=1}^{n}\mathbb E\left[ \sup_{|z| \ge \underline z}\left|\frac{N_i(z, {H}_j;T)-N_i(z, \wh H;T)}{D_i(\widehat{H}_j)}\right|~\mathds{1}(|A_i| \le \At_n,S^2_i \le \Bt_n) \cdot \mathds{1}(\mathcal A)\right]\\
  &\hskip 1em +\mathbb E\left[\sup_{k \in [N]}\left(\frac{1}{n}\sum_{i=1}^{n}\left\{\sup_{|z| \ge \underline z}\left|\frac{N_i(z, {H}_k;T)-N_i(z, H;T)}{D_i(\widehat{H}_k)}\right|~\mathds{1}(|A_i| \le \At_n,S^2_i \le \Bt_n)\right\}\right)\right]+\frac{(\log n)^{3/2}}{n^2}.
\end{align}
If $|a| \le \At_n$, then
\begin{align}
\label{eq:truncated_int_1_hell}
&\int_{0}^{T}\frac{C_{K,p}(s^2)^{(K-p-2)/2}(t)^{-(K-p-1)}}{\sqrt{(K-p+1)t^2 - (K-p)s^2}}\mathds{1} \left\{ t^2 \ge \frac{(K-p)s^2 + (z^2/\nu^2)}{K-p + 1} \right\}\\
& \hskip 15em \times\left|f_{\wh H,K-p+1}(a, t^2)-f_{H_j,K-p+1}(a, t^2)\right|\,\dd t^2\\
    & \le C_{K,p}(s^2)^{(K-p-2)/2}\,\left(\int_{0}^{T}\frac{(t^2)^{-(K-p-1)}}{(K-p+1)t^2 - (K-p)s^2}\mathds{1} \left\{ t^2 \ge \frac{(K-p)s^2 + (z^2/\nu^2)}{K-p + 1} \right\}\,\dd t^2\right)^{1/2} \times\\
    & \hskip 4em \left(\int_{0}^{T}(f_{\wh H,K-p+1}(a, t^2)-f_{H_j,K-p+1}(a, t^2))^2\,\dd t^2\right)^{1/2}\\
    & \lesssim_{K,p,M,\wb U,\underline L,\nu,\underline z}  \frac{C_{K,p}(s^2)^{(K-p-2)/2}}{\sqrt{z}((K-p+1)s^2+(z^2/\nu^2))^{\frac{K-p-2}{2}}} \times \frac{(\log n)^{1/2}}{n},
\end{align}
where the last inequality follows from the Cauchy Schwartz inequality, \eqref{eq:ncover_cons}, and the techniques outlined in the proof of Lemma S9 of \cite{ignatiadis2025empirical}. Therefore, for all $|z|\ge \underline z$ we have
\begin{align}
\label{eq:closet_cover_2d}
 & \left|N_i(z, {H}_j;T)-N_i(z, \wh H;T)\right| \lesssim_{K,p,M,\wb U,\underline L,\nu,\underline z}  \frac{(\log n)^{1/2}}{n}, \quad \mbox{if $|A_i| \le \At_n, S^2_i \le \Bt_n$ and $\mathcal A$ holds.}\nonumber\\
\end{align}

Observe that $\smash{D_i(\widehat{H}_\star) \ge \max\left\{\frac{D_i(\wh H)}{2},\frac{D_i(H)}{2}\right\}}$ implies
\[
\sup_{|z| \ge \underline z}\left(\frac{N_i(z,\wh H)}{D_i(\wh H_\star)}\right) \le 2, \quad \mbox{and} \quad \sup_{|z| \ge \underline z}\left(\frac{N_i(z, H)}{D_i(\widehat{H}_\star)}\right) \le 2, \quad \mbox{for all $i=1,\ldots,n$.}
\]
Using the foregoing inequalities and \eqref{eq:ncover_cons}, we have
\begin{align}
\label{eq:its_n_1}
&\frac{1}{n}\sum_{i=1}^{n}\mathbb E\left[\sup_{|z| \ge \underline z}\left(\frac{N_i(z,\mathrm H)}{D_i(\widehat{H}_\star)}\right)\left|\frac{D_i(\widehat{H}_\star)-D_i(\widehat{H}_j)}{D_i(\widehat{H}_j)}\right|~\mathds{1}(|A_i| \le \At_n,S^2_i \le \Bt_n) \cdot \mathds{1}(\mathcal A)\right] \\
&\lesssim_{K,p,M,\wb U,\underline L,\underline z} \frac{(\log n)^{3/2}}{n}, \quad \mbox{where $\mathrm H \in \{\wh H,H\}$.}
\end{align}
Furthermore, using \eqref{eq:closet_cover_2d}, we have
\begin{align}
\label{eq:its_n_2}
\frac{1}{n}\sum_{i=1}^{n}\mathbb E\left[ \sup_{|z| \ge \underline z}\left|\frac{N_i(z, {H}_j;T)-N_i(z, \wh H;T)}{D_i(\widehat{H}_j)}\right|~\mathds{1}(|A_i| \le \At_n,S^2_i \le \Bt_n) \cdot \mathds{1}(\mathcal A)\right] \lesssim_{K,p,M,\wb U,\underline L,\underline z} \frac{(\log n)^{2}}{n}.\nonumber\\
\end{align}
Combining \eqref{eq:its_n_1} and \eqref{eq:its_n_2} with \eqref{eq:reduction_2d_exp_max}, we get
  \begin{align}
  \label{eq:reduction_2d_exp_max_2}
  & \frac{1}{n}\mathbb E\left[\sum_{i=1}^{n} \sup_{|z| \ge \underline z}\left|\frac{N_i(z, {H})-N_i(z, \wh H)}{D_i(\widehat{H}_\star)}\right|~\mathds{1}(|A_i| \le \At_n,S^2_i \le \Bt_n) \cdot \mathds{1}(\mathcal A)\right]\\
  & \lesssim_{K,p,M,\wb U,\underline L,\underline z}  \frac{(\log n)^2}{n}+\mathbb E\left[\max_{k=1,\ldots,N}V_k\left((S^2_1,A_1),\ldots,(S^2_n,A_n)\right)\right],
\end{align}
where for any $k \in \{1,\ldots,N\}$
\[
V_k\left((S^2_1,A_1),\ldots,(S^2_n,A_n)\right) = \frac{1}{n}\sum_{i=1}^{n}\left\{\sup_{|z| \ge \underline z}\left|\frac{N_i(z, {H}_k;T)-N_i(z, H;T)}{D_i(\widehat{H}_k)}\right|~\mathds{1}(|A_i| \le \At_n,S^2_i \le \Bt_n)\right\}.
\]
Therefore, it suffices to control
\[
\mathbb E\left[\max_{k=1,\ldots,N}V_k\left((S^2_1,A_1),\ldots,(S^2_n,A_n)\right)\right].
\]
Next, observe that since
\[
\sup_{|z| \ge \underline z}\left|\frac{N_i(z, {H}_j;T)-N_i(z, H;T)}{D_i(\widehat{H}_j)}\right| \le 2, \quad \mbox{for all $i \in \{1,\ldots,n\}$ and $j \in \{1,\ldots,N\}$.}
\]
Therefore, for two sequences of points 
\[
\mathfrak a=\left\{(s^2_1,a_1),\ldots,(s^2_i,a_i),\ldots,(s^2_n,a_n)\right\}, \quad \mbox{and} \quad 
\mathfrak a^\prime=\left\{(s^2_1,a_1),\ldots,((s^\prime_i)^2,a^\prime_i),\ldots,(s^2_n,a_n)\right\}
\]
differeing only at the $i$-th entry we have
\begin{align}
 |V_k\left(\mathfrak a\right)-V_k\left(\mathfrak a^\prime\right)| \le \frac{2}{n}.   
\end{align}
Using the bounded differences inequality \cite[Theorem 6.2]{10.1093/acprof:oso/9780199535255.001.0001}, we get for all $j \in \{1,\ldots,N\}$
\[
\mathbb P\Big[\left|V_k\left((S^2_1,A_1),\ldots,(S^2_n,A_n)\right)-\mathbb E\left[V_k\left((S^2_1,A_1),\ldots,(S^2_n,A_n)\right)\right]\right|>t\Big] \le \exp\left(-\frac{nt^2}{4}\right).
\]
Therefore, we have
\begin{align}
\label{eq:2d_bdd_diff}
    &\mathbb E\Bigg[\max_{k \in \{1,\ldots,N\}}\left|V_k\left((S^2_1,A_1),\ldots,(S^2_n,A_n)\right)-\mathbb E\left[V_k\left((S^2_1,A_1),\ldots,(S^2_n,A_n)\right)\right]\right|\Bigg]\\
    &=\int_0^\infty\mathbb P\left[\max_{k \in \{1,\ldots,N\}}\left|V_k\left((S^2_1,A_1),\ldots,(S^2_n,A_n)\right)-\mathbb E\left[V_k\left((S^2_1,A_1),\ldots,(S^2_n,A_n)\right)\right]\right|>t\right]\,\dd t\\
    & = x_0 +N\int_{x_0}^\infty\exp\left(-\frac{nt^2}{4}\right)\,\dd t\le x_0 +\frac{2N}{nx_0}\exp\left(-\frac{nx^2_0}{4}\right), \quad \mbox{for any $x_0> 0$.}
\end{align}
Taking $x_0 \asymp_{K,p,M,\wb U,\underline L,\underline z} \sqrt{\frac{8\log N}{n}}$. Using Lemma \ref{lem:metric_entropy_2d} and some algebraic manipulation, we get
\begin{align}
\label{eq:2d_bdd_diff_1}
\mathbb E\Bigg[\max_{j \in \{1,\ldots,N\}}\left|V_j\left((S^2_1,A_1),\ldots,(S^2_n,A_n)\right)-\mathbb E\left[V_j\left((S^2_1,A_1),\ldots,(S^2_n,A_n)\right)\right]\right|\Bigg] \lesssim_{K,p,M,\wb U,\underline L,\underline z} \frac{(\log n)^{3/2}}{\sqrt{n}}.
\end{align}
Finally, we must control
\[
\max_{j \in \{1,\ldots,N\}}\mathbb E\left[V_j\left((S^2_1,A_1),\ldots,(S^2_n,A_n)\right)\right].
\]
In that direction, proceeding as in the proof of Lemma S9 of \cite{ignatiadis2025empirical} and using the definition of $V_1,\ldots,V_N$, we can get 
\begin{align}
\label{eq:its_n_3}
    &\mathbb E\left[V_j\left((S^2_1,A_1),\ldots,(S^2_n,A_n)\right)\right]\nonumber\\
    & \lesssim 2 \int_0^{\Bt_n}\int_{-\At_n}^{\At_n}\sup_{|z|\ge \underline z}\left|\int_{0}^{T}\frac{C_{K,p}\,(s^2)^{(K-p)/2-1}(t)^{-(K-p-1)}}{\sqrt{(K-p+1)t^2 - (K-p)s^2}} \times \right.\nonumber\\
    & \hskip 8em \left.\mathds{1} \left\{ t^2 \ge \frac{(K-p)s^2 + (z^2/\nu^2)}{K-p + 1} \right\}(f_{H,K-p+1}(t^2,a)-f_{H_j,K-p+1}(t^2,a))\,\dd t^2\right|\dd a\,\dd s^2\nonumber\\
    & \lesssim 2 \int_0^{\Bt_n}\int_{-\At_n}^{\At_n}\int_{0}^{T}\left|\frac{C_{K,p}\,(s^2)^{(K-p)/2-1}(t)^{-(K-p-1)}}{\sqrt{(K-p+1)t^2 - (K-p)s^2}} \times \right.\nonumber\\
    & \hskip 8em \left.\mathds{1} \left\{ t^2 \ge \frac{(K-p)s^2 + (\underline z^2/\nu^2)}{K-p + 1} \right\}(f_{H,K-p+1}(t^2,a)-f_{H_j,K-p+1}(t^2,a))\right|\,\dd t^2\,\dd a\,\dd s^2\nonumber\\
    & \lesssim_{K,p}\int_0^{\Bt_n} \Bigg\{\left(\int_{-\At_n}^{\At_n}\int_{0}^{T}\frac{(t^2)^{-(K-p-1)}}{(K-p+1)t^2 - (K-p)s^2}\mathds{1} \left\{ t^2 \ge \frac{(K-p)s^2 + (\underline z^2/\nu^2)}{K-p + 1} \right\}\,\dd t^2\,\dd a\right)\Bigg\}^{1/2}\nonumber\\
    &\hskip 15em \times (s^2)^{(K-p)/2-1} \left\| f_{{H}_j,K-p+1} - f_{H,K-p+1} \right\|_{L_2}\,\dd s^2\nonumber\\
    & \lesssim_{K,p,M,\wb U,\underline L,\underline z} \Bt_n\sqrt{\At_n}\times \left\| f_{{H}_j,K-p+1} - f_{H,K-p+1} \right\|_{L_2}\nonumber\\
    & \lesssim_{K,p,M,\wb U,\underline L,\underline z} \Bt_n\sqrt{\At_n}\times\mathcal H(f_{{H}_j,K-p},f_{H,K-p})\left|\log \mathcal H^2(f_{{H}_j,K-p},f_{H,K-p})\right|^{1/2}.
\end{align}
In the foregoing display, the penultimate inequality follows using the arguments similar to \eqref{eq:truncated_int_1_hell}.
Using Lemma \ref{lem:l_2_one_dof} and the definition of the class of densities $\mathcal A$, we can conclude from the above display that
\[
\mathbb E\left[V_j\left((A_1,S^2_1),\ldots,(A_n,S^2_n)\right)\right] \lesssim_{K,p,M,\wb U,\underline L,\underline z} \frac{(\log n)^{13/4}}{\sqrt{n}}, \quad \mbox{for all $j \in \{1,\ldots,N\}$.}
\]
Therefore,
\[
\max_{j \in \{1,\ldots,N\}}\mathbb E\left[V_j\left((A_1,S^2_1),\ldots,(A_n,S^2_n)\right)\right] \lesssim_{K,p,M,\wb U,\underline L,\underline z} \frac{(\log n)^{13/4}}{\sqrt{n}}.
\]
Combining the above display with \eqref{eq:reduction_2d_exp_max_2} and \eqref{eq:2d_bdd_diff_1}, the result follows.

\subsection{Proof of Lemma \ref{lemm:density_ratio_bound_2d}}
To show this theorem, recall $\mathcal S$, the $\eta=1/n$ cover of $\mathcal F_{\jt,\mathcal A}$ defined in \eqref{eq:def_f_joint} and $\mathfrak R=(0,\Bt_n] \times [-\At_n,\At_n]$ where $\At_n$ and $\Bt_n$ are defined in \eqref{eq:def_a_eps_here}. If $\mathcal A$ holds, then we can get an index $\wh j \in \{1,\ldots,N\}$ such that
    \[
    \sup_{(t^2,a) \in \mathfrak R}\left|f_{\wh H,K-p}(t^2,a)-f_{H_{\wh j},K-p}(t^2,a)\right| \le \frac{1}{n}.
    \]
    Let us further define the event
    \[
    \mathcal E:=\left\{|A_i| \le \At_n,\, S_i^2 \in (0,\Bt_n],~~~ \text{for all $i=1,\ldots,n$}\right\}.
    \]
    By the construction of $\At_n$ and $\Bt_n$ in \eqref{eq:def_a_eps_here}, one can get that $\mathbb P[\mathcal E^c] \le n^{-1}$.
    Now consider the decomposition
    \begin{align}
    \label{eq:diff_density_2d_master}
        &\frac{1}{n}\sum_{i=1}^n \mathbb{E}\!\left[
    \frac{|D_i(H)-D_i(\widehat H)|}{D_i(\widehat H_\star)}
    \,\mathds{1}(\mathcal A)\right]\\
    & \lesssim \mathbb P(\mathcal E^c)+\frac{2}{n}\sum_{i=1}^n\mathbb{E}\!\left[
    \frac{|D_i(H)-D_i(H_{\wh j})|}{D_i(H)+D_i(H_{\wh j})}
    \cdot\mathds{1}(\mathcal A \cap \mathcal E)\right]\\
    &~~~~~~~~~~+\frac{2}{n}\sum_{i=1}^n\mathbb{E}\!\left[
    \left|\frac{D_i(H)-D_i(H_{\wh j})}{D_i(H)+D_i(H_{\wh j})}-\frac{D_i(H)-D_i(\wh H)}{D_i(H)+D_i(\wh H)}\right|
    \cdot\mathds{1}(\mathcal A \cap \mathcal E)\right]\\
    & \lesssim \frac{1}{n}+\frac{2}{n}\sum_{i=1}^n\mathbb{E}\!\left[
    \frac{|D_i(H)-D_i(H_{\wh j})|}{D_i(H)+D_i(H_{\wh j})}
    \cdot\mathds{1}(\mathcal A \cap \mathcal E)\right]\\
    &~~~~~~~~~~+\frac{2}{n}\sum_{i=1}^n\mathbb{E}\!\left[
    \left|\frac{D_i(H)-D_i(H_{\wh j})}{D_i(H)+D_i(H_{\wh j})}-\frac{D_i(H)-D_i(\wh H)}{D_i(H)+D_i(\wh H)}\right|
    \cdot\mathds{1}(\mathcal A \cap \mathcal E)\right]
    \end{align}
    Proceeding as in the proof of Lemma S10 of \cite{ignatiadis2025empirical}, we can show that
    \[
    \left|\frac{D_i(H)-D_i(H_{\wh j})}{D_i(H)+D_i(H_{\wh j})}-\frac{D_i(H)-D_i(\wh H)}{D_i(H)+D_i(\wh H)}\right| \lesssim \frac{\eta}{D_i(H)}. 
    \]
    Using the above relation and the definition of $\mathcal E$, we can conclude that
    \begin{align}
        \label{eq:diff_density_2d_slave}
        \frac{2}{n}\sum_{i=1}^n\mathbb{E}\!\left[
    \left|\frac{D_i(H)-D_i(H_{\wh j})}{D_i(H)+D_i(H_{\wh j})}-\frac{D_i(H)-D_i(\wh H)}{D_i(H)+D_i(\wh H)}\right|
    \cdot\mathds{1}(\mathcal A \cap \mathcal E)\right] \lesssim \eta~B_\eta A_\eta \lesssim_{K,p,M,\wb U,\underline L}   \frac{(\log n)^{3/2}}{n}. 
    \end{align}
    Next, we consider 
    \[
    \frac{1}{n}\sum_{i=1}^n\mathbb{E}\!\left[
    \frac{\left|D_i(H)-D_i(H_{\wh j})\right|}{D_i(H)+D_i(H_{\wh j})}
    \,\mathds{1}(\mathcal A \cap \mathcal E)\right],
    \]
    which can be upper-bounded by
    \[
    \frac{1}{n}\sum_{i=1}^n\mathbb{E}\!\left[\max_{j \in \{1,\ldots,N\}}
    \frac{\left|D_i(H)-D_i(H_{j})\right|}{D_i(H)+D_i(H_{j})}
    \,\mathds{1}(\mathcal A \cap \mathcal E)\right]
    \]
    We can decompose the above expression as
    \begin{align}
    &\frac{1}{n}\sum_{i=1}^n\mathbb{E}\!\left[\max_{j \in \{1,\ldots,N\}}
    \left|\frac{D_i(H)-D_i(H_{j})}{D_i(H)+D_i(H_{j})}-\mathbb E\left[\frac{D_i(H)-D_i(H_{j})}{D_i(H)+D_i(H_{j})}\right]\right|\right]\\
    &~~~~~~~~+\max_{j \in \{1,\ldots,N\}}\left\{\frac{1}{n}\sum_{i=1}^n\mathbb E\left[\left|\frac{D_i(H)-D_i(H_{j})}{D_i(H)+D_i(H_{ j})}\right|\right]\right\}.
    \end{align}
    Using the definition of $\mathcal S$ along with the techniques used in the proof of Lemma S10 of \cite{ignatiadis2025empirical} and Theorem \ref{thm:double_hellinger}, we can show that
    \begin{align}
    \label{eq:diff_density_2d_slave_2}
     \max_{j \in \{1,\ldots,N\}}\left\{\frac{1}{n}\sum_{i=1}^n\mathbb E\left[\left|\frac{D_i(H)-D_i(H_{j})}{D_i(H)+D_i(H_{ j})}\right|\right]\right\} \lesssim_{K,p,M,\wb U,\underline L}  \max_{j \in \{1,\ldots,N\}}\mathcal{H}\left(f_{H,K-p}, f_{H_j,K-p}\right) \lesssim \frac{(\log n)^{3/2}}{\sqrt{n}},
    \end{align}
    and 
    \[
    \frac{1}{n}\sum_{i=1}^n\mathbb{E}\!\left[\max_{j \in \{1,\ldots,N\}}
    \left|\frac{D_i(H)-D_i(H_{\wh j})}{D_i(H)+D_i(H_{\wh j})}-\mathbb E\left[\frac{D_i(H)-D_i(H_{j})}{D_i(H)+D_i(H_{j})}\right]\right|\right] \lesssim_{K,p,M,\wb U,\underline L}   \frac{(\log n)^{3/2}}{\sqrt{n}}.
    \]
    Combining the above two displays with \eqref{eq:diff_density_2d_master} and \eqref{eq:diff_density_2d_slave}, the lemma follows.

\subsection{Proof of Proposition \ref{prop:valid_p_values}}
\label{sec:proof_valid_p_joint}
The proof of this proposition follows by combining the steps used to prove Proposition 15 of \cite{ignatiadis2025empirical} and using Proposition \ref{prop:asymp_p_val}.
\subsection{Proof of Theorem \ref{thm:final_rate}}
\label{sec:proof_fdr_joint}

Let $\widehat F_n^\jt$ denote the empirical distribution function of the empirical Bayes $p$-values $\{\widehat P_i^\jt\}_{i=1}^n$. Define the data-driven Benjamini–Hochberg threshold
\begin{align}
\label{eq:bh_threshold_est}
\wh t_\jt := \sup\left\{ t \in (0,1) : \frac{t}{\widehat F_n^\jt(t) \vee n^{-1}} \le \alpha \right\}.
\end{align}
At nominal level $\alpha \in (0,1)$, the Benjamini–Hochberg procedure rejects hypothesis $j \in \{1,\ldots,n\}$ whenever $\widehat P_j^\jt \le \wh t_\jt$.
Analogously, let $F_n^\jt$ denote the empirical distribution function of the oracle $p$-values $\{P^\jt_i\}_{i=1}^n$, and define the corresponding oracle threshold
\begin{align}
\label{eq:bh_threshold_oracle}
\wh t_{\jt,\orc} := \sup\left\{ t \in (0,1) : \frac{t}{F_n^\jt(t) \vee n^{-1}} \le \alpha \right\}.
\end{align}
Define the associated empirical processes
\begin{align}
\label{eq:emp_prc_jt_fdr}
V_{n,\jt}(t) := \sum_{i \in \mathcal H_0} \mathds{1}\{\widehat P_i^\jt \le t\}, \qquad R_{n,\jt}(t) := \sum_{i=1}^n \mathds{1}\{\widehat P_i^\jt \le t\},
\end{align}
where $\mathcal H_0$ denotes the set of true null hypotheses and $\pi_0 := |\mathcal H_0|/n$.
Let $V_n^\jt := V_{n,\jt}(\wh t_\jt)$ and $R_n^\jt := R_{n,\jt}(\wh t_\jt)$. The false discovery rate (FDR) of the procedure is given by
\begin{align}
\fdr_n^\jt := \mathbb E\left[ \frac{V_{n,\jt}(\wh t_\jt)}{R_{n,\jt}(\wh t_\jt) \vee 1} \right].
\end{align}

To establish Theorem~\ref{thm:final_rate}, we rely on the following auxiliary results.

\begin{lemm}[Lemma S11 of \citet{ignatiadis2025empirical}]
\label{lem:s11}
Fix $\zeta \in (\alpha,1)$ and let $t \in (t_0,t_1)$, where $t_0$ and $t_1$ are defined in Assumption~\ref{asm:1d_limma_trend_bh}. Suppose $\delta \in (0, t_0 \wedge (t_1 - t_0))$. Then, for each $i \in \{1,\ldots,n\}$,
\[
\mathds{1}(\wh P_i^\jt \le t) - \mathds{1}(P_i^\jt \le t+\delta)
\le \frac{1}{\delta} \left| \wh P_i^\jt \wedge \zeta - P_i^\jt \wedge \zeta \right|.
\]
Similarly, we have for each $i \in \{1,\ldots,n\}$,
\[
\mathds{1}(\wh P_i^\jt \le t) - \mathds{1}(P_i^\jt \le t-\delta)
\ge -\frac{1}{\delta} \left|\wh P_i^\jt \wedge \zeta - P_i^\jt \wedge \zeta \right|.
\]
\end{lemm}

\begin{lemm}[Lemma S12 of \citet{ignatiadis2025empirical} (Bretagnolle--Dvoretzky--Kiefer--Wolfowitz)]
\label{lem:bdkw}
Let $U_{i,n} \in [0,1]$, $i=1,\ldots,n$, be independent random variables (not necessarily identically distributed). Then, for every $\varepsilon \ge 0$,
\[
\mathbb P\left( \sup_{t \in [0,1]} \left| \frac{1}{n} \sum_{i=1}^n \mathds{1}\{U_{i,n} \le t\} - \mathbb P(U_{i,n} \le t) \right| \ge \varepsilon \right) \le 2 \exp(-n\varepsilon^2).
\]
In addition,
\[
\mathbb E\left[ \sup_{t \in [0,1]} \left| \frac{1}{n} \sum_{i=1}^n \mathds{1}\{U_{i,n} \le t\} - \mathbb P(U_{i,n} \le t) \right| \right] \le \sqrt{\frac{2e}{n}}.
\]
\end{lemm}

\begin{proof}[Proof of Theorem \ref{thm:final_rate}]
To prove Theorem \ref{thm:final_rate}, we adopt a leave-one-out argument in the spirit of \citet{ignatiadis2025empirical}. 
For each $i\in\{1,\ldots,n\}$, let $R^\jt_{n,i}$ denote the number of discoveries produced by the BH procedure applied to the leave-one-out p-values $\{\wh P^{\jt,-i}_j\}_{j=1}^n$, obtained from the modified data in which $Z_i$ is replaced by
\[
Z^\infty_i=\begin{cases}
\infty & \text{if $Z_i \ge 0$,}\\
-\infty & \text{if $Z_i < 0$.}
\end{cases}
\]
Note that the estimated prior $\wh H$ remains unchanged across these modified datasets, since $\wh H$ is computed only from the pairs $\{(S_i^2,A_i)\}_{i=1}^n$. 
Using arguments analogous to the proof of Lemma C.1 of \citet{roquain_verzelen}, we have that if the $i$-th hypothesis is rejected, then $R^\jt_n = R^\jt_{n,i}\ge 1$. Hence, for any $i\in\mathcal H_0$,
\begin{align}
\mathbb E_{H}\!\left[\frac{\mathds{1}\{\wh P^\jt_i \le \alpha R^\jt_n/n\}}{R^\jt_n \vee 1}\right]
&=
\mathbb E_{H}\!\left\{\mathbb E_{H}\!\left[\frac{\mathds{1}\{\wh P^\jt_i \le \alpha R^\jt_n/n\}}{R^\jt_n \vee 1}
~\Bigg|~(S_1^2,A_1), \ldots, (S_n^2,A_n), R^\jt_{n,i}\right]\right\}
\nonumber\\
&=
\mathbb E_{H}\!\left\{\mathbb E_{H}\!\left[\frac{\mathds{1}\{\wh P^\jt_i \le \alpha R^\jt_{n,i}/n\}}{R^\jt_{n,i}}
~\Bigg|~(S_1^2,A_1), \ldots, (S_n^2,A_n), R^\jt_{n,i}\right]\right\},
\end{align}
almost surely.

Let $\zeta:=\max\{3/4,(1+\alpha)/2\}$ and $\underline z=\underline L^{1/2}z_{1-\zeta/2}$. Define
\[
\Delta_i
:=
\sup_{z:|z|>\underline z}
\left|
\PjtFun(z,S^2_i,A_i;\wh H) - \PjtFun(z,S^2_i,A_i;H)
\right|.
\]
Repeating the arguments in the proof of Theorem 17 of \citet{ignatiadis2025empirical}, we obtain that for all $i\in\mathcal H_0$,
\[
\mathbb E_{H}\!\left[\frac{\mathds{1}\{\wh P^\jt_i \le \alpha R^\jt_{n,i}/n\}}{R^\jt_{n,i}}
~\Bigg|~ (S_1^2,A_1), \ldots, (S_n^2,A_n),R^\jt_{n,i}\right]
\le \frac{\alpha}{n}+\frac{\Delta_i}{R^\jt_{n,i}},
\qquad \text{a.s.}
\]
Moreover, by construction of the leave-one-out p-values, $R^\jt_{n,i}\ge R^\jt_n$ for all $i\in\mathcal H_0$. 
Retracing the remainder of the proof of Theorem 17 of \citet{ignatiadis2025empirical}, we conclude that
\begin{align}
\fdr^\jt_n
&=
\mathbb E_{H}\!\left[\sum_{i \in \mathcal H_0}\frac{\mathds{1}\{\wh P^\jt_i \le \alpha R^\jt_n/n\}}{R^\jt_n \vee 1}\right]
\nonumber\\
&\le
\frac{n_0}{n}\alpha
+
\mathbb E_H\!\left[\min\left\{\sum_{i \in \mathcal H_0}\frac{\Delta_i}{R^\jt_{n,i}},1\right\}\right]
\nonumber\\
&\le
\frac{n_0}{n}\alpha
+
\mathbb E_H\!\left[\min\left\{\sum_{i \in \mathcal H_0}\frac{\Delta_i}{R^\jt_{n}},1\right\}\right].
\end{align}
Consequently,
\[
\left(\fdr^\jt_n-\frac{n_0}{n}\alpha\right)_+
\le
\mathbb E_H\!\left[\min\left\{\sum_{i \in \mathcal H_0}\frac{\Delta_i}{R^\jt_{n}},1\right\}\right].
\]

Fix any sequence $\kappa_n\in(0,1)$ such that
\[
\mathbb P\!\left[R^\jt_n <n\kappa_n\right] \le \eta_n,
\qquad \eta_n\in(0,1).
\]
Then, using Proposition \ref{prop:asymp_p_val} together with the definition of $\zeta$ and $\underline z$, we have
\begin{align}
\label{eq:start_rate_fdr}
\left(\fdr_n^\jt-\frac{n_0}{n}\alpha\right)_+
&\le
\mathbb P[R^\jt_n < n\kappa_n]
+
\mathbb E\!\left[\sum_{i \in \mathcal H_0}\frac{\Delta_i}{n\kappa_n}\right]
\nonumber\\
&\lesssim_{\underline L,\wb U,M,K,p,\zeta}
\eta_n+\frac{(\log n)^{13/4}}{\kappa_n\sqrt{n}}.
\end{align}

It remains to choose $\kappa_n$ and $\eta_n$. 
Take any $t\in[t_0,\alpha]$ and $\delta\in(0,t_0\wedge(t_1-t_0))$, where $t_0,t_1$ are as in Assumption \ref{asm:1d_limma_trend_bh}. 
By Lemma \ref{lem:s11},
\begin{align}
\frac{R_{n,\jt}(t)}{n}
&=
\frac{1}{n}\sum_{i=1}^n\mathds{1}\{\wh P^\jt_i \le t\}
\nonumber\\
&\ge
\frac{1}{n}\sum_{i=1}^n\mathds{1}\{P^\jt_i \le t-\delta\}
-\frac{1}{n\delta}\sum_{i=1}^n\left|\wh P_i^\jt \wedge \zeta - P_i^\jt \wedge \zeta \right|.
\end{align}
In particular, taking $t=t_1$ yields
\begin{align}
\frac{R_{n,\jt}(t_1)}{n}
&=
\frac{1}{n}\sum_{i=1}^n\mathds{1}\{\wh P^\jt_i \le t_1\}
\nonumber\\
&\ge
\frac{1}{n}\sum_{i=1}^n\mathds{1}\{P^\jt_i \le t_1-\delta\}
-\frac{1}{n\delta}\sum_{i=1}^n\left|\wh P_i^\jt \wedge \zeta - P_i^\jt \wedge \zeta \right|
\nonumber\\
&=
\frac{1}{n}\sum_{i=1}^n\mathbb P\!\left[P^\jt_i \le t_1-\delta\right]
-\frac{1}{n\delta}\sum_{i=1}^n\left|\wh P_i^\jt \wedge \zeta - P_i^\jt \wedge \zeta \right|
\nonumber\\
&\quad
+\frac{1}{n}\sum_{i=1}^n\left(\mathds{1}\{P^\jt_i \le t_1-\delta\}-\mathbb P\!\left[P^\jt_i \le t_1-\delta\right]\right).
\label{eq:Rjt_lower_decomp}
\end{align}

By Assumption \ref{asm:1d_limma_trend_bh}, there exists
\[
\alpha_{0,\jt} \in \left(\liminf_{n\to\infty}\inf_{t \in [t_0,t_1]}\frac{1}{nt}\sum_{i=1}^n\mathbb P\!\left[P^\jt_i \le t\right],\alpha\right)
\]
such that
\begin{align}
\liminf_{n\to\infty}\frac{1}{n}\sum_{i=1}^n\mathbb P\!\left[P^\jt_i \le t_1-\delta\right]
\ge (t_1-\delta)\cdot\liminf_{n\to\infty}\inf_{t \in [t_0,t_1]}\frac{1}{nt}\sum_{i=1}^n\mathbb P\!\left[P^\jt_i \le t\right] \ge
\frac{t_1-\delta}{\alpha_{0,\jt}},
\end{align}
Equivalently, for any $\mathfrak P>0$ there exists $n_0(\mathfrak P)\ge 1$ such that for all $n\ge n_0(\mathfrak P)$,
\[
\frac{1}{n}\sum_{i=1}^n\mathbb P\!\left[P^\jt_i \le t_1-\delta\right]
\ge
\frac{t_1-\delta}{\alpha_{0,\jt}}+\mathfrak P.
\]

Define the random error term
\[
X_n
:=
\frac{1}{n\delta}\sum_{i=1}^n\left|\wh P_i^\jt \wedge \zeta - P_i^\jt \wedge \zeta \right|
+
\sup_{t\in(0,1)}
\left|
\frac{1}{n}\sum_{i=1}^n\left(\mathds{1}\{P^\jt_i \le t\}-\mathbb P\!\left[P^\jt_i \le t\right]\right)
\right|.
\]
By Markov's inequality together with Proposition \ref{prop:asymp_p_val} and Lemma \ref{lem:bdkw}, there exists a constant $C>0$ such that
\[
\mathbb P\!\left[X_n>\frac{\delta}{4\alpha_{0,\jt}}\right]
\le
\frac{4\alpha_{0,\jt}}{\delta}\,\mathbb E[X_n]
\le
C\cdot \frac{4\alpha_{0,\jt}\,(\log n)^{13/4}}{\delta^2\sqrt{n}}.
\]
On the event $\{X_n\le \delta/(4\alpha_{0,\jt})\}$, combining with \eqref{eq:Rjt_lower_decomp} and the deterministic lower bound above gives, for all $n\ge n_0(\mathfrak P)$,
\[
\frac{R_{n,\jt}(t_1)}{n}
\ge
\frac{t_1-\delta}{\alpha_{0,\jt}}+\mathfrak P-\frac{\delta}{4\alpha_{0,\jt}}.
\]
Now set
\[
\delta:=\frac{4t_1(\alpha-\alpha_{0,\jt})}{5\alpha},
\]
so that $(t_1-\delta)/\alpha_{0,\jt}-\delta/(4\alpha_{0,\jt})=t_1/\alpha$ and hence, for all $n\ge n_0(\mathfrak P)$,
\[
\frac{R_{n,\jt}(t_1)}{n}
\ge
\frac{t_1}{\alpha}+\mathfrak P
\ge
\frac{t_1}{\alpha}
\qquad \text{on the event }\left\{X_n\le\frac{\delta}{4\,\alpha_{0,\jt}}\right\}.
\]

By the self-consistency characterization of the BH threshold $\wh t_\jt$ at level $\alpha$,
\[
\wh t_\jt:=\sup\left\{t\in[0,1]:\frac{R_{n,\jt}(t)}{n}\ge \frac{t}{\alpha}\right\},
\]
the inequality $R_{n,\jt}(t_1)/n \ge t_1/\alpha$ implies $\wh t_\jt \ge t_1$. Since $R_{n,\jt}(t)$ is non-decreasing in $t$, it follows that
\[
\frac{R^\jt_n}{n}
=
\frac{R_{n,\jt}(\wh t_\jt)}{n}
\ge
\frac{R_{n,\jt}(t_1)}{n}
\ge
\frac{t_1}{\alpha},
\]
on the same event. Consequently, taking
\[
\kappa_n:=\frac{t_1}{\alpha},
\qquad
\eta_n:=C\cdot \frac{\alpha_{0,\jt}\,(\log n)^{13/4}}{\delta\sqrt{n}},
\]
and substituting into \eqref{eq:start_rate_fdr}, we obtain for all $n\ge n_0(\mathfrak P)$,
\begin{align}
\label{eq:start_fdr_2}
\left(\fdr^\jt_n-\frac{n_0}{n}\alpha\right)_+
&\lesssim_{\underline L,\wb U,M,K,p,\zeta}
C\cdot \frac{4\alpha_{0,\jt}\,(\log n)^{13/4}}{\delta^2\sqrt{n}}
+\frac{\alpha(\log n)^{13/4}}{t_1\sqrt{n}}.
\end{align}
In particular, for any $\varrho>13/4$,
\[
\sqrt{n}(\log n)^{-\varrho}\left(\fdr^\jt_n-\frac{n_0}{n}\alpha\right)_+
\longrightarrow 0.
\]
\end{proof}

\subsection{Power properties of the Benjamini Hochberg procedure using the estimated \jtlitrd{} p-values}
\label{sec:power_jt_bh}
To analyze the asymptotic power of the BH procedure, we impose a stronger assumption on the number of true null hypotheses
$n_0 := \operatorname{card}(\mathcal H_0)$ (where $\mathcal H_0$ denotes the subset of true null hypotheses) and on the mixing distribution, analogous to Assumption~8 of \citet{ignatiadis2025empirical}.

\begin{assumption}
\label{asm:2d_limma_trend_bh}
As $n \to \infty$, assume that $\frac{n_0}{n} \to \pi_0 \in (0,1)$, and that
\[
\frac{1}{n}\sum_{i=1}^n\mathbb P_{H}\!\left[
P^\jt_i \le t\right]\;\longrightarrow\;H_\infty^{\jt}(t),
\qquad t \in (0,1),
\]
for some distribution function $H_\infty^{\jt}$. Define
\[
t_\infty^{\jt}
:=
\sup\left\{
t \in (0,1): \frac{t}{H_\infty^{\jt}(t)} \le \alpha
\right\}.
\]
We further assume that $t_\infty^{\jt} \in (0,\alpha)$ and that the mapping
$t \mapsto t / H_\infty^{\jt}(t)$ is strictly increasing in a neighborhood of $t_\infty^{\jt}$.
\end{assumption}

The first condition ensures that the proportion of true null hypotheses is not asymptotically negligible, thereby yielding a nontrivial number of rejections. The second condition guarantees stability of the BH threshold by requiring local uniqueness of the solution to the defining inequality. More general versions of this assumption appear in Theorem~4 of \citet{storey_2004}, where the null proportion $\pi_0$ is estimated from the data and incorporated into the BH procedure. Related regularity conditions have also been studied by \citet{ferreira_zwinderman} and \citet{DuZhang2014SIM}.

We consider the following measures of power:
\begin{align}
\mathrm{Pow}_n^{\jt}
&:= \mathbb E\!\left[\frac{R_n^{\jt}-V_n^{\jt}}{n-n_0}\right],
\qquad
\mathrm{FNDR}_n^{\jt}
:= \mathbb E\!\left[
\frac{(\mathcal D_n^{\jt})^c \cap \mathcal H_0^c}{(n-R_n^{\jt}) \vee 1}
\right],
\end{align}
where $\mathcal D_n^{\jt}$ denotes the rejection set produced by the BH procedure applied to the estimated p-values $\{\widehat P_i^{\jt}\}$. The first measure of power was introduced in \citet{ferreira_zwinderman}, while the second was proposed by \citet{GenoveseWasserman2002FDR}. Both quantities were analyzed for the empirical Bayes \texttt{limma} procedure in \citet{ignatiadis2025empirical}.

For comparison, we define the corresponding quantities for the BH procedure applied to the oracle p-values $\{P_i^{\jt}\}$:
\begin{align}
\mathrm{Pow}_n^{\jt,\orc}
&:= \mathbb E\!\left[\frac{R_n^{\jt,\orc}-V_n^{\jt,\orc}}{n-n_0}\right],
\qquad
\mathrm{FNDR}_n^{\jt,\orc}
:= \mathbb E\!\left[
\frac{(\mathcal D_n^{\jt,\orc})^c \cap \mathcal H_0^c}{(n-R_n^{\jt,\orc}) \vee 1}
\right],
\end{align}
where $V_n^{\jt,\orc}$ and $R_n^{\jt,\orc}$ denote the numbers of false discoveries and total discoveries, respectively, produced by the BH procedure using oracle p-values, and $\mathcal D_n^{\jt,\orc}$ is the corresponding rejection set.

We now state our main result on asymptotic power equivalence.

\begin{theorem}
\label{thm:2d_power}
Under Assumptions~\ref{assu:compact_2d} and \ref{asm:2d_limma_trend_bh},
\[
\limsup_{n \to \infty}
\left|\mathrm{Pow}_n^{\jt}-\mathrm{Pow}_n^{\jt,\orc}\right| = 0,
\qquad
\limsup_{n \to \infty}
\left|\mathrm{FNDR}_n^{\jt}-\mathrm{FNDR}_n^{\jt,\orc}\right| = 0.
\]
\end{theorem}

The proof follows by retracing the arguments used in the proof of Proposition~20 of \citet{ignatiadis2025empirical}.

\section{Limma-trend in the compound partially Bayes framework}
\label{sec:cmp_bayes_ltrd}
\subsection{Compound partial Bayes and FDR control in \reglitrd{}}
\label{sec:reg_limma_trend_cmp}
In this section, instead of \eqref{eq:bayesian_cs_model_1}, we assume that the nuisance parameters $\tau^2_1,\ldots,\tau^2_n$ are fixed constants satisfying
\begin{align}
        \label{eq:bayesian_cp_model_1}
        \tau^2_i \in [\underline L_\trd,\overline U_\trd] \quad  \mbox{for all $i \in [n]$.}
\end{align}
Furthermore, we also adopt the data generating model described in \eqref{eq:misspecified-limma-reg} with a working trend $\xi_\mis \in \mathcal X$ such that the estimated trend $\wh \xi$ satisfies Assumption~\ref{assum:missp_trend_estimation} with $\xi_0$ replaced by $\xi_\mis$. In this framework, one can show that the NPMLE $\wh G_\trd$ targets the empirical distribution of the nuisance parameters  
\begin{align}
    \label{eq:def_g_trd_n}
    G_n^\trd:=\frac{1}{n}\sum_{i=1}^n\delta_{\tau_i^2}.
\end{align}
Therefore, the target quantity approximated by $\wh P^\mis_i$ is 
\begin{align}
    \label{eq:compound_trnd_p_value}
    Q^{\cmp}_i=\frac{\sum_{j=1}^n2\,\Phi\!\left(-\frac{|O_i|}{\nu \tau_j}\right)
\,p_{\chi^2}(\misV{i}\mid K-p,\tau^2_j)}{\sum_{j=1}^n p\!\left(\misV{i}\,\Big|\,K-p,\tau^2_j\right)}, \qquad \mbox{for $i \in [n]$.}
\end{align}
One can show the following property of the sequence $\{Q^{\cmp}_i\}$.
\begin{proposition}
    \label{eq:cmp_p_val_trd}
    Under Assumption~\ref{assu:cmp_trd}, the sequence $\{Q^{\cmp}_i\}$ satisfies the following average significance controlling property.
        Let
\(
\mathcal H_0 := \{ i \in [n] : \theta_i = 0 \}
\)
denote the set of null indices. Then, for all $t \in (0,1)$,
\[
\frac{1}{n}\sum_{i \in \mathcal H_0}\mathbb P_{\tau^2_i}\!\left[Q^{\cmp}_i \le t\right]\le t.
\]
In other words, $\{Q^{\cmp}_i\}$ forms a sequence of \emph{compound p-values} as defined in~\cite{ignatiadis_ramdas}.
\end{proposition}
\begin{proof}
Observe that 
\(
Q^\cmp_i:=\PregFun_\mis(O_i,\misV{i};G^\trd_n),
\)
where for any $x\ge 0$, $\PregFun_\mis(x,\misV{i};G^\trd_n)$ is defined through \eqref{eq:validity_seq_2}. Again, by the proof technique in \eqref{eq:validity_seq_2}, we have
\begin{align}
\label{eq:null_rel_unif}
\PregFun_\mis(O,V^2_\mis;\mathrm G) \mid \omega=0 \sim \mathrm{Unif}(0,1),
\end{align}
if $O \sim \dnorm(\omega,\tau^2\nu^2)$, $V^2_\mis \sim \tau^2\chi^2_{K-p}/(K-p)$ and $\tau^2 \sim \mathrm G$.
Define $O'_i\simiid \dnorm(0,\tau^2_i\nu^2)$. Observe that if $i \in \mathcal H_0$, then $O'_i\overset{d}{=}O_i \sim \dnorm(0,\nu^2\tau^2_i)$. Consequently, for any $x \in (0,1)$
    \begin{align}
        \frac{1}{n}\sum_{i \in \mathcal H_0}\mathbb P_{\tau^2_i}\!\left[Q^{\cmp}_i \le t\right] & \le\frac{1}{n}\sum_{i \in \mathcal H_0}\mathbb P_{\tau^2_i}\!\left[\PregFun_\mis(O_i,\misV{i};G^\trd_n) \le t\right]\\
        &\le\frac{1}{n}\sum_{i=1}^n\mathbb P_{\tau^2_i}\!\left[\PregFun_\mis(O'_i,\misV{i};G^\trd_n) \le t\right]\\
        & =\mathbb P_{\tau^2 \sim G^\trd_n}\left[\PregFun_\mis(O',V^2;G^\trd_n) \le t\right],
    \end{align}
    where $O' \sim \dnorm(0,\tau^2\nu^2)$ and $V^2\sim \tau^2\chi^2_{K-p}/(K-p)$. Using \eqref{eq:null_rel_unif}, the lemma follows.
\end{proof}
Next, consider the following assumption on the data generating model.
\begin{assumption}
    \label{assu:cmp_trd}
    The tuples $(M_i, Z_i, V_{i,\mis}^2)$ are generated according to~\eqref{eq:bayesian_cs_model_2}, and~\eqref{eq:tau_mis_v_mis} for $K -p\ge 2$ and are jointly independendent across $i \in [n]$. Furthermore, the random variables $M_1,\ldots,M_n$ satisfy
$\mathbb P\!\left(\max_{i \in [n]}|M_i| > \mathrm W\sqrt{\log n}\right)\le n^{-3}$ for some $W>0$. Next, the latent parameters $\tau^2_i$ are fixed constants satisfying \eqref{eq:bayesian_cp_model_1}. 
\end{assumption}
Under the foregoing assumption, one can show the following theorem, showing that the sequence $\{\wh P^\mis_i\}$ is a sequence of approximately compound p-values.  
\begin{proposition}
\label{prop:average_significance_control_cmp_p_val_trd}
Suppose Assumptions~\ref{assum:missp_trend_estimation} and~\ref{assu:cmp_trd} hold. Then there exists a constant $D_{\cmp}>0$ (depending only on $\underline L,\overline U,M,K,p,\nu$, and $\zeta$) such that, for all $n \ge n_\cmp \in \mathbb N_{\ge 1}$,
\begin{align}
\label{eq:cmp_trnd_p_val_2}
&\sup_{t \in [0,\zeta]}
\left\{
\frac{1}{n}\sum_{i \in \mathcal H_0}
\mathbb P_{\tau^2_i}\!\left[\wh P_i^{\trd} \le t\right]- t\right\}\\
&\le D_{\cmp}\cdot\max\left\{\Delta^{1/4}_n(\log n)^{5/4},\frac{(\log n)^{5/4}}{n^{1/8}}\Delta^{1/4}_n,\frac{(\log n)}{n^{1/8}}\cdot |\log \Delta_n|^{h_2/8}\Delta^{\frac{1}{4}(1-h_1/2)}_n,\frac{(\log n)^{5/4}}{n^{1/4}}\right\}.
\end{align}
\end{proposition}
\begin{proof}
    To prove the above proposition under Assumptions~\ref{assum:missp_trend_estimation} and~\ref{assu:cmp_trd}, one can retrace the arguments in the proof of Lemma~\ref{lem:hell_limm_trnd_ms} and Theorem 9* of \cite{ignatiadis2025empirical} to show that
    \[
    \mathbb P\left[\mathcal H^2(f_{\wh G_\trd,K-p},f_{G^\trd_n,K-p}) \ge \wt D_{\cmp}\lambda^2_{n,2}\right] \lesssim \frac{1}{n^2}+e^{-c_0(\log n)^2},
    \]
    where $\wt D_{\cmp}>0$ is an absolute constant and $\lambda^2_{n,2}$ is defined in Lemma~\ref{lem:hell_limm_trnd_ms}. This result can, in turn, be combined with the techniques adopted to prove Proposition~\ref{prop:avg_sign_limma_trnd_ms} and Theorem 13* to show that for any $\zeta \in (1/2,1)$
    \begin{align}
    \label{eq:cmp_trnd_p_val}
        \frac{1}{n}\sum_{i=1}^n\mathbb E\left[|Q^\cmp_i \wedge \zeta-\wh P^\trd_i \wedge \zeta|\right] \lesssim \mathfrak L_{n,\mis},
    \end{align}
    where $\mathfrak L_{n,\mis}$ is defined in the statement of Proposition~\ref{prop:avg_sign_limma_trnd_ms}.
    Now, retracing the arguments used to prove Proposition 15* of \cite{ignatiadis2025empirical} and using the foregoing relation, we can conclude the lemma.
\end{proof}
Furthermore, one can use the estimated compound p-values $\{\wh P_i^{\trd}\}$ in the Benjamini-Hochberg procedure to test $H_i:\theta_i=0$ with the level of significance $\alpha \in (0,1)$. In the following theorem, we show that the asymptotic FDR remains controlled at the desired nominal level of significance.
\begin{theorem}
\label{thm:final_rate_compound_p_val}
Suppose Assumptions~\ref{assum:missp_trend_estimation} and~\ref{assu:cmp_trd} hold. Fix $\alpha \in (0,1)$. Assume that the true p-values $\{Q^\cmp_i\}$ are critically dense at $\alpha$. Furthermore, also assume the conditions in Proposition~\ref{prop:average_significance_control_cmp_p_val_trd}. Then
\[
\limsup_{n \rightarrow \infty}\fdr^\mis_n\;\le\alpha,
\]
where $\fdr^\mis_n := \mathbb E\left[V^\mis_n/(R^\mis_n \vee 1)\right]$.
\end{theorem}
The proof of the theorem follows using \eqref{eq:cmp_trnd_p_val_2}, \eqref{eq:cmp_trnd_p_val}, and retracing the arguments of the proof of Theorem 17* of \cite{ignatiadis2025empirical}.

\subsection{Compound partial Bayes and FDR control in \jtlitrd{}}
\label{sec:cmp_ltrd_jt}
Under this framework, we relax the bivariate prior specification on the nuisance parameters in \eqref{eq:limma_trend_2d_prior} and instead assume that $(\mu_1,\sigma^2_1),\ldots,(\mu_n,\sigma^2_n)$ are fixed constants satisfying
\begin{align}
\label{eq:bayesian_joint_cmp}
    \mu_i \in [-M,M], \quad \mbox{and} \quad \sigma^2_i \in [\underline L,\overline U] \quad \mbox{for all $i \in [n]$.}
\end{align}
As shown in Section~\ref{sec:cmp_bayes_ltrd}, one can still use $\{\wh P_i^{\jt}\}$ computed by plugging in $\wh H$ (defined in \eqref{eq:joint_npmle}) in \eqref{eq:2d_pb_pvalue_general_case} to test the hypotheses $H_0:\theta_i=0$. As in Section~\ref{sec:cmp_bayes_ltrd}, within this framework, one can show that the sequence the $\{\wh P_i^{\jt}\}$ approximates the following quantities:
\begin{align}
\label{eq:compound_jt_p_val}
P_i^{\cmp}:=\frac{\sum_{j=1}^n 2\Phi\!\left(-|Z_i|/(\nu\sigma_j)\right)\,p_{K-p}(S_i^2,A_i \mid \mu_j,\sigma_j^2)}{\sum_{j=1}^n p_{K-p}(S_i^2,A_i \mid \mu_j,\sigma_j^2)},
\end{align}
where $p(s^2,a \mid \mu,\sigma^2)$ is defined in \eqref{eq:2_d_marginal}. The sequence $\{P_i^{\cmp}\}$ satisfies the following average significance controlling property.
\begin{lemm}
\label{lem:average_significance_control}
Consider the data-generating mechanism in \eqref{eq:def_z_i_dis_2d}, and \eqref{eq:limma_trend_2d_prior}. Also consider the p-values $\{P_i^{\cmp}\}$ defined in \eqref{eq:compound_jt_p_val}. Recall
\(
\mathcal H_0 := \{ i \in [n] : \theta_i = 0 \},
\)
the set of null indices. Then, under Assumption~\ref{assu:design}, for all $t \in (0,1)$,
\[
\frac{1}{n}\sum_{i \in \mathcal H_0}\mathbb P_{(\sigma_i^2,\mu_i)}\!\left[P_i^{\cmp} \le t\right]\le t.
\]
In other words, $\{P^\cmp_i\}$ forms a sequence of \emph{compound p-values} as defined in~\cite{ignatiadis_ramdas}.
\end{lemm}
The proof of this lemma follows by retracing the arguments in the proof of Theorem~21 of \cite{ignatiadis2025empirical}.
Next, consider the following assumption on the data generating model.

\begin{assumption}
    \label{assu:cmp_jt}
The tuples $(Z_i, S^2_i,A_i)$ are generated according to~\eqref{eq:def_z_i_dis_2d} for $K -p\ge 2$ and are jointly independendent across $i \in [n]$. The latent parameters $(\mu_i,\sigma^2_i)$ are fixed constants satisfying \eqref{eq:bayesian_joint_cmp}.
\end{assumption}

Furthermore, one can show that under~\eqref{eq:bayesian_joint_cmp}, the following proposition holds.
\begin{proposition}
\label{prop:average_significance_control_cmp_p_val}
If Assumption~\ref{assu:cmp_jt} holds, there exists a constant $C_{\cmp}>0$ (depending only on $\underline L,\overline U,M,K,p,\nu$, and $\zeta$) such that, for all $n \in \mathbb N_{\ge 1}$,
\[
\mathbb E\!\left[\sup_{t\in(0,\alpha)}\left|\frac{1}{n}\sum_{i\in\mathcal H_0}\mathds{1}\{\widehat P_i^{\jt}\le t\}-t\right|\right]\le C_{\cmp}\,\frac{(\log n)^{13/8}}{n^{1/4}}.
\]
\end{proposition}
\begin{proof}
    To prove the above theorem, one first shows that there exist constants $\mathfrak F_0>0$ and $n_0\in\mathbb N_{\ge1}$, depending only on $K,p,M,\underline L,\overline U$ and $\mathfrak d_0$, such that for all $n\ge n_0$,
\begin{align}
\label{eq:double_hell_cmp}
\mathbb P_{(\mu,\sigma)}\!\left[\mathcal H^2\!\left(f_{\widehat H,K-p},f_{H_n,K-p}\right)\ge \mathfrak F_0\frac{(\log n)^{3}}{n}\right]\le \exp\!\left(-\mathfrak d_0 \log n\right).
\end{align}
Here the probability is taken with respect to the joint distribution of $\{(S_i^2,A_i):i\in[n]\}$, where for each $i$, $(S_i^2,A_i)$ is generated according to \eqref{eq:def_z_i_dis} and \eqref{eq:limma_trend_summary} with parameters $(\mu_i,\sigma_i^2)$. The proof of the foregoing relation follows by using the definition of the NPMLE and modifying the arguments in the proof of Theorem~\ref{thm:double_hellinger} as in the proof of Theorem 9* of \cite{ignatiadis2025empirical}. Next, one can modify the arguments in the proof of Theorem~\ref{prop:asymp_p_val} along the lines of the proof of Theorem~13* of \cite{ignatiadis2025empirical} to show that for any $\zeta \in (1/2,1)$
\begin{align}
    \frac{1}{n}\sum_{i=1}^n\mathbb E\left[\left|\wh P^\jt_i\wedge \zeta - P^\cmp_i\wedge \zeta\right|\right] \lesssim \frac{(\log n)^{13/4}}{\sqrt{n}}.
\end{align}
Furthermore, one can also use \eqref{eq:double_hell_cmp}, Lemma~\ref{lem:average_significance_control} and retrace the arguments in the proof of Proposition~15* of \cite{ignatiadis2025empirical} to conclude the proposition.
\end{proof}
Now, consider a BH procedure to test $H_i:\theta_i=0$ using the NPMLE based \jtlitrd{} p-values $\{\wh P^\jt_i\}$. Let $\fdr^\jt_n$ be the false discovery rate of the resulting procedure.
\begin{theorem}
\label{thm:final_rate_compound_p_val_jt}
If Assumption~\ref{assu:cmp_jt} holds and the true p-values $\{P^\cmp_i\}$ are critically dense at $\alpha$ as defined in Assumption~\ref{asm:1d_limma_trend_bh}, then under \eqref{eq:bayesian_joint_cmp}, we have the following.
\[
\limsup_{n \rightarrow \infty}\fdr^\jt_n\;\le\alpha.
\]
\end{theorem}

\begin{proof}
    Since the conditional validity property analogous to \eqref{eq:jt_cond_valid_p_val} does not hold for the compound p-values $\{P_i^{\cmp}\}$, we cannot directly follow the proof of Theorem~\ref{thm:final_rate}. Instead, we decompose the false discovery proportion. For any deterministic sequence $\kappa_n>0$,
\begin{align}
\frac{V_{n,\jt}(\wh t_\jt)}{R_{n,\jt}(\wh t_\jt)\vee 1}\le \mathds{1}\{R_{n,\jt}(\wh t_\jt)<n\kappa_n\}+\frac{V_{n,\jt}(\wh t_\jt)}{R_{n,\jt}(\wh t_\jt)}\mathds{1}\{R_{n,\jt}(\wh t_\jt)\ge n\kappa_n\}.
\end{align}
Taking expectations yields
\begin{align}
\fdr_n^{\jt}\le \mathbb P\!\left(R_{n,\jt}(\wh t_\jt)<n\kappa_n\right)+\mathbb E\!\left[\frac{V_{n,\jt}(\wh t_\jt)}{R_{n,\jt}(\wh t_\jt)}\mathds{1}\{R_{n,\jt}(\wh t_\jt)\ge n\kappa_n\}\right].
\end{align}
Next observe that
\begin{align}
\frac{V_{n,\jt}(\wh t_\jt)}{R_{n,\jt}(\wh t_\jt)}=\frac{n\wh t_\jt}{R_{n,\jt}(\wh t_\jt)}+\frac{n}{R_{n,\jt}(\wh t_\jt)}\left\{\frac{1}{n}\sum_{i\in\mathcal H_0}\mathds{1}\{\widehat P_i^{\jt}\le\wh t_\jt\}-\wh t_\jt\right\}.
\end{align}
On the event $\{R_{n,\jt}(\wh t_\jt)\ge n\kappa_n\}$, the definition of the BH threshold implies $(n\wh t_\jt)/R_{n,\jt}(\wh t_\jt)\le\alpha$, and therefore
\begin{align}
\frac{V_{n,\jt}(\wh t_\jt)}{R_{n,\jt}(\wh t_\jt)}\le\alpha+\frac{1}{\kappa_n}\sup_{t\in(0,\alpha)}\left|\frac{1}{n}\sum_{i\in\mathcal H_0}\mathds{1}\{\widehat P_i^{\jt}\le t\}-t\right|.
\end{align}
Consequently,
\begin{align}
\fdr_n^{\jt}\le \mathbb P\!\left(R_{n,\jt}(\wh t_\jt)<n\kappa_n\right)+\alpha+\frac{1}{\kappa_n}\mathbb E\!\left[\sup_{t\in(0,\alpha)}\left|\frac{1}{n}\sum_{i\in\mathcal H_0}\mathds{1}\{\widehat P_i^{\jt}\le t\}-t\right|\right].
\end{align}
By Proposition~\ref{prop:average_significance_control_cmp_p_val},
\begin{align}
\mathbb E\!\left[\sup_{t\in(0,\alpha)}\left|\frac{1}{n}\sum_{i\in\mathcal H_0}\mathds{1}\{\widehat P_i^{\jt}\le t\}-t\right|\right]\lesssim_{K,p,\underline L,\overline U,\nu,\alpha}\frac{(\log n)^{13/8}}{n^{1/4}}.
\end{align}
Since the oracle compound p-values are critically dense at $\alpha$, retracing the final step in the proof of Theorem~\ref{thm:final_rate} shows that for $\kappa_n=5\alpha/4(\alpha-\alpha_0)$,
\begin{align}
\limsup_{n\to\infty}\mathbb P\!\left(R_{n,\jt}(\wh t_\jt)<n\kappa_n\right)=0.
\end{align}
Combining the above bounds yields
\[
\limsup_{n\to\infty}\fdr_n^{\jt}\le\alpha.
\]
This completes the proof.
\end{proof}

\end{document}